\documentclass[a4paper,UKenglish,cleveref,thm-restate]{lipics-v2021}
\hideLIPIcs
\EventLogo{}

\usepackage{stmaryrd}
\usepackage{tcolorbox}
\usepackage{cancel}
\usepackage{complexity}
\usepackage{amssymb}
\usepackage{amsmath}
\mathchardef\mhyphen="2D 
\usepackage{amsthm}
\usepackage{arydshln}
\usepackage{xcolor}
\usepackage{hyperref}
\usepackage{cleveref}

\usepackage[ruled,vlined,linesnumbered]{algorithm2e}
\SetKwInOut{Input}{Input}
\SetKwInOut{Output}{Output}
\SetKwFor{ForAll}{for all}{do}{end}

\usepackage{graphicx}
\usepackage{mathdots}
\usepackage{accents}
\usepackage[normalem]{ulem}
\usepackage{soul}
\usepackage{adjustbox}
\usepackage{xspace}
\usepackage{mathtools}

\usepackage{enumerate}
\usepackage{comment}
\usepackage{bigstrut}
\usepackage{centernot}
\usepackage{tikz}
\usetikzlibrary{positioning,arrows,fit,calc,shapes} 
\usetikzlibrary{decorations.pathreplacing}
\usetikzlibrary{arrows.meta}
\usetikzlibrary{decorations.markings}

\newcommand{\defeq}{\mathrel{\mathop:}=}

\newcommand{\red}[1]{\textcolor{red}{[#1]}}
\newcommand{\emp}{\bot}
\newcommand{\leafprune}{\textsc{LeafPrune}}
\newcommand{\leafeliminate}{\textsc{LeafElim}}
\newcommand{\atomextend}{\textsc{ConsEdgeElim}}
\newcommand{\atomeliminate}{\textsc{InconsEdgeElim}}
\newcommand{\Substitute}{\textsc{Substitute}}

\newcommand{\guard}[3]{%
\overset{\scriptscriptstyle #3}{#1 \Rightarrow #2}%
}

\newcommand{\aesat}{\forall\exists\textsc{3SAT}}
\newcommand{\false}{\mathsf{false}}
\newcommand{\true}{\mathsf{true}}
\newcommand{\Const}{\mathsf{Const}}
\newcommand{\Var}{\mathsf{Var}}
\newcommand{\Rel}{\mathsf{Rel}}
\newcommand{\Pred}{\mathsf{Pred}}

\newcommand{\db}{\mathbf{db}}
\newcommand{\DB}{\mathsf{DB}}

\newcommand{\rep}{\mathbf{r}}

\newcommand{\saturate}{\textsc{Saturate}}
\newcommand{\dbsaturate}{\textsc{DBSaturate}}
\newcommand{\saturatestep}{\textsc{SaturateStep}}
\newcommand{\normalizeone}{\textsc{Remove}}
\newcommand{\removereps}{\textsc{RemoveReps}}
\newcommand{\identify}{\textsc{Identify}}

\newcommand{\normalizetwo}{\textsc{Remove}}

\newcommand{\cqa}[1]{\mathsf{CERTAINTY}({#1})}
\newcommand{\cert}{\mathsf{CERTAINTY}}
\newcommand{\cqpk}{\mathsf{BCQ\text{-}PK\text{-}EVAL}}

\newcommand{\sub}{\subseteq}
\newcommand{\norep}[1]{{#1}^{\sf norep}}

\newcommand{\attacks}[1]{\stackrel{#1}{\rightsquigarrow}}
\newcommand{\attacksp}[1]{\stackrel{{#1}+}{\rightsquigarrow}}

\newcommand{\nattacks}[1]{\stackrel{#1}{\not\rightsquigarrow}}

\newcommand{\outvar}[2]{\mathsf{Out}_{#1}({#2})}
\newcommand{\outvarstar}[2]{\reach{#2}{#1}}

\newcommand{\drop}[2]{{#1}^{-{#2}}}
\newcommand{\reach}[2]{\mathsf{Reach}_{#2}(#1)}

\newcommand{\Saturate}{\textsc{Saturate}}
\newcommand{\Normalize}{\textsc{Normalize}}

\newcommand{\Try}{\textsc{Try}}

\newcommand{\Dom}{\mathrm{Dom}}

\newcommand{\keyvars}[1]{\mathsf{Source}({#1})}
\newcommand{\keyvarsvars}[1]{\mathsf{KeyVars}({#1})}
\newcommand{\notkeyvars}[1]{\mathsf{NonKeyVars}({#1})}

\newcommand{\adom}[1]{\mathsf{adom}({#1})}

\newcommand{\atomvars}[1]{\mathsf{Vars}({#1})}
\newcommand{\queryvars}[1]{\mathsf{Vars}({#1})}
\newcommand{\keyqueryvars}[1]{\mathsf{KeyVars}({#1})}
\newcommand{\nkeyqueryvars}[1]{\mathsf{NotKeyVars}({#1})}
\newcommand{\nrqueryvars}[1]{\mathsf{GuardedVars}({#1})}
\newcommand{\rqueryvars}[1]{\mathsf{NotGuardedVars}({#1})}
\newcommand{\rootqueryvars}[1]{\mathsf{RootVars}^{\attacks{}}({#1})}
\newcommand{\nrootqueryvars}[1]{\mathsf{NotRootVars}^{\attacks{}}({#1})}
\newcommand{\srcqueryvars}[1]{%
\mathsf{SrcVars}^{\to}(#1)
}
\newcommand{\sch}{\mathsf{Rel}}
\newcommand{\initqueryvars}[1]{\mathsf{GuardVars}({#1})}
\newcommand{\leafqueryvars}[1]{\mathsf{LeafVars}^{\attacks{}}({#1})}

\newcommand{\sib}{\sim}

\newcommand{\size}[1]{||{#1}||}

\newcommand{\FD}[1]{{\mathcal{K}}({#1})}
\newcommand{\FDclosure}[1]{{\mathcal{K}^*}({#1})}

\newcommand{\cl}[2]{{#1}^{+,{#2}}}
\newcommand{\SFD}[1]{{\mathcal{K}^{c}({#1})}}
\newcommand{\WFD}[1]{{\mathcal{K}^{i}({#1})}}
\newcommand{\fd}[2]{{#1}\rightarrow{#2}}

\newcommand{\blockof}[3]{{\mathsf{block}}({#1},{#2},{#3})}

\newcommand{\gaifman}[1]{\mathcal{G\mathsf{aifman}}({#1})}

\newcommand{\keyclosure}[2]{{#1}^{+,{#2}}}

\newcommand{\nscc}[1]{\#\mathrm{srcSCC}\!\left({#1}\right)}
\newcommand{\cgs}[1]{\mathsf{cgs}\!\left({#1}\right)}
\newcommand{\hcgs}[1]{\mathsf{hcgs}\!\left({#1}\right)}
\newcommand{\ctw}[1]{\mathsf{ctw}\!\left({#1}\right)}
\newcommand{\hctw}[1]{\mathsf{hctw}\!\left({#1}\right)}

\newcommand{\restrict}[2]{{#1}{\restriction}_{#2}}

\newcommand{\sjfbcq}{\mathsf{sjfBCQ}}
\newcommand{\sjfubcq}{\mathsf{sjfBCQ}_\mathsf{1}}
\newcommand{\sjfkbcq}[1]{\mathsf{sjfBCQ}_{#1}}
\newcommand{\bcq}{\mathsf{BCQ}}

\newcommand{\sjfabcq}{\mathsf{sjfBCQ}_\mathsf{A}}
\newcommand{\sjfuabcq}{\mathsf{sjfBCQ}_\mathsf{1,A}}
\newcommand{\sjfkabcq}[1]{\mathsf{sjfBCQ}_\mathsf{{#1},A}}

\newcommand{\dom}[1]{\mathsf{Dom}(#1)}

\newcommand{\uk}{-}

\newenvironment{redtext}{\color{red}}{\ignorespacesafterend}
\newenvironment{bluetext}{\color{blue}}{\ignorespacesafterend}
\newenvironment{olivetext}{\color{olive}}{\ignorespacesafterend}

\newcommand{\arity}[1]{{\mathsf{type}}(#1)}

\newcommand{\rootvar}[1]{\mathsf{RootOf}^{\attacks{}}(#1)} 

\newcommand{\rootof}[1]{\mathsf{RootOf}^{\attacks{}}(#1)}

\newcommand{\modec}[1]{{#1}^c}
\newcommand{\modei}[1]{{#1}^i}

\newcommand{\igraph}[1]{\mathcal{K}^{\attacks{}}(#1)}

\newcommand{\occq}[1]{\size{#1}} 
\newcommand{\subs}{\mathcal F}
\newcommand{\maxsubs}{\mathcal{M}}
\newcommand{\satmaxsubs}{\mathcal{SM}}
\usetikzlibrary{fit,positioning}
\usetikzlibrary{decorations.pathmorphing}
\usetikzlibrary{arrows.meta}

\tikzset{
  single arrow/.style={->, line width=0.6pt, >={Stealth}},
  double arrow/.style={
    ->,
    line width=0.6pt,
    double,
    double distance=1.2pt,
    >={Stealth}
  }
}
\tikzset{
  squig/.style={
    decorate,
    decoration={snake, amplitude=0.7mm, segment length=3mm}
  }
}

\newcommand{\sarr}[2]{\draw[single arrow] (#1) -- (#2);}
\newcommand{\bdarr}[2]{%
  \draw[double arrow] (#1) -- (#2);
}

\newcommand{\darr}[3][]{%
  \draw[double arrow,#1] (#2) to (#3);
}

\newcommand{\rdarr}[2]{%
  \draw[red, double arrow] (#1) -- (#2);
}

\title{Getting to the Root: A Combined Complexity Perspective on Consistent Query Answering} 

\titlerunning{Combined Complexity of Consistent Query Answering}
\author{Miika Hannula}{University of Tartu, Estonia}{}{}{}

\authorrunning{M. Hannula} 

\nolinenumbers

\makeatletter
\renewcommand{\subjclassHeading}{}
\renewcommand{\keywordsHeading}{}
\gdef\@ccsdescString{\relax}
\gdef\@keywords{\relax}
\makeatother
\begin{document}

\maketitle

\begin{abstract}
We study the combined complexity of consistent query answering for Boolean self-join-free conjunctive queries with unary primary keys and acyclic attack graphs. Although every fixed query in this class admits a first-order rewriting \cite{KoutrisW17}, we show that allowing the query to vary makes the problem $\Pi_2^{\P}$-complete. To isolate the query structure governing this transition, we introduce the \emph{closure generator size}: the minimum number of query variables whose closure under the functional dependencies induced by the primary keys contains every query variable. Our main theorem shows that every fixed bound on this parameter yields polynomial-time combined complexity. Parameterized by the closure generator size, the problem is in $\mathsf{XP}$ and $\mathrm{co}\text{-}\mathrm{W}[t]$-hard for every fixed $t$. 
The polynomial-time result is specific to unary primary keys: with binary primary keys, the problem becomes $\coNP$-hard already for self-join-free queries with acyclic attack graphs and closure generator size zero. 
Regarding the upper bound, CQA is in $\coNP$ for arbitrary Boolean conjunctive queries of bounded closure treewidth or bounded hyperclosure treewidth; these width measures originate in \cite{Adler08}.
Our polynomial-time algorithm realizes quantifier alternation dynamically by interleaving existential and universal substitutions with variable-level reductions. To implement this approach, we introduce concepts and techniques such as attack-propagation graphs, variable guarding and guarded pruning, database and query saturation, and query normalization.
  \end{abstract}

\tableofcontents
\clearpage

\section{Introduction}
Consistent query answering (CQA), introduced in \cite{ArenasBC99}, provides a principled semantics for queries over inconsistent databases. This semantics is often presented via the notions of repairs and {consistent answers}. A \emph{repair} of a database $\db$ is any consistent database $\rep$ that is as close to $\db$ as possible. Under the usual formal definition, $\rep$ is a repair whenever $\rep$ is a consistent database with a $\subseteq$-minimal symmetric difference $\rep \oplus \db$. In the context of primary-key constraints, or functional dependencies more generally, such a repair $\rep$ is always a subset of $\db$. Given a query $q(\vec{x})$, an answer $\vec{a}$ is called a  \emph{consistent answer} if it belongs to the output $q(\rep)$ of the query $q$ on each repair $\rep$ of $\db$. In particular, when $q$ is Boolean, the consistent answer is ``yes'' if the query answer is ``yes'' on every repair.

Since the number of repairs can be exponential in the size of the database, CQA is generally computationally intractable, even for conjunctive queries (CQs). In contrast to ordinary CQ evaluation, this intractability arises already in the data complexity setting, where the query and the constraint set are fixed and only the database is taken as an input. For instance, there are Boolean conjunctive queries consisting of only two atoms (over distinct relation symbols) for which the CQA problem is \coNP-hard. To date, several classifications of the data complexity of CQA have been established, typically along structural properties that depend on  the query and the constraint set. Most pertinent for our study is the trichotomy result by Wijsen and Koutris~\cite{KoutrisW17,KoutrisW21}, which partitions all self-join-free Boolean conjunctive queries (sjfBCQs) into three classes: those whose CQA problem is in $\FO$, $\L$-complete, or $\coNP$-complete. 
 Other similar data complexity classifications, fully or partially complete, concern path queries \cite{KoutrisOW21}, rooted tree queries \cite{KoutrisOW24},  aggregate queries \cite{KhalfiouiW23,KhalfiouiW24,KhalfiouiW26}, two-atom queries with self-joins \cite{PadmanabhaSS24}, queries under primary and foreign keys \cite{HannulaW22}, queries under semiring semantics \cite{KolaitisPVW25}, and counting complexity \cite{CalauttiLPS22a,MaslowskiW13}. Taken together, these works paint a remarkably comprehensive picture of the data complexity of CQA.

In this paper we shift attention from data complexity to combined complexity. In the CQA context, this point of view was given attention already in earlier papers such as \cite{ArenasB10,CaliLR03}.
Later, Arming et al. \cite{ArmingPS16} classified the CQA problem in various input configurations, including combined complexity (where all three input components are allowed to vary). As summarized therein, the combined complexity of CQA is $\Pi_2^{\P}$-complete over Boolean CQs and key constraints, and remains $\Pi_2^{\P}$-complete even when more general EGDs are considered instead of keys.  In addition, the same problem is $\NP$-complete over inclusion dependencies and undecidable over TGDs. Recently, \cite{abs-2412-08324} investigated parameterized complexity of CQA for data and combined complexity. These works notwithstanding—and there may be others that we fail to list here—it is safe to say that combined complexity for CQA has not yet received a similarly extensive structural analysis as data complexity. In particular, while we know that the combined complexity of CQA is in general intractable, it is not well understood what explains this intractability, and when it can be avoided. Attack graphs provide such explanations for data complexity, but no analogous theory exists for combined complexity. 



The following example illustrates the difference between data and combined complexity. The queries presented therein are first-order rewritable, however a naïve evaluation of the
  resulting alternating formulas takes exponential time in the query size. In contrast, a simple variable-level elimination procedure takes only linear time to compute the consistent answer. 

\begin{example}\label{ex:first}
Consider Boolean CQs
\begin{equation}\label{eq:gap}
q_k \defeq \exists x_1 \dots \exists x_{k+1} (R_1(\underline{x_1},x_2) \land R_2(\underline{x_2},x_3) \land \dots \land R_k(\underline{x_k},x_{k+1})),
\end{equation}
where the (underlined) first position in each atom is the primary-key position.
From the theory of attack graphs, we obtain a first-order rewriting to capture the CQA semantics of the queries $q_k$:
\begin{align*}
 \exists x_1 x_2 (R_1(\underline{x_1},x_2) \land &\forall x_2 (R_1(\underline{x_1},x_2) \to \exists x_3 (R_2(\underline{x_2},x_3) \land \\
&\forall x_3 (R_2(\underline{x_2},x_3) \to \exists x_4 (R_3(\underline{x_3},x_4) \land \dots \land \\
& \forall x_k (R_{k-1}(\underline{x_{k-1}},x_{k}) \to \exists x_{k+1}R_{k}(\underline{x_{k}}, x_{k+1}))\dots ))))).
\end{align*}
These alternating formulas can be evaluated na\"{i}vely in $O(n^k)$, where $n$ is the size of the database. This, however, is needlessly expensive. Given a database $\db$, an alternative process repeats the following steps, for each $i=k-1,\dots ,1$:
\begin{itemize}
\item Remove from $\db$ every $R_i$-block (i.e., a maximal set of $R_i$-facts with the same primary-key value) containing a fact $R_i(\underline{a},b)$ such that no fact of the form $R_{i+1}(\underline{b},c)$ exists in $\db$.
\item Remove 
 each fact $R_{i+1}(\underline{b},c)$ from $\db$.
\end{itemize}
In the end, return $\true$ if $\db$ is non-empty, and otherwise return $\false$. The procedure returns $\true$ precisely when each repair of $\db$ satisfies $q_k$. Since, assuming constant-time access to blocks by relation name and primary-key value, every block is examined at most once, the procedure runs in $O(n)$. 
\end{example}

The above example provides a simple combined complexity fragment that is tractable. In this paper, our aim is to understand when and how such fragments arise. The variable $x_1$ is what we call a closure generator of $q_k$: the primary-key dependencies successively determine $x_2,\ldots,x_{k+1}$. Hence, the closure generator size of $q_k$ is one, and our main result implies polynomial-time combined complexity for $\{q_k \mid k \geq 1\}$. 
 The elimination procedure exemplifies one part of what we call {database saturation}. The general polynomial-time algorithm is considerably more involved, as this additionally requires handling multiple incoming edges and cycles.
  
{\bf Contributions.} 
We identify two structural dimensions governing combined complexity of consistent query answering for Boolean self-join-free conjunctive queries with acyclic attack graphs: primary-key arity and closure generator size. We define the latter as the minimum size of a set of query variables whose closure under the functional dependencies induced by the query’s primary keys contains every variable of the query. Our main theorem shows that unary primary keys together with a fixed bound on the closure generator size yield polynomial-time combined complexity. Relaxing either restriction leads to a hardness result. With unary primary keys but no bound on the closure generator size, the problem is $\Pi_2^{\P}$-complete. Allowing binary primary keys makes the problem $\coNP$-complete even when the closure generator size is zero. The $\coNP$ upper bound holds even for general Boolean CQs under primary keys with bounded closure treewidth or hyperclosure treewidth, two notions originating in \cite{Adler08}. Additionally, 
treating the closure generator size as a parameter, in the case of unary keys and acyclic attack graphs we  prove
   membership in $\mathsf{XP}$, together with $\mathrm{co}\text{-}\mathrm{W}[t]$-hardness for every fixed $t$.


The polynomial-time result is our main technical contribution.
To achieve it, we develop attack-propagation graphs,
variable guarding and guarded reductions, database and query saturation, and
query normalization. Conceptually, the resulting algorithm moves beyond
static first-order rewritings and realizes quantifier alternation dynamically
through interleaved variable-level reductions.

{\bf Related work.} 
The attack-graph-based recursive rewriting construction, introduced in~\cite{Wijsen12} and later extended to the general self-join-free setting in~\cite{KoutrisW17}, provides the background for this work.
Perhaps closest in spirit to our approach is \cite{FanKOW23}, which identifies a class of sjfBCQs for which CQA can be solved via FO-rewriting in linear time. The class they introduce (queries admitting a pair-pruning join tree) is orthogonal to ours: it allows composite primary keys but requires the query to be ($\alpha$-)acyclic. Consequently, this class is a strict subclass of the sjfBCQs with an acyclic attack graph. For example, the query
\begin{equation}\label{eq:cycle}
q= \exists xyz (R(\underline{x},y) \land S(\underline{y},z) \land T(\underline{z},x) \land T'(\underline{z},x) \land U(\underline{y},x) \land U'(\underline{y},x))
\end{equation}
has unary primary keys, an acyclic attack graph, but is not acyclic. This illustrates also that tractability in our framework does not rely on hypergraph acyclicity.
Another work that studies the algorithmic aspects of CQA is \cite{FigueiraPSS25}, which presents a remarkably simple generic algorithm for all sjfBCQs (and path queries) whose CQA problem has polynomial-time data complexity. The algorithm runs in $O(n^k)$, where $n$ is the database size and $k$ is the number of atoms in the query. 

There is also a long line of work on the combined complexity of ordinary CQ
evaluation. In particular, tractability has been analyzed and characterized
via structural measures such as treewidth, hypertree width, generalized
hypertree width, and submodular width
\cite{GroheSS01,GottlobLS02,Marx13}. Our closure and hyperclosure treewidth
measures build on Adler's dependency-based framework \cite{Adler08}, and
  our generator sizes can be viewed as the one-bag versions of the width measures. This relates our results to the classical structural analysis of CQ evaluation.

\section{Preliminaries}
We use the convention of dropping the set parentheses for singleton sets in notation. That is, if our term involves a set argument  $U=\{u\}$, we write $u$ instead of $\{u\}$.
For two natural numbers $k\leq n$, we define $[k,n]\defeq \{k,k+1, \dots ,n\}$ and $[n]\defeq [1,n]$.
We assume disjoint countably infinite sets of variables $\Var$, constants $\Const$, and relation names $\Rel$. A \emph{term} is a variable or a constant.  We assume a fixed total order $\prec$ over $\Var\cup\Const\cup\Rel$, and extend it lexicographically to all composite expressions defined below. 


Each relation name $R$ is associated with primary-key arity $k$ and relation arity $n$, denoted $\arity{R}=(k,n)$, where $k\leq n$. 
An \emph{$R$-atom} (or simply an \emph{atom}) is an expression of the form $F\defeq R(\underline{t_1, \dots ,t_k},t_{k+1}, \dots ,t_n)$, where $t_1,\dots ,t_n$ are terms and $\arity{R}=(k,n)$. If $t_1,\dots ,t_n$ are constants, then this expression is an \emph{$R$-fact} (or a \emph{fact}). 
 Two facts $R(\underline{\vec{a}},\vec{c})$ and $R(\underline{\vec{b}},\vec{d})$ are called \emph{key-equal} if $\vec{a}=\vec{b}$. Throughout the proofs, we tacitly identify atoms that differ only by a permutation of their non-key terms whenever their order is irrelevant. Accordingly, we may denote atoms by writing
\(
R(\underline{t},\vec{u},\uk)
\)
where $\vec{u}$ is a distinguished sequence of non-key terms, and ``$\uk$'' stands for the non-distinguished non-key terms.


A \emph{Boolean conjunctive query (with primary-key constraints)}  $q$ is defined as a finite set of atoms $\{F_1, \dots ,F_n\}$.
The query is \emph{self-join-free} if $F_i$ and $F_j$ are associated with different relation names whenever $i\neq j$. We write $\sch(q)$ for the set of relation names appearing in $q$. 
The classes of Boolean conjunctive queries and self-join-free Boolean conjunctive queries are denoted $\mathsf{BCQ}$ and $\sjfbcq$, respectively. For a query $q$ from $\mathsf{BCQ}$, we write $\queryvars{q}$ for the set of variables appearing in it, 
$\keyvarsvars{q}$ for the set of variables appearing in a primary-key position of some atom in $q$, and $\notkeyvars{q}$
for the set of variables appearing in a non-primary-key position of some atom in $q$. Note that $\keyvarsvars{q}$ and $\notkeyvars{q}$ need not be disjoint. 
For  $V\subseteq \queryvars{q}$ and a function $f\colon V\to \Const \cup \Var$, the query $q_f$ is obtained from $q$ by replacing each variable $x\in V$ with $f(x)$. 
We let $\mathsf{BCQ}_{k}$ and $\sjfkbcq{k}$ denote the classes of queries in $\mathsf{BCQ}$ and  $\sjfbcq$
such that $\arity{R}=(\ell,n)$ with $\ell\leq k$ for each relation name $R$ appearing in $q$.



\subsection{Databases and Repairs}
Let $V$ be a set of variables.
A \emph{valuation (over $V$)} is a function $\theta$ from $V$ to the constants.
A valuation is extended to be the identity over constants, and it extends to atoms and queries of $\mathsf{BCQ}$ naturally.
If $q\in \sjfbcq$ and $R\in \sch(q)$, by slight abuse of notation we sometimes write $R$ for the unique $R$-atom in $q$.
A \emph{database} $\db$ is a finite set of facts, and the set of all databases is denoted by $\DB$.
The \emph{active domain of $\db$}, denoted by $\adom{\db}$, is the set of all constants appearing in $\db$.
Given a query $q$ from $\mathsf{BCQ}$, we say that $\db$ \emph{satisfies}
$q$ if there exists a valuation over $\queryvars{q}$ such that $\theta(q)\subseteq \db$.

Let $\db$ be a database.
A \emph{block} of $\db$ is any $\subseteq$-maximal set of key-equal facts from $\db$.
For an atom $R(\underline{\vec{x}},\vec{y})$, we write $\blockof{R}{\vec{a}}{\db}$ for the set of facts in $\db$ of the form $R(\underline{\vec{a}},\uk)$.
The expression \(R(\underline{\vec{a}},\vec{b},\uk)\in\db\) means that there exists a tuple of constants \(\vec{c}\) such that \(R(\underline{\vec{a}},\vec{b},\vec{c})\in\db\).
The database $\db$ is \emph{consistent} if for all $R(\underline{\vec{a}},\vec{b}),R(\underline{\vec{a}},\vec{c})\in \db$ it holds that $\vec{b}=\vec{c}$. A \emph{repair} of $\db$ is any $\subseteq$-maximal consistent subset of $\db$.

Let $\mathcal{R}$ be a set of relation names. We write $\restrict{\db}{\mathcal{R}}$ for the \emph{restriction} of $\db$ to $\mathcal{R}$. 
 If $q$ is a query from $\mathsf{BCQ}$, then we may write $\restrict{\db}{q}$  instead of  $\restrict{\db}{\sch(q)}$.
If $R$ is a  relation name, we may write $R^\db$ (or, by slight abuse of notation, simply $R$ when $\db$ is understood) instead of $\restrict{\db}{R}$ to denote the set of all $R$-facts in $\db$.

\subsection{Graphs and Closures}
Let $G=(V,E)$ be a directed graph, and let $U\subseteq V$ be a vertex set.
We write $\reach{U}{G}$ for the set of vertices \emph{reachable from} $U$ in $G$, assuming reachability is reflexive.
We also write 
 $\outvar{G}{U}$ for the set of vertices having an incoming edge from a vertex in $U$.
Let $P=x_1,\dots ,x_n$ be a simple path in $G$. For two vertices $x_i$ and $x_j$ with $1\leq i<j\leq n$, we write $P[x_i,x_j]$ for its unique subpath that starts at $x_i$ and ends in $x_j$. 
We say that $P$ is \emph{separated by $U$} if $x_i\in U$, for some $i\in [n]$. Two sets $U_0,U_1 \subseteq V$ are \emph{separated by $U$} if every non-empty path with one endpoint in $U_0$ and another in $U_1$ is separated by $U$. 
A \emph{strongly connected component} (SCC) of $G$ is a maximal strongly connected subgraph of $G$.
An SCC $C\subseteq V$ is a \emph{source SCC} if there is no edge $(u,v)\in E$ with $u\notin C$ and $v\in C$.
We write $\nscc{G}$ for the number of source SCCs of $G$.

{\bf Gaifman graphs.}
Let $q$ be a query from $\mathsf{BCQ}$.
The \emph{Gaifman graph of $q$}, denoted $\gaifman{q}$, is an undirected simple graph 
whose vertex set is $\queryvars{q}$. There is an edge between $x$ and $y$ 
if $x \neq y$ and some atom of $q$ contains both $x$ and $y$. 

{\bf Attack graphs.}
Let $q$ be a query from $\sjfbcq$. 
Given a set of variables $V$, the \emph{closure of $V$ under $q$}, denoted $\cl{V}{q}$, is the smallest set of variables $W$ 
 such that $V\subseteq W$ and, for every atom $F\in q$, if $\keyvarsvars{F}\subseteq W$, then $\atomvars{F}\subseteq W$.
Define  
$\keyclosure{F}{q}\defeq \cl{\keyvarsvars{F}}{q\setminus \{F\}}$, for each $F\in q$.
An atom $F$ of $q$ is said to \emph{attack} a variable $x$ (in $q$), denoted 
$F \attacks{q}x$, if $\notkeyvars{F}$ and $\{x\}$ are not separated by $\keyclosure{F}{q}$ in $\gaifman{q}$.
A variable $x$ is  \emph{unattacked} if no atom attacks $x$.  
The \emph{attack graph of $q$} is a directed simple graph whose vertices 
are the atoms of $q$. Given two distinct atoms $F,G\in q$, there is a directed edge from $F$ to $G$, denoted 
$F \attacks{q} G$, if $F$ attacks some variable of $\queryvars{G}$. We write  $F\attacksp{q} G$ if there is a non-empty sequence of attacks leading from $F$ to $G$. When $q$ is understood, we may abbreviate $\attacks{q}$ and $\attacksp{q}$ by $\attacks{}$ and $\attacksp{}$. 

{\bf Closure generators.}
The set $V$ is called a \emph{closure generator of $q$} if $\cl{V}{q} = \queryvars{q}$. A subquery $q'\subseteq q$ is called a \emph{hyperclosure generator of $q$} if $\queryvars{q'}$ is a closure generator.
The \emph{closure generator size of $q$} (resp. the \emph{hyperclosure generator size of $q$}), denoted  $\cgs{q}$  (resp. $\hcgs{q}$), is the minimum cardinality of a closure generator of $q$ (resp. hyperclosure generator of $q$).
For every query $q$ from $\mathsf{BCQ}_{1}$, it holds that $\cgs{q}=\hcgs{q}$.

{\bf Closure treewidths.}
A \emph{tree decomposition} of a query $q$ is a pair
$\mathcal{T}=(T,(B_t)_{t\in T})$, where $T$ is a tree and each bag
$B_t$ is a subset of $\queryvars{q}$, such that every atom $F\in q$
satisfies $\atomvars{F}\subseteq B_t$ for some $t\in T$, and, for every
$x\in\queryvars{q}$, the nodes $t$ satisfying $x\in B_t$ induce a
connected subtree of $T$.
The \emph{closure width} of a bag $B\subseteq\queryvars{q}$ is the
minimum cardinality of a set $V\subseteq\queryvars{q}$ such that
$B\subseteq\cl{V}{q}$. Its \emph{hyperclosure width} is the minimum
cardinality of a subquery $q'\subseteq q$ such that
$B\subseteq\cl{\queryvars{q'}}{q}$. The \emph{closure treewidth}
$\ctw{q}$ (respectively, \emph{hyperclosure treewidth} $\hctw{q}$) is
the minimum, over all tree decompositions of $q$, of the maximum closure
width (respectively, hyperclosure width) of a bag.
Unlike Adler's definition \cite{Adler08}, ours also uses composite-key dependencies.



\section{Main Results}\label{sect:mainstatement}
Fix a subclass $\mathcal{Q}$ of $\bcq$. We define $\cqa{\mathcal{Q}}$ as the following problem: given a database $\db$ and a query from $\mathcal{Q}$, determine whether every repair of $\db$ satisfies $q$.
Koutris and Wijsen  have shown that, for $q\in \sjfbcq$,  $\cqa{\{q\}}$ is in $\FO$ precisely when the attack graph of $q$ is acyclic; if the attack graph is cyclic, $\cqa{\{q\}}$ is either $\L$-complete (if no ``strong cycles'' exist) or $\coNP$-complete (otherwise) \cite{KoutrisW17,KoutrisW21}.

We write $\sjfabcq$ for the subclass of $\sjfbcq$
consisting of queries whose attack graph is acyclic.
Furthermore, let $\sjfkabcq{k} \defeq \sjfkbcq{k} \cap \sjfabcq$; that is, $\sjfkabcq{k}$ consists of Boolean self-join-free conjunctive queries with at most $k$-ary primary keys and acyclic attack graphs.

We prove that the combined complexity of consistent query answering over queries in $\sjfkabcq{1}$ is $\Pi^{\sf P}_2$-complete. 

\begin{restatable}{theorem}{hard}\label{thm:hard}
Let $\sjfkabcq{1} \subseteq\mathcal{Q}\subseteq \bcq$. Then
$\cqa{\mathcal{Q}}$ is $\Pi^{\sf P}_2$-complete.
\end{restatable}
\begin{proof}[Proof sketch]
For hardness, we reduce from $\aesat$. Let
\[
\psi = \forall x_1\dots \forall x_m\exists y_{1}\dots\exists y_{n}\varphi,
\]
where $\varphi = C_1 \land \dots \land C_\ell$ is a 3-CNF formula over the variables
$x_1,\ldots,x_{m},y_1,\dots ,y_n$. For a clause $C_i$, let $C_{i,j}$ denote its $j$th literal, and for a literal $l \in \{x, \lnot x\}$, let $v(l)=x$ denote the underlying variable.
Given $(b_1,b_2,b_3)\in\{0,1\}^3$, we write
$(b_1,b_2,b_3)\models C_i$
if assigning $b_j$ to the variable $v(C_{i,j})$ satisfies $C_i$.
We define
\begin{equation}\label{eq:reduction-sketch}
q\defeq \{R_i(\underline{c},x_i) \mid i\in [m]\} \cup \{S(\underline{z_i},v(C_{i,1}),v(C_{i,2}),v(C_{i,3})) \mid S\in \{T_i,U_{i}\}, i\in [\ell]\}
\end{equation}
and
\begin{align*}
\db \defeq &\{R_i(\underline{c},b) \mid i\in [m], b \in \{0,1\}\} \\
&\cup \{S(\underline{(b_1,b_2,b_3)},b_1,b_2,b_3) \mid S\in \{T_i,U_{i}\}, i\in [\ell], (b_1,b_2,b_3)\models C_i\},
\end{align*}
where $c$ is a constant and the triples in the primary-key position of $S$ are single constants.
Only the $R_i$-atoms can attack other atoms, and they do not attack each other; hence the attack graph is acyclic. 
It can be verified that
$\psi$  is true
if and only if $(\db,q)\in\cqa{\sjfuabcq}$.
\end{proof}

As our main technical result, \Cref{thm:main} shows that combined complexity over queries in $\sjfkabcq{1}$ is in polynomial time once we assume a bound on the closure generator size. When treating the closure generator size as a parameter, we obtain
$\mathrm{co}\text{-}\mathrm{W}[t]$-hardness for every fixed $t$; the proof adapts the above reduction by adding gadgets for quantifier alternation.
\begin{theorem}[Main Theorem]\label{thm:main}
Let $c\geq 0$ and $\kappa\in\{\mathsf{cgs},\mathsf{hcgs}\}$.
Let $\mathcal{Q}\subseteq\sjfuabcq$
  be such that $\kappa(q)\leq c$ for every $q\in\mathcal{Q}$.
  Then $\cqa{\mathcal{Q}}$ is in polynomial time.
\end{theorem}
\begin{restatable}{theorem}{paramresult}\label{thm:param}
$\cqa{\sjfuabcq}$ parameterized by $\cgs{q}$ is in $\mathsf{XP}$.
For every fixed $t\geq 1$, it is $\mathrm{co}\text{-}\mathrm{W}[t]$-hard under {\rm FPT} many-one reductions.
\end{restatable}
In \Cref{thm:param}, hardness for every level of the $\mathrm{co}\text{-}\mathrm{W}$
hierarchy suggests that the query-dependent quantifier alternation in the first-order rewritings
  is a feature of the problem itself, not a by-product of one particular rewriting method.
In \Cref{thm:main}, the restriction to unary primary keys  is necessary (unless $\P =\NP$); already with binary keys the same problem is \coNP-complete, even when $\cgs{q}=0$ for each query $q$.
For the lower bound, 
 the proof first constructs a circuit $C$ to compute an instance $\varphi(x_1, \dots ,x_n)$ of \textsc{3Sat}. Then, it builds the query $q$ from atoms $R_i(\underline{c},x_i)$ to encode universal quantification and binary-key atoms $S_g(\underline{a,b},g),S'_g(\underline{a,b},g)$ to encode binary gates $g$ with inputs $a,b$;  NOT gates are handled analogously with unary keys.  Finally, the atom $O(\underline{o})$ and the sole fact
  $O(\underline{0})$ force the output gate $o$ to take value $0$.
The upper bound holds even for arbitrary $\bcq$ queries with bounded closure treewidth; it adapts the dependency-based dynamic programming approach of \cite{Adler08}.
\begin{restatable}{theorem}{binaryhardness}\label{thm:binary-hardness}
Let $c\geq 0$ and $\kappa\in\{\mathsf{cgs},\mathsf{hcgs},\mathsf{ctw},\mathsf{hctw}\}$.
Let $\mathcal{Q}\subseteq\bcq$ satisfy $\kappa(q)\leq c$
for every $q\in\mathcal{Q}$. Then $\cqa{\mathcal{Q}}$ is in $\coNP$.
Moreover, the problem is $\coNP$-complete whenever
$\{q\in\sjfkabcq{2}\mid\cgs{q}=0\}\subseteq\mathcal{Q}$.
\end{restatable}

\begin{table}[t]
\caption{Overview of our combined-complexity results; $c\geq 0$ is
a fixed constant and $\kappa$ a fixed parameter function. We write $\sjfkabcq{2}^*\defeq \{q\in\sjfkabcq{2} \mid\cgs{q}=0\}$.}
\label{tab:results}
\centering
\small
\renewcommand{\arraystretch}{1.15}
\begin{tabularx}{\textwidth}{@{}p{0.34\textwidth}p{0.27\textwidth}X@{}}
\hline
\textbf{Query class $\mathcal{Q}$} & \textbf{Restriction on $\mathcal{Q}$} & \textbf{Complexity}\\
\hline
$\sjfkabcq{1}\subseteq \mathcal{Q} \subseteq \bcq$ & none
  & $\Pi_2^{\P}$-complete (\Cref{thm:hard})\\
$\sjfkabcq{2}^* \subseteq \mathcal{Q}\subseteq\bcq$ 
 & $\kappa(q)\leq c$\newline
  ($\kappa\in\{\mathsf{cgs},\mathsf{hcgs},\allowbreak\mathsf{ctw},\mathsf{hctw}\}$)
  & $\coNP$-complete (\Cref{thm:binary-hardness})\\
$\mathcal{Q} \subseteq \sjfkabcq{1}$ & $ \kappa(q)\leq c$ ($\kappa\in \{\mathsf{cgs},\mathsf{hcgs}\}$)
  & polynomial time (\Cref{thm:main})\\
$\mathcal{Q} =\sjfkabcq{1}$ & parameter $\cgs{q}$
  & in $\mathsf{XP}$ and $\mathrm{co}\text{-}\mathrm{W}[t]$-hard for every
    $t$ (\Cref{thm:param})\\
\hline
\end{tabularx}
\end{table}

\Cref{thm:hard,thm:param} are proven in \Cref{sect:apphard}, and
\Cref{thm:binary-hardness} is proven in \Cref{sect:apphard2}.
The remainder of the paper develops the argument underlying \Cref{thm:main}; the main ideas are presented in the body, while the technical proofs are deferred to the appendix. 
Informally, the algorithm behind \Cref{thm:main} implements quantifier alternation without materializing the underlying first-order rewriting. It interleaves existential substitution, guarded pruning, and universal substitution. Members of the closure generator act as guards, bounding the size of the intermediate instances as well as the number of computation branches.

We now proceed to present the ideas behind the proof of \Cref{thm:main}. 


  \section{Preprocessing}\label{sect:preprocessing}
We begin with the preprocessing step, which is divided into database saturation, query saturation, and query normalization, in this order.
We let $q$ be a query from $\sjfubcq$ (unless otherwise specified). Henceforward, we make use of the following graph representation of the query $q$.

  \paragraph*{Key-nonkey graph.}
 Fix a special symbol $\emp\notin\Var$. For an atom $F=R(\underline{t_1},t_2,\ldots,t_n)$, let
\[
\keyvars{F}\defeq
\begin{cases}
t_1 & \text{if $t_1\in\Var$,}\\
\emp & \text{if $t_1\in\Const$,}
\end{cases}
\qquad
\notkeyvars{F}\defeq\{t_i\mid i\in[2,n],\ t_i\in\Var\},
\]
and let $\atomvars{F}\defeq(\keyvars{F}\cup\notkeyvars{F})\setminus\{\emp\}$.
The \emph{key-nonkey graph of $q$}, denoted $\FD{q}$, has vertex set $\queryvars{q}\cup\{\emp\}$ and edge set
\begin{equation}\label{eq:FDdef}
\FD{q}\defeq\{\fd{\keyvars{F}}{y}\mid F\in q,\ y\in\notkeyvars{F}\}.
\end{equation}
We say that an edge $\fd{x}{y}\in\FD{F}$ is \emph{generated by} $F$. Since the head of every edge is a query variable, $\emp$ has indegree zero. Note that every minimum closure generator contains
  exactly one variable from each source SCC of $\FD{q}$ other than
  $\{\emp\}$.
In particular, we have that $\cgs{q}=\nscc{\FD{q}}-1$.

A crucial construct in the proof is a partitioning of the edges of $\FD{q}$ into the set of \emph{consistent edges} $\SFD{q}$ and the set of \emph{inconsistent edges}  $\WFD{q}$.  Additionally, we enrich $\SFD{q}$ with  dummy edges.
Formally, let
\begin{equation}\label{eq:transclosure}
\FDclosure{q}\defeq\{\fd{x}{y}\mid x\in\queryvars{q}\cup\{\emp\},\ y\in\cl{(\{x\}\setminus\{\emp\})}{q}\},
\end{equation}
where, by definition, 
 $\cl{V}{q} =\reach{V\cup\{\emp\}}{\FD{q}}\setminus\{\emp\}$, for $V\subseteq\queryvars{q}$.
Then, 
\begin{itemize}
\item  $\SFD{q}$ consists of the edges $\fd{x}{y}$ where $x=y$ or there exists
$F\in q$ and $G\in q\setminus\{F\}$ such that
$\fd{x}{y}\in \FD{F}$ and $\fd{x}{y}\in\FDclosure{G}$;
\item $\WFD{q}$ consists of the edges $\fd{x}{y}$ that belong to $\FD{q}\setminus \SFD{q}$.
\end{itemize}
In particular, every edge from $\FD{q}$ belongs either to $\WFD{q}$ or $\SFD{q}$, but not both.

\subsection{Database Saturation}\label{sect:dbsat}
Consider the following local consistency criteria for a database $\db$.
\begin{itemize}
\item \textbf{Non-emptiness.}
A database $\db$ is \emph{non-empty (over $q$)} if $\restrict{\db}{R}\neq \emptyset$ for each $R\in q$.

\item \textbf{Functional consistency.}
A database $\db$ is \emph{functionally consistent (over $q$)}
if for all $R(\underline{x},y,\uk)\in q$ such that $\fd{x}{y}\in \SFD{q}$
 and all $a,b,c\in \adom{\db}$, if $R(\underline{a},b,\uk), R(\underline{a},c,\uk) \in \db$, then  $ b=c$.

\item \textbf{Pairwise consistency.} 
A database $\db$ is \emph{pairwise consistent (over $q$)} if for any (not necessarily distinct) atoms $R,S\in q$ and any $R$-fact $A\in \db$, there is a valuation $\theta$ over $\atomvars{R} \cup\atomvars{S}$ such that $\theta(R)=A\in \db$
and  $\theta(S)\in \db$.

\item \textbf{Collision consistency.}
A database $\db$ is \emph{collision consistent (over $q$)} if 
for any $R(\underline{x_1},x_2, \dots ,x_\ell)$, $\fd{x_i}{z},\fd{x_j}{z}\in\SFD{q}$, and $T(\underline{x_j},z,\uk)$, where $i,j\in [\ell]$,
\begin{itemize}
\item if $x_i=z$ and $R(\underline{a_1},a_2,\dots ,a_\ell)\in \db$, then $T(\underline{a_j},a_i,\uk)\in  \db$;
\item if $S(\underline{x_i},z,\uk)\in q$ and $R(\underline{a_1},a_2,\dots ,a_\ell),S(\underline{a_i},b,\uk)\in \db$, then $T(\underline{a_j},b,\uk)\in  \db$. 
\end{itemize}

\item 
\textbf{Guardedness.}
A database $\db$ is \emph{guarded (over $q$)} if for any $R(\underline{x_1},x_2, \dots ,x_\ell)\in q$ and $S_1(\underline{z},x_{i_1},\uk),\dots ,S_n(\underline{z},x_{i_n},\uk)\in q$ such that $\fd{z}{x_{i_1}},\dots ,\fd{z}{x_{i_n}}\in\SFD{q}$, where $i_j\in [\ell]$ for $j\in [n]$, if $R(\underline{a_1},a_2, \dots ,a_\ell)\in \db$, then there is $b\in \adom{\db}$ such that $S_1(\underline{b},a_{i_1},\uk),\dots ,$ $S_n(\underline{b},a_{i_n},\uk)\in \db$.
\end{itemize}
Note that pairwise consistency implies that 
 for every atom $R \in q$ and every $R$-fact $A\in \db$, there exists
a valuation $\theta$  over $\atomvars{R}$  such that $\theta(R)=A$.

A database $\db$ that is non-empty, functionally consistent, pairwise consistent, collision consistent, and guarded over $q$ is called \emph{$q$-saturated} (or simply \emph{saturated}  when $q$ is understood).

These consistency criteria can be enforced by a general algorithm, called $\dbsaturate$ (see \Cref{alg:dbsaturate} in the appendix), that turns a database into a saturated one via repetitive removal of blocks.
To understand how this can help in finding the consistent answer, consider the following example, which expands
\eqref{eq:gap} with the kind of cyclicity exhibited in \eqref{eq:cycle}.
\smallskip
\begin{example}\label{ex:prune}
Consider queries
\[
q'_k \defeq q_k \,\cup \{S_i(\underline{x_i},x_1), S'_i(\underline{x_i},x_1)  \mid i\in [2,k+1]\},
\]
where $q_k$ is the query defined in \eqref{eq:gap}. It can be verified that the attack graph of $q'_k$ is acyclic. Suppose $\db$ is a $q'_k$-saturated database. Then, using functional consistency and collision consistency, we obtain that every repair of $\db$ satisfies $q'_k$ if and only if every repair of $\db$ satisfies $q'_k\setminus \{S_{k+1}(\underline{x_{k+1}},x_1), S'_{k+1}(\underline{x_{k+1}},x_1)\}$. By induction, it follows that every repair of $\db$ satisfies $q'_k$ if and only if every repair of $\db$ satisfies $q_k$. Hence, for saturated databases, the queries $q'_k$ are equally hard to the queries $q_k$, which were dealt in \Cref{ex:first}.
\end{example}

 We say that a many-one reduction $f\colon \DB\times \sjfubcq \to \DB \cup\{\true,\false\}$ is \emph{equivalence-preserving} if for every $(\db,q)$, $(\db,q)\in \cert$ whenever $f(\db,q)=\true$, $(\db,q)\notin \cert$ whenever $f(\db,q)=\false$, and otherwise $(\db,q)\in \cert$ if and only if $(f(\db,q),q)\in \cert$. 
The proof of the following lemma is deferred to \Cref{sect:app:dbsat}.

\smallskip
\begin{restatable}[Database Saturation Lemma]{lemma}{dbsat}\label{lem:dbsat}
\(
\dbsaturate \colon \DB\times \sjfubcq \to \DB \cup\{\true,\false\}
\)
is an equivalence-preserving polynomial-time reduction
such that for every pair $(\db,q)$, assuming $\db'= \dbsaturate(\db,q)\notin\{\true,\false\}$, we have that
\begin{enumerate}
\item $\db'$ is $q$-saturated,
\item $\db'\subseteq \db$.
\end{enumerate}
\end{restatable}

Before continuing to the next step, we make one technical assumption explicit.
  
  \paragraph*{Removing variable repetitions.}
We say that  $q$ is \emph{variable-repetition-free} if it contains only atoms 
$R(\underline{t_1},t_{2}, \dots ,t_n)$ such that the sequence $t_1, \dots ,t_n$ does not repeat any variables.
 Unless explicitly stated otherwise, every query considered from this point onward is assumed to be variable-repetition-free. This does not restrict the input class: in the proof of the main theorem, after applying database saturation, we apply a preprocessing step (\Cref{lem:repfreereal}) to remove variable repetitions.  Once $q$ is variable-repetition-free, $\FD{q}$ lacks self-loops.



\subsection{Query Saturation}\label{sect:sat}
Next, we move to query saturation. 
Informally, a query is saturated if it maximizes the number of consistent edges between two variables.

\begin{definition}[Saturatedness]\label{def:saturation}
    Let $q\in \sjfubcq$. We say that $q$ is saturated
 if for every atom \( F \in q \),  if
\begin{enumerate}
    \item\label{it:sat1} $\fd{x}{y}\in \FD{F}$,
    \item\label{it:sat2}  $\fd{y}{z}\in \SFD{q}$, and
    \item\label{it:sat3} $\fd{x}{z}\in \FDclosure{q\setminus \{F\}}$, 
\end{enumerate}
then $\fd{x}{z}\in \SFD{q}$.
\end{definition}
  
  Let us briefly explain the intuition behind query saturation. 
    At this point the database $\db$ is $q$-saturated.
  If $\fd{y}{z}\in\SFD{q}$ and $y\neq z$, then $q$ contains an atom
  $S(\underline{y},z,\uk)$, and functional consistency of $\db$ ensures
  that the value of $z$ in $S^\db$ is uniquely determined by the value of
  $y$. The case $y=z$ is analogous.
  Assume that $F=R(\underline{x},y,\uk)$ and that
  $\blockof{R}{a}{\db}$ contains facts
  $R(\underline{a},b,\uk)$ and $R(\underline{a},b',\uk)$ such that
  $S(\underline{b},c,\uk),S(\underline{b'},c',\uk)\in\db$ for distinct
  values $c$ and $c'$. Then, $\blockof{R}{a}{\db}$ can be safely removed from $\db$. Indeed, let
  $\rep$ be a repair of $\db\setminus\blockof{R}{a}{\db}$ falsifying
  $q$. If the two extensions of $\rep$ by $R(\underline{a},b,\uk)$ and $R(\underline{a},b',\uk)$ both satisfied
  $q$, the corresponding valuations would agree on $z$
  by $\fd{x}{z}\in\FDclosure{q\setminus\{F\}}$, while mapping $z$ to the
  distinct values $c$ and $c'$, a contradiction. Hence, $\rep$ can be extended 
  by one of these facts without making $q$ true.

  After exhaustively applying such removals, the query can be extended with fresh atoms,
  and the database with corresponding facts; this way, the missing consistent edges
  $\fd{x}{z}$ are added. The following lemma turns this informal description into a
  precise saturation procedure. An important feature of this procedure is that the size of the problem instance grows only polynomially, as the number of consistent relations added is at most quadratic.

Consider a many-one reduction $f\colon \DB\times \sjfuabcq \to \DB\times \sjfuabcq \cup\{\true,\false\}$. We say that $f$ is
\begin{itemize}
\item \emph{generator-monotone} if for every $(\db,q)$ such that $f(\db,q)\notin \{\true,\false\}$, letting $(\db',q') = f(\db,q)$, we have $\cgs{q'}\leq \cgs{q}$; and
\item \emph{equivalence-preserving} if for every $(\db,q)$, $(\db,q)\in \cert$ whenever $f(\db,q)=\true$, $(\db,q)\notin \cert$ whenever $f(\db,q)=\false$, and otherwise $(\db,q)\in \cert$ if and only if $f(\db,q)\in \cert$.
\end{itemize}
Let us also define an auxiliary technical condition: a query $q\in \sjfuabcq$ is \emph{binary-saturated} if $\fd{x}{y}\in \SFD{q}$ and $x\neq y$ imply
that $q$ contains two atoms $F,G$ such that $\arity{F}=\arity{G}=(1,2)$ and $\fd{x}{y}\in \FD{F}\cap\FD{G}$.
The proof of the following lemma is given in \Cref{sect:appsat}.

\smallskip

\begin{restatable}[Query Saturation Lemma]{lemma}{csat}\label{lem:c-sat}
Let $\mathcal{C}$ be the set of pairs $(\db,q)$, where $q$ is a query from $\sjfubcq$ and $\db$ is a $q$-saturated database.
Then $\Saturate$ is an equivalence-preserving, generator-monotone, polynomial-time reduction from $\mathcal{C}$ to $\mathcal{C}\cup\{\true,\false\}$.
Additionally, for each $(\db,q)$, assuming $(\db',q') = \Saturate(\db,q)\notin \{\true,\false\}$, we have that $q'$ is saturated and binary-saturated,
and $q'$ has acyclic attack-graph whenever $q$ has.
\end{restatable}

The following example demonstrates the usefulness of query saturation.
\begin{example}\label{ex:isnormal}
Consider an arbitrary database $\db$ and a query
\begin{align*}
q= \{&R_0(\underline{x},y),R_1(\underline{y},z),R_2(\underline{z},v),R'_0(\underline{x},y'),R'_1(\underline{y'},z'),R'_2(\underline{z'},v), \\
& S_0(\underline{y},x),S_1(\underline{z},x),S'_0(\underline{y'},x),S'_1(\underline{z'},x)  ,\\
& T_0(\underline{y},x),T_1(\underline{z},x),T'_0(\underline{y'},x),T'_1(\underline{z'},x)  \},
\end{align*}
illustrated by the leftmost graph in \Cref{fig:satprune}.
It can be observed that $q$ is a query from $\sjfuabcq$. However, it is not saturated. For instance, $\fd{z}{v}\in \FD{R_2}$, $\fd{v}{v}\in \SFD{q}$, and $\fd{z}{v}\in  \FDclosure{q\setminus \{R_2\}}$, while $\fd{z}{v}\notin \SFD{q}$. Thus, query saturation adds $\fd{z}{v}$  to $\SFD{q}$ via an introduction of fresh atoms $U(\underline{z},v),U'(\underline{z},v)$. Analogously, it adds $\fd{y}{v}$  to $\SFD{q}$ via $V(\underline{y},v),V'(\underline{y},v)$, and so on. After $q$ has been saturated, the database $\db$
is likewise saturated. Then, using functional consistency and collision consistency as in \Cref{ex:prune}, one can ``prune'' the query
in order to remove all the consistent edges from the key-nonkey graph in a top-down manner, leaving only the atoms $R_0,R_1,R'_0,R'_1$ in the query. 
The resulting query can then be evaluated in linear time analogously to \Cref{ex:first}.
\end{example}

\tikzset{
  guarded/.style={
    circle,
    draw=blue,
    thick,
    inner sep=1.5pt,
    minimum size=7mm
  },
  unguarded/.style={
    circle,
    draw=red,
    thick,
    inner sep=1.5pt,
    minimum size=7mm
  },
  empnode/.style={
    circle,
    draw=black,
    thick,
    inner sep=1.5pt,
    minimum size=7mm
  },
  levelbox/.style={
    draw=black,
    dashed,
    rounded corners,
    inner sep=7pt
  }
}

\begin{figure}[t]
\[
\begin{array}{ccccc}
 \hspace{0mm}
\scalebox{0.72}{%
\begin{tikzpicture}[
  every node/.style={font=\small}
]
\node[empnode] (x) at (0,0) {$x$};

\node[empnode] (y) at (-1,2) {$y$};
\node[empnode] (y') at (1,2) {$y'$};

\node[empnode] (z) at (-1,4) {$z$};
\node[empnode] (z') at (1,4) {$z'$};

\node[empnode] (v) at (0,6) {$v$};

\sarr{x}{y}
\sarr{x}{y'}

\sarr{y}{z}
\sarr{y'}{z'}
\sarr{z}{v}
\sarr{z'}{v}

\darr[bend right=20]{y}{x}        
\darr[bend left=20]{y'}{x}        
\darr[bend right=50]{z}{x}        
\darr[bend left=50]{z'}{x}        

\end{tikzpicture}
}
&
 \hspace{-3mm}
\begin{tikzpicture}
\draw[->]
    (0,0) -- (1.6,0)
    node[midway,above] {saturation};
    \node (y) at (0,-3) {};
\end{tikzpicture}
 \hspace{-3mm}
&
\scalebox{0.72}{%
\begin{tikzpicture}[
  every node/.style={font=\small}
]
\node[empnode] (x) at (0,0) {$x$};

\node[empnode] (y) at (-1,2) {$y$};
\node[empnode] (y') at (1,2) {$y'$};

\node[empnode] (z) at (-1,4) {$z$};
\node[empnode] (z') at (1,4) {$z'$};

\node[empnode] (v) at (0,6) {$v$};

\sarr{x}{y}
\sarr{x}{y'}

\sarr{y}{z}
\sarr{y'}{z'}
\darr{z}{v}
\darr{z'}{v}

\darr{y}{v}
\darr{y'}{v}

\darr{x}{v}

\darr[bend right=20]{y}{x}        
\darr[bend left=20]{y'}{x}        
\darr[bend right=50]{z}{x}        
\darr[bend left=50]{z'}{x}        

\end{tikzpicture}
}
&
  \hspace{-3mm}
\begin{tikzpicture}
\draw[->]
    (0,0) -- (1.6,0)
    node[midway,above] {``pruning''};
    \node (y) at (0,-3) {};
\end{tikzpicture}
  \hspace{0mm}
&
\scalebox{0.72}{%
\begin{tikzpicture}[
  every node/.style={font=\small}
]
\node[empnode] (x) at (0,0) {$x$};

\node[empnode] (y) at (-1,2) {$y$};
\node[empnode] (y') at (1,2) {$y'$};

\node[empnode] (z) at (-1,4) {$z$};
\node[empnode] (z') at (1,4) {$z'$};


\sarr{x}{y}
\sarr{x}{y'}

\sarr{y}{z}
\sarr{y'}{z'}

\end{tikzpicture}
}
\end{array}
\]
\caption{The key-nonkey-graphs $\FD{q}$ before and after query saturation and ``query pruning''; the saturated query is already normalized. Double arrows represent consistent edges and single arrows inconsistent edges.
 \label{fig:satprune}}
\end{figure}
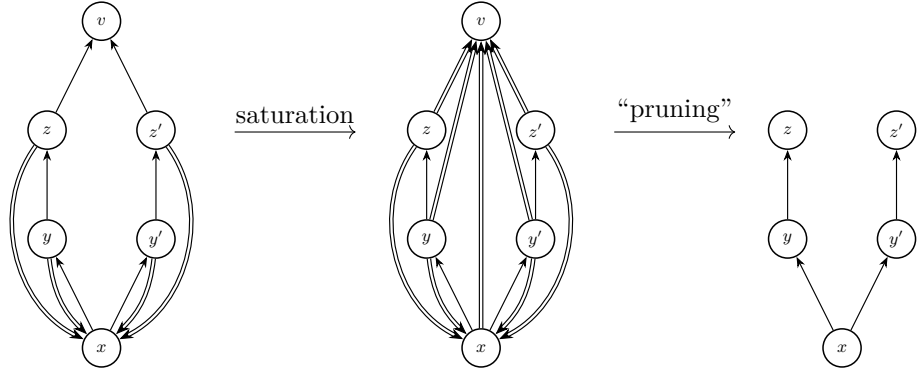

A saturated query has three useful properties: consistent edges are transitively closed, inconsistent edges propagate attacks, and two inconsistent edges cannot have distinct sources and a common target. The last property additionally relies on the acyclicity of the attack graph. Formal statements and proofs are given in \Cref{sect:appsat}.

\subsection{Query Normalization}\label{sect:normal}
As the next preprocessing step, we transform the query into a normal form. Informally, normalization ensures that whenever an atom generates more than one edge, all of them are inconsistent. Consequently, an atom generates either only consistent edges or only inconsistent edges, but never both. Furthermore, normalization guarantees that consistent edges are never reversed, never occur inside the non-primary-key part of an atom, and that the source of a consistent edge is never $\emp$. These properties will be used in the pruning stage.

\begin{definition}[Normality]\label{def:normal}
Let $q\in \sjfuabcq$. We say that $q$ is \emph{normal} if the following statements hold, for every $F\in q$ with $x= \keyvars{F}$.
\begin{enumerate}
\item\label{it:normal1} If $|\notkeyvars{F}|>1$, then $\fd{x}{y} \in \WFD{q}$, for each  $y\in \notkeyvars{F}$.
\item\label{it:normal2} $\fd{y}{z} \notin \SFD{q}$, for distinct $y,z\in \notkeyvars{F}$.
\item\label{it:normal4}   $\fd{y}{z} \in \SFD{q}$ implies $\fd{z}{y} \notin \SFD{q}$, for  distinct $y,z\in \queryvars{q}$.
\item\label{it:normal3} $\fd{\emp}{z} \notin \SFD{q}$, for each $z\in \queryvars{q}$.
\end{enumerate}
\end{definition}

The normalization process enforces one item at a time in the listed order; we briefly describe how this is done. An atom $R(\underline{x},y,\vec{z})$ with $\fd{x}{y} \in \SFD{q}$ and $\vec{z}$ containing at least one variable is replaced by a fresh atom $R'(\underline{x},\vec{z})$, with binary-saturation (\Cref{lem:c-sat}) guaranteeing that the edge $\fd{x}{y}$ remains consistent.
Similarly, an atom $R(\underline{x},y,z,\vec{u})$ with $\fd{y}{z} \in \SFD{q}$ is replaced by an atom $R'(\underline{x},y,\vec{u})$.
Any two distinct variables $y,z\in \queryvars{q}$ with $\fd{y}{z},\fd{z}{y}\in \SFD{q}$ are identified, and subsequently all variable repetitions in atoms are removed. Finally, any  $z\in \queryvars{q}$ with $\fd{\emp}{z} \in \SFD{q}$ is replaced by a uniquely determined constant. 

The following normalization lemma, proven in \Cref{sect:appnormal} and using a reduction  called $\Normalize$ (\Cref{alg:pre}), makes this process precise. Note that, by taking care of not removing consistent edges, after normalization the query and the database remain saturated.

\begin{restatable}[Normalization Lemma]{lemma}{normal}\label{lem:norm}
  Let $\mathcal{C}$ be the set of pairs $(\db,q)$, where $q$ is a saturated query from $\sjfuabcq$ and $\db$ is a $q$-saturated database.
Then $\Normalize$ is an equivalence-preserving, generator-monotone, polynomial-time reduction from $\mathcal{C}$ to $\mathcal{C}\cup\{\true,\false\}$.
Additionally, for each $(\db,q)$, assuming $(\db',q') = \Normalize(\db,q)\notin \{\true,\false\}$, we have that $q'$ is normal.
\end{restatable}

This completes the preprocessing stage. Henceforth, we may assume that $q$ is saturated and normal and that $\db$ is $q$-saturated. 

  \section{Attack Propagation and Guarded Pruning}\label{sect:ap-gp}
 The pruning stage makes use of the structural properties established during preprocessing. First, we introduce the attack-propagation graph to record how attacks
  propagate along key-nonkey edges.  
  This graph gives rise to a decomposition of the query into tree-like
  components. Intuitively, the roots of these components are existentially quantified variables, while the remaining variables are universally quantified. However, a na\"{i}ve evaluation of the latter would create an exponential blow-up, so after existential quantification the obtained
   $\emp$-rooted component is first pruned and only then universally quantified. 

 \subsection{Attack-Propagation Graphs}
 The attack-propagation graph is a subgraph of the key-nonkey graph. In particular, it does not include the dummy edges $\fd{x}{x}$ that were introduced for $\SFD{q}$.
 
 \smallskip
\begin{definition}[Attack-Propagation Graph]\label{def:attackprop}
Let $q$ be a saturated and normal query from $ \sjfuabcq$.
We define the \emph{attack-propagation graph of $q$}, denoted $\igraph{q}$, as follows
\begin{itemize}
\item the vertices consist of $\emp$ and all variables in $\queryvars{q}$; and
\item an edge $\fd{x}{y}\in \FD{q}$ is also an edge of $\igraph{q}$ if
\begin{itemize}
\item $\fd{x}{y}\in \WFD{q}$, or 
\item $\fd{x}{y}\in \SFD{q}$ and there exists an edge $\fd{z}{x}\in \WFD{q}$ such that $\fd{z}{y}\notin \SFD{q}$.
\end{itemize}
\end{itemize}
We denote by $\leafqueryvars{q}$ the set of all variables in $\queryvars{q}$ that have an incoming edge, but no outgoing edge, in $\igraph{q}$. Moreover, we denote by $\rootqueryvars{q}$ the set of all variables in $\queryvars{q}$ that have no incoming edge in $\igraph{q}$.
\end{definition}
Note that $\igraph{q}$ has no self-loops because $q$ is variable-repetition-free and $\emp$ has no incoming edge.
Moreover, by \Cref{lem:samesourcesat}, each variable $x$ has at most one incoming edge in $ \WFD{q}$. If such an edge exists, it is generated by a  unique atom. Furthermore, by normality, if $\fd{x}{y}\in \igraph{q}\cap \SFD{q}$, then $x$ is a variable (and not $\emp$).

 The attack-propagation graph $\igraph{q}$ is a DAG in which every vertex is reachable from exactly one element of $\rootqueryvars{q} \cup\{\emp\}$ (see \Cref{lem:forest,lem:distroot} in the appendix). With slight abuse of terminology, the induced subgraph on a $\subseteq$-maximal set of vertices sharing a root is called a \emph{tree} of $\igraph{q}$.
 
 \begin{example}
Consider a query
\begin{equation}\label{eq:qap}
q \defeq
\{
R(\underline{c},x_1,x_2),
S(\underline{x_1},y_1),
T(\underline{y_2},z_1,z_2)
\}
\cup\{U_i(\underline{x_1},y_2),
V_i(\underline{x_2},y_2),
W_i(\underline{y_1},y_2)
\mid i\in \{0,1\}\}\\
\end{equation}
where $c$ is a constant and all the remaining terms are variables. As shown in \Cref{fig:ap}, each edge of the key-nonkey graph $\FD{q}$ belongs also to the attack-propagation graph $\igraph{q}$, except for the consistent edge $\fd{y_1}{y_2}$. Thus, $\igraph{q}$ is ``tree-like'': it is acyclic, and its only vertex with multiple incoming edges is $y_2$; both incoming edges are consistent, and their sources $x_1$ and $x_2$ occur together in the non-key part of $R$.

 \end{example}
The behavior at $y_2$ is representative: in general, any two distinct edges of $\igraph{q}$ with a common target are consistent, and their sources occur together in the non-key part of an atom (see \Cref{lem:backward} in the appendix).

Every edge of $\WFD{q}$ propagates an attack in the sense of \Cref{lem:nextattack}. By definition, all such edges belong to $\igraph{q}$. The remaining edges of $\igraph{q}$, namely those in $\igraph{q}\cap \SFD{q}$, propagate attacks similarly. The following lemma makes this precise.
\medskip
\begin{lemma}\label{lem:upattack}
Let  $q\in \sjfuabcq$ be saturated and normal.
 Suppose $\fd{x}{y}\in \WFD{q}$ and $\fd{y}{z}\in \igraph{q}\cap \SFD{q}$. Let $F\in q$ be the (unique) atom
 generating $\fd{x}{y}$. 
  Then $F\attacks{} z$.
\end{lemma}
\begin{proof}
Toward contradiction, suppose $F\nattacks{} z$. Since $F\attacks{} y$ by \Cref{lem:nextattack}, and $\{y,z\}$ is an edge in the Gaifman graph of $q$, it follows that $\fd{x}{z}\in \FDclosure{q\setminus \{F\}}$. But then, by the saturation of $q$ we obtain $\fd{x}{z}\in \SFD{q}$, contradicting the assumption that $\fd{y}{z}\in \igraph{q}\cap \SFD{q}$.
\end{proof}
Returning to the previous example, we observe that $R$ attacks $y_2$ due to the edges $\fd{\emp}{x_1}\in \WFD{q}$ (generated by $R$) and $\fd{x_1}{y_2}\in \igraph{q}\cap \SFD{q}$.
In contrast, we have $\fd{x_1}{y_1}\in \WFD{q}$ (generated by $S$) and $\fd{y_1}{y_2}\in  \SFD{q} \setminus \igraph{q}$, reflecting the fact that $S$ does not attack $y_2$.

The attack-propagation graph determines the pruning order, but does not by itself control the size of intermediate instances. For this, we introduce guarding.

  \subsection{Guarding}\label{sect:guarding}
We begin by defining the set of \emph{source variables} for a query $q$: 
\[\srcqueryvars{q}\defeq \{x\in \queryvars{q}\mid x \text{ belongs to a source SCC in }\FD{q}\}.\]
A variable $x \in \queryvars{q}$ is called a \emph{guard variable} if  $x\in \rootqueryvars{q} \cap \srcqueryvars{q}$. The set of guard variables is denoted $\initqueryvars{q}$.
 We say that $y\in \queryvars{q}$ is \emph{guarded}
 if there exists $x\in \initqueryvars{q}$ such that $\fd{x}{y}\in \SFD{q}$. 
 The set of \emph{guarded variables} is denoted
 $\nrqueryvars{q}$ 
and the set of \emph{non-guarded variables} is $ \rqueryvars{q}\defeq \queryvars{q}\setminus \nrqueryvars{q}$.

\begin{example}\label{ex:q'}
Given $q$ from \eqref{eq:qap}, consider the following extended query
\begin{equation}\label{eq:q'}
q' \defeq
q  \cup
\{A_i(\underline{u},y_1),B_i(\underline{u},y_2),C_i(\underline{u},z_2)\mid i\in \{0,1\}\}.
\end{equation}
The variable $u$ is a guard since it does not have any incoming edges in the key-nonkey graph; in particular, it is both a root of the attack-propagation graph and a member of a source SCC in the key-nonkey graph. It is also the only guard variable associated with $q'$, since the remaining variables have incoming paths rooted at $\emp$. Furthermore, the variable $u$ is  guarded vacuously by itself, the variables $y_1,y_2,z_2$ are all guarded by it, and the remaining variables $x_1,x_2,z_1$ are unguarded, as illustrated in \Cref{fig:guarding}.

Since no two guard variables belong to the same source SCC, their number is bounded by $\cgs{q'}$ (see \Cref{lem:isone} in the appendix). One role of guarding is to control the size of the instances produced by reductions. We next see why this control is needed when flattening the tree rooted at $\emp$.

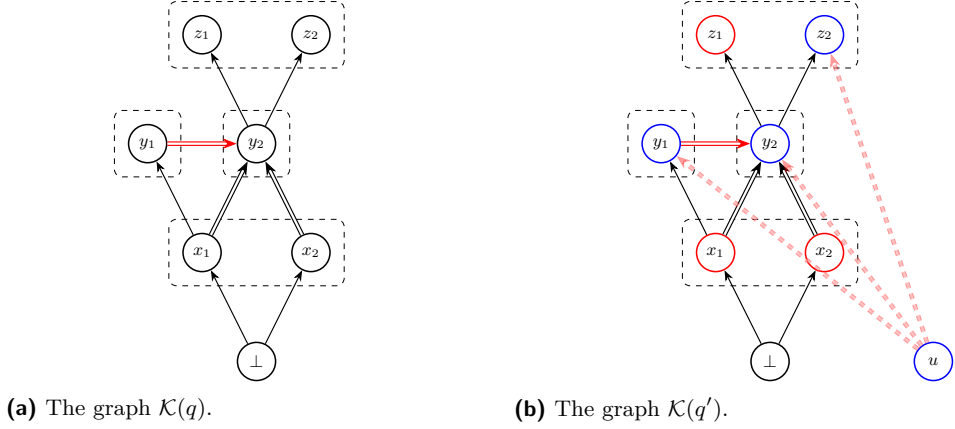
\begin{figure}
\centering
\begin{subfigure}[t]{0.42\textwidth}
\centering
\scalebox{0.72}{%
\begin{tikzpicture}[
  every node/.style={font=\small},
  plainvar/.style={
    circle,
    draw,
    thick,
    inner sep=1.5pt,
    minimum size=7mm
  }
]
\node[empnode] (emp) at (0,0) {$\bot$};

\node[plainvar] (x1) at (-1,2) {$x_1$};
\node[plainvar] (x2) at (1,2) {$x_2$};

\node[plainvar] (y1) at (-2,4) {$y_1$};
\node[plainvar] (y2) at (0,4) {$y_2$};

\node[plainvar] (z1) at (-1,6) {$z_1$};
\node[plainvar] (z2) at (1,6) {$z_2$};

\node[levelbox, fit=(x1)(x2)] {};
\node[levelbox, fit=(y1)] {};
\node[levelbox, fit=(y2)] {};
\node[levelbox, fit=(z1)(z2)] {};

\sarr{emp}{x1}
\sarr{emp}{x2}

\sarr{x1}{y1}
\bdarr{x1}{y2}
\bdarr{x2}{y2}
\rdarr{y1}{y2}

\sarr{y2}{z1}
\sarr{y2}{z2}
\end{tikzpicture}%
}
\caption{The graph $\FD{q}$.}\label{fig:ap}
\end{subfigure}
\hfill
\begin{subfigure}[t]{0.52\textwidth}
\centering
\scalebox{0.72}{%
\begin{tikzpicture}[
  every node/.style={font=\small}
]
\node[empnode] (emp) at (0,0) {$\bot$};
\node[guarded] (u) at (3,0) {$u$};

\node[unguarded] (x1) at (-1,2) {$x_1$};
\node[unguarded] (x2) at (1,2) {$x_2$};

\node[guarded] (y1) at (-2,4) {$y_1$};
\node[guarded] (y2) at (0,4) {$y_2$};

\node[unguarded] (z1) at (-1,6) {$z_1$};
\node[guarded] (z2) at (1,6) {$z_2$};

\node[levelbox, fit=(x1)(x2)] {};
\node[levelbox, fit=(y1)] {};
\node[levelbox, fit=(y2)] {};
\node[levelbox, fit=(z1)(z2)] {};

\sarr{emp}{x1}
\sarr{emp}{x2}

\sarr{x1}{y1}
\bdarr{x1}{y2}
\bdarr{x2}{y2}
\rdarr{y1}{y2}

\sarr{y2}{z1}
\sarr{y2}{z2}

\darr[red,opacity=.4,dashed]{u}{y1}
\darr[red,opacity=.4,dashed]{u}{y2}
\darr[red,opacity=.4,dashed]{u}{z2}
\end{tikzpicture}%
}
\caption{The graph $\FD{q'}$.}\label{fig:guarding}
\end{subfigure}
\caption{The key-nonkey graphs of $q$ and its guarding extension $q'$. Double arrows represent consistent edges and single arrows inconsistent edges; the red arrows do not belong to the attack-propagation graph. Dashed rectangles enclose the non-primary-key variables of a single atom. In the right panel, blue vertices are guarded and red vertices unguarded; the faded dashed arrows originating at $u$ are induced by the guarding atoms.}
\label{fig:ap-guarding}
\end{figure}
\end{example}


\subsection{Pruning}\label{sect:prune}
Pruning transforms the tree of $\igraph{q}$ rooted at $\emp$ into a tree of depth one and removes its unguarded variables. A na\"{i}ve implementation of this flattening may, however, materialize exponentially large  instances. The following example isolates this problem and shows how guarding prevents the blow-up.

\smallskip
\begin{example}\label{ex:guard}
Consider a query
\[
q_n \defeq  \{R_1(\underline{c},x_1), R_2(\underline{x_1},x_2), \dots ,R_n(\underline{x_{n-1}},x_n)\},
\]
where $c$ is a constant and all the remaining terms are variables. {At this point, none of the variables $x_1, \dots ,x_n$ is guarded.}
Let $\db$ be a database. 
Consider the result of flattening the path in $q_n$ into a single-atom query: $q^* \defeq \{R'(\underline{c},x_1,\dots ,x_n)\}$.
  A direct database transformation would introduce a fact $R'(\underline{c},d_1,\dots ,d_n)$ for every sequence of facts $R_1(\underline{c},d_1), R_2(\underline{d_1},d_2), \dots ,R_n(\underline{d_{n-1}},d_n)$ from $ \db$. Consequently, the size of $R'$ could be $2^n$ for a database whose size is linear in $n$. Now, extend $q_n$ to the query
  $\widehat q_n \defeq q_n\cup\{S_{i,j}(\underline{u},x_i)\mid i\in [n],j\in \{0,1\}\}$
  and suppose that $\db$ is $\widehat q_n$-saturated. The variables $x_1,\dots,x_n$ are then guarded by $u$. 
  Now, form $R'_{n-1}$ as the join of $R_{n-1}$ and $R_n$ (for both the database and the query), and saturate the database with respect to $\widehat q_{n-1} \defeq (\widehat q_n \setminus \{R_{n-1},R_n\})\cup\{R'_{n-1}\}$; this removes all triples $R'_{n-1}(\underline{d_{n-2}},d_{n-1},d_n)$ that are not guarded by the same value of $u$ (see the guardedness condition). 
  If saturation does not return $\false$, 
   we continue from the saturated database and $\widehat q_{n-1}$. Importantly, the size of $R'_{n-1}$ is bounded by the number of values taken by $u$ in the initial database $\db$, and hence by $|\adom{\db}|$. Continuing like this, if $\false$ is not returned, one obtains a relation $R'$ whose size is bounded by $|\adom{\db}|$.
\end{example}

The main idea in query pruning is to remove unguarded variables at leaves and subsequently propagate only guarded variables toward the root. Since the number of guard variables is bounded, every relation occurring in an intermediate instance remains polynomially bounded. Query pruning itself is decomposed into four separate steps; \Cref{fig:prune} in the appendix illustrates how the $\emp$-rooted tree of $\igraph{q'}$ (for the query $q'$ defined in \eqref{eq:q'}) is pruned. Unless pruning returns $\false$, applying it to $q'$ yields the flattened query
\begin{equation}\label{eq:q''}
q'' = 
\{
R'(\underline{c},y_1,y_2,z_2)
\}
  \cup
\{A_i(\underline{u},y_1),B_i(\underline{u},y_2),C_i(\underline{u},z_2)\mid i\in \{0,1\}\}.
\end{equation}

The following lemma formalizes this construction. The reduction $\textsc{Prune}$ is defined in \Cref{alg:prune}. Guarded reductions are defined in \Cref{sect:guarded}; intuitively, they preserve the guarding structure of the query.
We write $\keyqueryvars{q}\defeq \{x \in \queryvars{q} \mid \exists F \in q \colon x = \keyvars{F} \}$ and $\nkeyqueryvars{q} \defeq \queryvars{q}\setminus \keyqueryvars{q}$.
 \begin{restatable}[Pruning Lemma]{lemma}{prune}\label{lem:prune}
Let $c\in \mathbb{N}$. Let $\mathcal{C}$ consist of all pairs $(\db,q)$ such that $\db$ is a $q$-saturated database, $q$ is a saturated and normal query from  $\sjfuabcq$, and $|\initqueryvars{q}|\leq c$. 
Then, $\textsc{Prune}$ is a guarded, equivalence-preserving, and polynomial-time  reduction from $\mathcal{C}$ to $\mathcal{C}\cup\{\true,\false\}$.
Additionally, assuming $(\db',q')=\textsc{Prune}(\db,q)\notin\{\true, \false\}$, it holds that
\begin{equation}\label{eq:setprune}
\outvar{\igraph{q'}}{\emp} \subseteq \nrqueryvars{q'} \cap \leafqueryvars{q'} \cap \nkeyqueryvars{q'}.
\end{equation}
\end{restatable}


 \section{Variable Substitutions}\label{sect:sub}
Query pruning is used as a subroutine in our main algorithm, which otherwise follows the known first-order rewriting approach. The algorithm does not construct such a rewriting explicitly. Instead, it simulates its quantifier alternation by substituting constants for query variables through existential and universal branching. The reductions underlying these substitution steps are what we call \emph{substitution-guarded} (see \Cref{sect:appsub}), meaning that they preserve properties analogous to those of guarded reductions while allowing variables to be replaced by constants. They are also \emph{strict} in the sense that they decrease the query size (defined appropriately in \Cref{sect:appaux}).
 For the following lemma statement,  note that the effect of $\textsc{Substitute}_f$ (see \Cref{sect:aux}) on a query is simply to replace variables with constants according to $f$. Also, substitutable sets are defined in \Cref{def:sub}.
 
\begin{restatable}[Substitution Lemma]{lemma}{sub}\label{lem:pruneconst}
Let $\mathcal{C}$ consist of all pairs $(\db,q)\in \DB\times \sjfuabcq$, where $q$ is a saturated and normal query and $\db$ is a database. Let $\mathcal{C}^*$ be defined as $\mathcal{C}$, except that $\db$ has to be $q$-saturated.
 Then, for each $(\db,q)\in \mathcal{C}$, substitutable $Y\subseteq \queryvars{q}$, and function $f\colon Y\to \Const$, 
 $(\db,q)\mapsto \textsc{Substitute}_f(\db,q)$ is a strict, substitution-guarded, and polynomial-time reduction from $\mathcal{C}$ to $\mathcal{C}^*\cup\{\true,\false\}$.
\end{restatable}

Let us briefly explain the role of variable substitution.
Let $q^*$ be the query obtained from query $q'$ of \eqref{eq:q'} by replacing the constant $c$ with a fresh variable $v$, 
and let $\db$ be a database, which at this point is $q^*$-saturated.
Since $v$ is unattacked in $q^*$, there exists a constant $c$ such that every repair of $\db$ satisfies $q^*$ if and only if every repair of $\db$ satisfies $(q^*)_{v \mapsto c}$  (see \Cref{lem:non-attacked} in the appendix, due to \cite{KoutrisW17}). This is the existential substitution step. Then, the tree rooted at $\emp$ is pruned, resulting in the flattened query $q''$ of \eqref{eq:q''} and a corresponding database $\db'$.
 Now, every repair of $\db'$ satisfies $q''$ if and only if for each fact of the form $R'(\underline{c},d,e,f)$ from $\db'$, every repair of $ (\db'\setminus (R')^{\db'})\cup\{R'(\underline{c},d,e,f)\}$ satisfies $q''_{y_1,y_2,z_2 \mapsto d,e,f}$. 
 This is the universal substitution step. After this step, one tree of the attack-propagation graph has been removed. The process then continues from the next unattacked root variable; in our example, only $u$ remains.
 
Note that pruning a tree that is not rooted at $\emp$ may create cycles in the attack graph (see \Cref{ex:cycle} in the appendix). For this reason, each pruning step is performed after an existential substitution has created a $\emp$-rooted tree and before the corresponding universal substitutions eliminate its remaining guarded variables. 
 This dynamical interleaving of pruning inbetween substitutions is realized in our main lemma, to which we turn next.

\section{Proof of Main Theorem}\label{sect:main}
The lemma combines existential substitution, guarded pruning, and universal
  substitution. Strictness of substitutions and non-trivial pruning bound
  branch depth, while guarding polynomially bounds the sizes of intermediate
   instances. Additionally, the bounded number of guard combinations constrains also the number of possible substitution combinations.
   The proof and a schematic illustration appear in
   \Cref{sect:appmain}, building on \Cref{lem:prune,lem:pruneconst}.

\begin{restatable}[Main Lemma]{lemma}{main}\label{lem:main}
Let $c\in \mathbb{N}$. Let $\mathcal{C}$ consist of all pairs $(\db,q)$ such that $\db$ is a $q$-saturated database, $q$ is a saturated and normal query from $\sjfuabcq$, and $\cgs{q}\leq c$. Then $\textsc{Solve}$ is an equivalence-preserving, polynomial-time 
reduction from $\mathcal{C}$ to $\{\true,\false\}$.
\end{restatable}

\begin{proof}[Proof of \Cref{thm:main}]
Let $c\in\mathbb{N}$, and let $\mathcal{Q}$ be a class of queries $q$ from $\sjfuabcq$ such that $\cgs{q}\leq c$. By \Cref{lem:dbsat,lem:repfreereal}, we may first assume that the input query is variable-repetition-free. Applying \Cref{lem:dbsat,lem:c-sat,lem:norm} in sequence either resolves the instance or produces, in polynomial time, an equivalent pair $(\db',q')$, where $\db'$ is $q'$-saturated and $q'$ is saturated and normal. Moreover,
\(
\cgs{q'}\leq \cgs{q}\leq c.
\)
The pair $(\db',q')$ thus satisfies the assumptions of \Cref{lem:main}, and hence $\textsc{Solve}$ decides it in polynomial time. The claim follows by composing these equivalence-preserving polynomial-time reductions.
\end{proof}

\section{Conclusion}

We established polynomial-time combined complexity for consistent query answering over classes of self-join-free conjunctive queries with acyclic attack graphs and unary primary keys, provided that the query's closure generator size is bounded. The proof builds novel variable-level techniques on top of earlier attack graph theory.
Without the boundedness assumption the problem is $\Pi_2^{\sf P}$-complete, while dropping the arity restriction but keeping the boundedness assumption yields $\coNP$-completeness.
A full summary of our results is in \Cref{tab:results}.

We leave open the exact parameterized complexity, whether the polynomial-time result for unary primary keys extends from bounded closure generator size to bounded closure treewidth, the identification of a tractable fragment with
  composite primary keys, and extensions beyond the first-order-rewritable fragment of self-join-free conjunctive queries.


\bibliography{main}

\appendix

\paragraph*{Declaration on the use of generative AI.}
The paper was mostly written before the advent of AI generated mathematical proofs.
Generative AI tools were used for language editing and proofreading, while the research problem, concepts, algorithms, and technical results were developed without generative AI assistance. However, in the end there were two exceptions. OpenAI’s GPT-5.6 Sol produced a $W[1]$-hardness proof in response to our comment suggesting its likelihood. This
proof was based on our $\Pi^{\mathsf P}_2$-hardness construction, and it was subsequently strengthened to co-$\mathrm{W}[3]$-hardness by the same model. Following this, we further generalized the proof to co-$\mathrm{W}[t]$-hardness for every fixed $t$.
Finally, the tables were turned. OpenAI’s GPT-6 Astra produced the \coNP-hardness proof for \Cref{thm:binary-hardness}, and we performed language editing and proofreading for it.

\paragraph*{Appendix overview.}
\Cref{sect:appaux} collects auxiliary concepts and lemmas.
\Cref{sect:app:dbsat,sect:appsat,sect:appnormal} provide the omitted proofs for
database saturation, query saturation, and query normalization, respectively.
\Cref{sect:appprune,sect:appsub} develop the pruning and substitution reductions.
\Cref{sect:appmain} proves the main lemma, and \Cref{sect:apphard,sect:apphard2} establish the hardness results for unary and binary primary keys, respectively.

\section{Auxiliary Concepts and Results}\label{sect:appaux}
In the technical treatment, we partition a query $q$ from $\sjfuabcq$ into  two subqueries $\modei{q}$ and  $\modec{q}$, defined as follows:
\begin{itemize}
\item $F\in \modec{q}$ if for all $\fd{x}{y}\in \FD{F}$ it holds that $\fd{x}{y}\in \SFD{q}$.
\item $F\in \modei{q}$ if there exists $\fd{x}{y}\in \FD{F}$ such that $\fd{x}{y}\in \WFD{q}$.
\end{itemize}

Then, we define the size conventions used in the paper.
For a finite set of atoms $q$, we define:
\begin{align*}
\#\text{rel}(q) &\defeq \text{the number of occurrences of relation symbols in } q,\\
\#\text{var}(q) &\defeq \text{the number of occurrences of variables in } q,\\
\#\text{distvar}(q) &\defeq \text{the number of variables in } q,\\
\#\text{const}(q) &\defeq \text{the number of occurrences of constant symbols in } q.
\end{align*}
We define the \emph{size} of $q$ as
\[
\occq{q} \;\defeq\; \#\text{rel}(q) + 2\#\text{var}(q) + \#\text{distvar}(q)+ \#\text{const}(q).
\]
Thus, replacing a variable by a constant strictly decreases $\occq{q}$, and so does replacing two variables with one variable.
Observe that $\occq{q}$ remains linearly bounded by a more standard way of defining the size of $q$, such
as $\#\text{rel}(q) + \#\text{var}(q) + \#\text{const}(q)$.

The following lemma will be used in later reductions.
\smallskip
\begin{lemma}\label{lem:nscc}
Let $G=(V,E)$ be a directed graph. 
\begin{enumerate}
\item\label{it:ekaG} If $u,v\in V$ are two vertices such that $u$ is connected to $v$ in $G$, and $G' \defeq (V,E\cup\{(u,v)\})$, then $\nscc{G'}= \nscc{G}$.
\item\label{it:tokaG} If $x,y,z \in V$ are three vertices such that $(x,y),(y,z),(x,z)\in E$, and $G' \defeq (V,E\setminus \{(x,z)\})$, then $\nscc{G'}= \nscc{G}$.
\end{enumerate}
\end{lemma}
\begin{proof}
We claim that in both cases, a vertex $w_0$ is connected to another vertex $w_1$ in $G$ if and only if this holds in $G'$. The statements of the items then follow by a straightforward argument.
In {\Cref{it:ekaG}}, the claim holds clearly.
In {\Cref{it:tokaG}}, the claim holds because the edge $(x,z)$ in a path of $G$ can be replaced by two edges $(x,y),(y,z)$ in $G'$.
\end{proof}

\section{Lemmas for \Cref{sect:dbsat}}\label{sect:app:dbsat}
Consider the function $\dbsaturate$ given in \Cref{alg:dbsaturate}.  
The following example illustrates the violations detected by the algorithm and the blocks it removes.

\smallskip
\begin{example}
Consider the following query from $\sjfubcq$:
\[
q \;\defeq \;\{
R(\underline{c},x,y),
\;
S(\underline{z},x),
\;
S'(\underline{z},x),
\;
T(\underline{z},y),
\;
T'(\underline{z},y),
\;
U(\underline{x},v),
\;
U'(\underline{x},v),
\;
V(\underline{y},v),
\;
V'(\underline{y},v)
\}.
\]
Then, $\FD{q}=\{\fd{\emp}{x},\fd{\emp}{y},\fd{z}{x},\fd{z}{y},\fd{x}{v},\fd{y}{v}\}$ and $\SFD{q}=\{\fd{z}{x},\fd{z}{y},\fd{x}{v},\fd{y}{v}\}$.
Consider then a database $\db$ given as follows:
\[
\begin{array}{lllll}
\begin{array}[t]{r|lll}
R & \underline{c} & x & y\\
\hline
 & a & 1 & 20\\
  \cdashline{2-4}
 & c & 2 & 20\\\cdashline{2-4}
\end{array}
&
\begin{array}[t]{r|ll}
S & \underline{z} & x\\
\hline
 & d & 1\\
   & d & 2\\
  \cdashline{2-3}
 & e & 2\\\cdashline{2-3}
\end{array}
&
\begin{array}[t]{r|ll}
S' & \underline{z} & x\\
\hline
 & d & 1\\
 \cdashline{2-3}
 & e & 2\\\cdashline{2-3}
\end{array}
&
\begin{array}[t]{r|ll}
T & \underline{z} & y\\
\hline
 & d & 10\\
  \cdashline{2-3}
 & e & 20\\\cdashline{2-3}
\end{array}
&
\begin{array}[t]{r|ll}
T' & \underline{z} & y\\
\hline
 & d & 10\\
  \cdashline{2-3}
 & e & 20\\\cdashline{2-3}
\end{array}
\\[12ex]
\begin{array}[t]{r|ll}
U & \underline{x} & v\\
\hline
 & 1 & 100\\
  \cdashline{2-3}
 & 2 & 200\\\cdashline{2-3}
\end{array}
&
\begin{array}[t]{r|ll}
U' & \underline{x} & v\\
\hline
 & 1 & 100\\
  \cdashline{2-3}
 & 2 & 200\\\cdashline{2-3}
\end{array}
&
\begin{array}[t]{r|ll}
V & \underline{y} & v\\
\hline
 & 10 & 100\\
  \cdashline{2-3}
 & 20 & 200\\\cdashline{2-3}
\end{array}
&
\begin{array}[t]{r|ll}
V' & \underline{y} & v\\
\hline
 & 10 & 100\\
  \cdashline{2-3}
 & 20 & 200\\\cdashline{2-3}
\end{array}
\end{array}
\]
 The database violates $q$-saturation in several ways.
\begin{itemize}
\item $R(\underline{a},1,20)$ witnesses that $\db$ is not pairwise consistent, since there is no valuation $\theta$ mapping $R$ to this fact.
\item $S(\underline{d},1),S(\underline{d},2)$ witness that $\db$ is not functionally consistent.
\item $T(\underline{d},10)$ witnesses that $\db$ is not pairwise consistent, since is no valuation $\theta$ maps $T$ to this fact while simultaneously mapping $R$ to some fact in $\db$.
\item $R(\underline{a},1,20),U(\underline{1},100)$ witness that $\db$ is not collision consistent, since $V(\underline{20},100)$ does not belong to $\db$.
\item $R(\underline{a},1,20)$ witnesses that $\db$ is not guarded, since there exists no constant $b$ such that $S(\underline{b},1)$ and $T(\underline{b},20)$ are in $\db$.
\end{itemize}
Consider the database $\db'$ obtained from $\db$ as the output of \Cref{alg:dbsaturate}.
This process simply drops the top block from each relation in $\db$. One can verify that $(\db,q)\in \cert$ if and only if $(\db',q)\in \cert$.
Note that the same process over $\db^*\defeq \db\setminus \{V(\underline{10}, 100)\}$ would remove the $V$-relation entirely (in fact,  it would remove the entire database), in which case the algorithm would return $\false$. 
\end{example}

\smallskip

\begin{algorithm}[H]
\caption{$\dbsaturate(\db,q)$ \label{alg:dbsaturate}}
\Input{A query $q$ and a database $\db$}
\Output{A $q$-saturated database $\db'$}

\Repeat{no while-loop is entered}{

    \While{there exists a violation of non-emptiness, pairwise consistency, or guardedness witnessed by
    \(R(\underline{a},\uk)\in\db\)}{
        remove from \(\db\) all facts of the form
        \(R(\underline{a},\uk)\)\;
    }

    \While{there exists a violation of functional consistency witnessed by
    \(R(\underline{a},b,\uk),R(\underline{a},c,\uk)\in\db\)}{
        remove from \(\db\) all facts of the form
        \(R(\underline{a},\uk)\)\;
    }

    \While{there exists a violation of collision consistency witnessed by
    \(R(\underline{a},\uk)\in\db\), possibly together with
    \(S(\underline{a'},b,\uk)\in\db\)}{
        remove from \(\db\) all facts of the form
        \(R(\underline{a},\uk)\)\;
    }
}
\If{\(\restrict{\db}{R}=\emptyset\) for some \(R\in{q}\)}{
    \Return{\(\false\)}\;
}
\Return{\(\db\)}\;\end{algorithm}

The previous example generalizes to the following database saturation lemma.

\dbsat*
\begin{proof}
It is straightforward to verify that each individual step $\db \mapsto \db'$ in the algorithm is such that $(\db,q)\in \cert$ if and only if $(\db', q)\in \cert$. 
Specifically, a repair of $\db'$ that falsifies $q$ can be extended to one of $\db$ falsifying $q$ by picking the violating fact from the block of facts of the form $R(\underline{a},\uk)$. The only exception is the case of functional consistency having two witnesses; since $R(\underline{x},y,\uk)\in q$ and $\fd{x}{y}\in \SFD{q}$ implies the existence of another atom of the form $S(\underline{x},y,\uk)\in q$, one here picks from $R(\underline{a},b,\uk)$ if $S(\underline{a},b,\uk)$ does not belong to $\rep'$, and $R(\underline{a},c,\uk)$ if $S(\underline{a},c,\uk)$ does not belong to $\rep'$.  

If in the end $\restrict{\db}{R}= \emptyset$ for some $R \in q$, then clearly $(\db,q)\notin\cert$. 

It is straightforward to verify that the procedure gives rise to a an equivalence-preserving polynomial-time reduction $(\db,q)\mapsto \db \cup\{\true,\false\}$.
\end{proof}

\section{Lemmas for \Cref{sect:sat}}\label{sect:appsat}
We first record three basic consequences of query saturation.

\begin{lemma}\label{lem:ctrans}
Let $q\in \sjfubcq$ be saturated. If $\fd{x}{y},\fd{y}{z}\in \SFD{q}$, then $\fd{x}{z}\in \SFD{q}$.
\end{lemma}
\begin{proof}
Now, $\fd{x}{y}\in \SFD{q}$ entails we find $F\in q$ and $G\in q\setminus\{F\}$ such that $\fd{x}{y}\in \FD{F}$ and $\fd{x}{y}\in \FDclosure{G}$. Clearly, also $\fd{y}{z}\in \FD{q\setminus \{F\}}$, and thus $\fd{x}{z}\in \FDclosure{q\setminus \{F\}}$. Hence $\fd{x}{z}\in \SFD{q}$ follows due to saturation.
\end{proof}

\begin{lemma}\label{lem:nextattack}
Let $q\in \sjfubcq$ be saturated. Suppose $\fd{x}{y}\in \WFD{q}$.
If $F\in q$ generates $\fd{x}{y}$, then $F\attacks{}y$.
\end{lemma}
\begin{proof}
Toward contradiction, assume $F$ does not attack $y$.
Since the empty path connects $\notkeyvars{F}$ and $y$ in $\gaifman{q}$, we obtain
$y\in \keyclosure{F}{q}$. This means $\fd{x}{y}\in \FDclosure{q \setminus \{F\}}$.
Since $\fd{x}{y}\in \FD{F}$ and, by definition, $\fd{y}{y}\in \SFD{q}$, we obtain by saturation that
$\fd{x}{y}\in \SFD{q}$, a contradiction.
\end{proof}

\begin{lemma}\label{lem:samesourcesat}
Let $q\in \sjfuabcq$ be saturated. If $\fd{x}{z},\fd{y}{z}\in \WFD{q}$, then $x=y$.
\end{lemma}
\begin{proof}
Toward contradiction, suppose $\fd{x}{z},\fd{y}{z}\in \WFD{q}$ and $x\neq y$. Then we find two distinct atoms $F,G\in q$ such that $\fd{x}{z}\in \FD{F}$ and $\fd{y}{z}\in \FD{G}$. By \Cref{lem:nextattack}, this leads to a cycle in the attack graph of $q$, a contradiction. Hence the statement of the lemma holds.
\end{proof}

\subsection{Controlled and Size-Monotone Reductions}\label{sect:controlled}
Throughout the paper,  a reduction
\[
f\colon \DB\times \sjfuabcq \to (\DB\times \sjfuabcq)\cup\{\true,\false\}
\]
is called \emph{strict} if it strictly decreases the query size, i.e., if $\size{q'}< \size{q}$ whenever
$(\db',q')=f(\db,q)\notin \{\true,\false\}$.

\smallskip

\begin{definition}[$c$-Controlled and Size-Monotone Reduction]\label{def:sizecontrolled}
Let \(c\in\mathbb N\).
A reduction
\[
f\colon \DB\times \sjfuabcq \to (\DB\times \sjfuabcq)\cup\{\true,\false\}
\]
is \emph{$c$-controlled} if  
 for each \((\db,q)\in \DB\times \sjfuabcq\), having
\(
(\db',q')\defeq f(\db,q)\notin\{\true,\false\},
\)
the following hold:
\begin{enumerate}[(i)]
\item\label{it:sc-query-size}
\(\size{q'}\leq \size{q}+c\);

\item\label{it:sc-db-size}
\(\size{\db'}\leq \size{\db}+c|\adom{\db}|\);

\item\label{it:sc-adom}
\(\adom{\db'}\subseteq \adom{\db}\);

\item\label{it:sc-vars}
\(\queryvars{q'}\subseteq \queryvars{q}\).
\end{enumerate}
The reduction 
is \emph{size-monotone} if it is $0$-controlled.
\end{definition}

It follows immediately that, given a fixed constant $c$, any polynomial-length composition of polynomial-time $c$-controlled reductions remains computable in polynomial time.
\smallskip
\begin{lemma}\label{lem:sizecontrolled}
Let \(c\in\mathbb N\).
Any polynomial-length sequence of $c$-controlled
polynomial-time reductions can be evaluated in polynomial time.
\end{lemma}

\subsection{Saturation Lemma}
We illustrate the concept of query saturation via the following example. After the initial database saturation step, we have a problem instance $(\db,q)$ where $q$ is a query from $\sjfubcq$ and $\db$ is a $q$-saturated database.
\smallskip
\begin{example}
Consider a query
\[
q \defeq \{
R(\underline{x},y),
S(\underline{y},z),
S'(\underline{y},z),
U(\underline{x},v),
V(\underline{v},z)
\}.
\]
It holds that $\FD{q}=\{\fd{x}{y},\fd{y}{z},\fd{x}{v},\fd{v}{z}\}$ and $\SFD{q}=\{\fd{y}{z}\}$.
We also observe that $q$ is not saturated, since $\fd{x}{y}\in \FD{R}, \fd{y}{z}\in \SFD{q}, \fd{x}{z}\in \FDclosure{q\setminus \{R\}}$,
while $\fd{x}{z}\notin \SFD{q}$.

Consider a database $\db$ given as follows:
\[
\begin{array}{ccccc}
\begin{array}[t]{r|ll}
R & \underline{x} & y\\
\hline
 & a & 0\\
 & a & 1\\\cdashline{2-3}
 & b & 1\\
 & b & 2\\\cdashline{2-3}
\end{array}
&
\begin{array}[t]{r|ll}
S & \underline{y} & z\\
\hline
 & 0 & d\\\cdashline{2-3}
 & 1 & e\\\cdashline{2-3}
 & 2 & e\\\cdashline{2-3}
\end{array}
&
\begin{array}[t]{r|ll}
S' & \underline{y} & z\\
\hline
 & 0 & d\\\cdashline{2-3}
 & 1 & e\\\cdashline{2-3}
  & 2 & e\\\cdashline{2-3}
\end{array}
&
\begin{array}[t]{r|ll}
U & \underline{x} & v\\
\hline
 & a & f\\
 & a & g\\\cdashline{2-3}
 & b & g\\
 & b & h\\\cdashline{2-3}
\end{array}
&
\begin{array}[t]{r|ll}
V & \underline{v} & z\\
\hline
 & f & d\\\cdashline{2-3}
 & g & e\\\cdashline{2-3}
 & h & e\\\cdashline{2-3}
\end{array}
\end{array}
\]
We note that $\db$ is $q$-saturated. Specifically, the relations $S$ and $S'$ are consistent, following the fact that $\fd{y}{z}\in \SFD{q}$.
Now, $(\db,q)\in \cert$ if and only if $(\db',q)\in \cert$, where $\db'\defeq \db\setminus \{R(\underline{a},0),R(\underline{a},1)\}$.
To see why, consider a repair $\rep$ of $\db'$ falsifying $ q$. Then, $\rep$ picks either $U(\underline{a},f)$ or $U(\underline{a},g)$. If $\rep$ picks $U(\underline{a},f)$, let $\rep'$ be obtained from $\rep$ by adding $R(\underline{a},1)$. Then, $\rep'$ contains $R(\underline{a},1),S(\underline{1},e),U(\underline{a},f),V(\underline{f},d)$. In particular, the conflicting values for $z$ entail that no satisfying assignment can map $R(\underline{x},y)$ to $R(\underline{a},1)$. The case where $\rep$ picks $U(\underline{a},g)$ is treated symmetrically. Therefore, the repair $\rep'$ of  $\db$ obtained this way falsifies $q$.

We conclude that the block $\{R(\underline{a},0),R(\underline{a},1)\}$ can be safely removed from $\db$, resulting in $\db'$. A consequent application of database saturation yields $\db_0 \defeq \dbsaturate(\db',q)$, depicted below:
\[
\begin{array}{ccccc}
\begin{array}[t]{r|ll}
R & \underline{x} & y\\
\hline
 & b & 1\\
 & b & 2\\\cdashline{2-3}
\end{array}
&
\begin{array}[t]{r|ll}
S & \underline{y} & z\\
\hline
 & 1 & e\\\cdashline{2-3}
 & 2 & e\\\cdashline{2-3}
\end{array}
&
\begin{array}[t]{r|ll}
S' & \underline{y} & z\\
\hline
 & 1 & e\\\cdashline{2-3}
  & 2 & e\\\cdashline{2-3}
\end{array}
&
\begin{array}[t]{r|ll}
U & \underline{x} & v\\
\hline
 & b & g\\
 & b & h\\\cdashline{2-3}
\end{array}
&
\begin{array}[t]{r|ll}
V & \underline{v} & z\\
\hline
 & g & e\\\cdashline{2-3}
 & h & e\\\cdashline{2-3}
\end{array}
\end{array}
\]
By \Cref{lem:dbsat}, we have that $(\db',q)\in \cert$ if and only if $(\db_0,q)\in \cert$.
Now, we observe that any selection from the $R$-block is associated with a unique value of variable $z$. Hence, letting
\begin{align*}
q^* &\defeq q \cup \{T(\underline{x},z),T'(\underline{x},z)\},\\
\db^* &\defeq \db_0 \cup \{T(\underline{b},e),T'(\underline{b},e)\},
\end{align*}
we obtain that $(\db_0,q)\in \cert$ if and only if $(\db^*,q^*)\in \cert$, where $q^*$ is a saturated query and $\db^*$ is a $q^*$-saturated database. This concludes the example.
\end{example}

The construction below generalizes the example. We first describe an individual saturation step and prove its correctness. At this stage of the overall reduction,
the database $\db$ is $q$-saturated.

{\bf Saturation Step.} Let $q\in \sjfuabcq$. A pair 
\begin{equation}\label{eq:possiblyjoinable}
(x,z)\in (\queryvars{q}\cup\{\emp\})\times \queryvars{q} \setminus \{(x,x)\mid x\in \queryvars{q}\}
\end{equation}
 is called \emph{possible joinable (over $q$)}, and it is \emph{joinable (over $q$)} if there exists $F\in q$ and $y\in \queryvars{q}$
such that $F,x,y,z$ satisfy \Cref{it:sat1,it:sat2,it:sat3} of \Cref{def:saturation}. 
Note that $F$ is an atom of the form $T(\underline{t},y,\uk)$, where $t=x$ (if $x$ is a variable) or $t=a\in \Const$ (if $x=\emp$).
Given such $F,x,y,z$, we define
\[
  q'\defeq q \cup \{R(\underline{x},z),S(\underline{x},z)\},  
\]
if $x$ is a variable, and
\[
  q'\defeq q \cup \{R(\underline{a},z),S(\underline{a},z)\},  
\]
if $x=\emp$. Furthermore, $R$ and $S$ are fresh relation names, i.e., $R,S\notin \sch(q)$.

 Since $\fd{y}{z}\in \SFD{q}$, it holds that either (i) $y=z$ or (ii) there exists an atom of the form $U(\underline{y},z,\uk)$ in $q$. 
 Let $\db$ be a database. We let $\db_0$ be a database obtained from $\db$ by applying the following case-specific rule until no longer possible:
\begin{enumerate}[(i)]
\item {\bf Case $y= z$.} If $T(\underline{a},b,\uk),T(\underline{a},b',\uk)\in \db$, where $b\neq b'$, then  
remove $\blockof{T}{a}{\db}$ from $\db$.
\item {\bf Case $U(\underline{y},z)\in q$.} If $T(\underline{a},b,\uk),T(\underline{a},b',\uk),U(\underline{b},c,\uk),U(\underline{b'},c',\uk)\in \db$, where $c\neq c'$, then remove $\blockof{T}{a}{\db}$ from $\db$.
\end{enumerate}

Furthermore, we let $\db_1$ be obtained from $\db_0$ by applying the following case-specific rule until no longer possible:
\begin{enumerate}[(i)]
\item {\bf Case $y= z$.} If $T(\underline{a},b,\uk)\in \db_0$, then add $R(\underline{a},b),S(\underline{a},b)$ to $\db_0$.
\item {\bf Case $U(\underline{y},z)\in q$.} If $T(\underline{a},b,\uk),U(\underline{b},c,\uk)\in \db_0$, then add $R(\underline{a},c),S(\underline{a},c)$ to $\db_0$.
\end{enumerate}

Finally, we define $\db'\defeq \dbsaturate(\db_1,q')$. Given a joinable pair $(x,z)$, we let $\saturatestep_{x,z}(\db,q)=\false$ if $\db'\in \{\true,\false\}$,  and otherwise (when $\db'$ is a database), $\saturatestep_{x,z}(\db,q)=(\db',q')$.

We can now proceed to the saturation lemma, which is the key lemma in this work. Before it, we prove a simple helping lemma.

\begin{lemma}\label{lem:easy}
Let $q\in \sjfubcq$, and let $\fd{x}{y}\in \FDclosure{q}$. Let $\rep$ be a repair of a database $\db$. 
Let $\theta$ and $\theta'$ be two valuations such that $\theta(q)\subseteq \rep$, $\theta'(q)\subseteq \rep$. If $x$ is a variable and $\theta(x)=\theta'(x)$, then $\theta(y)=\theta'(y)$. If $x=\emp$, then $\theta(y)=\theta'(y)$.
\end{lemma}
\begin{proof}
Suppose $x=\emp$. From $\fd{x}{y}\in \FDclosure{q}$ we obtain a path $x_1,x_2,\dots ,x_n$ in $\FD{q}$ where $x_1 =\emp$ and $x_n=y$.
Since $\rep$ is consistent and $\theta(q)\subseteq \rep$ and $\theta'(q)\subseteq \rep$, we obtain $\theta(x_2)=\theta'(x_2)$. Informally, 
the edge $\fd{\emp}{x_2}$ arises from an atom of the form $R(\underline{a},x_2,\uk)$, where $a$ is a constant, and
both $\theta$ and $\theta'$ must pick the unique fact from the $R$-block of $\rep$ whose primary-key value is $a$.  Then, by simple induction we obtain
that $\theta(y)=\theta'(y)$.

Suppose $x$ is a variable. From $\fd{x}{y}\in \FDclosure{q}$ we obtain a path $x_1,x_2,\dots ,x_n$ in $\FD{q}$ where $x_1 \in \{x,\emp\}$ and $x_n=y$. Since $\rep$ is consistent and $\theta(x)=\theta'(x)$ it again follows by a simple induction that $\theta(y)=\theta'(y)$.
\end{proof}

\begin{lemma}[Saturation Lemma]\label{lem:addcfd}
Let $\mathcal{C}$ be a set of pairs $(\db,q)\in \DB\times \sjfabcq$ such that $\db$ is $q$-saturated.
There exists \(c\in\mathbb N\) such that for each $q\in \sjfabcq$ and each joinable pair $(x,z)$, $(\db,q) \mapsto \saturatestep_{x,z}(\db,q)$ is an equivalence-preserving, generator-monotone, $c$-controlled,
polynomial-time reduction from $\mathcal{C}$ to  $\mathcal{C} \cup \{\true,\false\}$. Moreover, $q'$ has an acyclic attack graph whenever $q$ has.
\end{lemma}

\begin{proof}
Clearly, we can select  \(c\in\mathbb N\)
such that $\saturatestep_{x,z}$ is $c$-controlled and in polynomial time. Furthermore, by \Cref{lem:nscc} and the fact that $\cgs{q}=\nscc{\FD{q}}-1$, the reduction is generator-monotone. It also maps problem instances into $\mathcal{C} \cup \{\true,\false\}$; in particular, if a pair $(\db',q')$ is returned, $\db'$ is $q'$-saturated by \Cref{lem:dbsat} and $q'$ is variable-repetition-free by construction.

It remains to prove the statements concerning attack graph acyclicity and equivalence preservation.
We consider the case where $x$ is a variable. The case where $x$ is $\emp$ is analogous. Note that $z$ is also a variable. By the lemma hypothesis, there exists $F\in q$ and a variable $y$ such that
$F,x,y,z$ satisfy \Cref{it:sat1,it:sat2,it:sat3} of \Cref{def:saturation}.



  {\bf (Preservation of  attack graph acyclicity) }
  To show that $q'\in\sjfuabcq$, it suffices to prove that $q'$ has an acyclic attack graph.
It is clear that $R(\underline{x},z)$ and $S(\underline{x},z)$ do not attack any atoms in $q'$. It thus suffices to show    $G\attacks{q'} H$ entails $G\attacks{q} H$ for all $G,H\in q$. Now, $G\attacks{q'} H$ means that $\gaifman{q'}$ has a path $U$
between $\notkeyvars{G}$ and $\atomvars{H}$ containing no variables from 
 $\keyclosure{G}{q'}$. Since $\keyclosure{G}{q}\subseteq \keyclosure{G}{q'}$, no variable from $U$ belongs to 
 $\keyclosure{G}{q}$ either.
 
 If $U$ is also a path in $\gaifman{q}$, we obtain $G \attacks{q} H$, as required. (The case where $x=\emp$ stops here, as
 atoms of the form $R(\underline{a},z),S(\underline{a},z)$, where $a$ is a constant, do not generate edges to the Gaifman graph.)

Thus, suppose $U$ is not a path in $\gaifman{q}$, in which case it contains an edge that does not belong to $\gaifman{q}$.
The only possibility is that this edge is $\{x,z\}$; that is, it is an edge introduced by the atom $R(\underline{x},z)$ (or $S(\underline{x},z)$).
 Note that $x$ and $z$ are not in $ \keyclosure{G}{q}$.
Therefore, we have  $F\neq G$, because otherwise \cref{it:sat1,it:sat3}
 would entail $z\in \keyclosure{G}{q}$, a contradiction; specifically, we would have $\keyvars{G}=x$ and $\fd{x}{z}\in \FDclosure{q\setminus\{G\}}$.

Now, let $U'$ be obtained from $U$ by replacing the edge $\{x,z\}$ with two edges $\{x,y\}$ and $\{y,z\}$. Clearly, $U'$ is a path in 
$\gaifman{q}$ between $\notkeyvars{G}$ and $\atomvars{H}$. We claim that no variable from $U'$ belongs to $ \keyclosure{G}{q}$.
Toward contradiction, suppose this is not the case. As $U$ is not separated by $\keyclosure{G}{q}$, it follows that $y\in \keyclosure{G}{q}$. But then, \cref{it:sat2} entails $z\in \keyclosure{G}{q}$, again a contradiction. We conclude by contradiction that $\notkeyvars{G}$ and $\atomvars{H}$ are not separated by $\keyclosure{G}{q}$. This establishes $G \attacks{q} H$. Hence the attack graph of $q'$ is acyclic.


  


{\bf (Equivalence preservation) }
We consider the case where $y\neq z$; the case where $y=z$ is analogous. In this case, $q$ contains an atom of the form $G\defeq U(\underline{y},z,\uk)$. Furthermore, we write  $F=T(\underline{x},y,\uk)$.

We prove two claims, reflecting the two stages in the construction of $\db'$. First we show that the block removal preserves equivalence.
\begin{claim}\label{claim:first1}
If $T(\underline{a},b,\uk),T(\underline{a},b',\uk),U(\underline{b},c,\uk),U(\underline{b'},c',\uk)\in \db$, where $c\neq c'$, then  
\[\text{$(\db,q)\in \cert$ if and only if $(\db\setminus \blockof{T}{a}{\db} ,q)\in \cert$.}\]
\end{claim}
\begin{claimproof}
It suffices to show that $(\db\setminus \blockof{T}{a}{\db},q)\notin \cert$ implies $(\db,q)\notin \cert$.
Let $\rep$ be a repair of $\db\setminus \blockof{T}{a}{\db}$ falsifying $q$. 
 Thus, $\rep \cup\{T(\underline{a},b,\uk)\}$ and $\rep \cup\{T(\underline{a},b',\uk)\}$ are two distinct repairs of $\db$. Assume toward contradiction that they both  satisfy $q$. Let $\theta$ and $\theta'$ be valuations such that $\theta(q) \subseteq \rep \cup\{T(\underline{a},b,\uk)\}$ and $\theta'(q) \subseteq \rep \cup\{T(\underline{a},b',\uk)\}$. Clearly, $\theta(F)=T(\underline{a},b,\uk)$ and $\theta'(F)=T(\underline{a},b',\uk)$, whence $\theta(x)=\theta'(x)$.
Furthermore, $\theta(q\setminus\{F\})\subseteq \rep$, $\theta'(q\setminus\{F\})\subseteq \rep$ and $\fd{x}{z}\in \FDclosure{q\setminus \{F\}}$, and thus \cref{lem:easy} implies $\theta(z)=\theta'(z)$. Since $\db$ is $q$-saturated, \cref{it:sat1} implies that $U^{\db}$ is consistent.
In particular, both $U(\underline{b},c,\uk)$ and $U(\underline{b'},c',\uk)$  belong to $\rep$, hence $\theta(G)=U(\underline{b},c,\uk)$ and $\theta'(G)=U(\underline{b'},c',\uk)$. But then $c\neq c'$ entails $\theta(z)\neq \theta'(z)$, a contradiction. We conclude by contradiction that $\db$ has a repair that falsifies $q$.
\end{claimproof}
Then, we show that the database can be extended, reflecting the second stage. 
\begin{claim}\label{claim:second2}
Suppose $T(\underline{a},b,\uk),T(\underline{a},b',\uk),U(\underline{b},c,\uk),U(\underline{b'},c',\uk)\in \db$ implies $c= c'$. Let
\[
\db' \defeq \{R(\underline{a},c),S(\underline{a},c)\mid \exists b: T(\underline{a},b,\uk),U(\underline{b},c,\uk)\in \db\} \cup\db
\]
Then,  $(\db,q)\in \cert$ if and only if $(\db',q')\in \cert$.
\end{claim}
\begin{claimproof}
It suffices to show that $(\db',q')\notin \cert$ implies $(\db,q)\notin \cert$.
Let $\rep'$ be a repair of $\db'$ falsifying $q'$. Let $\rep$ be obtained from $\rep'$ by dropping all the $R$-facts and $S$-facts. Clearly $\rep$ is a repair of $\db$. Assume toward contradiction that $\rep$ satisfies $q$. Let $\theta$ be a valuation such that $\theta(q)\subseteq \rep$. Suppose $\theta(F)=T(\underline{a},b,\uk)$  and  $\theta(G)=U(\underline{b},c,\uk)$. By the claim assumption, $\{R(\underline{a},c)\}$ and $\{S(\underline{a},c)\}$ are singleton blocks of $\db'$. Thus $\theta(q')\subseteq \rep'$, contradicting the assumption. We conclude by contradiction that $\rep$ falsifies $q$. 
\end{claimproof}

Now, by the definition of the reduction, equivalence-presevation follows by Claims \ref{claim:first1} and \ref{claim:second2} and \Cref{lem:dbsat}.
This concludes the proof of the lemma.
\end{proof}

\subsection{Saturation Step}
We begin by stating the following conventions, which will be used from hereafter.

\smallskip
\begin{remark}\label{rem:try}
For a function $f$, an assignment of the form
\(
(\db,q) \gets \Try(f(\db,q,\vec{x}))
\)
means the following: the current procedure
immediately returns $f(\db,q,\vec{x})$ if $f(\db,q,\vec{x})\in \{\true,\false\}$.
If $f(\db,q,\vec{x})=(\db',q') \notin \{\true,\false\}$, then 
 $(\db,q)\gets (\db',q')$.
 
 Furthermore, whenever an expression of the form $(g(\db,q),q')$ occurs and
$g(\db,q)\in\{\true,\false\}$, the expression is understood to
evaluate to this Boolean value rather than to a pair.
\end{remark}

Throughout the paper, reductions from $\mathcal{C} $ to $ \mathcal{C} \cup\{\true,\false\}$ are extended to be the identity on $\true$ and $\false$. In this way, we can also describe their composition properties. To start with, the following property is immediate.
\medskip
\begin{proposition}\label{prop:eqcompose}
  Generator-monotone and equivalence-preserving reductions  are both closed under composition.
\end{proposition}

The algorithm for $\Saturate$, given in \Cref{alg:sat}, applies \Cref{lem:addcfd} exhaustively to construct a saturated query.
Since \Cref{lem:addcfd} describes an equivalence-preserving and generator-monotone reduction, it gives rise to the following result.
\smallskip


\csat*
\begin{proof}
The algorithm applies  \Cref{lem:addcfd}
exhaustively. 
 Then, by \Cref{prop:eqcompose}, we obtain that $\Saturate$ is an equivalence-preserving and generator-monotone reduction from $\mathcal{C}$ to $\mathcal{C} \cup\{\true,\false\}$. In particular, this reduction does not introduces cycles into the attack graph.
 Furthermore, $q'$ is saturated and binary-saturated by construction.
 
It remains to show that $\Saturate$ runs in polynomial time. 
 Suppose $(\db',q')$ is some intermediate pair considered during computation, and suppose $(\db,q)$ is the initial input.
  Now, there exists a constant $c$ such that each individual step $\saturatestep_{x,z}(\db',q')$ is $c$-controlled. 
  After the step, $(x,z)$ no longer witnesses a failure of saturation or binary-saturation, and later steps cannot recreate such a failure for that pair.
  Thus, the number of iterations is bounded by $(|\queryvars{q}|+1)^2$ by \Cref{it:sc-vars} of \Cref{def:sizecontrolled}. Consequently, as the final output is a polynomial-length composition of $c$-controlled reductions, by \Cref{lem:sizecontrolled} it can be computed in polynomial time.
 To conclude, we need to verify that the condition of the while-loop is polynomial-time checkable. This condition consists of scans over $q'$ and, for each $F\in q'$, reachability checks over $\FD{q' \setminus\{F\}}$. 
 Since $q'$ is polynomially bounded in $q$, 
   this can be done in polynomial time in the initial input. We may thus conclude that $\saturate$ runs in polynomial time.
 \end{proof}

 \begin{algorithm}
\caption{$\Saturate(\db,q)$ \label{alg:sat}}
\Input{A database $\db$ and a query $q\in \sjfuabcq$.}
\Output{$\false$ or $(\db',q')$ where $\db'$ is a $q'$-saturated database and $q'$ a saturated and binary-saturated query from $\sjfuabcq$.}




\While{a possibly joinable $(x,z)$ witnesses that $q$ is not saturated and binary-saturated
}{
    $(\db,q) \gets\Try(\saturatestep_{x,z}(\db,q))$\;
}
\Return{$(\db,q)$}
\end{algorithm}

\section{Lemmas for \Cref{sect:normal}}\label{sect:appnormal}
In this section we prove the lemmas that are the building blocks of the normalization process.
We begin by illustrating the normalization process in detail. At this stage of the reduction, we may assume a problem instance $(\db,q)$ where $q$ is a saturated query from $\sjfuabcq$ and $\db$ is a $q$-saturated database.

\smallskip
\begin{example}
Consider a query
\[
q \defeq
\{
R(\underline{x},y,z,v),
R'(\underline{x'})
\}
\cup\{S_i(\underline{x},v),
T_i(\underline{y},z),
U_i(\underline{x},x'),
V_i(\underline{x'},x),
W_i(\underline{c},w)
\mid i\in \{0,1\}\}.
\]
We have that $\SFD{q}=\{\fd{x}{v}, \fd{y}{z},\fd{x}{x'},\fd{x'}{x}, \fd{\emp}{w}\}$ and $\WFD{q}=\{\fd{x}{y},\fd{x}{z}\}$.
The query is saturated but not normal.
In particular, all items of query normality are violated. Consider a $q$-saturated database $\db$ given as follows:
\[
\begin{array}{lllll}
\begin{array}[t]{r|llll}
R & \underline{x} & y & z & v\\
\hline
 & a & b & d & f\\
 & a & c & e & f
 \\ \cdashline{2-5}
\end{array}
&
\begin{array}[t]{r|ll}
R' & \underline{x'}\\
\hline
 & p\\
 \cdashline{2-2}
\end{array}
&
\begin{array}[t]{r|ll}
W_i & \underline{c} & w\\
\hline
 & c & 1\\ \cdashline{2-3}
\end{array}
\\[10ex]
\begin{array}[t]{r|ll}
S_i & \underline{x} & v\\
\hline
 & a & f\\
 \cdashline{2-3}
\end{array}
&
\begin{array}[t]{r|ll}
T_i & \underline{y} & z\\
\hline
 & b & d\\
 \cdashline{2-3}
 & c & e\\ \cdashline{2-3}
\end{array}
&
\begin{array}[t]{r|ll}
U_i & \underline{x} & x'\\
\hline
 & a & p\\
 \cdashline{2-3}
\end{array}
&
\begin{array}[t]{r|ll}
V_i & \underline{x'} & x\\
\hline
 & p & a\\
 \cdashline{2-3}
\end{array}
\end{array}
\]
 The problem instance $(\db,q)$ can be transformed into one where the query is additionally normal in few steps. 
Each numbered item addresses a violation of the condition with the same number in the definition of query normality.
\begin{enumerate}
\item Let $q_1 \defeq (q\setminus \{R(\underline{x},y,z,v)\}) \cup \{R_1(\underline{x},y,z)\}$ and $\db_1 \defeq (\db\setminus \{R\})\cup\{R_1\}$, where $R_1$ is obtained from $R$ by dropping the column that corresponds to $v$.
\item Let $q_2 \defeq (q_1\setminus \{R_1(\underline{x},y,z)\}) \cup \{R_2(\underline{x},y)\}$ and $\db_2 \defeq (\db_1\setminus \{R_1\})\cup\{R_2\}$, where $R_2$ is obtained from $R_1$ by dropping the column that corresponds to $z$.
\item Let $q_3 \defeq (q_2)_{x' \mapsto x}$. Let $\db_3$ be obtained from $\db_2$ by replacing values $b$ in columns corresponding to $x'$ with $f(b)$, where $f$ is the function determined by the consistent  relations $V_i$. Since $q_3$ is not variable-repetition-free, we further remove all repeated variable occurrences from atoms of $q_3$, keeping only the leftmost occurrence of each variable, and deleting the corresponding columns from the relations of $\db_3$.
\item Let $q' \defeq (q_3 \setminus \{W_i(\underline{c},w) \mid i\in\{0,1\}\}) \cup \{W_i(\underline{c},1) \mid i\in\{0,1\}\}$, and let $\db'\defeq \db_3$.
\end{enumerate}
The final query takes the form
\[
q' \defeq
\{
R(\underline{x},y),
R'(\underline{x})
\}
\cup\{S_i(\underline{x},v),
T_i(\underline{y},z),
U_i(\underline{x}),
V_i(\underline{x}),
W_i(\underline{c},1)
\mid i\in \{0,1\}\},
\]
and the final database $\db'$ can be depicted as
 \[
\begin{array}{lllll}
\begin{array}[t]{r|llll}
R & \underline{x} & y \\
\hline
 & a & b \\
 & a & c 
 \\ \cdashline{2-3}
\end{array}
&
\begin{array}[t]{r|ll}
R' & \underline{x}\\
\hline
 & a\\
 \cdashline{2-2}
\end{array}
&
\begin{array}[t]{r|ll}
W_i & \underline{c} & 1\\
\hline
 & c & 1\\ \cdashline{2-3}
\end{array}
\\[10ex]
\begin{array}[t]{r|ll}
S_i & \underline{x} & v\\
\hline
 & a & f\\
 \cdashline{2-3}
\end{array}
&
\begin{array}[t]{r|ll}
T_i & \underline{y} & z\\
\hline
 & b & d\\
 \cdashline{2-3}
 & c & e\\ \cdashline{2-3}
\end{array}
&
\begin{array}[t]{r|ll}
U_i & \underline{x} \\
\hline
 & a \\
 \cdashline{2-2}
\end{array}
&
\begin{array}[t]{r|ll}
V_i & \underline{x} \\
\hline
 & a \\
 \cdashline{2-2}
\end{array}
\end{array}
\]
The query $q'$ is now saturated and normal, and the database $\db'$ is $q'$-saturated.
\end{example}

We next set up some notational conventions, assuming an underlying query $q$ from $\sjfuabcq$.

\smallskip
\noindent
\textbf{Atoms.}
Let $F\in q$, and let $y\in \notkeyvars{F}$. Without loss of generality $F$ is of the form $R(\underline{x},y,\vec{z})\in q$.
Then, we define $\drop{F}{y}$ as the atom
\(
R'(\underline{x},\vec{z}),
\)
where $R'$ is a fresh relation symbol not occurring in $q$.
If $y\notin \notkeyvars{F}$, then $\drop{F}{y}\defeq F$. Moreover, then, $\drop{q}{y}\defeq \{\drop{F}{y}\mid F\in q\}$.

\smallskip
\noindent
\textbf{Facts.}
Let $A$ be a fact. Assume that $F$ is the unique underlying atom in $q$ with the same relation name. Suppose $y\in \notkeyvars{F}$, 
and suppose without loss of generality that $F$ is of the form $R(\underline{x},y,\vec{z})$.
Assuming $A$ is of the form $R(\underline{a},b,\vec{c})$,
define $\drop{A}{y}$ as the fact
\(
R'(\underline{a},\vec{c}),
\)
where $R'$ is the relation name of $\drop{F}{y}$. If $y\notin \notkeyvars{F}$, then $\drop{A}{y}\defeq A$.
Moreover, given a database $\db$, we write $\drop{\db}{y}\defeq \{\drop{A}{y}\mid A\in \db\}$.

\subsection{First Condition}
Given $q\in \sjfuabcq$, $F\in q$, and $y\in \notkeyvars{F}$,
define $q'\defeq (q \setminus \{F\})\cup \{\drop{F}{y}\}$.
Suppose $F$ is of the form
$T(\underline{x},y,\vec{z})$, in which case $\drop{F}{y}$ is of the form
$T'(\underline{x},\vec{z})$. Then, considering a database $\db$,  define
\begin{equation}\label{eq:db0}
\db_0 \defeq \{T'(\underline{a},\vec{c})\mid \exists b\colon T(\underline{a},b,\vec{c})\in \db\} \cup (\db\setminus \restrict{\db}{T}).
\end{equation}
We then let
\[
\normalizeone_{F,y}(\db,q)\defeq (\db',q'),
\]
where
\[
\db'\defeq \dbsaturate(\db_0,q').
\]

\begin{lemma}[Normality, \Cref{it:normal1}]\label{lem:removecfd}
Let $\mathcal{C}$ be the set of pairs $(\db,q)$, where $q$ is a saturated and binary-saturated query from $\sjfuabcq$ and $\db$ is a $q$-saturated database.
Then for each $F\in q$ such that $\notkeyvars{F}>1$ and
$\fd{x}{y}\in \SFD{q}\cap \FD{F}$, $(\db,q) \mapsto \normalizeone_{F,y}(\db,q)$ is a strict, size-monotone, generator-monotone, equivalence-preserving,
polynomial-time reduction from $\mathcal{C}$ to  $\mathcal{C} \cup \{\true,\false\}$.
\end{lemma}
\begin{proof}
Let us assume that $F$ is of the form
$T(\underline{x},y,\vec{z})$, in which case $\drop{F}{y}$ is of the form
$T'(\underline{x},\vec{z})$, and $\db_0$ is as in \eqref{eq:db0}.
 Given an atom $I\in q$, let us write $I'$ for its corresponding atom in $q'$; that is, $I'=I$ when $I\in q \setminus \{F\}$, while $F'= \drop{F}{y}$.

Since $q$ is binary-saturated and $\notkeyvars{F}>1$, it follows that $q\setminus \{F\}$ contains two atoms of the form $R(\underline{x},y),S(\underline{x},y)$.

Assume $(\db',q')=\normalizeone_{F,y}(\db,q)\notin \{\true,\false\}$.  
Now, $\normalizeone_{F,y}$ is generator-monotone due to $\FD{q}=\FD{q'}$. Moreover, that it is size-monotone and strict is clear by construction. Using \Cref{lem:dbsat}, the reduction is in polynomial time. It remains to prove that it is equivalence-preserving and that
$(\db',q')\in \mathcal{C}$.

  {\bf (Membership of $(\db',q')$ in $\mathcal{C}$)} Now, $\db'$ is $q'$-saturated by \Cref{lem:dbsat}.
  The saturation of $q'$ is straightforward to verify by the construction and the fact that $q$ is binary-saturated. For the binary-saturation of $q'$, 
 given $\fd{u}{v}\in \SFD{q'}$, we need to prove that  $q'$ contains two atoms of the form  $U(\underline{u},v),V(\underline{u},v)$.
 It is straightforward to verify that $\SFD{q'}\subseteq \SFD{q}$. Thus, by binary-saturation $q$ we obtain that $q$ contains two atoms of the form $U(\underline{u},v),V(\underline{u},v)$. Since $F$ is neither of these atoms, $q'$ contains these atoms too. Therefore, $q'$ is binary-saturated. 

To conclude that $(\db',q')\in \mathcal{C}$, it remains to prove the following claim.
\begin{claim}$q'\in \sjfuabcq$.
\end{claim}
\begin{claimproof}
  Clearly, $q'$ remains variable-repetition-free.

To prove that $q'$ has an acyclic attack graph, it suffices to show that $G'\attacks{q'} H'$ entails $G\attacks{q} H$, for all $G',H'\in q'$. Now, $G\attacks{q'} H$ means that $\gaifman{q'}$ has a path 
between some variables $x\in \notkeyvars{G'}$ and $y\in \atomvars{H'}$ such that no variable 
on the path belongs to $\keyclosure{(G')}{q'}$. 
Since every edge in $\gaifman{q'}$ belongs also to $\gaifman{q}$, the same path exists in $\gaifman{q}$ (between $x\in\notkeyvars{G}$ and $y\in \atomvars{H}$).
Note that $\keyclosure{G}{q}\subseteq \keyclosure{(G')}{q'}$ due to having both $R(\underline{x},y)$ and $S(\underline{x},y)$ in  $q$. Thus no variable on the path belongs to $\keyclosure{G}{q}$ either. This yields $G \attacks{q} H$ as required. 
\end{claimproof}




{\bf 
(Equivalence preservation)
}
It suffices to show that $(\db,q)\in \cert$ if and only if $(\db_0,q')\in \cert$.

Suppose each repair of $\db_0$ satisfies $q'$. 
Let $\rep$ be a repair of $\db$. Let $\rep'$ be obtained from $\rep$ by replacing each $T(\underline{a},b,\vec{c})$ with  $T'(\underline{a},\vec{c})$. Clearly $\rep'$ is a repair of $\db_0$. By assumption $\rep'$ satisfies $q'$. Let $\theta$ be a valuation such that $\theta(q')\subseteq \rep'$. In particular, $T(\underline{\theta(x)},\theta(\vec{z}))\in \rep'$, which means that there exists a constant $b$ such that
$T(\underline{\theta(x)},b,\theta(\vec{z})\in \rep$. Since $\fd{x}{y}\in \SFD{q\setminus\{T\}}$, there exists an atom of the form $R(\underline{x},y,\uk)\in q \setminus \{T\}$. In particular, we have $R(\underline{\theta(x)},\theta(y),\uk)\in \rep'$, and hence 
$R(\underline{\theta(x)},\theta(y),\uk)\in \rep$.
 Since $\db$ is saturated, we obtain $b=\theta(y)$. This shows $\theta(q)\subseteq \rep$. We conclude that each repair of $\db$ satisfies $q$. 

Suppose each repair of $\db$ satisfies $q$. 
Let $\rep'$ be a repair of $\db_0$. Let $\rep$ be obtained from $\rep'$ by  replacing each $T'(\underline{a},\vec{c})$ with an arbitrary fact of the form $T(\underline{a},b,\vec{c})\in \db$. It is clear that $\rep$ is a repair, and hence there exists a valuation $\theta$ such that $\theta(q)\subseteq \rep$. Then $\theta(q')\subseteq \rep'$, and we may conclude that each repair of $\db_0$ satisfies $q'$.
\end{proof}

 \medskip
 
 \subsection{Second Condition}

\bigskip

 \begin{lemma}[Normality, \cref{it:normal2}]\label{lem:removecfd2}
 Let $\mathcal{C}$ be the set of pairs $(\db,q)$, where $q$ is a saturated query from $\sjfuabcq$ such that it satisfies \Cref{it:normal1} of \Cref{def:normal}, and $\db$ is a $q$-saturated database. Then  for each $F \in q$ such that
  $\fd{y}{z}\in \SFD{q}$, for some distinct variables $y,z\in \notkeyvars{F}$,
 $(\db,q) \mapsto \normalizetwo_{F,z}(\db,q)$ is a strict, size-monotone, generator-monotone, equivalence-preserving,
polynomial-time reduction from $\mathcal{C}$ to  $\mathcal{C} \cup \{\true,\false\}$.
\end{lemma}
 \begin{proof}
 Let us assume that $F$ is of the form
  $T(\underline{x},y,z,\vec{v})\in q$, in which case $\drop{F}{z}$ takes the form
 $T(\underline{x},y,\vec{v})$. Let $\db_0$ be the database obtained by dropping the column of $T^{\db}$ corresponding to $z$, in the style of \eqref{eq:db0}.

 Given an atom $I\in q$, let us write $I'$ for its corresponding atom in $q'$; that is, $I'=I$ when $I\in q \setminus \{F\}$, while $F'= \drop{F}{z}$.

Assume $(\db',q')=\normalizetwo_{F,z}(\db,q)\notin \{\true,\false\}$.  
Now, $\normalizetwo_{F,z}$ is generator-monotone due to \Cref{lem:nscc}.
 Moreover, that it is size-monotone and strict is clear by construction. Using \Cref{lem:dbsat}, the reduction is in polynomial time. It remains to prove that it is equivalence-preserving and that
$(\db',q')\in \mathcal{C}$.

  {\bf (Membership of $(\db',q')$ in $\mathcal{C}$)}
   Now, $\db'$ is $q'$-saturated by \Cref{lem:dbsat}. It is also straightforward to verify that $q'$ satisfies \cref{it:normal1} of normality.
Consider then the following claim.
\begin{claim}$q'\in \sjfuabcq$.
\end{claim}
\begin{claimproof}
  Clearly, $q'$ remains variable-repetition-free.

For acyclicity of the attack graph over $q'$,
it suffices to show that $G'\attacks{q'} H'$ entails $G\attacks{q} H$ for all $G',H'\in q'$. Now, $G'\attacks{q'} H'$ means that $\gaifman{q'}$ has a path $U$ between some variables $u\in \notkeyvars{G'}$ and $v\in \atomvars{H'}$ that is not separated by $\keyclosure{(G')}{q'}$. 
Since every edge in $\gaifman{q'}$ belongs also to $\gaifman{q}$, the same path exists in $\gaifman{q}$ (between $u\in\notkeyvars{G}$ and $v\in \atomvars{H}$).
Moreover, it is easy to verify that $\keyclosure{G}{q}\subseteq \keyclosure{(G')}{q'}$. Thus no variable on $U$ belongs to $\keyclosure{G}{q}$ either. This yields $G \attacks{q} H$ as required. 
\end{claimproof}

To conclude that $(\db',q')\in \mathcal{C}$, 
it remains to prove the following claim.
\begin{claim}
 $q'$ is saturated.
\end{claim}
\begin{claimproof}
Let $F\in q$ be such that $\fd{u}{v}\in \FD{F'}$, $\fd{v}{w}\in \SFD{q'}$, and $\fd{u}{w}\in \FDclosure{q'\setminus \{F'\}}$. We need to show that $\fd{u}{w}\in \SFD{q'}$. Clearly, we have $\fd{u}{v}\in \FD{F}$, $\fd{v}{w}\in \SFD{q}$, and $\fd{u}{w}\in \FDclosure{q\setminus \{F\}}$, hence saturation of $q$ yields $\fd{u}{w}\in \SFD{q}$. If follows by construction that $\fd{u}{w}\in \SFD{q'}$ whenever $\fd{u}{w}$ is different from $\fd{y}{z}$. The case where $\fd{u}{w}$ equals $\fd{y}{z}$ is not possible because by
$|\notkeyvars{F}|>1$ and \Cref{it:normal1} of normality for $q$ we obtain $\fd{y}{z} \notin \SFD{q}$. 
\end{claimproof}


  {\bf (Equivalence preservation)} 
It suffices to prove  that $(\db,q)\in \cert$ if and only if $(\db_0,q')\in \cert$. 

First, suppose every repair of $\db_0$ satisfies $q'$. 
Let $\rep$ be a repair of $\db$. Let $\rep'$ be obtained from $\rep$ by replacing each $T(\underline{a},b,c,\vec{d})$ with  $T'(\underline{a},b,\vec{d})$. Clearly $\rep'$ is a repair of $\db_0$. By assumption $\rep'$ satisfies $q'$. Let $\theta$ be a valuation such that $\theta(q')\subseteq \rep'$. Then $T'(\underline{\theta(x)},\theta(y),\theta(\vec{v}))\in \rep'$. Note that $z\in \queryvars{q'}$. 
We claim that $T(\underline{\theta(x)},\theta(y),\theta(z),\theta(\vec{v}))\in \rep$.

Since $\fd{y}{z}\in \SFD{q\setminus \{F\}}$, there exists an atom of the form $R(\underline{y},z,\uk)\in q \setminus \{F\}$.
By definition, there exists a fact of the form 
$T(\underline{\theta(x)},\theta(y),c,\theta(\vec{v}))$ in $\rep$.
As $\db$ is $q$-saturated (in particular, it satisfies pairwise consistency), there exists a constant $c^*$ such that $R(\underline{\theta(y)},c^*,\uk)\in \db$. By functional consistency, this constant is unique, meaning that
$R(\underline{\theta(y)},c^*,\uk)\in \rep'$, hence $\theta(z)=c^*$.  Since $\db$ is $q$-saturated (collision consistency), 
it follows that $c=c^*$. It can now be seen that $\theta(q)\subseteq \rep$, and hence $\rep$ satisfies $q$.

Then, suppose every repair of $\db$ satisfies $q$. 
Let $\rep'$ be a repair of $\db_0$. Let $\rep$ be obtained from $\rep'$ by  replacing each $T'(\underline{a},b,\vec{d})$ with an arbitrary fact of the form $T(\underline{a},b,c,\vec{d})\in \db$. It is clear that $\rep$ is a repair of $\db$. By assumption
there is a valuation $\theta$ such that $\theta(q)\subseteq \rep$. Clearly, also $\theta(q')\subseteq \rep'$. We conclude that $\rep'$ satisfies $q'$.
 \end{proof}
 
 \medskip

  \subsection{Third Condition}
  To enforce the third condition of normality for a query $q$, we will identify variables $x$ and $y$ such that $\fd{x}{y},\fd{y}{x}\in \SFD{q}$.
  This is done two steps. First, we simply replace $x$ with $y$ in $q$. Since such a process may include variable repetition, we then remove such repetitions. We describe these two steps in this section. Before doing so, we introduce some auxiliary lemmas.
  
\subsubsection{Auxiliary Lemmas}  
  \begin{lemma}\label{lem:replaceconstant3}
Let $q\in \sjfubcq$, let $x,y\in \queryvars{q}$, and let $c$ be an arbitrary constant. Then the following statements hold:
\begin{enumerate}
\item\label{it:toc} $\nscc{\FD{q_{x\mapsto c}}}\leq \nscc{\FD{q}}$;
\item\label{it:toz} If $\fd{x}{y},\fd{y}{x}\in \FD{q}$, then $\nscc{\FD{q_{x\mapsto y}}}= \nscc{\FD{q}}$.
\end{enumerate}
\end{lemma}
\begin{proof}
Let us denote by $q'$ the query $q_{x\mapsto c}$ in \Cref{it:toc} and the query $q_{x\mapsto y}$ in \Cref{it:toz}.

{\bf (\Cref{it:toc})}
 Suppose $\nscc{\FD{q'}}=n$. Then we find vertices $x_1, \dots ,x_n$ from source SCCs $C_1, \dots ,C_n$ of $\FD{q'}$. In particular, $x_i$ and $x_j$ are not mutually reachable in $\FD{q'}$, for $i\neq j$. Moreover, $x$ is not listed among of $x_1, \dots ,x_n$.
  
  We first show that the vertices $x_1, \dots ,x_n$ are not mutually reachable in $\FD{q}$. Toward contradiction, suppose $\FD{q}$ has a simple path $D=u, \dots ,v$, where $u$ and $v$ are two vertices listed in  $x_1, \dots ,x_n$. Since $v\notin \reach{u}{\FD{q'}}$, it is the case that $x$ appears on $D$. This means we find a path from $\emp$ to $v$ in $\FD{q'}$. Since $\emp$ has no incoming edges, and $v$ appears on a source SCC, it follows that $v=\emp$. But then $\FD{q}$ has a non-empty path from $u$ to $\emp$, which contradicts the fact that $\emp$ has no incoming edges in $\FD{q}$. We conclude by contradiction that the vertices $x_1, \dots ,x_n$ are not mutually reachable in $\FD{q}$.
  
Then we show that the vertices  $x_1, \dots ,x_n$ belong to source SCCs in $\FD{q}$. Toward contradiction suppose this is not the case. Then for some vertex $v$ listed in $x_1, \dots ,x_n$, there exists an edge  $\fd{u}{v}\in {\FD{q}}$ such that $u\notin \reach{v}{\FD{q}}$.  

We claim first that $u\notin \reach{v}{\FD{q'}}$.
Toward contradiction, suppose $D=v, \dots ,u$ is a path in $\FD{q'}$. If $D$ contains variables only, then it is also a path in $\FD{q}$, a contradiction. Otherwise, $D$ contains $\emp$, in which case it follows that $v=\emp$. But then the edge $\fd{u}{v}$ contradicts the fact that $\emp$ has no incoming edges in $\FD{q}$. We conclude by contradiction that the claim holds.

Since $x \neq v$, we obtain either $\fd{u}{v}\in {\FD{q'}}$ or $\fd{\emp}{v}\in {\FD{q'}}$. The latter case contradicts the fact that $v$ belongs to a source SCC of $\FD{q'}$. Hence, the former case holds. It follows that $v$ does not belong to a source SCC in $\FD{q'}$ either, a contradiction. We conclude by contradiction that $x_1, \dots ,x_n$ belong to source SCCs in $\FD{q}$.
 
 It follows that the vertices $x_1, \dots ,x_n$ belong to distinct source SCCs in $\FD{q}$
 Hence $\nscc{\FD{q}}\geq n$, proving the claim.
 
 {\bf (\Cref{it:toz})} The graph $\FD{q_{x\mapsto y}}$ is obtained from ${\FD{q}}$ by identifying $x$ with $y$. Since $x$ and $y$ belong to the same SCC, the statement follows.
\end{proof}
 \smallskip

   \begin{lemma}\label{lem:boring}
   Let $q$ be a saturated query from $\sjfubcq$, 
 and let $x,y \in \queryvars{q}$ be such that $\fd{x}{y},\fd{y}{x}\in \SFD{q}$.
Let $q'\defeq q_{x\mapsto y}$.
Write $x'=y$, and $u'=u$ for $u\in ( \queryvars{q}\setminus \{x\})\cup \{\emp\}$.
 Then, for $u\in \queryvars{q}\cup \{\emp\}$ and $v\in \queryvars{q}$ it holds that
 \[\fd{u}{v}\in \SFD{q} \iff \fd{u'}{v'}\in \SFD{q'}.\]
 \end{lemma}
 \begin{proof}
 Since ``$\Rightarrow$'' is straightforward, we focus on  ``$\Leftarrow$''. For this, assume that $\fd{u'}{v'}\in \SFD{q'}$.
 Let us write $F'$ to denote $F_{x\mapsto y}$, for $F\in q$.
 
 Case 1: $u=v=x$. It holds by definition that $\fd{x}{x}\in \SFD{q}$. 
 
 Case 2: $u= x$ and $v\neq x$. Then $u'=y$ and $v'=v$, and the assumption is that $\fd{y}{v}\in \SFD{q'}$.
 Let $F'\in q'$ and $G'\in q'\setminus \{F'\}$ be such that $\fd{y}{v}\in \FD{F'}$ and $\fd{y}{v}\in \FDclosure{G'}$. Note that also $F\neq G$.
 
 Suppose first that $\fd{x}{v}\in \FD{F}$. If $\fd{x}{v}\in \FDclosure{G}$, then $\fd{x}{v}\in \SFD{q}$. Otherwise, we have $\fd{x}{v}\notin \FDclosure{G}$ which entails $\fd{y}{v}\in \FD{G}$. 
 Then, by $\fd{x}{y}\in \SFD{q}$ and $G\neq F$ we obtain that $\fd{x}{v}\in \FDclosure{q \setminus\{F\}}$. Since $\fd{v}{v}\in \SFD{q}$ by definition, by the saturation of $q$ we obtain $\fd{x}{v}\in \SFD{q}$. 
 
 Suppose then that $\fd{x}{v}\notin \FD{F}$. Then we obtain $\fd{y}{v}\in \FD{F}$. If $\fd{y}{v}\in \FDclosure{G}$, then we obtain that $\fd{y}{v}\in \SFD{q}$, and furthermore by saturation and $\fd{x}{y}\in \SFD{q}$ that $\fd{x}{v}\in \SFD{q}$.
 Otherwise, if $\fd{y}{v}\notin \FDclosure{G}$, we obtain $\fd{x}{v}\in \FD{G}$. In this case, we obtain, by composing $\fd{x}{y}\in \SFD{q}$ and $\fd{y}{v}\in \FD{q \setminus \{G\}}$ to obtain $\fd{x}{v}\in \FD{q \setminus \{G\}}$, and using saturation with $\fd{v}{v}\in \SFD{q}$, that $\fd{x}{v}\in \SFD{q}$.
 
 Hence we conclude that $\fd{u}{v}\in \SFD{q}$.
 
  Case 3: $u\neq x$ and $v=x$.  Then $u'=u$ and $v'=y$, and the assumption is that $\fd{u}{y}\in \SFD{q'}$.  Let $F'\in q'$ and $G'\in q'\setminus \{F'\}$ be such that $\fd{u}{y}\in \FD{F'}$ and $\fd{u}{y}\in \FDclosure{G'}$.
  
  Suppose first that $\fd{u}{x}\in \FD{F}$. If $\fd{u}{x}\in \FDclosure{G}$, then $\fd{u}{x}\in \SFD{q}$. Otherwise, we have $\fd{u}{y}\in \FDclosure{G}$ which entails that $\fd{u}{y}\in \FD{G}$ or $\fd{\emp}{y}\in \FD{G}$. In both cases, using saturation, $G\neq F$, and $\fd{y}{x}\in \SFD{q}$ we obtain $\fd{u}{x}\in \SFD{q}$.
  
  Suppose then that $\fd{u}{x}\notin \FD{F}$. Then we have $\fd{u}{y}\in \FD{F}$. If $\fd{u}{y}\in \FDclosure{G}$, then $\fd{u}{y}\in \SFD{q}$, in which case saturation and $\fd{y}{x}\in \SFD{q}$ leads to $\fd{u}{x}\in \SFD{q}$.
  Otherwise, if $\fd{u}{y}\notin \FDclosure{G}$, we obtain $\fd{u}{x}\in \FDclosure{G}$. Then, by saturation, $G\neq F$, and $\fd{y}{x}\in \SFD{q}$ we obtain $\fd{u}{x}\in \SFD{q}$.
  
   Therefore $\fd{u}{v}\in \SFD{q}$ follows.
  
  Case 4: $u\neq x \neq v$. This case is clear, concluding the proof.
 \end{proof}

\subsubsection{Replacing Variables}
Let $q\in \sjfuabcq$ be saturated, 
 and let $x,y \in \queryvars{q}$ be such that $\fd{x}{y},\fd{y}{x}\in \SFD{q}$.
Let $q'\defeq (q_{x\mapsto y})$. Let $\db$ be a $q$-saturated database.

Since $q$ is saturated and $\db$ is $q$-saturated,
there exist an atoms of the form $S(\underline{x},y,\uk),T(\underline{y},x,\uk)\in q$ 
such that $S^{\db}$ and $T^{\db}$ are 
consistent. 

For an atom $R(\underline{v_1},v_2, \dots ,v_n)\in q$ and a fact $A=R(\underline{a_1},a_2, \dots ,a_n) \in \db$, define  $f(A)\defeq R(\underline{a'_1},a'_2, \dots ,a'_n)$ where
\[
a'_i \defeq \begin{cases}
a_i &\text{ if $v_i\neq x$,}\\
b &\text{ if $v_i=x$ and there exists a fact of the form $S(\underline{a_i},b,\uk)\in \db$.}\\
\end{cases}
\]
Note that $f(A)$ is well-defined by the $q$-saturation of $\db$. 
We define
\[
\identify_{x,y}(\db,q)\defeq (\db',q'),
\]
where $\db'\defeq \dbsaturate(f[\db],q')$, and $g[X]$ denotes the image of a function $g$ on a set $X$ as usual.
\bigskip

Let $\sjfuabcq^=$ be defined otherwise as $\sjfuabcq$, except that it includes also queries that are not variable-repetition-free.

\smallskip
 \begin{lemma}[Normality, \cref{it:normal4}]\label{lem:removecfd4}
   Let $\mathcal{C}$ (resp. $\mathcal{C}^=$) be the set of pairs $(\db,q)$, where $q$ is a saturated query from $\sjfuabcq$ (resp. from  $\sjfuabcq^=$) such that it satisfies \Cref{it:normal1,it:normal2} of \Cref{def:normal}, and $\db$ is a $q$-saturated database. 
  Then, for each distinct $x,y \in \queryvars{q}$ with
  $\fd{x}{y},\fd{y}{x}\in \SFD{q}$,
 $(\db,q) \mapsto \identify_{x,y}(\db,q)$ is a strict, size-monotone, generator-monotone, equivalence-preserving,
polynomial-time reduction from $\mathcal{C}$ to  $\mathcal{C^=} \cup \{\true,\false\}$.
\end{lemma}
 \begin{proof}
 
  Let us write  $x'\defeq y$, and $u'\defeq u$ for $u\in (\queryvars{q}\setminus \{x\})\cup\{\emp\}$. 
Given an atom $F\in q$, let us write $F'$ for its corresponding atom in $q'$, that is, $F'=F_{x \mapsto y}$.
By assumption there are atoms of the form $S(\underline{x},y,\uk),T(\underline{y},x,\uk)\in q$.

Generator-monotonicity follows by \Cref{lem:replaceconstant3}, while size-monotonicity and strictness are clear. Moreover, the reduction is in polynomial time due to \Cref{lem:dbsat}. 
 It remains to prove that the reduction is equivalence-preserving while it maintains membership in $\mathcal{C^=} \cup \{\true,\false\}$.
Suppose $(\db',q') = \identify_{x,y}(\db,q) \notin \{\true,\false\}$.

{ \bf
(Membership of $(\db',q')$ in $\mathcal{C^=}$)
 }
 We divide this into several claims.
  \begin{claim}
  $q'\in \sjfuabcq^=$.
  \end{claim}
  \begin{claimproof}
  It suffices to show that $G'\attacks{q'} H'$ entails $G\attacks{q} H$ for all $G',H'\in q'$. Assume $G'\attacks{q'} H'$. Then  $\gaifman{q'}$ has a path $U'=x'_1, \dots ,x'_n$ between $ \notkeyvars{G'}$ and $ \atomvars{H'}$ that is not separated by $\keyclosure{(G')}{q'}$. 
If  $x_1, \dots ,x_n$ is a path in $\gaifman{q}$ between $ \notkeyvars{G}$ and $ \atomvars{H}$, define $U\defeq x_1, \dots ,x_n$.
Otherwise, by our assumption, for some $i\in [2,n-1]$ with $x_i=y$ it holds that $\{x_{i-1},y\}$ and $\{x,x_{i+1}\}$ (or $\{x_{i-1},x\}$ and $\{y,x_{i+1}\}$) are edges in $\gaifman{q}$; in this case, define  $U \defeq x_1, \dots x_{i-1},x_i,x,x_{i+1},\dots ,x_n$ (or $U \defeq x_1, \dots x_{i-1},x,x_i,x_{i+1},\dots ,x_n$).
Since $\fd{x}{y}\in \SFD{q}$, we obtain that $U$ is a path in $\gaifman{q}$ between $ \notkeyvars{G}$ and $ \atomvars{H}$.

Toward contradiction, suppose $G\nattacks{q} H$. Then $u \in \keyclosure{G}{q}$, for some $u$ appearing in $U$. 
Since $\fd{x}{y}\in \SFD{q}$ and $y$ appears in $U$, we may assume without loss of generality that $u \neq x$. 
Now, if $D=t, \dots ,u$ is the directed graph witnessing this, then $D \defeq t', \dots ,u'$ is a directed graph witnessing 
 $u' \in \keyclosure{(G')}{q'}$. Since $u'$ appears in $U'$, this contradicts the fact that $U'$ is not separated by $\keyclosure{(G')}{q'}$. Hence we conclude by contradiction that $G\attacks{q} H$.
\end{claimproof}

  \begin{claim}
 $q'$ is saturated.
  \end{claim}
    \begin{claimproof}
 Let $F'\in q'$ be such that $\fd{u'}{v'}\in \FD{F'}$, $\fd{v'}{w'}\in \SFD{q'}$, and $\fd{u'}{w'}\in \FDclosure{q'\setminus \{F'\}}$. We need to show that $\fd{u'}{w'}\in \SFD{q'}$.
 We may select $u,v$ such that $\fd{u}{v}\in \FD{F}$. Then also $\fd{v}{w}\in \SFD{q}$ by \Cref{lem:boring}. Furthermore, $\fd{u'}{w'}\in \FDclosure{q'\setminus \{F'\}}$ together with $\fd{x}{y},\fd{y}{x}\in \SFD{q}$ imply
 $\fd{u}{w}\in \FDclosure{q\setminus \{F\}}$; in particular, for any two edges $\fd{q}{y},\fd{y}{r}\in \FD{q'\setminus \{F'\}}$, one of the following four possibilities holds: $\fd{q}{x},\fd{x}{r}\in \FD{q\setminus \{F\}}$, $\fd{q}{y},\fd{y}{r}\in \FD{q\setminus \{F\}}$, $\fd{q}{x},\fd{x}{y},\fd{y}{r}\in \FD{q\setminus \{F\}}$, or $\fd{q}{y},\fd{y}{x},\fd{x}{r}\in \FD{q\setminus \{F\}}$. Thus, since $q$ is saturated, we obtain $\fd{u}{w}\in \SFD{q}$, hence $\fd{u'}{w'}\in \SFD{q'}$ by  \Cref{lem:boring}.
  \end{claimproof}

\begin{claim}
$q'$ satisfies \cref{it:normal1,it:normal2} of normality
\end{claim}
\begin{claimproof}
(\Cref{it:normal1}) Suppose toward contradiction that $\fd{u'}{v'} \in \SFD{q'}$ where $u'=\keyvars{F'}$,  $v'\in \notkeyvars{F'}$, while $|\notkeyvars{F'}|>1$.
Then \Cref{lem:boring} implies $\fd{u}{v} \in \SFD{q}$ where $u=\keyvars{F}$ and  $v\in \notkeyvars{F}$. Furthermore, it is clear that $|\notkeyvars{F}|\geq |\notkeyvars{F'}|$. This contradicts the normality of $q$, hence \Cref{it:normal1} holds for $q'$.

(\Cref{it:normal2}) Suppose toward contradiction that $\fd{u'}{v'} \in \SFD{q'}$, for distinct $u',v'\in \notkeyvars{F'}$.  Then \Cref{lem:boring} implies $\fd{u}{v} \in \SFD{q}$. Since 
 $u$ and $v$ are clearly distinct, and $u,v\in \notkeyvars{F'}$, we obtain a contradiction with the normality of $q$, so \Cref{it:normal2} follows for $q'$.
\end{claimproof}
This concludes the proof of the membership of $(\db',q')$ in $\mathcal{C^=}$.

{ \bf
(Equivalence preservation)}
By \Cref{lem:dbsat},
it suffices to prove  $(\db,q)\in \cert$ if and only if $(f[\db],q')\in \cert$.

 We first prove that $f$ is injective. Suppose toward contradiction that $f(A)= f(B)$ for some $A\neq B$ from $\db$. 
If $A=R(\underline{a_1},a_2, \dots ,a_n) $ and  $B=R(\underline{b_1},b_2, \dots ,b_n) $, where $R(\underline{v_1},v_2, \dots ,v_n)\in q$, then for some $i\in [n]$ we have $v_i=x$, $a_i \neq b_i$ such that $S(\underline{a_i},c,\uk),S(\underline{b_i},c,\uk)\in \db$ for some $c$.
 By $q$-saturation it holds that one of $T(\underline{c},a_i,\uk)$ or $T(\underline{c},b_i,\uk)$ does not belong to $\db$. However, this contradicts another criterion of $q$-saturation, pairwise consistency, which entails that any $S$-fact can be joined with some $T$-fact. From the contradiction, it follows that $f$ is injective.  
 
 By an analogous reasoning, if $f(A)$ and $ f(B)$ are key-equal, then $A$ and $B$ are key-equal. Hence a subset $\db'$ of $\db$ forms a block of $\db$ if and only if $f[\db']$ forms a block of $f[\db]$.

($\Rightarrow$) Suppose $(\db,q)\in \cert$, and $\rep'$ be a repair of $f[\db]$. Since $f$ is injective, it has an inverse function $f^{-1}\colon f[\db]\to \db$. Define $\rep \defeq f^{-1}[\rep']$. By our arguments above, $\rep$ is a repair of $\db$. By assumption there exists a valuation $\theta$ such that $\theta(q)\subseteq \rep$.   We claim that $\theta(q')\subseteq \rep'$. Indeed, using $S(\underline{\theta(x)},\theta(y),\uk)\in \db$,  if $\theta(R)\in \rep$, then 
 $\theta(R_{x \mapsto y})= f(\theta(R))\in f(\rep)=(f \circ f^{-1})[\rep']=\rep'$. Hence the claim holds, leading to $(f[\db],q')\in \cert$.

($\Leftarrow$) Suppose $(f[\db],q')\in \cert$, and let $\rep$ be a repair of $\db$. Again, we see that $f[\rep]$ is a repair of $f[\db]$, hence we find a valuation $\theta$ such that $\theta(q')\subseteq f[\rep]$. In particular, since $S(\underline{y},y,\uk)\in q'$, we have $S(\underline{\theta(y)},\theta(y),\uk)\in f[\rep]$. Now, consider the fact $S(\underline{a},\theta(y),\uk) = f^{-1}(S(\underline{\theta(y)},\theta(y),\uk))\in (f^{-1}\circ f)[\rep]=\rep$. Let $\theta'$ be otherwise as $\theta$, except that it maps $x$ to $a$.
Then, for any $R\in q$, we obtain $\theta'(R)=f^{-1}(\theta(R_{x \mapsto y}))$. Since $\theta(R_{x \mapsto y})\in f[\rep]$, this leads to $\theta'(R)\in(f^{-1}\circ f)[\rep]= \rep$. Hence $\theta'(q)\subseteq \rep$, leading to $(\db,q)\in \cert$.
\end{proof}

  \subsubsection{Removing Repeated Variables}
  
  For an atom
\(
F=R(\underline{t_1},t_2,\ldots,t_n),
\)
let $\norep{F}$ be obtained by deleting every argument
position \(i\) such that \(t_i\) is a variable and
\(t_i=t_j\) for some \(j<i\).
Then, let $\norep{q}\defeq \{\norep{F}\mid F\in q\}$.
Also, given a database $\db$ and query $q$, 
 let $\db_0$ be obtained from $\db$ by dropping those columns that correspond to the positions dropped in the reduction $q\mapsto \norep{q}$.
(Formally, $\norep{*}$  changes also relation names when the arity changes.)
Then, we set $\removereps(\db,q)\defeq (\db',q')$, where $q'\defeq \norep{q}$ and $\db'\defeq \dbsaturate(\db_0,q')$.


Recall the notion of the size-monotone reduction from \Cref{sect:controlled}.
\smallskip
 \begin{lemma}\label{lem:repfreereal}
  Let $\mathcal{C}$ (resp. $\mathcal{C}^=$) be the set of pairs $(\db,q)$, where $q$ is a query from $\sjfuabcq$ (resp. from  $\sjfuabcq^=$), and $\db$ is a $q$-saturated database. 
  Then 
 $(\db,q) \mapsto \removereps(\db,q)$ is a size-monotone, generator-monotone, equivalence-preserving,
polynomial-time reduction from $\mathcal{C}^=$ to  $\mathcal{C} \cup \{\true,\false\}$.
Additionally, if $q$ is saturated and satisfies \Cref{it:normal1,it:normal2}, then so does $q'$.
 %
%
\end{lemma}
\begin{proof}
We provide a sketch and leave the details to the reader. 
Since $\norep{*}$ only removes self-loops from $\FD{q}$, the reduction is generator-monotone.
Clearly, it is also size-monotone and in polynomial time.

Assume $(\db',q')=  \removereps(\db,q)\notin \{\true,\false\}$.
Clearly, $q'$ is variable-repetition-free. Also, the attack graphs of $q$ and $q'$ are identical (modulo changes in the relation names).
That saturation and  satisfaction of \Cref{it:normal1,it:normal2} is preserved is straightforward to verify.
Moreover, $\db'$ is $q'$-saturated by \Cref{lem:dbsat}.
 Hence $(\db',q')$ belongs to $\mathcal{C} \cup \{\true,\false\}$.
 
 For Equivalence preservation, observe that due to $q$-saturation of $\db$, any $R$-fact $A\in \db$ contains repeated constants in those positions where the corresponding $R$-atom of $q$ contains repeated variables. Therefore, $(\db,q)\in \cert$ if and only if $(\db_0,q')\in \cert$. It then follows by \Cref{lem:dbsat} that the reduction is equivalent-preserving.
\end{proof}

  \subsection{Fourth Condition and Substitutable Sets}
This section proceeds in two steps. First, we provide some auxiliary  concepts and lemmas. 
One such concept is that of a substitutable set of variables. Such a set can be replaced with constants in a way that is safe with respect to  the overall reduction.
 
 \subsubsection{Auxiliary  Concepts and Lemmas}\label{sect:aux}
 
 We use the following result, which states that replacing a variable with a constant in a query does not introduce new attacks.
This result is a minor variant of \cite[Lemma~3.7]{KoutrisW17} (which includes also composite primary keys) and follows by the same argument.
\begin{lemma}[\cite{KoutrisW17}]\label{lem:replaceconstant2}
Let $q\in \sjfuabcq$, let $x\in \queryvars{q}$, and let $c$ be an arbitrary constant. 
If  $F_{x\mapsto c} \attacks{q_{x\mapsto c}} G_{x\mapsto c}$, then $F \attacks{q} G$.
Consequently,  $q_{x\mapsto c}\in \sjfuabcq$.
\end{lemma}

 \medskip

\begin{algorithm}[t]
\caption{$\textsc{GetFunction}(\db,q)$ \label{alg:substitute}} 
\Input{A pair $(\db,q)$, where $q\in\sjfuabcq$ and  $\db$ is a $q$-saturated database}
\Output{A function $f\colon \outvar{\SFD{q}}{\emp} \to \adom{\db}$}
\ForEach{$y\in \outvar{\SFD{q}}{\emp}$}{
    Select the $\prec$-least atom of the form $S(\underline{c},y,\uk)\in q$, where $c$ is a constant\;
   Set $f(y)\defeq b$, where $b$ is the unique constant such that $S(\underline{c},b,\uk)\in \db$\;
}
\Return{$f$} 
\end{algorithm}

The following lemma concerning edge preservation upon substitution is also needed. Note that we only replace sets of variables that are closed under the consistent edges of the key-nonkey graph.
\smallskip
   \begin{lemma}\label{claim:notsimple0}
   Let $q\in \sjfuabcq$ be saturated.
Let $Y\subseteq \queryvars{q}$ be such that {$\outvar{\SFD{q}}{\emp}\subseteq Y$} and $\outvar{\SFD{q}}{Y} \subseteq Y$,
and let $f\colon Y \to \Const$ be a function.
Define $q'\defeq q_{f}$.
Let $u,v\in \queryvars{q} \cup\{\emp\}$ be such that $v\notin Y$. Then
\begin{enumerate}
\item\label{it:not1} $\fd{u}{v}\in \FD{q}$ implies $\fd{u'}{v}\in \FD{q'}$;
\item\label{it:not2} $\fd{u}{v}\in \WFD{q}$ implies $\fd{u'}{v}\in \WFD{q'}$;
\item\label{it:not3} $\fd{u}{v}\in \SFD{q}$ if and only if $\fd{u'}{v}\in \SFD{q'}$;
\end{enumerate}
where $u'\defeq \emp$ if $u\in Y$, and otherwise  $u'\defeq u$. Consequently, if additionally $u\notin Y$, then 
$\fd{u}{v}\in \SFD{q}$ if and only if $\fd{u}{v}\in \SFD{q'}$.
\end{lemma}
\begin{proof}
\Cref{it:not1} is clear, and \Cref{it:not2} follows from Items \ref{it:not1} and \ref{it:not3}. Thus, we prove the last item.
Let us write $F'$ for the atom $F_{Y\mapsto f(Y)}$, for each $F\in q$.

($\Rightarrow$) 
 Suppose $\fd{u}{v}\in \SFD{q}$. Then we find $F\in q$ and $G\in q\setminus\{F\}$ such that $\fd{u}{v}\in \FD{F}$ and $\fd{u}{v}\in \FDclosure{G}$. Clearly, $\fd{u'}{v}\in \FD{F'}$ and $\fd{u'}{v}\in \FDclosure{G'}$.
Hence $\fd{u'}{v}\in \SFD{q'}$.
 
($\Leftarrow$)  
Suppose $\fd{u'}{v}\in \SFD{q'}$. Then, there exists $F'\in q$ such that $\fd{u'}{v}\in \FD{F'}$ and $G'\in q'\setminus \{F'\}$ such that $\fd{u'}{v}\in \FDclosure{G'}$.
Consequently, there exists $p$ such that $\fd{p}{v}\in \FD{F}$ and $p'=u'$. In particular, it holds that $\fd{p'}{v}\in \FDclosure{G'}$.

Suppose toward contradiction that $p\in Y$. Since $\outvar{\SFD{q}}{Y} \subseteq Y$ and $v\notin Y$, we obtain $\fd{p}{v}\in \WFD{q}$. 
If $\fd{p}{v}\in \FDclosure{G}$, we obtain 
$\fd{p}{v}\in \SFD{q}$, a contradiction. Thus, $\fd{p}{v}\notin \FDclosure{G}$, which entails  $\fd{y}{v} \in \FD{G}$ for some $y\in Y$
such that $y\neq p$. But then, $\fd{y}{v}\in \WFD{q}$, as otherwise by $\outvar{\SFD{q}}{Y} \subseteq Y$ we would obtain $v\in Y$, a contradiction. However, since $p\neq y$, we obtain then a contradiction with \Cref{lem:samesourcesat}. Hence we conclude that $p\notin Y$.

Assume toward contradiction that $p=\emp$. Then also $u'=p'=\emp$. Thus, we obtain either $\fd{\emp}{v}\in \FD{G} $ or, for some $y\in Y$, $\fd{y}{v}\in \FD{G} $. Since $\fd{\emp}{v}\in \FD{F}$, both contradict $v\notin Y$, {$\outvar{\SFD{q}}{\emp}\subseteq Y$}, and $\outvar{\SFD{q}}{Y} \subseteq Y$. Hence $p\neq \emp$.

It follows that $p$ is a variable not from $Y$. Consequently, $p=p'=u'=u$, hence $\fd{u}{v}\in \FD{F}$.
If $\fd{u}{v}\in \FDclosure{G}$, then $\fd{u}{v}\in \SFD{q}$, as required. Otherwise, if $\fd{u}{v}\notin \FDclosure{G}$, we obtain $\fd{y}{v}\in\FD{G}$, for some $y\in Y$. Since $v\notin Y$ and $\outvar{\SFD{q}}{Y} \subseteq Y$, we obtain $\fd{y}{v}\in\WFD{G}$. Then, as $u\neq y$,  \Cref{lem:samesourcesat} implies $\fd{u}{v}\in \SFD{q}$, as required. 
\end{proof}

\bigskip

\begin{definition}[Substitutable Sets]\label{def:sub}
      Let $q\in \sjfuabcq$ be saturated query such that it satisfies \Cref{it:normal1,it:normal2,it:normal4} of \Cref{def:normal}. We say that $Y\subseteq \queryvars{q}$ is \emph{substitutable} if
    \begin{enumerate}
    \item\label{it:sub1} $  Y\neq  \emptyset$;
    \item\label{it:sub2} $\outvar{\SFD{q}}{\emp}\subseteq Y$ and $\outvar{\SFD{q}}{Y} \subseteq Y$;
    \item\label{it:sub3} either of the following holds
 \begin{enumerate}[(i)]
\item\label{it:ii} $Y=\{y\}$, where $y$ is not attacked in the attack graph of $q$, or
\item\label{it:iii}  $Y= \bigcup \{\notkeyvars{H} \mid H\in q_0,\keyvars{H}=\emp\}$, for some $q_0\subseteq q$.
\end{enumerate}
\end{enumerate}
\end{definition}

\bigskip

Given an instance $(\db,q)$, a substitutable set
$Y\subseteq \queryvars{q}$, and a substitution
$f:Y\to\adom{\db}$, define
\begin{equation}\label{eq:subf}
\textsc{Substitute}_f(\db,q)
  \defeq
  (\dbsaturate(\db,q_f),q_f).
\end{equation}

\bigskip

\begin{lemma}\label{lem:guarantees}
Let $q\in \sjfuabcq$ be saturated query such that it satisfies \Cref{it:normal1,it:normal2,it:normal4} of \Cref{def:normal}. 
If $\outvar{\SFD{q}}{\emp}$ is non-empty, then it is substitutable. 
\end{lemma}
\begin{proof}
Let us write $Y\defeq \outvar{\SFD{q}}{\emp}$.
 
 Considering \Cref{it:sub2} of \Cref{def:sub}, for any $z\in \outvar{\SFD{q}}{Y} $ there exists $y\in Y$ such that $\fd{y}{z}\in \SFD{q}$. Furthermore, $\fd{\emp}{y}\in \SFD{q}$. By saturation it follows then that $\fd{\emp}{z}\in \SFD{q}$, i.e., $z\in Y$.
Hence $\outvar{\SFD{q}}{Y} \subseteq Y$.

 Considering \Cref{it:sub3},  \Cref{it:normal1} of \Cref{def:normal}  implies that each $\fd{\emp}{y}\in \SFD{q}$ is witnessed by an atom $F_y\in q$ such that $\keyvars{F_y}=\emp$ and $\notkeyvars{F_y}=\{y\}$. We can then define $q_0\defeq \{F_y \mid y\in Y\}$ to obtain \Cref{it:iii}.

Hence, $\outvar{\SFD{q}}{\emp}$ is substitutable whenever it is non-empty.
\end{proof}
\medskip

\subsubsection{Substituting Constants for Variables}

We can now consider the fourth normalization condition, which stipulates that the special vertex $\emp$ has no outgoing consistent edges. The following lemma, especially its first item, is used to replace the target variables of such edges with constants. Additionally, the second item is used later when variables are replaced with constants in the alternating computation that follows the rewriting strategy; there, Equivalence preservation is not needed.

\smallskip
 \begin{lemma}[Normality, \Cref{it:normal3}]\label{lem:removecfd3}
 Let  $\mathcal{C}$ be the set of pairs $(\db,q)$, where $q$ is a saturated query from  $\sjfuabcq$ such that it satisfies \Cref{it:normal1,it:normal2,it:normal4} of \Cref{def:normal}, and $\db$ is a database. 
 Let $\mathcal{C}^*$ be defined as $\mathcal{C}$, except that it requires $\db$ to be $q$-saturated.
 Then, for each $(\db,q)\in \mathcal{C}$, substitutable $Y\subseteq \queryvars{q}$, and function $f\colon Y\to \Const$, 
 $(\db,q) \mapsto \Substitute_f(\db,q)$ is a strict, size-monotone, generator-monotone,
polynomial-time reduction from $\mathcal{C}$ to $\mathcal{C}^* \cup\{\true,\false\}$.
 Additionally, the following claims hold:
 \begin{enumerate} 
 \item\label{it:add1}
  If $\db$ is $q$-saturated and $f\defeq \textsc{GetFunction}(\db,q)$, then $\Substitute_f$ is equivalence-preserving, 
\item \label{it:add2}
The statement remains true if \Cref{it:normal3} of \Cref{def:normal} is additionally assumed for the queries \(q\) with \((\db,q)\in\mathcal C\).
\end{enumerate}
\end{lemma}
\begin{proof}
Given an atom $F\in q$, let us write $F'$ for its corresponding atom in $q'$, that is, $F'=F_{f}$.

Strictness  follows by non-emptiness  of $Y$ and generator-monotonicity  by \Cref{lem:replaceconstant3}.
Moreover, the reduction is clearly size-monotone, and it is in polynomial time due to \Cref{lem:dbsat}. 
 It remains to prove the additional statements and the claim that the reduction is into $\mathcal{C}^* \cup \{\true,\false\}$.
For this, suppose $(\db',q') = \Substitute_f(\db,q) \notin \{\true,\false\}$.

{ \bf
(Membership of $(\db',q')$ in $\mathcal{C}^*$)
 }
 First, observe that $\db'$ is $q'$-saturated by \Cref{lem:dbsat}.  Also, it is clear that $q'$ remains variable-repetition-free, hence
 by \Cref{lem:replaceconstant2} we obtain $q'\in \sjfuabcq$.
 
 We divide the remaining parts into several claims.
   \begin{claim}\label{claim:issat}
$q'$ is saturated.
  \end{claim}
 \begin{claimproof}
 Let $F\in q$ be such that $\fd{u}{v}\in \FD{F'}$, $\fd{v}{w}\in \SFD{q'}$, and $\fd{u}{w}\in \FDclosure{q'\setminus \{F'\}}$. We need to show that $\fd{u}{w}\in \SFD{q'}$. 

 We observe that $u$ is a variable or $\emp$, and $v,w$ are variables. 
 Since $Y\cap (\queryvars{q'} \cup\{\emp\})= \emptyset$, we have $u,v,w\notin Y$. 
  \begin{claim}
$\fd{v}{w}\in \SFD{q}$.
  \end{claim}
  \begin{claimproof}
  By $\fd{v}{w}\in \SFD{q'}$ we find $F'\in q'$ and $G'\in q'\setminus \{F'\}$ such that $\fd{v}{w}\in \FD{F'}$ and $\fd{v}{w}\in \FDclosure{G'}$. Clearly, $\fd{v}{w}\in \FD{F}$, and if additionally $\fd{v}{w}\in \FDclosure{G}$, then the claim follows. Otherwise, if $\fd{v}{w}\notin \FDclosure{G}$, we find $y\in Y$ such that $\fd{y}{w}\in \FD{G}$, leading to $\fd{\emp}{w}\in \FD{G'}$ and furthermore to $\fd{v}{w}\in \FDclosure{G'}$. However, since $Y$ is closed under the edges of $\SFD{q}$, and $w\notin Y$, we obtain $\fd{y}{w}\in \WFD{q}$. Since $v\notin Y$, we have $v\neq y$. Hence, by \Cref{lem:samesourcesat}, it follows that $\fd{v}{w}\in \SFD{q}$, which was our claim.
  \end{claimproof}
  
 Let $u'=\keyvars{F}$. Then $\fd{u'}{v}\in \FD{F}$, and $u'\neq u$ if and only if $u'\in Y$. In what follows, we prove that $u'\notin Y$.
 %
Toward this, we first prove the following claim.
 \begin{claim}
 $\fd{u'}{w}\in \FDclosure{q\setminus \{F\}}$.
  \end{claim}
  \begin{claimproof}
 Let $D=(x_1, \dots ,x_{n-1},w)$ be a simple path witnessing $\fd{u}{w}\in \FDclosure{q'\setminus \{F'\}}$. In particular, we have that $x_1=u$ or $x_1=\emp$ (or both). Since both $u$ and $\emp$ are not from $Y$, also $x_1$ is not from $Y$. Since the remaining elements of $Y$ are variables from $\queryvars{q'}$, they are not from $Y$ either.
 Hence the suffix $D[x_2,w]$ belongs to $\FD{q\setminus \{F\}}$. 
 
 Consider the atom $G'\in q'\setminus \{ F'\}$ such that $\fd{x_1}{x_2}\in \FD{G'}$.
 If $\keyvars{G}\notin Y$, then $\fd{x_1}{x_2}\in \FD{G}$, i.e., $\fd{x_1}{x_2}\in \FD{q\setminus \{F\}}$, and consequently, it can be confirmed that $\fd{u'}{w}\in \FDclosure{q\setminus \{F\}}$, as required
 
 Thus, suppose $\keyvars{G}=y \in Y$. Then $\fd{y}{x_2} \in \FD{q\setminus \{F\}}$, and in particular, 
 \[D' \defeq (y, x_2,\dots ,x_{n-1},w)\]
  is a simple path witnessing  $\fd{y}{w}\in \FDclosure{q\setminus \{F\}}$. 
 
 We claim that $F\nattacks{q} y$.
 In the case of \cref{it:ii} this is immediate. In the case of \cref{it:iii}, we find  $ H\in q$ such that $y\in \notkeyvars{H}\subseteq Y$ and $\keyvars{H}=\emp$. We
 observe $F\neq H$ because $\FD{F'}$ is non-empty but $\FD{H'}$ is empty. Since $\keyvars{H}=\emp$, we have $y\in  \keyclosure{F}{q}$, hence $F\nattacks{q} y$. Thus the claim that $F\nattacks{q} y$ holds.
 
 Next, we claim that $F\nattacks{q} w$. Toward contradiction, suppose $F\attacks{q} w$. Then,  
 since $F\nattacks{q} y$ and $F\attacks{q} w$, we obtain that $y$ and $w$ are separated by $\keyclosure{F}{q}$. In particular, as $D'$ is also a path in $\gaifman{q}$, some element of $D'$ belongs to $\keyclosure{F}{q}$. Consequently, $w$ belongs to $\keyclosure{F}{q}$, as $D'$ is a path in $\FD{q\setminus \{F\}}$. This, however, contradicts $F\attacks{q} w$, so the claim follows.
 
  Now, 
  $\{v,w\}$ is an undirected (possibly singleton) edge generated by $\fd{v}{w}\in \SFD{q}$, and
  also we have $\fd{u'}{v}\in \FD{F}$. Therefore, we obtain from $F\nattacks{q} w$ that one of $v$ or $w$ belongs to
  $\keyclosure{F}{q}$. Moreover, if $v$ belongs to this set, then by $\SFD{q} \subseteq \FD{q\setminus\{F\}}$
  also so does $w$. We conclude that $w\in \keyclosure{F}{q}$, which is tantamount to the claim that
  $\fd{u'}{w}\in \FDclosure{q\setminus \{F\}}$.
  
   This concludes the proof of the claim. 
     \end{claimproof}

 From $\fd{u'}{v}\in \FD{F}$,  $\fd{v}{w}\in \SFD{q}$, and $\fd{u'}{w}\in \FDclosure{q\setminus \{F\}}$, we obtain by saturation of $q$ that
  $\fd{u'}{w}\in \SFD{q}$. By the same token $u'$ is a variable.
  
  It follows that $u'\notin Y$, because $w \notin Y$ and $Y$ is closed under $\SFD{q}$. Hence $u=u'$, and therefore by $\fd{u'}{w}\in \SFD{q}$ we obtain $\fd{u}{w}\in \SFD{q}$. Since $u,w\notin Y$, it follows by \Cref{claim:notsimple0} that $\fd{u}{w}\in \SFD{q'}$, as required. This concludes the proof of \Cref{claim:issat}.
\end{claimproof}

    \begin{claim}\label{claim:isnorm}
$q'$ satisfies \Cref{it:normal1,it:normal2,it:normal4} of \Cref{def:normal}.
  \end{claim}
 \begin{claimproof}
 Note that we cannot assume \cref{it:normal3} of normality for $q$ in the proof.

 (\Cref{it:normal1}) Let $F'\in q'$ be such that $|\notkeyvars{F'}|>1$. Let $p= \keyvars{F'}$. We need to show that $\fd{p}{v} \in \WFD{q'}$, for each $v\in \notkeyvars{F'}$. Suppose toward contradiction that $\fd{p}{v} \in \SFD{q'}$ for some $v\in \notkeyvars{F'}$. 
 Note that in this case $v\notin Y$.
 Then we find $G'\in q'\setminus \{F'\}$ such that $\fd{p}{v} \in \FDclosure{G'}$. Note that $\notkeyvars{F'}\subseteq \notkeyvars{F}$, i.e., $|\notkeyvars{F}|>1$. Furthermore, we clearly have $v\in \notkeyvars{F}$.
 
 We claim that $p=\emp$. Toward contradiction, suppose $p$ is a variable. Then $p= \keyvars{F}$,
 and $\fd{p}{v} \in \FDclosure{G}$, leading to $\fd{p}{v} \in \SFD{q}$, which contradicts \cref{it:normal1} of normality for $q$. Thus $p=\emp$. 
 This entails $\keyvars{F}\in Y\cup\{\emp\}$ and $\keyvars{G}\in Y\cup\{\emp\}$.
 
 If  $\keyvars{G}=\emp$, then $\fd{\keyvars{F}}{v} \in \SFD{q}$, a contradiction with \cref{it:normal1} of normality for $q$.

 Suppose then $\keyvars{G}=y \in Y$. If $\keyvars{F}=y$, we obtain a contradiction with \cref{it:normal1} as above. Thus suppose $\keyvars{F}\neq y$. 
 Then \Cref{lem:samesourcesat} entails $\fd{\keyvars{F}}{v} \in \SFD{q}$ or $\fd{y}{v} \in \SFD{q}$. 
 The former case contradicts normality of $q$, and
 the latter case is not possible due to $v\notin Y$. 
 Thus, by contradiction, we obtain $\fd{p}{v} \in \WFD{q'}$, for each $v\in \notkeyvars{F'}$.
 
  (\Cref{it:normal2}) Assume toward contradiction that $\fd{u}{v} \in \SFD{q'}$, for some distinct $u,v\in \notkeyvars{F'}$ and $F'\in q'$.
  Then $u,v\in \notkeyvars{F}$ where $F\in q$. Clearly we have $\fd{u}{v} \in \FD{q}$. 
  Denote $p=\keyvars{F}$, in which case we have $\fd{p}{v}\in \FD{F}$. Since $q$ satisfies \cref{it:normal1} of normality, and $|\notkeyvars{F}|>1$, it follows that $\fd{p}{v}\in \WFD{F}$.
  Hence, by \Cref{lem:samesourcesat}, and since $p\neq u$ (as $q$ is variable-repetition-free), it cannot be the case that $\fd{u}{v} \in \WFD{q}$. Therefore, $\fd{u}{v} \in \SFD{q}$, but this contradicts \cref{it:normal2} of normality for $q$. Thus the assumption is false, and \cref{it:normal2} follows.

  (\Cref{it:normal4}) Assume toward contradiction that $\fd{u}{v} \in \SFD{q'}$ and $\fd{v}{u} \in \SFD{q'}$, for two distinct variables $u,v\in \queryvars{q'}$.   Then \Cref{claim:notsimple0} implies $\fd{u}{v} \in \SFD{q}$ and $\fd{v}{u} \in \SFD{q}$, in contradiction with \Cref{it:normal4} of normality for $q$.
      \end{claimproof}

In the following, we consider the additional statements.

{\bf (\Cref{it:add1})}
Fix $f\defeq \textsc{GetFunction}(\db,q)$.
 Since $\db$ is $q$-saturated, we have that if $\theta$ is a valuation such that $\theta(q)\subseteq \db$, then $\restrict{\theta}{Y}=f$. Hence,  $(\db,q)\mapsto (\db,q_{f})$ is equivalence-preserving. Since Equivalence preservation holds for database saturation and is preserved under composition, $(\db,q) \mapsto \Substitute_f(\db,q)$ is equivalence-preserving.


  {\bf (\Cref{it:add2})}
Assume toward contradiction that $\fd{\emp}{v} \in \SFD{q'}$ for some $v\in \queryvars{q'}$. Let $F',G'\in q'$ be such that $\fd{\emp}{v} \in \FD{F'} \cap\FD{G'}$. Since $q$ satisfies \cref{it:normal3} of normality, we have $\fd{\emp}{v} \notin \FD{F} \cap\FD{G}$. Thus $\keyvars{F}\in Y$ or $\keyvars{G}\in Y$. By symmetry, we may assume that $\keyvars{F}=y \in Y$. If $\keyvars{G}\in \{\emp,y\}$, we obtain $\fd{y}{v} \in \SFD{q}$, which leads to $v \in Y$ by \Cref{it:sub2} of \Cref{def:sub}, a contradiction with $v\in \queryvars{q'}$. Since $\keyvars{G}\subseteq \{\bot\}\cup Y$, we may conclude that $\keyvars{G}=y'\in Y\setminus \{y\}$. Summarizing, we have that $\fd{y}{v},\fd{y'}{v}\in \FD{q}$, where $y\neq y'$. 
By \Cref{lem:samesourcesat} it follows that $\fd{y}{v}\in \SFD{q}$ or $\fd{y'}{v}\in \SFD{q}$, which again leads to $v \in Y$, a contradiction. We conclude by contradiction that $\fd{\emp}{v} \notin \SFD{q'}$ for each $v\in \queryvars{q'}$. Hence \cref{it:normal3} holds.
\end{proof}

 \subsection{Normalization Step}
We can now compose all the individual normalization steps into one lemma. Note that $\Normalize$ is the function described in \Cref{alg:pre}.
\smallskip

\normal*
\begin{proof}
The equivalence preservation and generator-monotonicity follow from \Cref{prop:eqcompose} and from exhaustively applying the reductions of \Cref{lem:removecfd,lem:removecfd2,lem:removecfd4,lem:removecfd3} sequentially in this order, as described in \Cref{alg:pre}.
(After \Cref{lem:removecfd4}, note that \Cref{lem:repfreereal} is applied to remove variable repetitions.)
The lemmas also guarantee that the output query $q'$ is normal and from $\mathcal{C}$.
It remains to show that the computation runs in polynomial time. 

 Because the size of the query shrinks at each iteration (due to strictness of each reduction used in the lemmas), the total number of iterations is at most $\size{q}$. Moreover, the size of the database does not increase at any iteration (due to size-monotonizity of each reduction). By the lemmas, each iteration is in polynomial time over the temporary input, so each iteration is in polynomial time over the input. Hence, the overall computation is in polynomial time.

From these considerations, it can be seen that the lemma statement holds.
 \end{proof}
 
 \begin{algorithm}
\caption{$\Normalize(\db,q)$ \label{alg:pre}}
\Input{A saturated query $q\in \sjfuabcq$ and a $q$-saturated database $\db$.}
\Output{$\false$ or $(\db',q')$, where $\db'$ is a saturated database and $q'$ a saturated and normal query from $ \sjfuabcq$.}
    \While{
    some $F\in q,y\in \notkeyvars{F}$ witness that \Cref{it:normal1} of normality is false for $q$
   }{
        $(\db,q)\gets \Try(\normalizeone_{F,y}(\db,q))$\;
        \tcp{\Cref{lem:removecfd}}
    }

    \While{
        some $F\in q,y\in \notkeyvars{F}$ witness that \Cref{it:normal2} of normality is false for $q$
}{ 
        $(\db,q)\gets \Try(\normalizetwo_{F,y}(\db,q))$\;
        \tcp{\Cref{lem:removecfd2}}
}

      \While{
       some $x,y\in \queryvars{q}$ witness that \Cref{it:normal4} of normality is  false for $q$
      }{ 
        $(\db,q)\gets \Try(\identify_{x,y}(\db,q))$\;
        \tcp{\Cref{lem:removecfd4}}
        $(\db,q)\gets \Try(\removereps(\db,q))$\;
       \tcp{\Cref{lem:repfreereal}}
    }

    \While{\Cref{it:normal3} of normality is  false for $q$
    }{
        $f \gets \textsc{GetFunction}(\db,q)$\;
        \tcp{\Cref{lem:guarantees} entails the assumptions for  \Cref{lem:removecfd3} w.r.t. $Y\defeq \outvar{\SFD{q}}{\emp}$}
                $(\db,q)\gets \Try(\Substitute_f(\db,q))$\;
        \tcp{\Cref{lem:removecfd3}}
    }


\Return{$(\db,q)$}\;
\end{algorithm}

\subsection{Basic Properties of Saturated and Normal Queries}
After the preprocessing stage, the query $q$ is saturated and normal, and the acyclicity of its attack graph has been preserved.
In this section, we consider some basic properties of such queries.

We say that two (not necessarily distinct) variables $x$ and $y$ are \emph{siblings (w.r.t. $q$)}, denoted $x\sib_q y$ (or $x\sib y$ if $q$ is understood) if $x=y$ or there exists an atom $F\in q$ such that $x$ and $y$ are both in $\notkeyvars{F}$.

\begin{figure}
\centering
\begin{tabular}{c c c}
\begin{tikzpicture}[
    baseline=(current bounding box.center),
    dot/.style={circle, draw, inner sep=1.2pt, minimum size=16pt}
]
    \node[dot] (x)  at (0,0) {$x$};
    \node[dot] (y)  at (0,1.5) {$y$};
    \node[dot] (xp) at (2,0) {$u$};
    \node[dot] (yp) at (2,1.5) {$v$};

    \sarr{x}{y}
    \sarr{xp}{yp}
    
        \node[levelbox, fit=(y)(yp)] {};

\end{tikzpicture}
&
\hspace{.8cm} $\implies$ \hspace{.8cm}
&
\begin{tikzpicture}[
    baseline=(current bounding box.center),
    dot/.style={circle, draw, inner sep=1.2pt, minimum size=16pt}
]
    \node[dot] (x)  at (1,0) {$x$};
    \node (u)  at (1.9,0) {$(=u)$};
    \node[dot] (y)  at (0,1.5) {$y$};
    \node[dot] (yp) at (2,1.5) {$v$};

    \sarr{x}{y}
    \sarr{x}{yp}
    
        \node[levelbox, fit=(y)(yp)] {};

\end{tikzpicture}
\end{tabular}
\caption{Illustration for \Cref{lem:samesource}\label{fig:illustration1}}
\end{figure}

\smallskip
\begin{lemma}\label{lem:samesource}
Let $q\in \sjfuabcq$ be saturated and normal. If $\fd{x}{y},\fd{u}{v}\in \WFD{q}$, where $y\sib v$, then $x=u$.
\end{lemma}
\begin{proof}
The case where $y=v$  is covered by \Cref{lem:samesourcesat}.

Consider then the case where $y\neq v$ (see \Cref{fig:illustration1}). Suppose $\fd{x}{y},\fd{u}{v}\in \WFD{q}$ and $y \sib v$. Let $F\in q$ be such that $y,v\in \notkeyvars{F}$. 
Let $w=\keyvars{F}$. By normality of $q$, we have $\fd{w}{y},\fd{w}{v}\in \WFD{q}$. By \Cref{lem:samesourcesat}, we obtain $x=w=u$.
\end{proof}

\smallskip
\begin{lemma}\label{lem:nextattack2}
Let $q\in \sjfuabcq$ be saturated and normal. Suppose $\fd{x}{y}\in \WFD{q}$, $\fd{y}{z}\in \SFD{q}$, and $\fd{x}{z}\notin \SFD{q}$.
If $F\in q$ is the atom generating $\fd{x}{y}$, then $F\attacks{}z$.
\end{lemma}
\begin{proof}
The case where $y=z$ is covered by \Cref{lem:nextattack}. Thus, we have $y\neq z$. It follows that
 $y$ and $z$ cannot be siblings, because otherwise we obtain a contradiction with the normality of $q$.
Suppose toward contradiction that $F\not\attacks{}z$. Since $F\attacks{}y$ by \Cref{lem:nextattack}, we obtain  $\fd{x}{y}\notin \FDclosure{q\setminus \{F\}}$. As $\{y,z\}$ is an edge in $\gaifman{q}$, we obtain $\fd{x}{z}\in \FDclosure{q\setminus \{F\}}$. By saturation we obtain $\fd{x}{z}\in \SFD{q}$, which contradicts the assumption. Hence $F\attacks{}z$.
\end{proof}

\smallskip
\begin{lemma}\label{lem:shortestpath}
Let $q\in \sjfuabcq$ be saturated and normal. 
Let $D=x_1, \dots ,x_n$ be a path of $\FD{q}$ such that $x_1=x$, $x_n=y$, and $\fd{x_1}{x_2}\in \WFD{q}$. 
Assume that $D$ is shortest among all paths from $x_1$ to $x_n$ in $\FD{q}$.
Then $x_n\in \reach{x_1}{\igraph{q}}$.
\end{lemma}
\begin{proof}
By \Cref{lem:ctrans}, the path $D$ does not contain two consecutive edges from  $\SFD{q}$.
By definition, then, any edge $\fd{u}{v}\in \SFD{q}$ that appears on $D$ is preceded by an edge $\fd{w}{u}$ 
such that $\fd{w}{v}\notin \SFD{q}$, meaning that $\fd{u}{v}\in \igraph{q}$. The statement of the lemma follows from these observations.
\end{proof}

\begin{figure}
\centering
\begin{tabular}{c c c}
\begin{tikzpicture}[
    baseline=(current bounding box.center),
    dot/.style={circle, draw, inner sep=1.2pt, minimum size=16pt}
]
    \node[dot] (x) at (0,0) {$x$};
    \node[dot] (y) at (0,1.5) {$y$};
    \node[dot] (u) at (2,0) {$u$};
    \node[dot] (v) at (2,1.5) {$v$};

    \sarr{x}{y}
    \sarr{u}{v}
    \darr{y}{v}
\end{tikzpicture}
&
\hspace{.8cm} $\implies$ \hspace{.8cm}
&
\begin{tikzpicture}[
    baseline=(current bounding box.center),
    dot/.style={circle, draw, inner sep=1.2pt, minimum size=16pt}
]
    \node[dot] (x) at (0,0) {$x$};
    \node[dot] (y) at (0,1.5) {$y$};
    \node[dot] (u) at (2,0) {$u$};
    \node[dot] (v) at (2,1.5) {$v$};

    \sarr{x}{y}
    \sarr{u}{v}
    \darr{y}{v}
    \darr{x}{v}
\end{tikzpicture}
\end{tabular}
\caption{Illustration for \Cref{lem:swcollide}. \label{fig:illustration2}}
\end{figure}

\smallskip
\begin{lemma}\label{lem:swcollide}
Let $q\in \sjfuabcq$ be saturated and normal. 
Suppose
$\fd{x}{y},\fd{u}{v}\in \WFD{q}$
and  $\fd{y}{v}\in \SFD{q}$ where $y\neq v$. Then  $\fd{x}{v}$ belongs to $\SFD{q}$.
\end{lemma}
\begin{proof}
Note that $y\neq u$ by $\fd{u}{v}\in \WFD{q}$
and  $\fd{y}{v}\in \SFD{q}$. 

We first prove that $x\neq u$.
Toward contradiction, suppose  $x=u$. Then we have $x\neq y$ (by $y\neq u$) and $\fd{x}{y},\fd{x}{v}\in \WFD{q}$ and  $\fd{y}{v}\in \SFD{q}$.
Suppose first that there exists a single atom $F\in q$ such that  $y,v\in \notkeyvars{F}$. 
Since $y\neq v$, this contradicts normality of $q$ (specifically, \cref{it:normal2}). 
Then, suppose there exists two distinct $F,G\in q$ such that $\fd{x}{y}\in \FD{F}$ and $\fd{x}{v}\in \FD{G}$. In this case, 
as $\fd{y}{v}\in \SFD{q}$, the saturation of $q$ implies $\fd{x}{v}\in \SFD{q}$, a contradiction.
We conclude by contradiction that $x \neq u$ (see \Cref{fig:illustration2}).
%

Let $F\in q$ be the atom generating $\fd{u}{v}$. Since $\fd{y}{v}\in \SFD{q}$, there is an  edge $\{y,v\}$ in $\gaifman{q}$. By \Cref{lem:nextattack}, $F$ attacks $v$, and hence $\fd{u}{v}\notin \FDclosure{q \setminus \{F\}}$. By $\fd{y}{v}\in \SFD{q}\subseteq \FD{q \setminus \{F\}}$ we also observe that $\fd{u}{y}\notin \FDclosure{q \setminus \{F\}}$. In particular, $\{y,v\}$ is not separated by  $\keyclosure{F}{q}$. If $G$ is the atom generating $\fd{x}{y}$, we thus obtain $F\attacks{} G$. By acyclicity of the attack graph, we have $G\not\attacks{} F$.
By \Cref{lem:nextattack}, $G$ attacks $y$, and hence $\fd{x}{y}\notin \FDclosure{q \setminus \{G\}}$. However, due to the edge $\{y,v\}$ in $\gaifman{q}$, and since $v\in \atomvars{F}$, it must be the case $\fd{x}{v}\in \FDclosure{q \setminus \{G\}}$.  Since $\fd{x}{y}\in \FD{G}$, $\fd{y}{v}\in \SFD{q}$, it follows by saturation that $\fd{x}{v}\in \SFD{q}$.
\end{proof}

\begin{figure}
\centering
\begin{tabular}[t]{ccccc}
\begin{tikzpicture}[
    baseline=(current bounding box.center),
    dot/.style={circle, draw, inner sep=1.2pt, minimum size=16pt}
]
    \node[dot] (a) at (-1,0) {$x$};
    \node[dot] (b) at (-1,1.5) {$y$};
    \node[dot] (c) at (1,0) {$u$};
    \node[dot] (d) at (1,1.5) {$v$};
    \node[dot] (e) at (0,3) {$z$};
        \node[levelbox, fit=(b)] {};
        \node[levelbox, fit=(d)] {};

    \sarr{a}{b}
    \sarr{c}{d}
    \darr{b}{e}
    \darr{d}{e}
\end{tikzpicture}
&
\hspace{.8cm} $\implies$ \hspace{.8cm}
&
\begin{tikzpicture}[
    baseline=(current bounding box.center),
    dot/.style={circle, draw, inner sep=1.2pt, minimum size=16pt}
]    
    \node[dot] (a) at (-1,0) {$x$};
    \node[dot] (b) at (-1,1.5) {$y$};
    \node[dot] (c) at (1,0) {$u$};
    \node[dot] (d) at (1,1.5) {$v$};
    \node[dot] (e) at (0,3) {$z$};

    \sarr{a}{b}
    \sarr{c}{d}
    \darr{b}{e}
    \darr{d}{e}
    \darr{a}{e}
    
            \node[levelbox, fit=(b)] {};
        \node[levelbox, fit=(d)] {};

\end{tikzpicture}
&
\text{or}
&
\begin{tikzpicture}[
    baseline=(current bounding box.center),
    dot/.style={circle, draw, inner sep=1.2pt, minimum size=16pt}
]    
    \node[dot] (a) at (-1,0) {$x$};
    \node[dot] (b) at (-1,1.5) {$y$};
    \node[dot] (c) at (1,0) {$u$};
    \node[dot] (d) at (1,1.5) {$v$};
    \node[dot] (e) at (0,3) {$z$};

        \node[levelbox, fit=(b)] {};
        \node[levelbox, fit=(d)] {};

    \sarr{a}{b}
    \sarr{c}{d}
    \darr{b}{e}
    \darr{d}{e}
    \darr{c}{e}

\end{tikzpicture}
\\[2cm]
\begin{tikzpicture}[
    baseline=(current bounding box.center),
    dot/.style={circle, draw, inner sep=1.2pt, minimum size=16pt}
]
    \node[dot] (a) at (-1,0) {$x$};
    \node[dot] (b) at (-1,1.5) {$y$};
    \node[dot] (c) at (1,0) {$u$};
    \node[dot] (d) at (1,1.5) {$v$};
    \node[dot] (e) at (0,3) {$z$};

    \sarr{a}{b}
    \sarr{c}{d}
    \darr{b}{e}
    \darr{d}{e}
    
        \node[levelbox, fit=(b)(d)] {};

\end{tikzpicture}
&
\hspace{.8cm} $\implies$ \hspace{.8cm}
&
\begin{tikzpicture}[
    baseline=(current bounding box.center),
    dot/.style={circle, draw, inner sep=1.2pt, minimum size=16pt}
]    
    \node[dot] (a) at (0,0) {$x$};
    \node (u)  at (0.9,0) {$(=u)$};
    \node[dot] (b) at (-1,1.5) {$y$};
    \node[dot] (d) at (1,1.5) {$v$};
    \node[dot] (e) at (0,3) {$z$};

    \sarr{a}{b}
    \sarr{a}{d}
    \darr{b}{e}
    \darr{d}{e}
    
        \node[levelbox, fit=(b)(d)] {};

\end{tikzpicture}

\end{tabular}
\caption{Illustration for \Cref{lem:shortcut}.\label{fig:shortcut}}
\end{figure}

\smallskip
\begin{lemma}\label{lem:shortcut}
Let $q\in \sjfuabcq$ be saturated and normal. Let 
 $\fd{x}{y},\fd{u}{v}\in \WFD{q}$
and  $\fd{y}{z},\fd{v}{z}\in \SFD{q}$. Then one of the following holds:
\begin{itemize}
 \item $x=u$ and $y\sib v$, 
 \item $\fd{x}{z}\in \SFD{q}$,
 \item $\fd{u}{z}\in \SFD{q}$. 
 \end{itemize}
\end{lemma}
\begin{proof}
Note that $y\sib v$ implies $x=u$ by \Cref{lem:samesource} (see \Cref{fig:shortcut}). Thus, suppose $y\not\sib v$. 

It is possible that $y=z$ or $z=v$ (but not both). Then, by \Cref{lem:swcollide}, we obtain that  $\fd{x}{z}$ or $\fd{y}{z}$ is in $\SFD{q}$, respectively. Thus we may assume that $y\neq z$ and $z\neq v$.

By the assumption we find atoms of the form $F\defeq R(\underline{x},y,\uk),G\defeq R'(\underline{u},v,\uk), H\defeq S(\underline{y},z,\uk),I\defeq S'(\underline{v},z,\uk)\in {q}$. (Note that possibly $x=\emp$ or $y=\emp$; these cases are treated analogously.)

Since $y\not\sib v$, we have  $F \neq G$.
In particular, $y,z,v$ is a path in $\gaifman{q}$ connecting $\notkeyvars{F}$ and $\notkeyvars{G}$. Since the attack graph of $q$ is acyclic, one of $F$ or $G$ does not attack the other. By symmetry, we may assume $F$ does not attack $G$. Since $F\attacks{} y$ by \Cref{lem:nextattack}, it follows that at least one of $z$ or $v$ belongs to $\keyclosure{F}{q}$. In fact, we obtain $z\in \keyclosure{F}{q}$, because this is implied by $v\in \keyclosure{F}{q}$ and $\fd{v}{z}\in \SFD{q}$. Thus we have $\fd{x}{z}\in \FDclosure{q\setminus \{F\}}$, $\fd{x}{y}\in \FD{F}$, and $\fd{y}{z}\in \SFD{q}$. By the saturation of $q$, this leads to $\fd{x}{z}\in\SFD{q}$, showing the statement of the lemma.
\end{proof}

\medskip

The following lemma is a consequence of \Cref{lem:swcollide} and the definition of $ \igraph{q}$ (\Cref{def:attackprop}).
\smallskip
\begin{lemma}\label{lem:consecutive}
Let  $q\in \sjfuabcq$ be saturated and normal. The following statements hold:
\begin{enumerate}
\item\label{it:eksu} If $\fd{x}{y}$ and $\fd{y}{z}$ are in $ \igraph{q}$, then one of them belongs to $ \WFD{q}$.
\item\label{it:toksu} If $\fd{x}{y}$ is in $ \igraph{q}$ and $x\in \rootqueryvars{q}$, then $\fd{x}{y}$ is in $\WFD{q}$.
\end{enumerate}
\end{lemma}
\begin{proof}
{\bf (\Cref{it:eksu})} 
 Toward contradiction, suppose both $\fd{x}{y}$ and $\fd{y}{z}$ belong to $\SFD{q}$. Then we find edges of the form $\fd{x_0}{x},\fd{y_0}{y}$ from $\WFD{q}$. Since $x\neq y$, it follows by \Cref{lem:swcollide} that $\fd{x_0}{y} \in \SFD{q}$, a contradiction with $\fd{x}{y}$ being in $ \igraph{q}$.
 
 {\bf (\Cref{it:toksu})} If $\fd{x}{y}$ is in $ \igraph{q}\cap \SFD{q}$, then by definition it has an incoming edge in $\WFD{q}$, and thus in $\igraph{q}$.
\end{proof}

\medskip

\begin{lemma}\label{lem:forest}
    If $q\in \sjfuabcq$ is saturated and normal, then 
 $\igraph{q}$ is acyclic. 
\end{lemma}
\begin{proof}
By \Cref{lem:consecutive} a cycle in $\igraph{q}$ would contain at least one edge from $\WFD{q}$, and no two consecutive edges from $\SFD{q}$.
 Then, by \Cref{lem:nextattack,lem:nextattack2}, each atom generating an edge in $\WFD{q}$ in this cycle would attack the atom generating the next edge in $\WFD{q}$. This contradicts the acyclicity of the attack graph. Hence $\igraph{q}$ is acyclic.
\end{proof}

\medskip

\begin{lemma}\label{lem:noupstrong}
Let  $q\in \sjfuabcq$ be saturated and normal.  Suppose ${x}$ is connected to ${y}$ by a path in $\igraph{q}$ that is of length at least $2$. 
 Then $\fd{x}{y}\notin \FD{q}$.
\end{lemma}
\begin{proof}
By \Cref{lem:forest}, we have
 that $x\neq y$. 
 Moreover, $\fd{x}{y}\notin \WFD{q}$ by \Cref{lem:samesourcesat,lem:forest}.

Assume toward contradiction that $\fd{x}{y}\in \SFD{q}$. 
By normality of $q$,  we observe that $x$ is a variable (and not $\bot$). 
Let $D=x, \dots ,y_0,y$ be the path of $ \igraph{q}$ as in the assumption.
We have $x\neq y_0$; otherwise, $\igraph{q}$ has a cycle in contradiction with \Cref{lem:forest}. We consider two different possibilities.

First, suppose $\fd{y_0}{y}\in \WFD{q}$. Let $F\in {q}$ be (the unique) atom generating $\fd{y_0}{y}$. Then $F\attacks{} y$ by \Cref{lem:nextattack}. In particular, ${y}\notin \keyclosure{F}{q}$, and consequently, by $\fd{x}{y}\in \SFD{q}\subseteq \FD{q\setminus \{F\}}$, it holds that ${x}\notin \keyclosure{F}{q}$. Since $\gaifman{q}$ has an edge $\{x,y\}$, we obtain $F\attacks{} x$. Since
$x\neq y_0$, there exists an atom $G\in q\setminus \{F\}$ with $x=\keyvars{G}$ such that some $x_0\in \notkeyvars{G}$ is connected to $y_0$ in $\igraph{q}$ (possibly by the empty path). We obtain that $F\attacks{} G$ and by \Cref{lem:nextattack,lem:upattack,lem:consecutive} that $G\attacksp{} F$, contradicting the acyclicity of the attack graph.

Then, suppose $\fd{y_0}{y}\in \SFD{q}$. By definition, there exists an edge of the form $\fd{z}{y_0}\in \WFD{q}$. 
 Let $F\in q$ be (the unique) atom generating $\fd{z}{y_0}$.
Since $\fd{y_0}{y}$ belongs to both $ \SFD{q}$ and $\igraph{q}$, it follows that $\fd{z}{y}\notin \SFD{q}$. 
In particular, this entails $z\neq y$.  
Moreover, by saturation of $q$, we obtain 
$\fd{z}{y}\notin \FDclosure{q\setminus \{F\}}$, i.e., $y\notin \keyclosure{F}{q}$. Hence, $x\notin \keyclosure{F}{q}$ by $\fd{x}{y}\in \SFD{q}$. Since $F\attacks{} y$ by \Cref{lem:upattack}, this leads to $F\attacks{} x$, similarly to above. Specifically, $x\neq z$ as a consequence.
Again, we find an atom $G \in q\setminus \{F\}$ with $x=\keyvars{G}$ such that some $x_0\in \notkeyvars{G}$ is connected to $y_0$ in $\igraph{q}$.
As above, we obtain $F\attacks{} G$ and $G\attacksp{} F$, a contradiction.

Hence $\fd{x}{y}\notin \SFD{q}$. Concluding, we have established that $\fd{x}{y}\notin \FD{q}$.
 \end{proof}
 
 \medskip
 
 \begin{lemma}\label{lem:verysimple}
 Let  $q\in \sjfuabcq$ be saturated and normal.
 If $\fd{u}{v}\in \FD{q}$,  $\fd{v}{w}\in \SFD{q}\setminus \igraph{q}$, and $v\neq w$, 
 then $\fd{u}{w}\in \SFD{q}$.
\end{lemma}
\begin{proof}
If $\fd{u}{v}\in \SFD{q}$, then $\fd{v}{w}\in \SFD{q}$ implies $\fd{u}{w}\in \SFD{q}$ by \Cref{lem:ctrans}.
Thus, suppose $\fd{u}{v}\in \WFD{q}$.
Since $\fd{v}{w}\in \SFD{q}$ and $v\neq w$, we have that $\fd{v}{w}\in \FD{q}$. Then, by definition, $\fd{u}{w}\in \SFD{q}$.
\end{proof}
 \medskip

\begin{lemma}\label{lem:verysimple2}
Let  $q\in \sjfuabcq$ be saturated and normal.
 If $\fd{u}{v},\fd{v}{w}\in \FD{q}$ and $v\in \leafqueryvars{q}$, 
  then $\fd{u}{w}\in \SFD{q}$.
\end{lemma}
\begin{proof}
It follows by the assumption that $\fd{v}{w}\in \SFD{q}\setminus \igraph{q}$.
Moreover, $v\neq w$ since $q$ is variable-repetition-free.
Hence, \Cref{lem:verysimple} entails $\fd{u}{w}\in \SFD{q}$.
\end{proof}

\medskip

\begin{lemma}\label{lem:backward}
    Let $q\in \sjfuabcq$ be saturated and normal. Then the following statements hold: 
    \begin{enumerate}
    \item\label{it:sib1} If $\fd{x}{z},\fd{y}{z}\in\igraph{q}$ and $x\neq y$, then $\fd{x}{z},\fd{y}{z}\in \SFD{q}$, where $x$ and $y$ are variables such that $x\sib y$;
    \item\label{it:sib2} If $\fd{x}{y},\fd{x'}{z}\in\igraph{q}$, $y\neq z$, and $y\sib z$, then $\fd{x}{y},\fd{x'}{z}\in \WFD{q}$ and $x= x'$.
        \end{enumerate}
        Consequently, if $\fd{x}{y},\fd{x'}{y'}\in\igraph{q}$ and $y\sib y'$, then $x\sib x'$.
\end{lemma}
\begin{proof}
\textbf{(\Cref{it:sib1})}
By \Cref{lem:nextattack}, it cannot be the case that both $\fd{x}{z}$ and $\fd{y}{z}$ belong to $\WFD{q}$. 
Hence one of $\fd{x}{z}$ and $\fd{y}{z}$ belongs to $\SFD{q}$.

By symmetry, we may assume that  $\fd{x}{z}\in \SFD{q}$. We claim that this entails, $\fd{y}{z}\in \SFD{q}$. 
First, observe $\fd{x}{z}\in \SFD{q}$ entails $\fd{x_0}{x}\in \WFD{q}$ for some $x_0\in \queryvars{q}\cup\{\emp\}$. 
Then, $\fd{y}{z}\in \WFD{q}$ is not possible, because in such a case by $x\neq z$ and \Cref{lem:swcollide} we would obtain $\fd{x_0}{z}\in \SFD{q}$, contradicting the assumption that $\fd{x}{z}\in\igraph{q}$.  Therefore, it follows that $\fd{y}{z}\in \SFD{q}$. 

Now, $\fd{y}{z}\in \SFD{q}\cap \igraph{q}$  implies $\fd{y_0}{y}\in \WFD{q}$ for some $y_0\in \queryvars{q}\cup\{\emp\}$. From $\fd{x}{z},\fd{y}{z}\in\igraph{q}$ we obtain $\fd{x_0}{z},\fd{y_0}{z}\notin\SFD{q}$. Hence, by \Cref{lem:shortcut}, we obtain also $x\sib y$, as required. In particular, $x$ and $y$ are variables.

\textbf{(\Cref{it:sib2})}
Suppose $\fd{x}{y},\fd{x'}{z}\in\igraph{q}$, $y\neq z$, and $y\sim z$. Let $F\in q$ be such that $y,z\in \notkeyvars{F}$. 
By normality of $q$, we have $\fd{u}{y},\fd{u}{z}\in \WFD{q}\subseteq \igraph{q}$, where $u\defeq \keyvars{F}$. 
By \cref{it:sib1}, it holds that $x=u=x'$, which shows the claim.
\end{proof}

\medskip

 \begin{lemma}\label{lem:rootnotattacked}
 Let $q\in \sjfuabcq$ be saturated and normal. 
 Suppose $\outvar{\igraph{q}}{\emp}=\emptyset$. 
 There exists a variable $x\in \rootqueryvars{q}$ such that $x$ is not attacked and $\outvar{\SFD{q}}{x}=\{x\}$.
 \end{lemma}
\begin{proof}
Assume toward contradiction that all variables in $ \rootqueryvars{q}$ are attacked.
Since $\outvar{\igraph{q}}{\emp}=\emptyset$ and $\igraph{q}$ is acyclic, there exists at least one element $x$ in $ \rootqueryvars{q}$. Then $F\attacks{} x$, for some $F\in q$. In particular, there exists a variable $v\in \notkeyvars{F}$ such that $\fd{\keyvars{F}}{v}\in \WFD{q}\subseteq \igraph{q}$. Hence, $\keyvars{F}\neq \bot$ by $\outvar{\igraph{q}}{\emp}=\emptyset$.
Now, either $\keyvars{F}\in  \rootqueryvars{q}$, or there exists a non-empty path from a variable $y\in \rootqueryvars{q}$ to $\keyvars{F}$. 
In the former case, letting $y$ be the variable $\keyvars{F}$, it holds that $F'\attacks{} y$ for some $F'\in q$.
In the latter case, the first edge of the non-empty path, denote it by $\fd{y}{z}$, belongs to $\WFD{q}$ by \Cref{lem:consecutive}. 
Suppose $\fd{y}{z}$ is generated by $G\in q$.
By \Cref{lem:nextattack,lem:nextattack2}, we have $G \attacksp{} F$. 
Then, $F'\attacks{} y$ for some $F'\in q$, as in the previous case. In both cases, we repeat the reasoning above, substituting $y$ for $x$.
Eventually, this leads to a cycle in the attack graph, a contradiction.

We conclude by contradiction that there exists a variable $x\in \rootqueryvars{q}$ that is not attacked.
If $\outvar{\SFD{q}}{x}= \{x\}$, we are done. Thus, suppose  $\outvar{\SFD{q}}{x}\neq \{x\}$, i.e., $\fd{x}{y}\in \SFD{q}$, for some $y\neq x$. 

First, we claim that $y\in \rootqueryvars{q}$.
Suppose this is not the case. Then there exists an edge of the form $\fd{y_0}{y}\in \igraph{q}$. Suppose first that $\fd{y_0}{y}\in \WFD{q}$, 
and let $F\in q$ be the atom generating $\fd{y_0}{y}$. Then $F\attacks{} y$ by \Cref{lem:nextattack}, hence $y\notin \keyclosure{F}{q}$.
Consequently, as $\fd{x}{y}\in \SFD{q}$, we obtain $x\notin \keyclosure{F}{q}$. However, since $\{x,y\}$ is an edge $\gaifman{q}$, we obtain $F\attacks{} x$, contradicting the fact that $x$ is unattacked. Suppose then that $\fd{y_0}{y}\in \SFD{q}$, in which case there exists an edge of the form $\fd{z}{y_0}\in \WFD{q}$, generated by $F'$. In this case, we obtain by \Cref{lem:upattack} that $F'\attacks{} y$, and the contradiction obtains as above. Hence the claim that $y\in \rootqueryvars{q}$ holds.

Then, we claim that $y$ is likewise not attacked. Toward contradiction, suppose $F\attacks{} y$, for some $F\in q$. Then, $y\notin \keyclosure{F}{q}$, hence $x\notin \keyclosure{F}{q}$ due to $\fd{x}{y}\in \SFD{q}$. By the same token, $\{x,y\}$ is an edge $\gaifman{q}$, leading to $F\attacks{} x$, a contradiction. Therefore, $y$ is not attacked.

By \Cref{lem:ctrans} and \Cref{it:normal4} of normality, $\SFD{q}$ is acyclic apart from the dummy self-loops. Consequently, by the above two paragraphs, there exists $y\in \rootqueryvars{q}$ such that $y$ is not attacked and  $\outvar{\SFD{q}}{y}=\{y\}$. This concludes the proof.
\end{proof}

\medskip

\begin{lemma}\label{lem:distroot}
 Let $q\in \sjfuabcq$ be saturated and normal. 
Let $u$ and $v $ be distinct elements from $\rootqueryvars{q}\cup \{\emp\}$. 
Then $\reach{u}{\igraph{q}}\cap \reach{v}{\igraph{q}}= \emptyset$. 
\end{lemma}
\begin{proof}
Assume toward contradiction that $\reach{u}{\igraph{q}}\cap \reach{v}{\igraph{q}}\neq \emptyset$, for some  distinct variables 
$u$ and $v $ from $\rootqueryvars{q}$.
Choose a vertex $w\in \reach{u}{\igraph{q}}\cap \reach{v}{\igraph{q}}$.
By \Cref{lem:backward}, if $\fd{x}{y},\fd{x'}{y'}\in\igraph{q}$ and $y\sib y'$, then $x\sib x'$.
Hence, there exist two vertices $u_0\in \reach{u}{\igraph{q}}$ and $v_0\in \reach{v}{\igraph{q}}$ such that $u_0\sim v_0$, and 
$u_0=u$ or $v_0=v$. By symmetry we may assume that $u_0=u$, meaning that $u\sim v_0$.
If $u=v_0$, then $u\in \reach{v}{\igraph{q}}$, a contradiction. Hence $u\neq v_0$, which means
that there exists $F\in q$ such that $u,v_0\in \notkeyvars{F}$. By normality of $q$, we obtain $\fd{x}{u}\in \WFD{q}$, for
$x=\keyvars{F}$. By definition, then, $u$ has an incoming edge in $\igraph{q}$, a contradiction.
\end{proof}

The previous lemma establishes that for each vertex in $v\in \queryvars{q}\cup\{\bot\}$, there exists
a unique $u\in \rootqueryvars{q}\cup\{\emp\}$ such that $v\in \reach{u}{\igraph{q}}$. We then call this vertex
$u$ the \emph{root of} $v$, denoted $\rootof{v}$. Observe that $\rootof{\bot}=\bot$.



\medskip

\begin{lemma}\label{lem:puuh}
Let $q\in \sjfuabcq$ be saturated and normal.
Let $F\in q$ and let $x,y\in \notkeyvars{F}$. Then $\fd{x}{y}\notin\FD{q}$.
\end{lemma}
\begin{proof}
Toward contradiction, suppose $\fd{x}{y}\in\FD{q}$. Then by definition there exists $G\in q$ such that $\fd{x}{y}\in\FD{G}$.
Since $q$ is variable-repetition-free, we obtain $x\neq y$. Hence, by normality, $\fd{z}{y}\in \WFD{q}$, for $z\defeq \keyvars{F}$. Also by normality, $\fd{x}{y}\notin\SFD{G}$, so $\fd{x}{y}\in\WFD{G}$. Moreover, variable-repetition-freeness guarantees that $z\neq x$. Thus, a contradiction with \Cref{lem:samesourcesat} arises, whereby we conclude that $\fd{x}{y}\notin\FD{q}$.
\end{proof}

\section{Definitions and Lemmas for \Cref{sect:prune}}\label{sect:appprune}
This section contains the  missing definitions, proofs, and lemmas for \Cref{sect:prune}.

\subsection{Properties of Guards}
First, we consider various properties related to guarding.
\smallskip
\begin{lemma}\label{lem:isedge}
Let $q\in \sjfuabcq$ be saturated and normal. Let $x,y\in \rootqueryvars{q}$. If $y \in \reach{x}{\FD{q}}$, then $\fd{x}{y} \in\SFD{q}$.
\end{lemma}
\begin{proof}
If $x=y$, the statement holds trivially. Thus suppose $x\neq y$.
Assume toward contradiction that $\fd{x}{y}\notin  \SFD{q}$. 
By $y \in \reach{x}{\FD{q}}$ there exists a (non-empty) path $D$ from $x$ to $y$ in $\FD{q}$.
Assume $D$ is shortest among such paths. By \Cref{lem:ctrans} and the assumption, some edge $\fd{u}{v}$ on $D$ belongs to $ \WFD{q}$.
Moreover, $x\in \rootqueryvars{q}$ entails that $\fd{u}{v}$ is not the last edge on $D$. We may select $\fd{u}{v}$ so that none of the edges that follow it in $D$ belongs to $\WFD{q}$. Then, due to \Cref{lem:ctrans} and $D$ being shortest, we obtain that
 $\fd{u}{v}$ is followed by $\fd{v}{y}\in \SFD{q}$ in $D$. Since $v\neq y$ and $\fd{v}{y}\notin \igraph{q}$ (by $y\in \rootqueryvars{q}$),
it follows by \Cref{lem:verysimple} that $\fd{u}{y}\in \SFD{q}$, contradicting the assumption that $D$ is shortest.
Hence we conclude that $\fd{x}{y}\in  \SFD{q}$.
\end{proof}

\smallskip
\begin{lemma}\label{lem:isone}
If $q\in \sjfuabcq$ is saturated and normal, then $|\initqueryvars{q}|\leq \cgs{q}$.
\end{lemma}
\begin{proof}
Let $x,y\in \initqueryvars{q}$. Then, by definition, $x,y\in \rootqueryvars{q}$. By \Cref{lem:isedge}, if $x$ and $y$ belong to the same SCC of $\FD{q}$, then
$\fd{x}{y} \in\SFD{q}$ and $\fd{y}{x} \in\SFD{q}$. By normality, this is only possible when $x=y$. Thus, the variables in $\initqueryvars{q}$
are distributed to distinct source SCCs other than $\{\emp\}$. Since $\cgs{q}=\nscc{\FD{q}}-1$, this entails the lemma statement.
\end{proof}

\smallskip
\begin{lemma}\label{lem:rootguarded}
If $q\in \sjfuabcq$ is saturated and normal, then $\rootqueryvars{q}\subseteq \nrqueryvars{q}$. 
\end{lemma}
\begin{proof}
Let $x\in \rootqueryvars{q}$. If $x\in \initqueryvars{q}$, then $x$ is guarded by itself. 
If $x\notin \initqueryvars{q}$, 
 then there exists a variable $y\in  \initqueryvars{q}$ such that $x \in \reach{y}{\FD{q}}$ and $y \notin \reach{x}{\FD{q}}$. For this, observe that every source SCC intersects with $ \rootqueryvars{q}$.
 Since $y\in \rootqueryvars{q}$, we obtain by \Cref{lem:isedge} that $\fd{y}{x} \in\SFD{q}$. That is, $y$ guards $x$, which means that $x\in \nrqueryvars{q}$. 
\end{proof}

\smallskip
\begin{lemma}\label{lem:path}
Let $q\in \sjfuabcq$ be saturated and normal. Let $y\in \rqueryvars{q}\cap \nkeyqueryvars{q}$, $x\in\queryvars{q}$, $x\neq y$, and $\fd{x}{y}\in \SFD{q}$. Then $\FD{q}$ contains a path $z_1, \dots ,z_n$ such that 
\begin{enumerate}
\item $\fd{z_1}{y}\in \igraph{q}$; 
\item $z_n=x$; and 
\item $\fd{z_i}{y}\in \SFD{q}$ and $z_i\neq y$ for each $i\in [n]$.
\end{enumerate}
\end{lemma}
\begin{proof}
Note that $\fd{x}{y}\in \FD{q}$ due to $x\neq y$ and $\fd{x}{y}\in \SFD{q}$.

If $\fd{x}{y}\in  \igraph{q}$, 
  we can select the path that consists of $x$ only. Thus suppose $\fd{x}{y}\in \SFD{q}\setminus \igraph{q}$.

Toward contradiction, suppose the statement of the lemma is false.
Now, $\FD{q}$ contains a  path $z_1, \dots ,z_n$  where $z_1\in \initqueryvars{q}$ and
 $z_n=x$. 
 Since $y\in  \nkeyqueryvars{q}$, we have $z_i\neq y$ for each $i\in [n]$.
 Suppose $\fd{z_{k+1}}{y},\dots ,\fd{z_{n}}{y}\in \SFD{q}$ and $\fd{z_{k+1}}{y} \notin \igraph{q}$. Since 
 $z_{k+1}\neq y$, we obtain $\fd{z_{k}}{y}\in \SFD{q}$ by \Cref{lem:verysimple}. If $\fd{z_{k}}{y}\in \igraph{q}$, 
  the statement is true for the truncated path $(z_k, \dots ,z_n)$, a contradiction.
  Hence, we have that  $\fd{z_{k}}{y},\dots ,\fd{z_{n}}{y}\in \SFD{q}$ and  $\fd{z_{k}}{y}\notin \igraph{q}$.

By induction,  it follows that  $\fd{z_{1}}{y}\in \SFD{q}$, contradicting $y\in \rqueryvars{q}$. We conclude by contradiction that the statement is true.
\end{proof}

\smallskip
\begin{lemma}\label{lem:grounding}
Let $q\in \sjfuabcq$ be saturated and normal, and let $R\in q$.
Let $\db$ be a $q$-saturated database.
Then $|R^{\db}| \leq  |\restrict{R^\db}{V}| \cdot |\adom{\db}|^{|\initqueryvars{q}|}$, where $V\defeq \atomvars{R}\cap \rqueryvars{q}$.
\end{lemma}
\begin{proof}
Choose an inclusion-minimal set $\{z_1,\dots,z_n\}\subseteq \initqueryvars{q}$ such that, for every $y\in \atomvars{R}\cap \nrqueryvars{q}$, we have $\fd{z_i}{y}\in \SFD{q}$ for some $i\in[n]$. Write
\[
W_i\defeq \{y\in \atomvars{R}\cap \nrqueryvars{q}\mid \fd{z_i}{y}\in \SFD{q}\}.
\]
By minimality, every $W_i$ is non-empty. If $W_i=\{z_i\}$, then clearly $|\restrict{R^{\db}}{W_i}|\leq |\adom{\db}|$. Otherwise, we can choose $y\in W_i\setminus\{z_i\}$ and an atom $S_i\in q$ generating $\fd{z_i}{y}$.
By the $q$-saturation of $\db$, we have  
$ |\restrict{R^{\db}}{W_i}|
\leq |\restrict{S_i^{\db}}{\{z_i\}}|
\leq |\adom{\db}|$. 
Consequently,
\begin{align*}
 |R^{\db}|
 &\leq |\restrict{R^{\db}}{V}|\cdot
   \prod_{i=1}^{n}|\restrict{R^{\db}}{W_i}| \\
 &\leq |\restrict{R^{\db}}{V}|\cdot |\adom{\db}|^n
 \leq |\restrict{R^{\db}}{V}|\cdot
 |\adom{\db}|^{|\initqueryvars{q}|}.
\end{align*}
 This proves the lemma statement.
\end{proof}

\subsection{Guarded Reductions}\label{sect:guarded}
Next, we present what we call guarded reductions. These reductions rely on a size-invariance property which is given first.
This property ensures that repeated applications of guarded reductions do not cause an uncontrolled blow-up in the intermediate instances.

Given an atom $F=R(t_1,\ldots,t_n)\in q$, a set of variables
  $V\subseteq\atomvars{F}$, and a database $\db$, let
  $I_V\defeq\{i\in[n]\mid t_i\in V\}.$
  Then, we write $\restrict{R^{\db}}{V}$ for the projection of $R^{\db}$ to the
  positions in $I_V$. If $V=\emptyset$, this
  projection is $\{()\}$ when $R^{\db}\neq\emptyset$, and $\emptyset$
  otherwise.

\smallskip
\begin{definition}[$N$-unguarded-bounded instance]\label{def:Nguarded}
Let $\db$ be a database and $q\in\sjfuabcq$. Let $N\in\mathbb{N}$.
The pair $(\db,q)$ is called \emph{$N$-unguarded-bounded} if, for each $R\in q$,
\[
|\restrict{R^{\db}}{V}| \leq N,
\]
where $V\defeq \atomvars{R}\cap \rqueryvars{q}$.
\end{definition}

\smallskip
\begin{definition}[Guarded Reduction]\label{def:faith}
A many-one reduction $f\colon \DB\times \sjfuabcq \to \DB\times \sjfuabcq \cup \{\true,\false\}$ is \emph{guarded} if for every $(\db,q)\in \DB\times \sjfuabcq$, assuming $(\db',q')= f(\db,q) \notin \{\true,\false\}$,
the following statements hold:
\begin{enumerate}
\item\label{it:inv:1} If $(\db,q)$ is $N$-unguarded-bounded, then $(\db',q')$ is $N$-unguarded-bounded.
\item\label{it:inv:2} $\adom{\db'}\subseteq \adom{\db}$.
\item\label{it:sizeinv} 
\(
\restrict{\db'}{\sch(q)\cap \sch(q')} \subseteq \restrict{\db}{\sch(q)\cap \sch(q')}.
\) 
\item\label{it:inv:const} If $F\in \modec{q}$, and either $\keyvars{F}\in \initqueryvars{q}$ or $\atomvars{F}= \emptyset$, then $F\in \modec{(q')}$.
\item\label{it:inv:guarded} $\nrqueryvars{q}=\nrqueryvars{q'}$.
\item\label{it:inv:guarding} $\initqueryvars{q}=\initqueryvars{q'}$.
\item\label{it:dummy} $\queryvars{q'}\subseteq \queryvars{q}$.
\item\label{it:dummy2} $\size{q'}\leq \size{q}$. 

\end{enumerate}
\end{definition}

\begin{proposition}\label{prop:dbrestriction}
Every reduction mapping $(\db,q)$ to $(\db',q)$ with $\db'\subseteq\db$ is guarded.
\end{proposition}
\begin{proof}
The query is unchanged, and deleting facts cannot increase the active domain or the cardinality of any relation projection. Thus all conditions of guardedness hold.
\end{proof}

By slight abuse of terminology, when we state preservation properties for $\dbsaturate$, we mean the reduction $(\db,q)\mapsto (\dbsaturate(\db,q),q)$.
The following lemma is an immediate consequence of \Cref{lem:dbsat}.

\smallskip
\begin{lemma}\label{lem:dbsatinv}
$\dbsaturate$ is a guarded reduction.
\end{lemma}

The following proposition shows a useful property.  
\begin{proposition}\label{prop:compose}
 Guarded reductions are  closed under composition, provided that newly
introduced relation names are fresh.
\end{proposition}
 \begin{proof}
All conditions in the definition of guardedness compose. 
\Cref{it:inv:const} additionally relies on \Cref{it:inv:guarding}, while
\Cref{it:sizeinv} requires that introduced relation names are fresh. Informally, freshness guarantees that relation symbols have contiguous lifetimes under guarded reductions.
\end{proof}

\subsection{Leaf Pruning}
A variable $y\in 
\leafqueryvars{q}$ is called \emph{prunable} if  $y\in \keyqueryvars{q}$.
The operator $\leafprune_y$ is defined on pairs $(\db,q)$ such that
$y$ is a prunable variable of $q$.
For such a pair, given $q_0\defeq \{F\in {q}   \mid y= \keyvars{F}\}$ and $\db_0 \defeq \db \setminus \restrict{\db}{q_0}$,
 we define
\begin{equation}
\leafprune_y(\db,q)\defeq (\db',q'),
\end{equation}
where
\begin{equation}\label{eq:leafprune-q0}
q'\defeq \, q\setminus q_0 \text{ and }\db' \defeq \, \dbsaturate(\db_0,q').
\end{equation}
In particular, we have that $y\notin \keyqueryvars{q'}$. Furthermore, $q'$ inherits variable-repetition-freeness from $q$.

\smallskip
\begin{lemma}\label{lem:little0}
Let $q\in \sjfuabcq$ be saturated and normal. 
Let $y$ be a prunable variable.
Let $ u,v \in \queryvars{q}\cup \{\emp\}$ be such that $u\neq y$.
If $(\db',q')=\leafprune_y(\db,q)\notin\{\true,\false\}$,
then
\begin{itemize}
\item $\fd{u}{v}\in \FD{q}$ if and only if $\fd{u}{v}\in \FD{q'}$,
\item $\fd{u}{v}\in \WFD{q}$ if and only if $\fd{u}{v}\in \WFD{q'}$,
\item $\fd{u}{v}\in \SFD{q}$ if and only if $\fd{u}{v}\in \SFD{q'}$,
\item $\fd{u}{v}\in \igraph{q}$ if and only if $\fd{u}{v}\in \igraph{q'}$.
\end{itemize}
\end{lemma}
\begin{proof}
The case where $u=v$ is clear since both $q$ and $q'$ are variable-repetition-free. Thus, we assume that $u\neq v$.
Then, $\fd{u}{v}\in \FD{F}$ implies $F\in q'$. 
Furthermore, the assumption is that $ u,v \in \queryvars{q}\cup \{\emp\}$ are such that $u\neq y$.
We then obtain that $\fd{u}{v}\in \FD{q}$ implies $\fd{u}{v}\in \FD{q'}$, and $\fd{u}{v}\in \SFD{q}$ implies $\fd{u}{v}\in \SFD{q'}$.
The converse implications hold trivially. The second equivalence is then a consequence of the first and the third one.

The last equivalence follows from the second and third. For instance, if $\fd{u}{v}\in \igraph{q}$ due to $\fd{u}{v}\in \SFD{q}$, $\fd{w}{u}\in \WFD{q}$, and $\fd{w}{v}\notin \SFD{q}$, then $w\neq y$ (as $w\notin \leafqueryvars{q}$), hence 
$\fd{u}{v}\in \SFD{q'}$, $\fd{w}{u}\in \WFD{q'}$, and $\fd{w}{v}\notin \SFD{q'}$. The remaining cases can be checked analogously.
\end{proof}

\smallskip
\begin{lemma}\label{lem:root1}
Let $q\in \sjfuabcq$ be saturated and normal. 
Let $y$ be a prunable variable. 
If $(\db',q')=\leafprune_y(\db,q)\notin\{\true,\false\}$,
then $\initqueryvars{q}=\initqueryvars{q'}$.
\end{lemma}
\begin{proof}
 We need to prove $\rootqueryvars{q}\cap \srcqueryvars{q} =\rootqueryvars{q'} \cap \srcqueryvars{q'}$.

We first prove $\rootqueryvars{q}=\rootqueryvars{q'}$. Assuming $v\notin \rootqueryvars{q}$, there exists an edge of the form $\fd{u}{v}$ in $\igraph{q}$. Since $y\in \leafqueryvars{q}$, we obtain $u\neq y$, hence $\fd{u}{v}$ belongs to $\igraph{q'}$ by \Cref{lem:little0}, i.e., $v\notin \rootqueryvars{q'}$. Conversely, assuming $v\notin \rootqueryvars{q'}$, there exists an edge of the form $\fd{u}{v}$ in $\igraph{q'}$. By construction $u\neq y$, hence $\fd{u}{v}$ belongs to $\igraph{q}$ by \Cref{lem:little0}, i.e., $v\notin \rootqueryvars{q'}$.

Then, we show $\srcqueryvars{q} \cap \rootqueryvars{q} \subseteq \srcqueryvars{q'}$.  
Assuming $u\in \srcqueryvars{q}\cap \rootqueryvars{q}$, suppose toward contradiction that  $u\notin \srcqueryvars{q'}$. 
Then we find an edge $\fd{u_0}{u} \in \FD{q'}$ such that $u_0 \notin \reach{u}{ \FD{q'}}$. Clearly, $\fd{u_0}{u} \in \FD{q}$. Since $u \in \srcqueryvars{q}$, it holds that $u_0 \in \reach{u}{ \FD{q}}$. In particular, there exists a simple directed path $D$ from $u$ to $u_0$ in $\FD{q}$. Since no such path exists in $\FD{q'}$, this path contains an edge of the form $\fd{y}{v}$. 
Moreover, as $y\neq u$ due to $y\in \leafqueryvars{q}$ and $u\in \rootqueryvars{q}$ (while $\leafqueryvars{q}\cap \rootqueryvars{q}=\emptyset$), this edge is preceded in $D$ by an edge of the form $\fd{p}{y}$.
Now, since $y\in \leafqueryvars{q}$, we obtain $\fd{y}{v}\in \SFD{q}\setminus \igraph{q}$. 
Since $y\neq v$, it follows by \Cref{lem:verysimple} that $\fd{p}{v}\in \SFD{q}$. Thus, by \Cref{lem:little0}, the path $D'$ obtained from $D$ by replacing $\fd{p}{y}$ and $\fd{y}{v}$ with the shortcut edge 
$\fd{p}{v}$ is a path from $u$ to $u_0$ in $\FD{q'}$. This raises a contradiction, by which we conclude that $u\in \srcqueryvars{q'}$.

Finally, we prove $\srcqueryvars{q'} \subseteq \srcqueryvars{q}$.
Assuming $u\in \srcqueryvars{q'}$, suppose toward contradiction that $u\notin \srcqueryvars{q}$. Then we find an edge $\fd{u_0}{u} \in \FD{q}$ such that $u_0 \notin \reach{u}{ \FD{q}}$.
If $\fd{u_0}{u} \in \FD{q'}$, then by assumption $u_0 \in \reach{u}{ \FD{q'}}$, which implies $u_0 \in \reach{u}{ \FD{q}}$ by the construction of $q'$, a contradiction. Thus, suppose $\fd{u_0}{u} \notin \FD{q'}$. This entails $u_0=y$.
Therefore,  by $y\in \leafqueryvars{q}$ there exists an edge of the form $\fd{w}{u_0}\in \igraph{q}$. 
Since $\fd{w}{u_0}\in \FD{q}$ by definition, it follows by \Cref{lem:verysimple2} that $\fd{w}{u}\in \SFD{q}$. Now, $w\neq y$ as $q$ is variable-repetition-free, so either $w=u$ or $\fd{w}{u}\in \FD{q'}$ obtains. But then, by assumption, 
$w \in \reach{u}{ \FD{q'}}$, which implies $w \in \reach{u}{ \FD{q}}$. Composing with the edge $\fd{w}{u_0}$, we obtain $u_0 \in \reach{u}{ \FD{q}}$, a contradiction. Thus, we conclude that $u\in \srcqueryvars{q}$.
\end{proof}

\smallskip

\begin{lemma}\label{lem:nr}
Let $q\in \sjfuabcq$ be saturated and normal. 
Let  $y$ be a prunable variable. 
If $(\db',q')=\leafprune_y(\db,q)\notin\{\true,\false\}$,
then
 $\nrqueryvars{q}=\nrqueryvars{q'}$.
\end{lemma}
\begin{proof}
Let $v\in\nrqueryvars{q}$. Then we find $S(\underline{u},v,\uk)\in q$ such that $u\in \initqueryvars{q}$ and $\fd{u}{v}\in \SFD{q}$. 
Since $u\in \rootqueryvars{q}$ and $y\in \leafqueryvars{q}$, we obtain $y\neq u$. Thus, by \Cref{lem:little0} we have that $\fd{u}{v}\in \SFD{q'}$, and by construction we obtain $S(\underline{u},v,\uk)\in q'$.
Moreover, by \Cref{lem:root1}, we obtain $u\in \initqueryvars{q'}$.

Assume then that $v\in\nrqueryvars{q'}$. Then we find  $S(\underline{u},v,\uk)\in q'$ such that $u\in \initqueryvars{q'}$ and $\fd{u}{v}\in \SFD{q'}$.
Now $y\neq u$ by the construction of $q'$, hence $\fd{u}{v}\in \SFD{q}$ by \Cref{lem:little0}. Moreover, $S(\underline{u},v,\uk)\in q$ is clear, and 
$u\in \initqueryvars{q}$ follows by \Cref{lem:root1}.
\end{proof}

\smallskip
\begin{lemma}\label{lem:vequery1}
Let $q\in \sjfuabcq$ be a saturated and normal.
Let $y$ be a prunable variable. 
If $\leafprune_y(\db,q)= (\db',q')$, then the following statements hold:
\begin{enumerate}[(a)]
\item\label{lem:veqtwo11} $q'\in\sjfuabcq$; 
\item\label{lem:veqone11} $q'$ is saturated;
\item\label{lem:veqoneprime11} $q'$ is normal.
\end{enumerate}
\end{lemma}
\begin{proof}
\textbf{(\cref{lem:veqtwo11})} 
Clearly, $q'$ remains variable-repetition-free.
For acyclicity of the attack graph, it suffices to prove that $F\attacks{q'} G$ implies $F\attacks{q} G$.

Suppose $F\attacks{q'} G$. Then  $\gaifman{q'}$ contains a (possibly empty and undirected) path $U$
between $\notkeyvars{F}$ and   $\atomvars{G}$ that is not separated by $\keyclosure{F}{q}$. Clearly, $U$ belongs also to 
$\gaifman{q}$. Assume toward contradiction that $F\not\attacks{q} G$. Then, some variable $v$ in $U$ 
belongs to $\keyclosure{F}{q}$. Writing $u=\keyvars{F}$, this means $\FD{q\setminus \{F\}}$ contains a (directed) path $D$ from $t$ to $v$, where $y\neq t\in \{\emp,u\} $; for this, recall that $y\notin \queryvars{q'}$.
We may assume $D$ is shortest among all such paths.

Since $D$ is not included in  $\FD{q'\setminus \{F\}}$, 
an edge of the form $\fd{y}{z}$ appears on it. Since $y\in \leafqueryvars{q}$, it holds that $\fd{y}{z}\in \SFD{q} \setminus \igraph{q}$. 
Since $t\neq y$, some edge $\fd{v}{y}$ precedes $y$  in $D$. Clearly $y\neq z$, so we obtain by \Cref{lem:verysimple} that $\fd{v}{z}\in \SFD{q} \subseteq \FD{q\setminus \{F\}}$, leading to a shortcut in $D$, a contradiction.
%
We conclude from the contradiction that $F\attacks{q} G$.

\textbf{(\cref{lem:veqone1})}
Suppose \( F \in q' \),  
     $\fd{u}{v}\in \FD{F}$,
      $\fd{v}{w}\in \SFD{q'}$, and
     $\fd{u}{w}\in \FDclosure{q'\setminus \{F\}}$. 
We need to prove that $\fd{u}{w}\in \SFD{q'}$.  

By \Cref{lem:little0}, $\fd{v}{w}\in \SFD{q}$, and clearly, $\fd{u}{w}\in \FDclosure{q\setminus \{F\}}$. 
Since $F\in q$, we obtain by saturation of $q$ that $\fd{u}{w}\in \SFD{q}$.
Since $u\in \keyqueryvars{q'}$, we have that $u\neq y$. Thus, by \Cref{lem:little0}, we obtain that $\fd{u}{w}\in \SFD{q'}$.

\textbf{(\cref{lem:veqoneprime11})} 
In what follows we consider an atom $F\in q'$. Assuming that there is a contradiction with the normality of $q'$, possibly witnessed by $F$, we construct a contradiction with the assumption that $q$ is normal.

(\cref{it:normal1}) Suppose toward contradiction that $|\notkeyvars{F}|>1$ and $\fd{p}{u} \in \SFD{q'}$, where $p= \keyvars{F}$
and $u\in \notkeyvars{F}$. Since $F\in q$ and, by \Cref{lem:little0}, $\fd{p}{u} \in \SFD{q}$, we obtain a contradiction with the normality of $q$.

(\cref{it:normal2}) Suppose toward contradiction that $\fd{u}{v} \in \SFD{q'}$, for some distinct $u,v\in \notkeyvars{F}$.  Then 
$\fd{u}{v} \in \FD{q'}$, and in particular $u\in \keyqueryvars{q'}$.  Since $y\notin \keyqueryvars{q'}$, we have 
$u\neq y$, hence \Cref{lem:little0} entails $\fd{u}{v} \in \SFD{q}$ in contradiction with the normality of $q$. 

(\cref{it:normal4}) If $\fd{u}{v},\fd{v}{u}\in \SFD{q'}$, for some distinct $u,v\in \queryvars{q'}$, then analogously to the previous case we obtain $\fd{u}{v},\fd{v}{u}\in \SFD{q}$ using \Cref{lem:little0}, a contradiction with the normality of $q$.

(\cref{it:normal3}) If $\fd{\emp}{z} \in \SFD{q'}$, for some $z\in \queryvars{q'}$, then by \Cref{lem:little0} we also obtain $\fd{\emp}{z} \in \SFD{q}$, contradicting the normality of $q$.
\end{proof}

The main lemma of this section follows.
\smallskip
\begin{lemma}[Leaf Pruning]\label{lem:vedb111}
Let $\mathcal{C}$ consist of all $(\db,q)\in \DB\times \sjfuabcq$, where $q$ is a saturated and normal query and $\db$ is a $q$-saturated database.
For each $(\db,q)\in \mathcal{C}$ and prunable $y\in \queryvars{q}$, $(\db,q)\mapsto \leafprune_y(\db,q)$ is a strict, guarded, equivalence-preserving, and  polynomial-time reduction from $\mathcal{C}$ to $\mathcal{C}\cup\{\true,\false\}$.
\end{lemma}
\begin{proof}
Let $y$ be a prunable variable, and suppose $\leafprune_y(\db,q)= (\db',q')$.
The reduction $\leafprune_y$ is clearly strict and computable in polynomial time. It is into $\mathcal{C}\cup\{\true,\false\}$
by \Cref{lem:vequery1,lem:dbsat}. It suffices to prove that it is guarded and equivalence-preserving.

{\bf (Guardedness)}
\Cref{it:inv:guarded,it:inv:guarding} of guardedness follow 
by \Cref{lem:root1,lem:nr}. \Cref{it:sizeinv,it:inv:2,it:sizeinv,it:dummy,it:dummy2} are clear by construction.
We prove  
  \Cref{it:inv:const,it:inv:1}.
  
For \Cref{it:inv:const}, assume that $F\in \modec{q}$, and either $\keyvars{F}\in \initqueryvars{q}$ or $\atomvars{F}= \emptyset$.
Since $\keyvars{F}\neq y$, as $y\in \leafqueryvars{q}$ (and $\initqueryvars{q}$ is a subset of $\rootqueryvars{q}$, which is disjoint with $\leafqueryvars{q}$), we obtain $F\in {q'}$.
 Moreover, 
$F\in \modec{(q')}$, because $F\in \modei{(q')}$ means that
there exists an edge $\fd{u}{v}$ in  $\WFD{q'}\cap \FD{F}$.
Since this implies $u\neq y$, we would obtain $\fd{u}{v}\in \WFD{q}$ by \Cref{lem:little0}, which entails $F\in \modei{q}$, a contradiction.
Hence,  \Cref{it:inv:const} holds. 

It remains to consider \Cref{it:inv:1}, that is, $N$-boundedness.
Assume that for each $F\in q$,
\(
|\restrict{F^{\db}}{V}| \leq N,
\)
where $V\defeq \atomvars{F}\cap \rqueryvars{q}$.
Now, \Cref{lem:nr} entails $\rqueryvars{q} \cap \queryvars{q'} = \rqueryvars{q'}$.
Thus, for each for each $F\in q'$, given $V'\defeq \atomvars{F'}\cap \rqueryvars{q'}$, we obtain
$V' = V$. Hence,
\[
|\restrict{F^{\db'}}{V'}| \leq |\restrict{F^{\db}}{V}| \leq N,
\]
proving \Cref{it:inv:1}. 

We conclude that $\leafprune_y$ is guarded.

{\bf (Equivalence preservation)}
 Given $q_0\defeq \{F\in {q}   \mid y= \keyvars{F}\}$ and $\db_0 \defeq \db \setminus \restrict{\db}{q_0}$, by \Cref{lem:dbsat} it suffices to prove that $(\db,q)\in \cert$ if and only if $(\db_0,q')\in \cert$.

($\Rightarrow$) 
Suppose  $(\db,q)\in \cert$, and let  $\rep'$ be an arbitrary repair of $\db_0$. 
Define
\[
\rep \defeq \rep' \cup \restrict{\db}{q_0}.
\]
Since $q$ is saturated and normal, and $y\in \leafqueryvars{q}$, it holds that every $F\in q_0$ is of the form $S(\underline{y}, \vec{c})$ or   $S(\underline{y},z, \vec{c})$, where $z$ is a variable such that $\fd{y}{z}\in \SFD{q}$ and
$\vec{c}$ is a (possibly empty) sequence of constants.
Since $\db$ is $q$-saturated, $\restrict{\db}{q_0}$ is consistent.
In particular, $\rep$ is a repair of $\db$. 
Since $\rep\models q$ by assumption, we obtain $\rep'\models q'$. Hence $(\db_0,q')\in \cert$.

($\Leftarrow$)
Suppose  $(\db_0,q')\in \cert$, and
let $\rep$ be a repair of $\db$. Then $\rep'\defeq \rep \setminus \restrict{\rep}{q_0}$ is a repair of $\db_0$. 
Now, $(\db_0,q')\in \cert$ implies
we find a valuation $\theta$ such that $\theta(q')\subseteq \rep'$. 
We show that
 $\theta(F)\in \rep$, for each $F \in q_0$. For this, consider an atom $F \in q_0$. As observed above, $F$ is of the form $S(\underline{y}, \vec{c})$ or   $S(\underline{y},z, \vec{c})$, where $z$ is a variable such that $\fd{y}{z}\in \SFD{q}$ and
$\vec{c}$ is a (possibly empty) sequence of constants. Note that $S^{\db}$ is consistent due to  $q$-saturation.

Since $y\in \leafqueryvars{q}$, 
 there exists an edge of the form $\fd{x}{y} \in \igraph{q}\subseteq \FD{q}$ (where, as $q$ is variable-repetition-free, $x\neq y$).
Then, there exists an atom of the form $R(\underline{x},y,\uk)\in q$ generating $\fd{x}{y}$.
Since $x\neq y$, we have $R(\underline{x},y,\uk)\in q'$, hence $R(\underline{\theta(x)},\theta(y),\uk)\in \rep'$.
If $\notkeyvars{F}=\emptyset$, then from the $q$-saturation of $\db$ it readily follows that $\theta(F)\in \db$. By the consistency of $S^{\db}$ we then obtain $\theta(F)\in \rep$.

Thus, suppose $\notkeyvars{F}\neq \emptyset$. Then, $F$ is of the form $S(\underline{y},z, \vec{c})$, where $z$ is a variable such that $\fd{y}{z}\in \SFD{q}$ and
$\vec{c}$ is a (possibly empty) sequence of constants.
Since $y\in \leafqueryvars{q}$, we obtain by \Cref{lem:verysimple2} that $\fd{x}{z}\in \SFD{q}$. 
 Consequently, there exists an atom of the form $T(\underline{x},z,\uk)\in q$. 
  Since $x\neq y$, we have $T(\underline{x},z,\uk)\in q'$, hence $T(\underline{\theta(x)},\theta(z),\uk)\in \rep'$. 
 As $\db$ is $q$-saturated and $\rep'\subseteq \db$, we obtain $S(\underline{\theta(y)},\theta(z),\uk)\in \db$. 
 Moreover, by the $q$-saturation of $\db$, this is tantamount to $\theta(F)=S(\underline{\theta(y)},\theta(z), \vec{c})\in \db$.
 Since $S^{\db}$ is consistent, we obtain  $\theta(F) \in \rep$. We conclude that $\theta(q)\subseteq \rep$, hence $(\db,q)\in \cert$.
\end{proof}

\subsection{Leaf Elimination}\label{sect:le}
A variable $y\in 
 \rqueryvars{q}\cap\leafqueryvars{q}\cap  \nkeyqueryvars{q}$ 
is called \emph{eliminable}.
The operator $\leafeliminate_y$ is defined on pairs $(\db,q)$ such that
$y$ is an eliminable variable of $q$.
For such a pair, given
\begin{equation}\label{eq:and}
q' \defeq \{\drop{F}{y}\mid F\in q\}\text{ and }
\db_0 \defeq \{\drop{A}{y}\mid A\in \db\},
\end{equation}
 we define
\begin{equation}
\leafeliminate_y(\db,q)\defeq (\db',q'),
\end{equation}
where $\db' \defeq \dbsaturate(\db_0,q')$.
The reader is referred to the beginning of \Cref{sect:normal} for the definition of the notation $F^{-y}$.
Note that $y\notin 
   \queryvars{q'}$, and in particular, $\queryvars{q}\setminus \{y\}=\queryvars{q'}$.
   Furthermore, $q'$ inherits variable-repetition-freeness from $q$. 
The helping lemmas that follow are analogous to the previous section.

\smallskip
\begin{lemma}\label{lem:little}
Let $q\in \sjfbcq$ be saturated and normal. 
Let  $y$ be eliminable. 
Let $u,v\in \queryvars{q} \cup \{\emp\}$ be such that $v\neq y$. 
If $(\db',q')=\leafeliminate_y(\db,q)\notin\{\true,\false\}$,
then
\begin{itemize}
\item $\fd{u}{v} \in \FD{q}$ if and only if $\fd{u}{v} \in \FD{q'}$;
\item $\fd{u}{v} \in \WFD{q}$ if and only if $\fd{u}{v} \in \WFD{q'}$;
\item $\fd{u}{v} \in \SFD{q}$ if and only if $\fd{u}{v} \in \SFD{q'}$;
\item $\fd{u}{v}\in \igraph{q}$ if and only if $\fd{u}{v}\in \igraph{q'}$.
\end{itemize}
\end{lemma}
\begin{proof}
The case where $u=v$ is clear, so we assume that $u\neq v$.
The first and third items are straightforward. Then, the second item follows.
The fourth item is a consequence of the second and third one. For instance, if $\fd{u}{v}\in \igraph{q}$ due to $\fd{u}{v}\in \SFD{q} \cap \FD{q}$, $\fd{w}{u}\in \WFD{q}$, and $\fd{w}{v}\notin \SFD{q}$, then $u\neq y$ by $y\in \nkeyqueryvars{q}$, 
  hence we obtain $\fd{u}{v}\in \SFD{q'}\cap \FD{q}$, $\fd{w}{u}\in \WFD{q'}$, and $\fd{w}{v}\notin \SFD{q'}$ by the second and third item. The converse direction is analogous, as $y\notin \queryvars{q'}$. 
\end{proof}

\smallskip
\begin{lemma}\label{lem:root11}
Let $q\in \sjfbcq$ be saturated and normal. 
Let $y$ be eliminable.
If $(\db',q')=\leafeliminate_y(\db,q)\notin\{\true,\false\}$,
then $\initqueryvars{q}=\initqueryvars{q'}$.
\end{lemma}
\begin{proof}
 We need to prove $\rootqueryvars{q}\cap \srcqueryvars{q} =\rootqueryvars{q'} \cap \srcqueryvars{q'}$.

We first prove $\rootqueryvars{q} \subseteq \rootqueryvars{q'}$. 
Assuming $u\notin \rootqueryvars{q'}$, there exists an edge of the form $\fd{w}{u}\in\igraph{q'}$.  Then $u\neq y$ by the construction of $q'$.
If $\fd{w}{u}\in\WFD{q'}$, then $\fd{w}{u}\in\WFD{q}$  by \Cref{lem:little}, i.e., $u\notin \rootqueryvars{q'}$.
 Thus, suppose $\fd{w}{u}\in\SFD{q'}\cap \FD{q'}$. Then we find another edge $\fd{p}{w}\in \WFD{q'}$ such that $\fd{p}{u}\notin\SFD{q'}$.
 By construction, $w\neq y$, hence by \Cref{lem:little} we have $\fd{p}{w}\in \WFD{q}$, $\fd{p}{u}\in\SFD{q}$ and $\fd{p}{u}\notin\SFD{q}$, which is tantamount
to $\fd{w}{u}\in\igraph{q}$, i.e., $u \notin \rootqueryvars{q}$.

We then prove $\rootqueryvars{q'} \subseteq \rootqueryvars{q}$. 
Assume $u\notin \rootqueryvars{q}$. If $u=y$, then by $y\notin \queryvars{q'}$ we obtain $u\notin \rootqueryvars{q'}$.
Thus, suppose $u\neq y$. By assumption there exists an edge of the form $\fd{w}{u}\in\igraph{q}$.  
If $\fd{w}{u}\in\WFD{q}$, then $\fd{w}{u}\in\WFD{q'}$  by \Cref{lem:little}, i.e., $u\notin \rootqueryvars{q'}$.
 Thus, suppose $\fd{w}{u}\in\SFD{q}\cap \FD{q}$. Then we find another edge $\fd{p}{w}\in \WFD{q}$ such that $\fd{p}{u}\notin\SFD{q}$.
 We have $w\in \keyqueryvars{q}$, hence $w\neq y$. Consequently,
   \Cref{lem:little} implies $\fd{p}{w}\in \WFD{q'}$, $\fd{p}{u}\in\SFD{q'}$, and $\fd{p}{u}\notin\SFD{q'}$, which is tantamount
to $\fd{w}{u}\in\igraph{q'}$, i.e., $u \notin \rootqueryvars{q'}$.

Next, we  prove $\srcqueryvars{q} \cap \rootqueryvars{q} \subseteq \srcqueryvars{q'}$.
Assuming  $u\notin \srcqueryvars{q'}$, there exists an edge of the form $\fd{u_0}{u}\in {\FD{q'}}$ such that $u_0\notin \reach{u}{\FD{q'}}$.
Clearly, we have $\fd{u_0}{u}\in {\FD{q}}$. Now, suppose toward contradiction that $u\in\srcqueryvars{q} \cap \rootqueryvars{q}$.
 Then we obtain $u_0\in \reach{u}{\FD{q}}$. In particular, there exists a simple directed path $D$ from $u$ to $u_0$ in $\FD{q}$. Since no such path exists in $\FD{q'}$, this path contains an edge of the form $\fd{w}{y}$ as an intermediate vertex. Since $y\in \nkeyqueryvars{q}$, this edge must be the last one in $D$, i.e., $y=u$.  But then,  $u\in \rootqueryvars{q}$ and $y\in \leafqueryvars{q}$ raise a contradiction, as $\rootqueryvars{q}\cap \leafqueryvars{q}=\emptyset$. Thus the claim follows.

Finally, we show $\srcqueryvars{q'} \subseteq \srcqueryvars{q}$.
Assuming $u\in \srcqueryvars{q'}$, suppose toward contradiction that $u\notin \srcqueryvars{q}$. Then we find an edge $\fd{u_0}{u} \in \FD{q}$ such that $u_0 \notin \reach{u}{ \FD{q}}$.
If $\fd{u_0}{u} \in \FD{q'}$, then by assumption $u_0 \in \reach{u}{ \FD{q'}}$, which implies $u_0 \in \reach{u}{ \FD{q}}$ by the construction of $q'$, a contradiction. Thus, suppose $\fd{u_0}{u} \notin \FD{q'}$. This entails $u=y$, in which case $u\notin \queryvars{q'}$, and in particular $u\notin \srcqueryvars{q'}$, a contradiction. Therefore, by contradiction we obtain $u\in \srcqueryvars{q}$.
\end{proof}

\smallskip
\begin{lemma}\label{lem:nr12}
Let $q\in \sjfbcq$ be saturated and normal. 
Let $y$ be eliminable. 
If $(\db',q')=\leafeliminate_y(\db,q)\notin\{\true,\false\}$,
then $\nrqueryvars{q}=\nrqueryvars{q'}$.
\end{lemma}
\begin{proof}
Let $v\in\nrqueryvars{q}$. Then we find 
 $u\in \initqueryvars{q}$ such that $\fd{u}{v}\in \SFD{q}$. 
Since $y\in \rqueryvars{q}$, we obtain $y\neq v$. 
Then, by \Cref{lem:little} we obtain $u\in \initqueryvars{q'}$ and by \Cref{lem:root11} that $\fd{u}{v}\in \SFD{q'}$,
 hence $v\in\nrqueryvars{q'}$.

Assume then that $v\in\nrqueryvars{q'}$. Then we find  
 $u\in \initqueryvars{q'}$ such that $\fd{u}{v}\in \SFD{q'}$.
Now $y\neq u$ by the construction of $q'$, hence $\fd{u}{v}\in \SFD{q}$ by \Cref{lem:little}. Moreover, 
$u\in \initqueryvars{q}$ follows by \Cref{lem:root11}. Thus $v\in\nrqueryvars{q}$.
\end{proof}

\smallskip
\begin{lemma}\label{lem:vequery11}
Let $q\in \sjfuabcq$ be saturated and normal. Let $y$ be eliminable. 
If $(\db',q')=\leafeliminate_y(\db,q)\notin\{\true,\false\}$,
then
\begin{enumerate}[(a)]
\item\label{lem:veqtwo1} $q'\in\sjfuabcq$; 
\item\label{lem:veqone1} $q'$ is saturated;
\item\label{lem:veqoneprime1} $q'$ is normal.
\end{enumerate}
\end{lemma}
\begin{proof}
\textbf{(\cref{lem:veqtwo1})} 
Clearly, $q'$ remains variable-repetition-free.
For acyclicity of the attack graph,
it  suffices to prove that $F\attacks{q'} G$ implies $F\attacks{q} G$.

Suppose $F\attacks{q'} G$. Then  $\gaifman{q'}$ contains a (possibly empty) path $U$
between $\notkeyvars{F}$ and   $\atomvars{G}$ that is not separated by $\keyclosure{F}{q}$. Clearly, $U$ belongs also to 
$\gaifman{q}$. Assume toward contradiction that $F\not\attacks{q} G$. Then, some variable $v$ in $U$ 
belongs to $\keyclosure{F}{q}$. Writing $u=\keyvars{F}$, this means $\FD{q\setminus \{F\}}$ contains a (directed) path $D$ from $t$ to $v$, where $t\in \{u,\emp\}$. We may assume $D$ is shortest among all such paths.
Since $D$ is not included in  $\FD{q'\setminus \{F\}}$, we observe that an edge of the form  $\fd{x}{y}$ appears on $D$.
Since $y\in \nkeyqueryvars{q}$, it must be the case that $y$ is the last element of $D$, in which case $y$ appears on $U$.
However, $U$ is a path on $\gaifman{q'}$, and $y\notin \queryvars{q'}$, meaning that $y$ cannot appear on $U$. From this contradiction, we conclude that $F\attacks{q} G$.

\textbf{(\cref{lem:veqone1})}
Suppose \( F \in q' \),  
     $\fd{u}{v}\in \FD{F}$,
      $\fd{v}{w}\in \SFD{q'}$, and
     $\fd{u}{w}\in \FDclosure{q'\setminus \{F\}}$. 
We need to prove $\fd{u}{w}\in \SFD{q'}$.  

 Let $F=\drop{G}{y}$, where $G\in q$. Now $u,v,w\in \queryvars{q'}$, which entails that none of $u,v,w$ is $y$. Hence, by \Cref{lem:little}, we have $\fd{v}{w}\in \SFD{q}$, and  clearly $\fd{u}{w}\in \FDclosure{q\setminus \{G\}}$.  Since also $\fd{u}{v}\in \FD{G}$, 
by saturation of $q$ it follows that $\fd{u}{w}\in \SFD{q}$, hence $\fd{u}{w}\in \SFD{q'}$ by \Cref{lem:little}.


\textbf{(\cref{lem:veqoneprime1})} 
In what follows we write $H' \defeq \drop{H}{y}$, for each $H\in q$. We assume, for the sake of a contradiction, that $q'$ is not normal.
There are four different cases.

(\cref{it:normal1}) Suppose toward contradiction that $|\notkeyvars{F'}|>1$, and $\fd{p}{u} \in \SFD{q'}$, for $p= \keyvars{F'}$
and some $u\in \notkeyvars{F'}$. Then $p,u$ are distinct from $y$, and hence \Cref{lem:little} entails $\fd{u}{v} \in \SFD{q}$. 
Clearly also $|\notkeyvars{F}|>1$, $p= \keyvars{F'}$
and  $u\in \notkeyvars{F'}$. This raises a contradiction with the normality of $q$. 

(\cref{it:normal2}) If $\fd{u}{v} \in \SFD{q'}$, for some distinct $u,v\in \notkeyvars{F'}$, then $u,v$ are disctinct from $y$, 
leading to $\fd{u}{v} \in \SFD{q}$ by \Cref{lem:little}. Moreover $u,v\in \notkeyvars{F}$, leading to a contradiction with the normality of $q$. Therefore this items follows, too.

(\cref{it:normal4}) If $\fd{u}{v},\fd{v}{u}\in \SFD{q'}$, for some distinct $u,v\in \queryvars{q'}$, then $u,v\in \queryvars{q}\setminus\{y\}$, and
$\fd{u}{v},\fd{v}{u}\in \SFD{q}$ obtains by \Cref{lem:little}, a contradiction with the normality of $q$.

(\cref{it:normal3}) If $\fd{\emp}{z} \in \SFD{q'}$, for some $z\in \queryvars{q'}$, then $\fd{\emp}{z} \in \SFD{q}$ by \Cref{lem:little}, a contradiction.
Hence this item follows.
\end{proof}

\smallskip
\begin{lemma}[Leaf Elimination]\label{lem:vedb0}
Let $\mathcal{C}$ consist of all $(\db,q)\in \DB\times \sjfuabcq$, where $q$ is a saturated and normal query and $\db$ is a $q$-saturated database.
For each $(\db,q)\in \mathcal{C}$ and eliminable $y\in \queryvars{q}$, 
$(\db,q)\mapsto \leafeliminate_y(\db,q)$ is strict, guarded, equivalence-preserving, and  polynomial-time reduction from $\mathcal{C}$ to $\mathcal{C}\cup\{\true,\false\}$.
\end{lemma}
\begin{proof}
Let $y$ be an eliminable variable, and suppose $\leafeliminate_y(\db,q)= (\db',q')$.
The reduction $\leafeliminate_y$ is clearly strict and computable in polynomial time. It is into $\mathcal{C}\cup\{\true,\false\}$
by \Cref{lem:vequery11,lem:dbsat}. It suffices to prove that it is guarded and equivalence-preserving.

{\bf (Guardedness)}
\Cref{it:inv:guarded,it:inv:guarding} of guardedness follow by
by \Cref{lem:root11,lem:nr12}. \Cref{it:sizeinv,it:inv:2,it:sizeinv,it:dummy,it:dummy2} are clear by construction.
We prove  
  \Cref{it:inv:const,it:inv:1}.

For \Cref{it:inv:const}, assume that $F\in \modec{q}$, and either $\keyvars{F}\in \initqueryvars{q}$ or $\atomvars{F}= \emptyset$.
 Toward contradiction, suppose first that $F\notin {q'}$.
Then, $y\in \notkeyvars{F}$ by the construction of $q'$. In particular,  $\atomvars{F} \neq \emptyset$, so we obtain 
$x\defeq \keyvars{F}\in \initqueryvars{q}$. Since $F\in \modec{q}$, this entails $\fd{x}{y}\in \SFD{q}$. Hence, $y\in \nrqueryvars{q}$ in contradiction with the assumption that $y$ is eliminable. By contradiction, we obtain $F\in {q'}$.
Moreover, 
$F\in \modec{(q')}$, because $F\in \modei{(q')}$ means that there exists an edge $\fd{u}{v}\in \WFD{q'}\cap \FD{F}\neq \emptyset$.
Since this implies $v\neq y$, we would obtain $\fd{u}{v}\in \WFD{q}$ by \Cref{lem:little}, which entails $F\in \modei{q}$, a contradiction.
Hence,  \Cref{it:inv:const} holds. 

For \Cref{it:inv:1}
assume that for each $F\in q$,
\(
|\restrict{F^{\db}}{V}| \leq N,
\)
where $V\defeq \atomvars{F}\cap \rqueryvars{q}$.
By \Cref{lem:nr12}, we obtain $\rqueryvars{q} \cap \queryvars{q'} = \rqueryvars{q'}$.
Thus, for each for each $F\in q'$, given $V'\defeq \atomvars{F'}\cap \rqueryvars{q'}$, we obtain
$V' = V$
\[
|\restrict{F^{\db'}}{V'}| \leq |\restrict{F^{\db}}{V}| \leq N.
\]
We may thus conclude that the stated reduction  is guarded.

{\bf (Equivalence preservation)}
By \Cref{lem:dbsat} it suffices to prove that $(\db,q)\in \cert$ if and only if $(\db_0,q')\in \cert$, given \eqref{eq:and}.

($\Rightarrow$) 
Suppose  $(\db,q)\in \cert$, and let  $\rep'$ be an arbitrary repair of $\db_0$. 
We first construct a repair $\rep$ of $\db$ from $\rep'$ as follows:
\begin{itemize}
\item  For each $A'\in \rep'$, add to $\rep$ an arbitrary fact $A\in \db_0$ such that $\drop{A}{y}=A'$. 
\end{itemize}
Clearly, $\rep$ is a repair of $\db$. 
Since $\rep\models q$ by assumption, we obtain $\rep' \models q'$. Hence $(\db_0,q')\in \cert$.

($\Leftarrow$)
Let $\rep$ be a repair of $\db$. 
Then $\rep'\defeq \drop{\rep}{y}$ is a repair of $\db_0$.
Now, $(\db_0,q')\in \cert$ implies
we find a valuation $\theta$ such that $\theta(q')\subseteq \rep'$. 
Our aim is to find a constant $b$ such that $\theta'(q)\subseteq \rep$, where $\theta'\defeq \theta_{y\mapsto b}$.

{\bf Construction of $\theta'$. }    
Since $y\in \leafqueryvars{q}$, we find an edge of the form $\fd{x}{y}$ in $\igraph{q}$. 
Note that $x\neq y$ by definition of $\igraph{q}$ and the fact that $q$ is variable-repetition-free.
In particular, 
 there exists an atom of the form $F=R(\underline{x},y,\vec{z})$, where $\vec{z}$ is possibly empty.
Then, $\drop{F}{y}=R'(\underline{x},\vec{z})$, and we have $R'(\underline{\theta(x)},\theta(\vec{z}))\in \rep'$. By construction, there exists $b$ 
  such that $R(\underline{\theta(x)},b,\theta(\vec{z}))\in \rep$. Then,  setting $\theta'\defeq \theta_{y\mapsto b}$, we have $\theta'(F)\in \rep$. 

Let $G\in q \setminus \{F\}$. We need show that $\theta'(G)\in \rep$. 
  Note that $y\notin \keyvars{G}$ since $y\in \nkeyqueryvars{q}$.
 If $y\notin \atomvars{G}$, then we clearly obtain $\theta'(G)\in \rep$.
Thus, suppose that $y\in \notkeyvars{G}$.
 Note that $y\neq \keyvars{G}$, as $q$ is variable-repetition-free.
 In particular, we can  write
$G=S(\underline{u},y,\vec{v})$, where $u\neq y$. 
 Then $\fd{u}{y}\in \FD{q}$. 
Note that the case that $\fd{u}{y}\in \WFD{q}$ is not possible due to the normality of $q$ and \Cref{lem:backward}. 
Thus $\fd{u}{y}\in \SFD{q}$. 

Since $y\in \rqueryvars{q}\cap \nkeyqueryvars{q}$ and $u\neq y$, we obtain by  \Cref{lem:path}
 a path $z_1, \dots ,z_n$ in $\FD{q}$ such that 
 \begin{enumerate}
\item $\fd{z_1}{y}\in \igraph{q}$, 
\item $z_n=u$, and 
\item $\fd{z_i}{y}\in \SFD{q}$ and $z_i\neq y$ for each $i\in [n]$.
\end{enumerate}

Let us write $S_i(\underline{z_i},y,\uk)$ for the atom of $q$ generating $\fd{z_i}{y}$, for each $i \in [n]$. In particular, let us set $S_n$ as $S$. 
Now we have $ \fd{x}{y}\in \igraph{q}$ and $\fd{z_1}{y} \in \SFD{q} \cap\igraph{q}$.

If $x=z_1$, then $S_1(\underline{\theta(z_1)},b,\uk)\in \rep$ by $R(\underline{\theta(x)},b,\uk)$ and the fact that $\db$ is $q$-saturated.

Suppose $x\neq z_1$. 
By \Cref{lem:backward}, we obtain $\fd{x}{y}\in \SFD{q}$ and $z_1 \sib x$, where $x$ is a variable.
 Then, $q$ contains an atom of the form $T(\underline{t},x,z_1,\vec{w})$. 
By normality of $q$ we have $T\in \modei{q}$. 
Note that $y\notin \atomvars{T}$: if $y=t$, then $\igraph{q}$ contains a cycle $y,x,y$ (contradicting \Cref{lem:forest}), and if $y$ appears in $(z_1,\vec{w})$ (recall $y\neq x$), then $\igraph{q}$ contains two edges colliding at $y$, one edge from $\WFD{q}$ and the other from $\SFD{q}$ (contradicting \Cref{lem:backward}).
This entails $\theta(T)\in \rep$. Since $\db$ is $q$-saturated, we obtain $S_1(\underline{\theta(z_1)},b,\uk)\in \rep$.

For an inductive argument, suppose $S_i(\underline{\theta(z_i)},b,\uk)\in \rep$. Since $\fd{z_i}{z_{i+1}}\in \FD{q}$, there exists an atom of the form
$U(\underline{z_{i}},z_{i+1},\vec{u})$ in $q$. Recall $\fd{z_{i+1}}{y}\in \SFD{q}$, which entails, by normality of $q$, that $y$ does not appear on $\vec{u}$. Moreover, $y$ is neither of $z_i$ or $z_{i+1}$. 
Hence it is correct to conclude that $\theta(U)\in \rep$. By saturation of $\db$, we obtain $S_{i+1}(\underline{\theta(z_{i+1})},b,\uk)\in \rep$.

In this way, we obtain by induction that $S(\underline{\theta(u)},b,\uk)\in \rep$. To conclude the argument, observe that $\theta(\drop{G}{y})=S'(\underline{\theta(u)},\theta(\vec{v}))\in \rep'$. Consequently,  $S(\underline{\theta(u)},c,\theta(\vec{v}))\in \rep$ for some constant $c$. 
Then, by $\fd{u}{y}\in \SFD{q}$ and the $q$-saturation of $\db$ we obtain $b=c$, leading to $\theta'(G)\in \rep$.



 We conclude that $\theta(q)\subseteq \rep$, hence $(\db,q)\in \cert$.
\end{proof}

\subsection{Consistent Edge Elimination}
We say that $F\in \modei{q}$, with $Y\defeq \notkeyvars{F}$ and $Z\defeq \outvar{\igraph{q}}{Y}$, 
is \emph{extendable} if following conditions hold:
 \begin{enumerate}
 \item\label{it:Y} $\outvar{\WFD{q}}{Y}=\emptyset$,
   \item\label{it:Z} $Z\subseteq  
    \leafqueryvars{q} \cap \nkeyqueryvars{q} \cap \nrqueryvars{q}$,
     \item\label{it:Ztwo} $Z \neq \emptyset$.
 \end{enumerate}
By \cref{it:Ztwo} also $Y$ is non-empty.
Moreover, by \cref{it:Y} we have $\outvar{\igraph{q}}{y}\subseteq \outvar{\SFD{q}}{y}$ for each $y\in Y$.

The operator $\atomextend_F$ is defined on pairs $(\db,q)$ such that
$F$ is an extendable atom of $q$.
Assuming $F= R(\underline{x},\vec{y})$, define  first 
\begin{equation}\label{eq:q0}
q_0 \defeq  \{G\in q \mid \keyvars{G}\in Y\}.
\end{equation}
 Then, let $\vec{z}=(z_1, \dots ,z_\ell)$ list the variables in $Z$.
Define
\begin{equation}\label{eq:strongel}
q' \defeq (q \setminus (q_0 \cup 
 \{R(\underline{x},\vec{y})\}))  \cup \{R'(\underline{x},\vec{y},\vec{z})\}.
\end{equation}
Note that 
\begin{equation}\label{eq:subset}
q_0\subseteq \modec{q}
\end{equation}
by \cref{it:Y} and normality of $q$. For each $i\in [\ell]$, choose the $\prec$-least atom $S_i(\underline{y_{j_i}},z_i,\uk)\in \modec{q}$
that generates $\fd{y_{j_i}}{z_i}\in \SFD{q}$, for some variable $y_{j_i}\in Y$.
Then, let
\begin{equation}\label{eq:strongel2}
\db_0 \defeq (\db \setminus \restrict{\db}{q_0}) \cup 
\{R'(\underline{a},\vec{b},\vec{c}) \mid
R(\underline{a},\vec{b}) \in \db,\ 
\forall i\in [\ell]\colon S_i(\underline{b_{j_i}},c_i,\uk)\in \db\}.
\end{equation}
We define
\begin{equation}
\atomextend_F(\db,q)\defeq (\db',q'),
\end{equation}
where $\db' \defeq \dbsaturate(\db_0,q')$.

 It can be observed that $\queryvars{q}=\queryvars{q'}$. Therefore, and since $Z\subseteq \nkeyqueryvars{q}$ and $x\notin Y\cup Z$ by \Cref{lem:forest}, we have $Y \cup Z\subseteq \nkeyqueryvars{q'}$. 
 Due to $x\notin Y\cup Z$ we have that $q'$ inherits variable-repetition-freeness from $q$.
Also, due to normality, if $F\in \modei{q}$, then $\fd{x}{y}\in \FD{F}$ implies $\fd{x}{y}\in \WFD{q}$.
 
\smallskip
\begin{lemma}\label{lem:repfree}
Let $q\in \sjfuabcq$ be a saturated and normal. 
 Let $I\in \modei{q}$ be extendable.
 If $(\db',q')=\atomextend_I(\db,q)\notin\{\true,\false\}$, then $q'$ is variable-repetition-free.
\end{lemma}
\begin{proof}
Let us assume that
$F=R(\underline{x},\vec{y})$.
Recall that, beside dropping atoms, $q'$ only replaces $R(\underline{x},\vec{y})$ with $R'(\underline{x},\vec{y},\vec{z})$. Thus we only need to exclude the possibility that $Y\cap Z\neq \emptyset $ or $x\in Z$.
The first case implies that there are $y,y'\in Y$ such that $\fd{y}{y'}\in \igraph{q}$.
In particular, $\fd{y}{y'}\in \FD{q}$.
Since then $y\neq y'$ due to $q$ being variable-repetition-free, we obtain a contradiction with \Cref{lem:puuh}.
 The second case, in turn, leads to a cycle in $\igraph{q}$, in contradiction with \Cref{lem:forest}.
Therefore, $Y\cap Z= \emptyset $ and $x\notin Z$, and it is correct to conclude that $q'$ remains variable-repetition-free.
\end{proof}

\smallskip
\begin{lemma}\label{lem:little2}
Let $q\in \sjfuabcq$ be saturated and normal. 
Let $F\in \modei{q}$ be extendable.
 Let $u\in \queryvars{q} \setminus Y$ and $v \in \queryvars{q}$, where $Y\defeq \notkeyvars{F}$.
 If $(\db',q')=\atomextend_F(\db,q)\notin\{\true,\false\}$, then
 \begin{itemize}
\item $\fd{u}{v} \in \FD{q}$ implies $\fd{u}{v} \in \FD{q'}$,
\item $\fd{u}{v} \in \WFD{q}$ implies $\fd{u}{v} \in \WFD{q'}$,
\item $\fd{u}{v} \in \SFD{q}$ if and only if $\fd{u}{v} \in \SFD{q'}$.
\end{itemize}
\end{lemma}
\begin{proof}
The case where $u=v$ follows by \Cref{lem:repfree}, so we assume that $u\neq v$.
The first item is clear, and the second item follows from the first and the third item. Thus, we prove the third item.
Let us write $F'$ for the atom corresponding to $F\in q$ in $q'$; that is, $F'=F$ for $F\neq R$ and $F'=R'(\underline{x},\vec{y},\vec{z})$ for $F=R(\underline{x},\vec{y})$.

Assume that $\fd{u}{v}\in \SFD{q}$. 
Then we find distinct $G,H\in q$ such that $\fd{u}{v}\in \FD{G}$ and $\fd{u}{v}\in \FDclosure{H}$. Since  $u\notin Y$, it follows that $\fd{u}{v}\in \FD{G'}$ and $\fd{u}{v}\in \FDclosure{H'}$. 
 Thus $\fd{u}{v}\in \SFD{q'}$.

Assume that $\fd{u}{v}\in \SFD{q'}$. Then we find distinct $G',H'\in q'$ such that $\fd{u}{v}\in \FD{G'}$ and $\fd{u}{v}\in \FDclosure{H'}$.
If $\fd{u}{v}\in \FD{G}$ and $\fd{u}{v}\in \FDclosure{H}$, then $\fd{u}{v}\in \SFD{q}$. 

Thus, suppose first that $\fd{u}{v}\notin \FD{G}$. This is only possible when $G=R$, $u=x$, and $v\in Z$. Since $G$ and $H$ are distinct, we have $H=H'$ and consequently $\fd{u}{v}\in \FDclosure{H}$. Then we have $\fd{x}{v}\in \FD{H}$ or $\fd{\emp}{v}\in \FD{H}$. The first case contradicts \Cref{lem:noupstrong}, and the second case entails $\fd{\emp}{v}\in \SFD{q}$ by \Cref{lem:backward}, contradicting the normality of $q$. Thus this is case is not possible. 

Suppose then that $\fd{u}{v}\in \FD{G}$ and $\fd{u}{v}\notin \FDclosure{H}$. This is only possible when $H=R$ and $v\in Z$. Now, we have $\fd{y}{v}\in \igraph{q}\cap \SFD{q}$, for some $y\in Y$, and $\fd{u}{v}\in \FD{q}$, which entails using \Cref{lem:backward} that $\fd{u}{v}\in \SFD{q}$, as required.
\end{proof}

\smallskip
\begin{lemma}\label{lem:root2}
Let $q\in \sjfuabcq$ be saturated and normal. 
Let $F\in \modei{q}$ be extendable.
 If $(\db',q')=\atomextend_F(\db,q)\notin\{\true,\false\}$, 
then $\initqueryvars{q}=\initqueryvars{q'}$.
\end{lemma}
\begin{proof}
 We need to prove $\rootqueryvars{q}\cap \srcqueryvars{q} =\rootqueryvars{q'} \cap \srcqueryvars{q'}$.
 Let us write $F'$ for the atom corresponding to $F\in q$ in $q'$; that is, $F'=F$ for $F\neq R$ and $F'=R'(\underline{x},\vec{y},\vec{z})$ for $F=R(\underline{x},\vec{y})$. We divide the proof into separate claims, from which the lemma statement follow.
\begin{claim}
$\rootqueryvars{q} \subseteq \rootqueryvars{q'}$.
\end{claim}
\begin{claimproof}
Toward contradiction, suppose that $u\in \rootqueryvars{q}\setminus \rootqueryvars{q'}$. Then there exists an edge of the form $\fd{w}{u}\in\igraph{q'}$.  
Since $Y \cup Z\subseteq \nkeyqueryvars{q'}$, we have $w\notin Y\cup Z$. We consider two cases.

Suppose first that $\fd{w}{u}\in\WFD{q'}$.  Since $u\in \rootqueryvars{q}$, we also have $\fd{w}{u}\notin\WFD{q}$. 
Furthermore, we have $\fd{w}{u}\notin\SFD{q}$  by $w\notin Y$, 
 and  \Cref{lem:little2}.
 Therefore, $\fd{w}{u}\in \FD{q'}\setminus \FD{q}$, which is only possible when $w=x$ and $u\in Z$. However,
 $Z\cap  \rootqueryvars{q}=\emptyset$, while $u\in \rootqueryvars{q}$, a contradiction.

Suppose then that $\fd{w}{u}\in\FD{q'}\cap\SFD{q'}$. Then we find another edge $\fd{p}{w}\in \WFD{q'}$ such that $\fd{p}{u}\notin\SFD{q'}$.
Since $Y \cup Z\subseteq \nkeyqueryvars{q'}$ and $\fd{w}{u},\fd{p}{w}\in\FD{q'}$,
 we have $w,p\notin Y$. Hence, \Cref{lem:little2} implies $\fd{w}{u}\in\FD{q}\cap\SFD{q}$ and  $\fd{p}{u}\notin\SFD{q}$.
Since $u\in \rootqueryvars{q}$, it holds that
$\fd{p}{w}\notin \WFD{q}$ (otherwise, we would obtain $\fd{p}{u}\in \igraph{q}$). Moreover, by \Cref{lem:ctrans} we have $\fd{p}{w}\notin \SFD{q}$. Summing up, $\fd{p}{w}\in \FD{q'}\setminus \FD{q}$.  This entails $p=x$ and $w\in Z$. However, this contradicts $w\notin Z$, which was established above.
We conclude by contradiction that neither the case that $\fd{w}{u}\in\FD{q'}\cap\SFD{q'}$ is possible. Hence $\rootqueryvars{q} \subseteq \rootqueryvars{q'}$.
\end{claimproof}

\begin{claim}
$\rootqueryvars{q'} \subseteq \rootqueryvars{q}$.
\end{claim}
\begin{claimproof}
  Suppose $u\notin \rootqueryvars{q}$.  
By assumption there exists an edge of the form $\fd{w}{u}\in\igraph{q}$.  
Assume first that $\fd{w}{u}\in\WFD{q}$. Since $\outvar{\WFD{q}}{Y}=\emptyset$, this implies  $w\notin Y$, hence
 $\fd{w}{u}\in\WFD{q'}$  by \Cref{lem:little2}, i.e., $u\notin \rootqueryvars{q'}$.
 Then, suppose that $\fd{w}{u}\in\FD{q}\cap\SFD{q}$. This entails that we find another edge $\fd{p}{w}\in \WFD{q}$ such that $\fd{p}{u}\notin\SFD{q}$.
 In particular, we obtain $p\notin Y$ by $\outvar{\WFD{q}}{Y}=\emptyset$. 
 
 Assume first that $w\notin Y$. Then,
   \Cref{lem:little2} implies $\fd{p}{w}\in \WFD{q'}$, $\fd{w}{u}\in\SFD{q'}$, and $\fd{p}{u}\notin\SFD{q'}$, which is tantamount
to $\fd{w}{u}\in\igraph{q'}$, i.e., $u \notin \rootqueryvars{q'}$.

 Assume then that $w\in Y$. Then we obtain $u\in Z$ by $\fd{w}{u}\in\igraph{q}$. Toward contradiction, suppose $u\in \rootqueryvars{q'}$. Then, the edge $\fd{x}{u}$ that belongs to $\FD{q'}$ cannot belong to $\WFD{q'}$, i.e., it belongs to $\SFD{q'}$. In particular, there exists an edge $\fd{t}{u}$ in $\FD{q'\setminus \{F'\}}\subseteq \FD{q\setminus \{F\}}$, where $t\in \{x,\emp\}$. Note that $t=x$ would contradict \Cref{lem:noupstrong}, as the sequence $x,y,z$ form a path of length $2$ in $\igraph{q}$. 
 Hence $t=\emp$, in which case by normality we obtain $\fd{\emp}{u}\in \WFD{q}$. However, we also have $\fd{w}{u}\in\igraph{q}\cap\SFD{q}$, in which case $w\neq \emp$. Hence, we have a contradiction with \Cref{lem:backward}, by which we conclude that $u\notin \rootqueryvars{q'}$. This concludes the proof of the claim. 
  \end{claimproof}
  
 \begin{claim}
 $\srcqueryvars{q} \cap \rootqueryvars{q}\subseteq \srcqueryvars{q'}$.
 \end{claim}
 \begin{claimproof}
Suppose toward contradiction that the claim is false. 
Let $u\in  (\srcqueryvars{q} \cap \rootqueryvars{q}) \setminus \srcqueryvars{q'}$.
This means we find an edge $\fd{u_0}{u}\in \FD{q'}$ such that $u_0 \notin \reach{u}{\FD{q'}}$.
Moreover, $u\notin Y$ by $Y\cap \rootqueryvars{q}=\emptyset$.

By construction, it is not difficult to see that $u \in \reach{u_0}{\FD{q}}$; in particular, if $\fd{u_0}{u}$ is an edge from $x$ to a variable in $Z$, it can be replaced with two edges: one from $x$ to a variable $y\in Y$ such that $\fd{y}{z}\in \igraph{q}$, and the other being the edge $\fd{y}{z}$.

Now,
 let $D=(u, \dots ,u_0)$ be a shortest path witnessing $u \in \reach{u_0}{\FD{q}}$. 
We claim that $D$ gives rise to a path $D'$ in $\FD{q'}$ from $u$ to $u_0$, contradicting $u_0 \notin \reach{u}{\FD{q'}}$. For this, suppose $D$ contains an edge $\fd{p}{w}$ generated by $G\in q\setminus q'$. Then either $G=R$ or $p\in Y$. In the first case $\fd{p}{w}$ is generated by $R'\in q'$, so assume the second
case where $p\in Y$.  In this case, $p \neq u$ 
(as $u\notin Y$), and moreover 
\Cref{it:Y} entails $\fd{p}{w}\in\SFD{q}$. Let $\fd{r}{p}$ be the edge preceding $\fd{p}{w}$ in $D$. Then $\fd{r}{p}\notin \SFD{q}$, because otherwise by \Cref{lem:ctrans} we obtain $\fd{r}{p}\in \SFD{q}$ in contradiction with the minimality of $D$. Hence $\fd{r}{p}\in \WFD{q}$. 
Then  $p\in Y$, $\fd{x}{p}\in \WFD{q}$, and  
 \Cref{lem:samesourcesat}  entail $r=x$. 
 
 Summing up, we have shown that $\fd{p}{w}$ belongs to $\SFD{q}$ and is preceded by $\fd{x}{p}\in \WFD{q}$ in $D$, and $p\in Y$. 
  Note that $\fd{p}{w}\notin \igraph{q}$ implies $\fd{x}{w}\in \SFD{q}$, contradicting the minimality of $D$. 
  Hence we have $\fd{p}{w}\in \igraph{q}$. In this case, we obtain $w\in Z$ and $\fd{x}{w}\in \FD{R'}$, meaning that the edges $\fd{x}{p}$ and $\fd{p}{w}$ in $D$ can be replaced with the edge $\fd{x}{w}$ from $\FD{q'}$. In this way, $D$ gives rise to a path $D'$ in $\FD{q'}$ from $u$ to $u_0$, in contradiction with  $u_0 \notin \reach{u}{\FD{q'}}$. Therefore we conclude that $\srcqueryvars{q} \cap \rootqueryvars{q}\subseteq \srcqueryvars{q'}$.
  \end{claimproof}
  
  \begin{claim}
 $\srcqueryvars{q'} \subseteq \srcqueryvars{q}$.
 \end{claim}
 \begin{claimproof}
Assuming $u\in \srcqueryvars{q'}$, suppose toward contradiction that $u\notin \srcqueryvars{q}$. Then we find an edge $\fd{u_0}{u} \in \FD{q}$ such that $u_0 \notin \reach{u}{ \FD{q}}$.
If $\fd{u_0}{u} \in \FD{q'}$, then by assumption $u_0 \in \reach{u}{ \FD{q'}}$, which implies $u_0 \in \reach{u}{ \FD{q}}$; for this, recall that any edge in $\FD{q'}\setminus \FD{q}$ proceeds from $x$ to $Z$ and is replaceable with two edges in $\FD{q}$. However, this contradicts the fact that $u_0 \notin \reach{u}{ \FD{q}}$.
 Thus, suppose $\fd{u_0}{u} \notin \FD{q'}$. This entails $u_0\in Y$. 
 By \Cref{it:Y} we then obtain $\fd{u_0}{u} \in \SFD{q}$, and moreover, we have $\fd{x}{u_0}\in \WFD{q}$.
 
 Consider first the possibility that $\fd{u_0}{u} \in \igraph{q}$. This means that $u\in Z$. 
 However, as $\fd{x}{u}\in \FD{q'}$ and $ Z\subseteq \nkeyqueryvars{q'}$, this creates a contradiction with $u\in \srcqueryvars{q'}$.
 Hence, suppose $\fd{u_0}{u} \notin \igraph{q}$. In this case, however, $\fd{x}{u} \in \SFD{q}$ obtains. Since $x\notin Y$, we then have 
 $\fd{x}{u} \in \SFD{q'}$ by \Cref{lem:little2}.  Thus, by $u\in \srcqueryvars{q'}$ we obtain $x \in \reach{u}{ \FD{q'}}$, which implies $x \in \reach{u}{ \FD{q}}$. Hence, composing with the edge $\fd{x}{u_0}\in \WFD{q}$ we obtain $u_0 \in \reach{u}{ \FD{q}}$, a contradiction.
 We conclude by contradiction that $\srcqueryvars{q'} \subseteq \srcqueryvars{q}$.
 \end{claimproof}
 This concludes the proof.
\end{proof}

\smallskip
\begin{lemma}\label{lem:nr5}
Let $q\in \sjfuabcq$ be a saturated and normal. 
 Let $F\in \modei{q}$ be extendable.
 If $(\db',q')=\atomextend_F(\db,q)\notin\{\true,\false\}$, then
  $\nrqueryvars{q}=\nrqueryvars{q'}$.
\end{lemma}
\begin{proof}
Let $v\in \nrqueryvars{q}$. 
Then, there exists 
 $u\in \initqueryvars{q}$ such that $\fd{u}{v}\in \SFD{q}$.
In particular, $u\notin Y$, since $Y\cap \rootqueryvars{q}=\emptyset$ (and, by definition, $ \initqueryvars{q}\subseteq \rootqueryvars{q}$).
By \Cref{lem:root2} we obtain $u\in \initqueryvars{q'}$, and by \Cref{lem:little2} we have $\fd{u}{v}\in \SFD{q'}$. 
Hence $v\in \nrqueryvars{q'}$.
  
Conversely, assume $v\in \nrqueryvars{q'}$. 
Then, we find 
 $u\in \initqueryvars{q'}$ such that $\fd{u}{v}\in \SFD{q'}$.  
Note that $u\notin Y$; otherwise, if $u\in Y$, we obtain $\fd{x}{u}\in \FD{q'}$, which together with $Y\subseteq \nkeyqueryvars{q'}$
implies $u\notin \srcqueryvars{q}$, a contradiction. (Recall that by definition, $\initqueryvars{q'}\subseteq \srcqueryvars{q'} $.)
Now, by \Cref{lem:root2} we obtain $u\in \initqueryvars{q}$, and by \Cref{lem:little2} we have $\fd{u}{v}\in \SFD{q}$. Hence $v\in \nrqueryvars{q}$.
\end{proof}

\smallskip
\begin{lemma}\label{lem:strongel}
Let $q\in \sjfuabcq$ be a saturated and normal. 
 Let $I\in \modei{q}$ be extendable.
 If $(\db',q')=\atomextend_I(\db,q)\notin\{\true,\false\}$, then
\begin{enumerate}[(a)]
\item\label{it:saetwo}  $q'\in \sjfuabcq$; 
\item\label{it:saeone} $q'$ is saturated;
\item\label{it:saeoneone} $q'$ is normal.
\end{enumerate}
\end{lemma}
\begin{proof}
Let us write $F'$ for the atom corresponding to $F\in q$ in $q'$; that is, $F'=F$ for $F\neq R$ and $F'=R'(\underline{x},\vec{y},\vec{z})$ for $F=R(\underline{x},\vec{y})$.

\textbf{(\cref{it:saetwo})} 
By  \Cref{lem:repfree}, it suffices to show that if $F'\attacks{q'} G'$, then $F\attacks{q} G$. 
    Suppose $F'\attacks{q'} G'$.
Then  $\gaifman{q'}$ contains an (undirected) path 
\[
U=x_1, \dots ,x_\ell \quad (\ell \geq 1)
\]
between some $x_1 \in \notkeyvars{F'}$ and   $x_\ell\in \atomvars{G'}$
such that $U$ is not separated by
$\keyclosure{(F')}{q'}$. 
Assume toward contradiction that $U$ is separated by $\keyclosure{F}{q}$. Then $\FD{q\setminus \{F\}}$ contains a directed graph 
\[
D=t, \dots ,x_i,
\]
 where $t\in\{\keyvars{F},\emp\}$. Since, for each $H'\in q'$, $ \keyvars{H}=\keyvars{H'} \notin Y$, we obtain $t\notin Y$.
 
 We may assume that $D$ is the shortest path witnessing $x_i\in \keyclosure{F}{q}$. In particular, $D$ is simple.
Then, an edge $\fd{u}{v}$, where $u\in Y$, appears in $D$.  In particular, $u\neq t$. 
\begin{claim}\label{claim:write}
Assume that an edge $\fd{u}{v}$, where $u\in Y$, appears in $D$. Then $v\in Z$, and $u$ is preceded in $D$ by $x$. Moreover, $\fd{u}{v}\in \igraph{q}$.
\end{claim}
\begin{claimproof}
Since $t\notin Y$, we have $u\neq t$. Hence $u$ is preceded in $D$ by some vertex $w$. That is, the edge $\fd{w}{u}$ precedes the edge $\fd{u}{v}$ in $D$.


 We prove that $w=x$. Since $u\in Y$, by \cref{it:Y} we have $\fd{u}{v}\in \SFD{q}$. If $\fd{w}{u}\in \SFD{q}$, we obtain by \Cref{lem:ctrans} that $\fd{w}{v}\in \SFD{q}$, which implies $D$ is not a shortest path witnessing $x_i\in \keyclosure{F}{q}$, a contradiction.
 Hence, assume $\fd{w}{u}\in \WFD{q}$. Since also $\fd{x}{u}\in \WFD{q}$, we obtain by \cref{lem:samesource} that $x=w$.
 
 To complete the proof, we prove $\fd{u}{v}\in \igraph{q}$ and $v\in Z$. We now have  $\fd{x}{u}\in \WFD{q}$ and $\fd{u}{v}\in\FD{q}\cap \SFD{q}$. If $\fd{u}{v}\notin \igraph{q}$, then by \Cref{lem:verysimple} $\fd{x}{v}\in \SFD{q}$, in which case the path 
$D^* \defeq t, \dots x,v, \dots ,x_i$ witnesses $x_i\in \keyclosure{F}{q}$, contradicting the minimality of $D$. Hence $\fd{u}{v}\in \igraph{q}$, which also entails $v\in Z$.
\end{claimproof}

By \Cref{claim:write}, we can write
\[
D=t, \dots x,u,v, \dots ,x_i,
\]
where $v\in Z$, and possibly $t=x$ or $v=x_i$. 

Since $\fd{x}{u}\in \WFD{q}$, and $R\in q$ is the unique atom in $q$ generating $\fd{x}{u}$, we obtain $R\neq F$. In particular, $R'\neq F'$, which entails $\fd{x}{v}\in \FD{q\setminus \{F'\}}$ due to $v\in Z$.

We now claim that the paths $D[t:x]$ and $D[v:x_i]$ belong also to $\FD{q\setminus \{F'\}}$. This entails $x_i \in \keyclosure{(F')}{q'}$, contradicting the assumption. 

Recall now $q_0$ from \eqref{eq:q0}.
By \Cref{claim:write}, if some edge in $D[t:x]$ or $D[v:x_i]$ is generated by an atom of $q_0$, then $x$ appears twice in $D$, contradicting the minimality assumption. Since $\FD{q\setminus (\{F\} \cup q_0)}  \subseteq \FD{q'\setminus \{F'\}}$, we obtain that each  edge in $D[t:x]$ or $D[v:x_i]$ belongs to $\FD{q'\setminus \{F'\}}$. This creates a contradition, as described above. We conclude by contradiction that $U$ is not separated by $\keyclosure{F}{q}$

If $U$ is a path between $\notkeyvars{F}$ and $\atomvars{G}$, then $F\attacks{q} G$.

 Otherwise, suppose $U$ is not a path of $\gaifman{q}$ between $\notkeyvars{F}$ and $\atomvars{G}$. Then there is at least one $z\in Z$ such that
 one of the following cases applies:
 \begin{itemize}
 \item Case A: $z$ appears as an endpoint in $U$ and belongs to $\notkeyvars{F'}$ (resp. $\atomvars{G'}$), while $z$ does not belong to $\notkeyvars{F}$ (resp. $\atomvars{G}$). (Informally, $F$ or respectively $G$ is $R$.)
 \item Case B: $\{u,z\}$, for some $u\in \{x\}\cup Y$, is an edge in $U$ which does not belong to $\gaifman{q}$. (Informally, $U$ contains an edge generated by $R$ but not $R'$.)
 \end{itemize} 
 For each such $z$, choose an arbitrary $y\in Y$ such that $\fd{y}{z}\in \igraph{q}$. Then, proceed in cases as follows:
 \begin{itemize}
 \item Case A: extend the endpoint $z$ with the  edge $\{y,z\}$.
  \item Case B: replace the edge $\{u,z\}$ with the edges $\{u,x\},\{x,y\},\{y,z\}$ when $u\in Y$, and with the edges $\{u,y\},\{y,z\}$ when $u=x$.
 \end{itemize} 
  Denote by $U'$ the undirected path obtained from $U$ by treating each occurring case in this way. Then, $U'$ is a path between $\notkeyvars{F}$ and  $\atomvars{G}$.
  
We claim that $U'$ is not separated by $\keyclosure{F}{q}$. Assume toward contradiction that $U'$ is separated by $\keyclosure{F}{q}$. As we have already established that $U$ is not separated by $\keyclosure{F}{q}$, we obtain that $x$ or $y$ belongs to $\keyclosure{F}{q}$. Since $\fd{x}{y}\in \FD{q\setminus \{F\}}$ (due to $R\neq F$), and $\fd{y}{z}\in \SFD{q}\subseteq \FD{q\setminus \{F\}}$ (due to \cref{it:Y}), we obtain $z\in \keyclosure{F}{q}$. As $z$ belongs to $U$, this contradicts the fact that $U$ is not separated by $\keyclosure{F}{q}$. Hence $U'$ is not separated by $\keyclosure{F}{q}$. In particular, we obtain $F\attacks{q} G$. This concludes the proof of the first item.

\textbf{(\cref{it:saeone})} 
 Suppose  that for some $F'\in q'$, we have    $\fd{u}{v}\in \FD{F'}$,
      $\fd{v}{w}\in \SFD{q'}$, and
    $\fd{u}{w}\in \FDclosure{q'\setminus \{F'\}}$.  We need to prove  that  $\fd{u}{w}\in \SFD{q'}$
    
    We show first that  $\fd{u}{w}\in \SFD{q}$. For this, we first prove two helping claims.

    \begin{claim}\label{claim:first}
    $\fd{v}{w}\in \SFD{q}$.
    \end{claim}
    \begin{claimproof}
    By assumption there are $G'\in q'$ and $H'\in q'\setminus \{G'\}$ such that $\fd{v}{w}\in \FD{G'}$ and $\fd{v}{w}\in \FDclosure{H'}$.
     
     We claim that  $\fd{v}{w}\in \FD{G}$.
     Toward contradiction, suppose $\fd{v}{w}\notin \FD{G}$.  Then  $G'=R'$, $v=x$, and $w\in Z$. In this case, since $G' \neq H'$, we have $H'=H\in q$. Now, either $\fd{x}{w}$ or $\fd{\emp}{w}$ belongs to $\FD{H'}$. The former case violates \cref{lem:noupstrong}, because there exists $y\in Y$ such that $\fd{x}{y}\in \WFD{q}$ and $\fd{y}{w}\in \igraph{q}$. In the latter case, we have $\fd{y}{w}\in \igraph{q}\cap \SFD{q}$ (by \cref{it:Y}), for some $y\in Y$.
     However, by normality of $q$, we also have $\fd{\emp}{w}\in \WFD{q}$, which creates a contradiction with \cref{lem:backward}.  Hence $\fd{v}{w}\in \FD{G}$. 

If also $\fd{v}{w}\in \FDclosure{H}$, then the claim follows.
    Thus, assume  $\fd{v}{w}\notin \FDclosure{H}$. Then $H'=R'$ and $w\in Z$. In this case, since $G' \neq H'$, we have $G'=G\in q$. In particular, $\fd{v}{w}\in \FD{q}$. Since $w\in Z$, we obtain  $\fd{y}{w}\in \igraph{q}$, for some $y\in Y$. By \cref{it:Y}, it holds that $\fd{y}{w}\in \SFD{q}$. Hence, 
    $\fd{v}{w}\in \SFD{q}$ by \Cref{lem:backward}, establishing the claim.
\end{claimproof}

\begin{claim}\label{claim:second}
 $\fd{u}{w}\in \FDclosure{q\setminus \{F\}}$. 
  \end{claim}
  \begin{claimproof}
  Let $D=t, \dots ,w$ be a path in $\FD{q'\setminus \{F'\}}$ where $t\in\{ u,\emp\}$. Then, every edge from $D$ belongs to $\FD{q\setminus \{F\}}$, except possibly for an edge of the form $\fd{x}{z}$, $z\in Z$, that has  been generated by $R'$. In this case, $R'\neq F'$, and the edge can be replaced with two edges $\fd{x}{y},\fd{y}{z}$, where $y\in Y$. In particular, the edge $\fd{x}{y}$ is generated by $R\neq F$, and the edge $\fd{y}{z}$ is generated by an atom $G\in q$ that does not appear in $q'$ (hence $G\neq F$). This proves the claim.
  \end{claimproof}

Let us then consider two alternative cases.

     {\bf Case $\fd{u}{v}\in \FD{F}$.} Then, Claims \ref{claim:first} and \ref{claim:second} and the saturation of $q$ imply $\fd{u}{w}\in \SFD{q}$.

    {\bf Case $\fd{u}{v}\notin \FD{F}$.} Then $F'=R'$, $u=x$, and $v\in Z$. In this case, however, we find $y\in Y$ such that $\fd{u}{y}\in \FD{F}$ and $\fd{y}{v} \in \igraph{q}$. Moreover, we have $\fd{y}{v}\in \SFD{q}$ by \cref{it:Y}. Then, Claim \ref{claim:first} and \Cref{lem:ctrans} yield $\fd{y}{w}\in \SFD{q}$. Then, using \Cref{claim:second} and the saturation of $q$, we obtain $\fd{u}{w}\in \SFD{q}$.
    
    Hence we have proven that $\fd{u}{w}\in \SFD{q}$. To finish the proof of the item,
 assume toward contradiction that $\fd{u}{w}\notin \SFD{q'}$. By the construction of $q'$, it is correct to conclude that $u\in Y$.
 However, by assumption $q'$ includes the atom $F'$ such that $\keyvars{F'}=u$. This contradicts the fact that $Y\subseteq \nkeyqueryvars{q'}$. We conclude by contradiction that  $\fd{u}{w}\in \SFD{q'}$.

     \textbf{(\cref{it:saeoneone})} 
In what follows we consider an atom $F'\in q'$. Toward contradiction, suppose $q'$ is not normal. We consider the different cases of normality in the following.

(\cref{it:normal1}) Suppose that $|\notkeyvars{F'}|>1$, and $\fd{p}{u} \in \SFD{q'}$, for $p= \keyvars{F'}$
and some $u\in \notkeyvars{F'}$. 
Then $\fd{p}{u} \in \FD{F'}$, and there exists $G'\in q'\setminus \{F'\}$ such that $\fd{p}{u} \in \FDclosure{G'}$.

{\bf Case 1: $F'\neq R'$}. Then $F=F'$, hence $|\notkeyvars{F}|>1$ and $\fd{p}{u} \in \FD{F}$. If $\fd{p}{u} \in \FDclosure{G}$, 
then $\fd{p}{u} \in \SFD{q}$, contradicting the normality of $q$. Suppose thus that $\fd{p}{u} \notin \FDclosure{G}$. This implies $G=R$ and $u\in Z$. 
Hence, we find $y\in Y$ such that $\fd{y}{u}\in \igraph{q}$,  and consequently, by \cref{it:Y}, $\fd{y}{u}\in \SFD{q}$. Thus, by $\fd{p}{u} \in \FD{q}$ and \Cref{lem:backward} we obtain $\fd{p}{u} \in \SFD{q}$.  This contradicts the normality of $q$.

{\bf Case 2: $F'= R'$}. Then $p=x$ and $u\in Y\cup Z$. Since $G'\neq F'$, we obtain $G=G'$ and hence $\fd{x}{u} \in \FDclosure{G}$.
Moreover, we have $F=R$. 

Suppose that $\fd{x}{u} \notin \FD{F}$. Since $\fd{x}{u} \in \FD{F'}$,  this entails $u\in Z$. From $\fd{x}{u} \in \FDclosure{G}$ we obtain that there exists $v\in \{x,\emp\}$ such that $\fd{v}{u}\in \FD{q}$. From this and $u\in Z$, we obtain $\fd{v}{u}\in \SFD{q}$ using again \Cref{lem:backward}. By normality of $q$, then, $v\neq \emp$, i.e., $v=x$. However, $\fd{x}{u}\in \FD{q}$ contradicts \Cref{lem:noupstrong}.

Suppose then that $\fd{x}{u} \in \FD{F}$. From $\fd{x}{u} \in \FDclosure{G}$ and $G\neq F$ it follows that $\fd{x}{u}\in \SFD{q}$. However, we have $F=R$, and thus by the assumption that $R\in \modei{q}$, it follows that $\fd{x}{u} \in \WFD{q}$, a contradiction.

(\cref{it:normal2}) Suppose toward contradiction that $\fd{u}{v} \in \SFD{q'}$, for some distinct $u,v\in \notkeyvars{F'}$. 
Then, in particular, $\fd{u}{v} \in \FD{q'}$.
It follows that $u\notin Y\cup Z$, since $ Y\cup Z\subseteq \nkeyqueryvars{q'}$. 
Since $\notkeyvars{R'}\subseteq Y \cup Z$, we obtain $F'\neq R'$, and hence $F'=F$. Now, $\fd{u}{v} \in \SFD{q'}$ means that there exist $G'\in q'$ and $H'\in q'\setminus \{G'\}$ such that $\fd{u}{v} \in \FD{G'}$ and $\fd{u}{v} \in\FDclosure{H'}$. 

We claim that $\fd{u}{v} \in \FD{G}$. Toward contradiction, suppose $\fd{u}{v} \notin \FD{G}$. Since $\fd{u}{v} \in \FD{G'}$, this entails $G=R$,  $u=x$, and $v\in Z$. 
Then $H'\neq R'$, i.e., $H'=H$ and we have that $\fd{x}{v}$ or $\fd{\emp}{v}$ belongs to  $\FD{H}$.
The first case contradicts \Cref{lem:noupstrong}. The second case entails, by normality, that $\fd{\emp}{v}\in \WFD{q}$; this, however, contradicts \Cref{lem:backward} due to $\fd{y}{v}\in \SFD{q}\cap \igraph{q}$, for some $y\in Y$. The claim thus follows by contradiction.

If $\fd{u}{v} \in\FDclosure{H}$, then by the above claim we obtain $\fd{u}{v} \in \SFD{q}$. Since $u,v$ are distinct variables in $\notkeyvars{F}$, this contradicts the normality of $q$.
 Thus, assume 
  $\fd{u}{v} \notin\FDclosure{H}$. Since $\fd{u}{v} \in\FDclosure{H'}$, this entails $H=R$ and $v\in Z$.  
Since $\fd{u}{v} \in \FD{q}$ and $v\in Z$, we obtain $\fd{u}{v}\in \SFD{q}$ using again \Cref{lem:backward}.
As above, this contradicts the normality of $q$.

(\cref{it:normal4}) Suppose, toward contradiction that $\fd{u}{v},\fd{v}{u}\in \SFD{q'}$, for some distinct $u,v\in \queryvars{q'}$. In particular,
$\fd{u}{v},\fd{v}{u}\in \FD{q'}$. 
Since $ Y\subseteq \nkeyqueryvars{q'}$, we obtain $u,v\in \queryvars{q}\setminus Y$, hence \Cref{lem:little2} entails
$\fd{u}{v},\fd{v}{u}\in \SFD{q}$, a contradiction with the normality of $q$.

(\cref{it:normal3}) Suppose toward contradiction that $\fd{\emp}{z} \in \SFD{q'}$, for some $z\in \queryvars{q'}$. Then $\fd{\emp}{z} \in \FD{G'}\cap \FD{H'}$, for some distinct $G',H'\in q'$. Now, $\fd{\emp}{z} \notin \FD{G}\cap \FD{H}$ due to the normality of $q$. 
Suppose, by symmetry, that $\fd{\emp}{z} \notin \FD{G}$. This is only possible when $G=R$ and $z\in Z$. 
Since $G',H'$ are distinct, so are $G,H$, and we obtain $H'=H$. In particular, we have $\fd{\emp}{z} \in  \FD{H}$, i.e., $\fd{\emp}{z} \in  \FD{q}$. Using  \Cref{lem:backward}, we again obtain $\fd{\emp}{z} \in  \SFD{q}$, but this contradicts the normality of $q$.
\end{proof}

\smallskip
\begin{lemma}[Consistent Edge Elimination]\label{lem:vedb01}
Let $\mathcal{C}$ consist of all $(\db,q)\in \DB\times \sjfuabcq$, where $q$ is a saturated and normal query and $\db$ is a $q$-saturated database.
For each $(\db,q)\in \mathcal{C}$ and extendable $F\in \modei{q}$, 
$(\db,q)\mapsto \atomextend_F(\db,q)$ is a strict, guarded, equivalence-preserving, and  polynomial-time reduction from $\mathcal{C}$ to $\mathcal{C}\cup\{\true,\false\}$.
 \end{lemma}
\begin{proof}
Recall that we write $Y=\notkeyvars{F}$ and $Z= \outvar{\igraph{q}}{Y}$. Recall that both $Y$ and $Z$ are non-empty by definition.

The reduction $\atomextend_F$ is clearly strict and computable in polynomial time. It is into $\mathcal{C}\cup\{\true,\false\}$
by \Cref{lem:strongel,lem:dbsat}. It suffices to prove that it is guarded and equivalence-preserving.

{\bf (Guardedness)}
\Cref{it:inv:guarded,it:inv:guarding} of guardedness follow by
by \Cref{lem:root2,lem:nr5}. \Cref{it:sizeinv,it:inv:2,it:sizeinv,it:dummy,it:dummy2} are clear by construction.
We prove  
  \Cref{it:inv:const,it:inv:1}.

For \Cref{it:inv:const}, assume that $G\in \modec{q}$, and either $\keyvars{G}\in \initqueryvars{q}$ or $\atomvars{G}= \emptyset$.
 Toward contradiction, suppose first that $G\notin {q'}$.
 Then, by construction, $\keyvars{G}\in Y$ or $G=F$. The former case is not possible since $Y\cap \initqueryvars{q}=\emptyset$. The latter case is not possible since $F\in \modei{q}$ by the definition of extendability. This proves that $G\in q'$. Moreover, $G\in \modec{(q')}$, because $G\in \modei{(q')}$ means that there exists an edge $\fd{u}{v}\in \WFD{q'}\cap \FD{G}$.
 Since $Y\subseteq \nkeyqueryvars{q'}$, this would imply $u\notin Y$. Thus, by \Cref{lem:little2}, we would obtain $\fd{u}{v}\in \WFD{q}$, which entails $G\in \modei{q}$, a contradiction.
Hence,  \Cref{it:inv:const} holds.  
 
For \Cref{it:inv:1}, 
assume that for each $S\in q$,
\(
|\restrict{S^{\db}}{V}| \leq N,
\)
where $V\defeq \atomvars{S}\cap \rqueryvars{q}$.
By \Cref{lem:nr5}, and due to  $\queryvars{q}= \queryvars{q'}$, we have $\rqueryvars{q}= \rqueryvars{q'}$. 

Let $S\in q\setminus (q_0 \cup\{R\})$. Then by \cref{it:sizeinv} of guardedness (whose verification we omitted), we obtain $S^{\db'}\leq S^{\db}$. Hence, by assumption we obtain
\(|\restrict{S^{\db'}}{V'}| \leq |\restrict{S^{\db}}{V}|\leq   N,\)
where $V'\defeq \atomvars{S}\cap \rqueryvars{q'}$ and $V\defeq \atomvars{S}\cap \rqueryvars{q}$ (i.e., $V'=V$).

Consider then the atom $R'\in q'$. 
Since $Z\subseteq \nrqueryvars{q}$, by \Cref{lem:nr5} we obtain $Z\subseteq \nrqueryvars{q'}$.
In particular, $Z\cap \rqueryvars{q'}=\emptyset$.
Hence, writing $V'\defeq \atomvars{R'}\cap \rqueryvars{q'}$ and $V\defeq \atomvars{R}\cap \rqueryvars{q}$, 
we have $V'= V$ by \Cref{lem:nr5}. By assumption and the construction of $\db'$, this leads to
\(|\restrict{(R')^{\db'}}{V'}| = |\restrict{R^{\db}}{V}|\leq  N\). Hence, \Cref{it:inv:1} holds, too.


We may thus conclude that the stated reduction  is guarded.

{\bf (Equivalence preservation)}
By \Cref{lem:dbsat} it suffices to prove that $(\db,q)\in \cert$ if and only if $(\db_0,q')\in \cert$, assuming $\db_0$ and $q'$ are as in \eqref{eq:strongel2} and
\eqref{eq:strongel}.

($\Rightarrow$) 
Suppose  $(\db,q)\in \cert$, and let  $\rep'$ be an arbitrary repair of $\db_0$. 
Let
\[
\rep \defeq (\rep' \setminus  \restrict{\rep'}{R'}) \cup \{R(\underline{a},\vec{b})\in \db \mid \exists \vec{c}\colon R'(\underline{a},\vec{b},\vec{c})\in \rep'\} \cup  \restrict{\db}{q_0}.
\]
Note that $\restrict{\db}{q_0}$ is consistent by \eqref{eq:subset} and the $q$-saturation of $\db$. Also, $\{R(\underline{a},\vec{b})\in \db \mid \exists \vec{c}\colon R(\underline{a},\vec{b},\vec{c})\in \rep'\}$ is consistent since $\rep'$ is consistent. Moreover, since $\db$ is $q$-saturated, for each $R(\underline{a},\vec{b}) \in \db$ there exists a (unique) extension $R'(\underline{a},\vec{b},\vec{c})\in \db_0$.
Hence it can be seen that $\rep$ is a repair of $\db$.

Since $\rep\models q$ by assumption, there exists a valuation $\theta$ on $\queryvars{q}$ such that $\theta(q)\subseteq \rep$. 
In particular, we obtain $\theta(q' \setminus \{R'\})\subseteq \rep'$ and $\theta(R)\in \rep$. 
We claim that  $\theta(R')\in \rep'$. 

Let $\vec{c}$ be the (unique) sequence of constants such that $R'(\underline{\theta(x)},\theta(\vec{y}),\vec{c})\in \rep'$.
Let $i\in [\ell]$. By definition we have $S_i(\underline{\theta(y_{j_i})},c_i,\uk)\in \db$.  
Since $\fd{y_{j_i}}{z_i}\in \SFD{q}$, $q$ is normal, and $\db$ is $q$-saturated, it follows that $\restrict{\db}{S_i}$ is consistent.
In particular, we have
 $S_i(\underline{\theta(y_{j_i})},c_i,\uk)\in \rep$. Since $S_i(\underline{\theta(y_{j_i})},\theta(z_i),\uk)\in \rep$, it is the case that $\theta(z_i)=c_i$. This proves $\theta(R')\in \rep'$. 
Hence $\rep'\models q'$, which leads to $(\db_0,q')\in \cert$.

($\Leftarrow$)
Suppose  $(\db_0,q')\in \cert$, and let  $\rep$ be an arbitrary repair of $\db$. 
Let
\[
\rep' \defeq (\rep \setminus  \restrict{\rep}{q_0 \cup\{R\}}) \cup \{R'(\underline{a},\vec{b},\vec{c})\in \db_0 \mid  R(\underline{a},\vec{b})\in \rep\}.
\]
Since $\rep$ is consistent, for each $\vec{a}$ there exists at most one $\vec{b}$ such that $R(\underline{a},\vec{b})\in \rep$.
Moreover, for each $R(\underline{a},\vec{b})\in \rep$, there is precisely one extension $R(\underline{a},\vec{b},\vec{c})\in \db_0$; this follows
by definition since $\fd{y_{j_i}}{z_i}\in \SFD{q}$, for each $i\in [\ell]$, and $\db$ is $q$-saturated. 
Hence, $\{R'(\underline{a},\vec{b},\vec{c})\in \db_0 \mid  R(\underline{a},\vec{b})\in \rep\}$ is consistent.
Furthermore, since $\rep$ is a repair, for each $R'(\underline{a},\vec{b},\vec{c})\in \db_0$ there exists a fact $R(\underline{a},\vec{b}')\in \rep$
with the same primary-key value.
Now, for each $i\in [\ell]$, as $S_i(\underline{y_{j_i}},z_i,\uk)\in q$  and $\db$ is $q$-saturated, 
there is a (unique) fact of the form $S_i(\underline{b'_{j_i}},c'_i,\uk)\in \db$. Hence, setting $ \vec{c}'\defeq (c'_1, \dots ,c'_\ell)$,
we obtain $R'(\underline{a},\vec{b}',\vec{c}')\in \db_0$. Thus we may conclude that $\rep'$ picks precisely one fact from each $R'$-block of $\db_0$.
It follows that $\rep'$ is a repair of $\db_0$.

By assumption
we find a valuation $\theta$ such that $\theta(q')\subseteq \rep'$. We prove that $\theta(q)\subseteq \rep$.

Clearly, it holds $\theta(q\setminus q_0)\subseteq \rep$. Let $H\in q_0$. Then $\keyvars{H}\in Y$, and
as observed in \eqref{eq:subset}, $H\in \modec{q}$. In particular, we can write
 $H=S(\underline{y},v_1, \dots ,v_\ell)$, where $y\in Y$, $y\notin \{v_1, \dots ,v_\ell\}$, and $\fd{y}{v_i}\in \SFD{q}$ whenever $v_i$ is a variable. 
 Note that $\restrict{\db}{S} $ is consistent, which means that $\restrict{\db}{S} =\restrict{\rep}{S}$.
 We prove $\theta(H)\in \rep$.
 
 If $\notkeyvars{H}=\emptyset$, then, since $\db$ is $q$-saturated and $\theta(R)\in \rep$, the reader can verify that $\theta(H)\in \rep$. Thus, suppose $\notkeyvars{H}\neq \emptyset$.
 Then, by the saturation and normality of $q$, $H$ is of the form $S(\underline{y},z, \vec{c})$, where $z$ is a variable such that $\fd{y}{z}\in \SFD{q}$ and
$\vec{c}$ is a (possibly empty) sequence of constants

 Let us first isolate a single variable $v\in \notkeyvars{H}$, and write $H=S(\underline{y},v, \uk)$. 
 We  prove the following claim. 
 \begin{claim}\label{claim:V}
 $S(\underline{\theta(y)},\theta(v),\uk)\in \db$.
 \end{claim}
 \begin{claimproof}
 Suppose first that $\fd{y}{v}\in \igraph{q}$. Then $v\in Z$, and in particular $v = z_i $, for some $i\in [\ell]$. From $R'(\underline{\theta(x)}, \theta(\vec{y}), \theta(\vec{z}))\in \rep'$ we obtain $S_i(\underline{\theta(y_{j_i})},\theta(v),\uk)\in \db$, where $\fd{y_{j_i}}{v}\in \SFD{q}$.  Note that we also have $R(\underline{\theta(x)}, \theta(\vec{y}))\in \db$.
 Since $y\in Y$ (i.e., $y$ is listed in $\vec{y}$) and $\fd{y}{v}\in \SFD{q}$, we obtain $S(\underline{\theta(y)},\theta(v),\uk))\in \db$ by $q$-saturation of $\db$ (collision consistency).

Suppose then that  $\fd{y}{v}\notin \igraph{q}$. Since $\fd{y}{v}\in \FD{q}\cap\SFD{q}$ and $\fd{x}{y}\in \WFD{q}$, we obtain $\fd{x}{v}\in \SFD{q}$. Let $T(\underline{x},v,\uk)\in \modec{q}$ be an atom generating $\fd{x}{v}$.



Note that $x\notin Y$ due to variable-repetition-freeness of $q$.  
Then, as  $T\neq R$ due to $R\in \modei{q}$, we obtain  $T\in q\setminus (q_0 \cup\{R\}) \subseteq q'$. Hence $\theta(q')\subseteq \rep'$ implies $T(\underline{\theta(x)},\theta(v),\uk)\in \rep$. 
Since furthermore $R(\underline{\theta(x)}, \theta(\vec{y}))\in \db$, $\fd{y}{v}\in \SFD{q}$, and $y\in Y$ (i.e., $y$ is a variable listed in $\vec{y}$), it follows by $q$-saturation of $\db$ that $S(\underline{\theta(y)},\theta(v),\uk)\in \db$.
This proves the claim.
\end{claimproof}


 Since $\restrict{\db}{S}$ is consistent and $\db$ is $q$-saturated, the claim implies  $\theta(H)=S(\underline{\theta(y)},\theta(v),\vec{c}) \in \rep$. 
 Thus $\theta(q)\subseteq \rep$, concluding the proof.
\end{proof}

\subsection{Inconsistent Edge Elimination}
We say that $S\in \modei{q}$, with 
$Z \defeq \notkeyvars{S}$,
 is \emph{eliminable} if the following conditions hold:
 \begin{enumerate}
 \item\label{it:aeone} $\keyvars{S} \in \outvarstar{\igraph{q}}{\bot}\cap \nrootqueryvars{q}$;
 \item\label{it:aetwo} $Z\subseteq  \leafqueryvars{q}\cap \nkeyqueryvars{q}\cap \nrqueryvars{q} $; 
\item\label{it:aenotempty} $Z\neq \emptyset$.
 \end{enumerate}
 
Let $ S(\underline{y},\vec{z})\in \modei{q}$ be eliminable. 
 We define a query $q'$ by two cases.
\begin{enumerate}[(i)]
\item\label{it:atomel}
The incoming edge of $y$ in $\igraph{q}$ is of the form $\fd{x}{y} \in \WFD{q}$ generated by $R(\underline{x},y,\vec{u}) \in \modei{q}$.
 We define
\begin{equation}\label{eq:atomel}
q' \defeq (q \setminus  \{R(\underline{x},y,\vec{u}) ,S(\underline{y},\vec{z})\}) \cup \{R'(\underline{x},y,\vec{u},\vec{z})\}.
\end{equation}

 \item\label{it:atomel2}
The incoming edge of $y$ in $\igraph{q}$ is of the form $\fd{y'}{y} \in \SFD{q}$
generated by $T(\underline{y'},y,\uk) \in \modec{q}$. Consequently, there exists an edge of the form $\fd{x}{y'} \in \WFD{q}$ generated by $R(\underline{x},y',\vec{u}) \in \modei{q}$.
We define 
\begin{equation}\label{eq:atomel2}
q' \defeq (q \setminus  \{R(\underline{x},y',\vec{u}) ,S(\underline{y},\vec{z})\}) \cup \{R'(\underline{x},y',\vec{u},\vec{z})\}.
\end{equation}
\end{enumerate}
Note that $\queryvars{q}=\queryvars{q'}$ and $Z\subseteq   \nkeyqueryvars{q'} $.

For a database $\db$, define first a database $\db_0$ in the above two cases.
\begin{enumerate}[(i)]
\item\label{eq:dbel2}
$\db_0 \defeq (\db \setminus \restrict{\db}{\{R,S\}})
\cup
\{R'(\underline{a},b,\vec{c},\vec{d})
\mid
R(\underline{a},b,\vec{c}),
S(\underline{b},\vec{d})\in \db\}.
$
\item\label{eq:dbel3}
$\db_0 \defeq (\db \setminus \restrict{\db}{\{R,S\}})
\cup
\{R'(\underline{a},b',\vec{c},\vec{d})
\mid
R(\underline{a},b',\vec{c}), T(\underline{b'},b,\uk), 
S(\underline{b},\vec{d})\in \db\}.
$
\end{enumerate}
Then, define
\begin{equation}
\atomeliminate_S(\db,q)\defeq (\db',q'),
\end{equation}
where $\db' \defeq \dbsaturate(\db_0,q')$.

Before moving on to the reduction lemma concerning $\atomeliminate_S$, let us first consider the following example which illustrates why the root of $\keyvars{S}$ in $\igraph{q}$ needs to be $\bot$. Note that this restriction is unique to the inconsistent edge elimination; all other reduction steps allow the root of the modified variable or atom to be a variable. 

\smallskip
\begin{example}\label{ex:cycle}
Consider the following query from $\sjfubcq$:
\[
q \;\defeq \;\{
R(\underline{x},y),
S(\underline{y},z),
T(\underline{u},v)
\}
\cup
\{
U_i(\underline{w},z),V_i(\underline{w},v),W_i(\underline{u},z)\mid i\in \{0,1\}
\}.
\]
Observe that $q$ is saturated and normal. The root variables are $x,u,w$, the guarded variables are $x,u,w,z,v$, and the only unguarded variable is $y$. However, in contrast to the requirements of inconsistent edge elimination w.r.t. $S$, we have that 
$\keyvars{S}= x\notin \outvarstar{\igraph{q}}{\bot}$ as $x$ is a root variable itself. Now, applying inconsistent edge elimination to $q$, would yield the query 
\[
q' \;\defeq \;\{
R(\underline{x},y,z),
T(\underline{u},v)
\}
\cup
\{
U_i(\underline{w},z),V_i(\underline{w},v),W_i(\underline{u},z)\mid i\in \{0,1\}
\}.
\]
For illustrations, see \Cref{fig:cycle}.  Now, $q$ has an acyclic attack graph, but $q'$ has a cyclic attack graph due to $R\attacks{} T$ and $T\attacks{} R$. In particular, $T$ does not attack $R$ in the query $q$.

Consider then the queries $q_{x\mapsto c}$ and $q'_{x\mapsto c}$, for some constant. The attack-propagation graphs  $\igraph{q_{x\mapsto c}}$ and $\igraph{q'_{x\mapsto c}}$ are then obtained from those of $q$ and $q'$ by replacing the vertex $x$ with $\emp$. In this case, atom elimination w.r.t. $S$ on $q_{x\mapsto c}$ yields $q'_{x\mapsto c}$, and both queries have an acyclic attack graph.
\end{example}

\begin{figure}
\vspace{5mm}
\[
\hspace{0cm}
\begin{array}{ccc}
\scalebox{.72}{%
\begin{tikzpicture}[
  every node/.style={font=\small}
]

\node[guarded] (x) at (0,0) {$x$};

\node[unguarded] (y) at (0,2) {$y$};

\node[guarded] (z) at (0,4) {$z$};

\node[guarded] (w) at (2,4) {$w$};

\node[guarded] (u) at (2,0) {$u$};

\node[guarded] (v) at (2,2) {$v$};


\node[levelbox, fit=(y)] {};
\node[levelbox, fit=(z)] {};
\node[levelbox, fit=(v)] {};

\sarr{x}{y}
\sarr{y}{z}
\sarr{u}{v}
\bdarr{u}{y}
\bdarr{w}{z}
\bdarr{w}{v}

\end{tikzpicture}
}
&
\hspace{.3cm}
\begin{tikzpicture}
\draw[->]
    (0,0) -- (2,0)
    node[midway,above] {``$\atomeliminate_S$''};
    \node (y) at (0,-3) {};
\end{tikzpicture}
&
\hspace{.3cm}

\scalebox{.72}{%
\begin{tikzpicture}[
  every node/.style={font=\small}
]

\node[guarded] (x) at (0,0) {$x$};

\node[guarded] (z) at (-1,2) {$z$};

\node[unguarded] (y) at (1,2) {$y$};

\node[guarded] (w) at (1,4) {$w$};

\node[guarded] (u) at (3,0) {$u$};

\node[guarded] (v) at (3,2) {$v$};


\node[levelbox, fit=(y)(z)] {};
\node[levelbox, fit=(v)] {};

\sarr{x}{y}
\sarr{x}{z}
\sarr{u}{v}
\bdarr{u}{y}
\bdarr{w}{z}
\bdarr{w}{v}

\end{tikzpicture}
}
\end{array}
\]
\caption{Illustration of attack-propagation graph: left $\igraph{q}$, right $\igraph{q'}$. Guarded variables are in blue and unguarded ones in red. \label{fig:cycle}}
\end{figure}
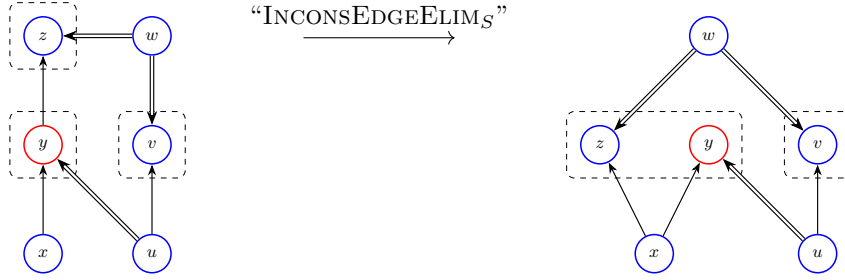

Now, let us start with the preparatory results, as in earlier sections.

\smallskip
\begin{lemma}\label{lem:nonewreps}
Let $q\in \sjfbcq$ be saturated and normal. 
Let $S\in \modei{q}$ be eliminable. 
 If $(\db',q')=\atomeliminate_S(\db,q)\notin\{\true,\false\}$, then
 $q'$ is variable-repetition-free.
 \end{lemma}
\begin{proof}
Consider first \Cref{it:atomel}. 
Given that $q$ is repetition-free, repetition can only arise if $x$ appears in $\vec{z}$ or if $\vec{z}$ and $\vec{u}$ share a variable.
In the first case, we obtain that $\igraph{q}$ contains two edges $\fd{x}{y}$ and $\fd{y}{x}$, contradicting \Cref{lem:forest}. In the second case, we obtain that $\notkeyvars{R}$ contains two distinct variables $y$ and $z$, where $z$ is the common variable of $\vec{z}$ and $\vec{u}$, such that $\fd{y}{z}\in\FD{q}$, contradicting \Cref{lem:puuh}.

Consider then \Cref{it:atomel2}. In this case, repetition can only arise if ${x}$ or $y'$ appears in $\vec{z}$ or if $\vec{u}$ and $\vec{z}$ share a variable. In the first case, we obtain a contradiction with \Cref{lem:forest} analogously to the previous case. Considering the second case, suppose $z$ is a variable appearing in both $\vec{u}$ and $\vec{z}$. Then, $\fd{x}{z}$ and $\fd{y}{z}$ are both edges in $\WFD{q}$. Since $x\neq y$, due to acyclicity of $\igraph{q}$, we obtain a contradiction with \Cref{lem:samesourcesat}. 

We conclude that in both cases $q'$ is variable-repetition-free.
\end{proof}

\smallskip
\begin{lemma}\label{claim:simple}
Let $q\in \sjfbcq$ be saturated and normal. 
Let $S\in \modei{q}$ be eliminable. 
Let $u,v\in \queryvars{q}$ be such that $v\notin Z$, where $Z \defeq \notkeyvars{S}$.
 If $(\db',q')=\atomeliminate_S(\db,q)\notin\{\true,\false\}$, then
  \begin{itemize}
\item $\fd{u}{v} \in \FD{q}$ implies $\fd{u}{v} \in \FD{q'}$
\item $\fd{u}{v} \in \WFD{q}$ implies $\fd{u}{v} \in \WFD{q'}$
\item $\fd{u}{v} \in \SFD{q}$ if and only if $\fd{u}{v} \in \SFD{q'}$
\end{itemize}
The third item holds unconditionally, that is, even if $v\in Z$.
\end{lemma}
\begin{proof}
By \Cref{lem:nonewreps} it suffices to focus on the case where $u\neq v$.
The first item is clear, and the second item follows from the first and the third item. Thus, we prove the third item.

Let us write $I'$ for the atom in $q'$ corresponding to an atom $I$ in $q$. That is, $R'$ is the extended atom from \cref{eq:atomel,eq:atomel2}, $S'$ is undefined,  and $I'\defeq I$ for $I\in q\setminus \{R,S\}$.

Assume that $\fd{u}{v}\in \SFD{q}$. Then we find distinct $G,H\in q$ such that $\fd{u}{v}\in \FD{G}$ and $\fd{u}{v}\in \FDclosure{H}$. 
We observe that $G\notin\{R,S\}$ because $R,S\in \modei{q}$, hence $G'=G$ and $\fd{u}{v}\in \FD{G'}$.
Furthermore, $\fd{u}{v}\in \FDclosure{H}$ implies $\fd{u}{v}\in \FD{H}$ or $\fd{\emp}{v}\in \FD{H}$. In the first case, we obtain 
$\fd{u}{v}\in \FD{H'}$, as above. In the second case, since $y$ is a variable we have $H\neq S$, hence $\fd{\emp}{v}\in \FD{H'}$ follows. Thus we conclude that $\fd{u}{v}\in \FDclosure{H}$ implies $\fd{u}{v}\in \FDclosure{H'}$. This leads to $\fd{u}{v}\in \SFD{q'}$.

For the converse direction, assume that $\fd{u}{v}\in \SFD{q'}$. Then we find distinct $G',H'\in q'$ such that $\fd{u}{v}\in \FD{G'}$ and $\fd{u}{v}\in \FDclosure{H'}$. If $\fd{u}{v}\in \FD{G}$ and $\fd{u}{v}\in \FDclosure{H}$, then $\fd{u}{v}\in \SFD{q}$.

Thus, suppose first that $\fd{u}{v}\notin \FD{G}$. This is only possible when $G=R$, $u=x$, and $v\in Z$. Since $G$ and $H$ are distinct, we have $H=H'$ and consequently $\fd{u}{v}\in \FDclosure{H}$. Then we have $\fd{x}{v}\in \FD{H}$ or $\fd{\emp}{v}\in \FD{H}$. The first case contradicts \Cref{lem:noupstrong}, and the second case entails $\fd{\emp}{v}\in \SFD{q}$ by \Cref{lem:samesourcesat}, $\fd{y}{v}\in \WFD{q}$, and $y\neq \emp$, contradicting the normality of $q$. Thus this is case is not possible. 

Suppose then that $\fd{u}{v}\in \FD{G}$ and $\fd{u}{v}\notin \FDclosure{H}$. This is only possible when $H=R$ and $v\in Z$. 
Moreover, we have that $x=u$ or $x=\emp$. If $x=u$, then $\fd{x}{v}\in \FD{q}$ and \Cref{lem:noupstrong} lead to a contradiction.
If $x=\emp$, then again we obtain $\fd{\emp}{v}\in \SFD{q}$ in contradiction with the normality of $q$.

Hence we conclude that $\fd{u}{v}\in \SFD{q}$.
\end{proof}

\begin{lemma}\label{lem:root3}
Let $q\in \sjfbcq$ be saturated and normal. 
Let $S\in \modei{q}$ be eliminable. 
 If $(\db',q')=\atomeliminate_S(\db,q)\notin\{\true,\false\}$, then
 $\initqueryvars{q}=\initqueryvars{q'}$.
\end{lemma}
\begin{proof}
 We need to prove $\rootqueryvars{q}\cap \srcqueryvars{q} =\rootqueryvars{q'} \cap \srcqueryvars{q'}$. The proof is divided into separate claims.

\begin{claim} $\rootqueryvars{q} \subseteq \rootqueryvars{q'}$. 
\end{claim}
\begin{claimproof}
Toward contradiction, suppose that $u\in \rootqueryvars{q}\setminus \rootqueryvars{q'}$. Then there exists an edge of the form $\fd{w}{u}\in\igraph{q'}$.  In particular, $\fd{w}{u}\in\FD{q'}$, which leads to $w\neq u$ due to \Cref{lem:nonewreps}.
Consider two cases.

Suppose first that $\fd{w}{u}\in\WFD{q'}$. Then we have $\fd{w}{u}\notin\SFD{q}$ by \Cref{claim:simple}. 
Since $u\in \rootqueryvars{q}$, we also have $\fd{w}{u}\notin\WFD{q}$. Therefore, $\fd{w}{u}\in \FD{q'}\setminus \FD{q}$, which is only possible when $w=x$ and $u\in Z$. However, in this case $u\notin \rootqueryvars{q}$ due to the edge $\fd{y}{u}\in \WFD{q}$, for some $y\in Y$, a contradiction. Thus the case that $\fd{w}{u}\in\WFD{q'}$ is not possible.

Suppose then that $\fd{w}{u}\in\SFD{q'}$. 
Since $w\neq u$, we have that $\fd{w}{u}\in \FD{q}$.
Now, we find another edge $\fd{p}{w}\in \WFD{q'}$ such that $\fd{p}{u}\notin\SFD{q'}$.
By \Cref{claim:simple} we also have $\fd{w}{u}\in\SFD{q}$ and $\fd{p}{u}\notin\SFD{q}$. Since $u\in \rootqueryvars{q}$, it holds that
$\fd{p}{w}\notin \WFD{q}$ (otherwise, $\fd{w}{u}\notin \igraph{q}$).
Moreover, by \Cref{lem:ctrans} we have $\fd{p}{w}\notin \SFD{q}$. Summing up, $\fd{p}{w}\in \FD{q'}\setminus \FD{q}$, and thus $p=x$ and $w\in Z$. In particular, it follows that $\fd{x}{u}\notin \SFD{q}$. We deal with the two cases separately.

In the case of \Cref{it:atomel}, we find $y\in Y$ such that $\fd{y}{w}\in \WFD{q}$ and $\fd{w}{u}\in\SFD{q}$, while $\fd{w}{u}\notin\igraph{q}$ due to $u\in\rootqueryvars{q}$. Recalling that $\fd{w}{u}\in \FD{q}$, this leads to $\fd{y}{u}\in\SFD{q}$. 
Since $y\in \nrootqueryvars{q}$, we have that $u\neq y$, leading to $\fd{y}{u}\in\FD{q}$. Then, by $\fd{x}{y}\in \WFD{q}$ and $\fd{y}{u}\notin \igraph{q}$, it follows that $\fd{x}{u}\in\SFD{q}$, a contradiction.

In the case of \Cref{it:atomel2}, we find $y\in Y$ such that $\fd{y}{u}\in\SFD{q}$ obtains similarly, and consequently also $\fd{y'}{u}\in\SFD{q}$ transitively via $\fd{y'}{y}\in\SFD{q}$. 
Since $y'\in \nrootqueryvars{q}$ due to the edge $\fd{x}{y}\in \WFD{q}$, we obtain $y'\neq u$, leading to $\fd{y'}{u}\in\FD{q}$
Then, we obtain $\fd{x}{u}\in\SFD{q}$ from $\fd{x}{y'}\in \WFD{q}$, $\fd{y'}{u}\in\SFD{q}\cap \FD{q}$, and $\fd{y'}{u}\notin\igraph{q}$.
Again, this leads to a contradiction.


Hence neither the case that $\fd{w}{u}\in\SFD{q'}$ is  possible. We conclude by contradiction that $\rootqueryvars{q} \subseteq \rootqueryvars{q'}$.
\end{claimproof}

\begin{claim}
 $\rootqueryvars{q'} \subseteq \rootqueryvars{q}$.
 \end{claim}
 \begin{claimproof}
   Suppose $u\notin \rootqueryvars{q}$.  
By assumption there exists an edge of the form $\fd{w}{u}\in\igraph{q}$.  
Since $q$ is variable-repetition-free, we have $w\neq u$.

Suppose first that $\fd{w}{u}\in\WFD{q}$. If $u\notin Z$, then $\fd{w}{u}\in\WFD{q'}$  by \Cref{claim:simple}, i.e., $u\notin \rootqueryvars{q'}$. Thus, suppose $u\in Z$. Then, we have $\fd{x}{u}\in \FD{R'}$. If $\fd{x}{u}\in \WFD{q'}$, we are done, so assume that
$\fd{x}{u}\in \SFD{q'}$. Then, we find $G\in q'\setminus \{R'\}$ such that $\fd{x}{u}\in \FDclosure{G}$. 
In particular, we have $\fd{t}{u}\in \FD{G}$, where $t\in \{x,\emp\}$.
Note that it holds that $G\in q\setminus \{R\}$. Thus, by \Cref{lem:noupstrong}, the case where $t=x$ is not possible, i.e., $\fd{\emp}{u}\in \FD{G}$. Since $q$ is normal, it holds that $\fd{\emp}{u}\in \WFD{G}$. But then, since $y\neq \emp$ and $\fd{y}{u}\in \WFD{q}$, we obtain a contradiction with \Cref{lem:samesourcesat}.

 Then, suppose $\fd{w}{u}\in\SFD{q}$.
  Then,  we find another edge $\fd{p}{w}\in \WFD{q}$ such that $\fd{p}{u}\notin\SFD{q}$.
  In this case, by $Z\subseteq  \leafqueryvars{q}$ we obtain $w\notin Z$. 
 Note that also $u\notin Z$; if this were not the case,  we would have $\fd{w}{u}\in\SFD{q}\cap \igraph{q}$ and $\fd{y}{u}\in \WFD{q}$, contradicting \Cref{lem:backward}.
Consequently,
   \Cref{claim:simple} implies $\fd{p}{w}\in \WFD{q'}$, $\fd{w}{u}\in\SFD{q'}$, and $\fd{p}{u}\notin\SFD{q'}$. Since $w\neq u$, we also have $\fd{w}{u}\in\FD{q'}$. This leads to
to $\fd{w}{u}\in\igraph{q'}$, i.e., $u \notin \rootqueryvars{q'}$.
\end{claimproof}

 \begin{claim}
  $
  \srcqueryvars{q} \subseteq  \srcqueryvars{q'}$.
  \end{claim}
  \begin{claimproof}
 Toward contradiction, suppose $u\in  \srcqueryvars{q} \setminus  \srcqueryvars{q'}$.
Then we find an edge $\fd{u_0}{u}\in {\FD{q'}}$ such that $u_0 \notin \reach{u}{\FD{q'}}$.
Note that the edges in $\FD{q'}\setminus\FD{q}$ are of the form $\fd{x}{z}$, where $z\in Z$. However, $\fd{u_0}{u}$ is not of this form; otherwise, due to \Cref{it:aeone}, $u \in \reach{\emp}{\FD{q}}$,  contradicting the fact that $u\in \srcqueryvars{q}$, as $\emp$ has no incoming edges in $\FD{q}$.
 Hence 
 $\fd{u_0}{u}\in {\FD{q}}$, 
 so $u_0 \in \reach{u}{\FD{q}}$ obtains by assumption.
 As above, we note that any path in $\FD{q}$ from $u$ to $u_0$  does not include edges from $\FD{q}\setminus\FD{q'}$, all of which take the form $\fd{y}{z}$, where $z\in Z$. It is then correct to conclude that $u_0 \in \reach{u}{\FD{q'}}$, a contradiction. Hence $ \srcqueryvars{q} \subseteq  \srcqueryvars{q'}$ obtains.
\end{claimproof}

 \begin{claim}
  $\srcqueryvars{q'} \subseteq \srcqueryvars{q}$.
  \end{claim}
  \begin{claimproof}
Assuming $u\in \srcqueryvars{q'}$, suppose toward contradiction that $u\notin \srcqueryvars{q}$. Then we find an edge $\fd{u_0}{u} \in \FD{q}$ such that $u_0 \notin \reach{u}{ \FD{q}}$.
If $\fd{u_0}{u} \in \FD{q'}$, then by assumption $u_0 \in \reach{u}{ \FD{q'}}$, which implies $u_0 \in \reach{u}{ \FD{q}}$; for this, observe that any edge in $\FD{q'}\setminus \FD{q}$ proceeds from $x$ to $Z$ and is replaceable with two  or three edges in $\FD{q}$, depending on the case. However, this contradicts the fact that $u_0 \notin \reach{u}{ \FD{q}}$.
 Thus, suppose $\fd{u_0}{u} \notin \FD{q'}$. This entails $u_0= y$ and $u\in Z$. Then, we observe $\fd{x}{u} \in \FD{q'}$ while $Z\subseteq   \nkeyqueryvars{q'} $ (by \Cref{lem:nonewreps} and \Cref{it:aetwo}),  contradicting $u\in \srcqueryvars{q'}$.
 We conclude by contradiction that $\srcqueryvars{q'} \subseteq \srcqueryvars{q}$.
 \end{claimproof}
 This concludes the proof.
\end{proof}

\smallskip
\begin{lemma}\label{lem:nr2}
Let $q\in \sjfbcq$ be saturated and normal. 
Let $S\in \modei{q}$ be eliminable. 
 If $(\db',q')=\atomeliminate_S(\db,q)\notin\{\true,\false\}$, then
 $\nrqueryvars{q}=\nrqueryvars{q'}$.
\end{lemma}
\begin{proof}
Assume first that $v\in \nrqueryvars{q}$. 
Then there exists 
  $u\in \initqueryvars{q}$ such that $\fd{u}{v}\in \SFD{q}$.
By \Cref{lem:root3} we obtain $u\in \initqueryvars{q'}$, and by \Cref{claim:simple} we obtain $\fd{u}{v}\in \SFD{q'}$. 
 Hence $v\in \nrqueryvars{q'}$. 
  
Conversely, assume $v\in \nrqueryvars{q'}$. 
Then there exists $u\in \initqueryvars{q'}$ such that $\fd{u}{v}\in \SFD{q'}$. 
Again, by \Cref{lem:root3,claim:simple} we obtain $u\in \initqueryvars{q}$ and $\fd{u}{v}\in \SFD{q}$, hence $v\in \nrqueryvars{q}$.
\end{proof}

\smallskip
\begin{lemma}\label{lem:qnew}
Let $q\in \sjfuabcq$ be a saturated and normal. 
 Let $S\in q$ be eliminable. 
 If $(\db',q')=\atomeliminate_S(\db,q)\notin\{\true,\false\}$, then
\begin{enumerate}[(a)]
\item\label{it:qaetwo}  $q'\in \sjfuabcq$; 
\item\label{it:qaeone} $q'$ is saturated;
\item\label{it:qaeoneone} $q'$ is normal.
\end{enumerate}
\end{lemma}
\begin{proof}
We consider the two cases of inconsistent edge elimination simultaneously.
In \cref{it:atomel}, let $R(\underline{x},y,\vec{u}), S(\underline{y},\vec{z}) \in \modei{q}$ be the atoms witnessing that $S$ is eliminable.
In \cref{it:atomel2}, let $R(\underline{x},y',\vec{u}),S(\underline{y},\vec{z}) \in \modei{q}$ and $T(\underline{y'},y,\uk)\in \modec{q}$ be the atoms witnessing that $S$ is eliminable.

Let us write $Z=\notkeyvars{S}$.

For $F\in q\setminus\{S\}$, let $F'$ be the corresponding atom in $q'$. In the case of \cref{it:atomel}, $F'=F$ if $F\notin \{R,S\}$, $R'$ is as in \cref{it:atomel}, and $S'$ is not defined.
We define $F'$ analogously in the case of \cref{it:atomel2}.
    Since $S'$ is undefined, we obtain $F\neq S \neq G$. Furthermore, we have $\keyvars{F}= \keyvars{F'}$ for all $F\in q\setminus\{S\}$.

\textbf{(\cref{it:qaetwo})} 
By \Cref{lem:nonewreps},
it suffices to prove that $q'$ has an acyclic attack graph. For this, we show that if $F'\attacks{q'} G'$, then $F\attacks{q} G$. 
    Toward contradiction, suppose this is not the case, i.e., $F'\attacks{q'} G'$ and $F\nattacks{q} G$, for some $F,G\in q$.

Now, $F'\attacks{q'} G'$ implies  $\gaifman{q'}$ contains an (undirected) path 
\[
U=x_1, \dots ,x_\ell \quad (\ell \geq 1)
\]
between some $x_1 \in \notkeyvars{F'}$ and   $x_\ell\in \atomvars{G'}$
such that $U$ is not separated by
$\keyclosure{(F')}{q'}$.

Now, $U$ is a path between $\notkeyvars{F}$ and $\atomvars{G}$, unless, for some $z\in Z$, we have that $z$ occurs as an endpoint of $U$ without belonging to $\notkeyvars{F}\cup\atomvars{G}$, or $U$ contains an edge $\{w,z\}$, where $w$ is a variable occurring in $(x,\vec{u})$.
We proceed in cases for each occurrence of such $z\in Z$ in $U$.
\begin{itemize}
\item  \cref{it:atomel}:  Replace the edge $\{w,z\}$ with the two edges $\{w,y\}$ and $\{y,z\}$, or extend the endpoint $z$ with the edge $\{y,z\}$
\item  \cref{it:atomel2}: Replace the edge $\{w,z\}$ with the three edges $\{w,y'\}$, $\{y',y\}$, and $\{y,z\}$, or extend the endpoint $z$ with the two edges $\{y,z\}$ and $\{y',y\}$.
\end{itemize}
Denote by $U'$ the undirected path obtained from $U$ in this way.

Since $F\nattacks{q} G$, 
it holds that $x_i\in \keyclosure{F}{q}$ for some $i\in [\ell]$. In the case that we have extended $U$ to $U'$, we may have $y\in \keyclosure{F}{q}$ (or, additionally, $y'\in \keyclosure{F}{q}$. However, in this case we obtain $z\in \keyclosure{F}{q}$, where the variable $z$ appears in $U$, by $\fd{y}{z}\in \FD{q\setminus \{F\}}$ (using $S\neq F$, and, in the second case, also $\fd{y'}{y}\in \SFD{q}$). 

Note that $\keyvars{F}\neq x_i$, because otherwise by $\keyvars{F}=\keyvars{F'}$ we obtain $x_i\in \keyclosure{(F')}{q'}$, a contradiction.

Hence 
$\FD{q\setminus \{F\}}$ contains a non-empty path 
\[D=t, \dots ,x_i,\]
 where $t\in\{\emp, \keyvars{F}\}$. 
Suppose $D$ is shortest among all paths witnessing $x_i\in \keyclosure{F}{q}$.
In particular, it is a
simple path. 
Moreover, we observe  
\begin{equation}\label{eq:huh}
\FD{q\setminus \{F\}} \subseteq\FD{q'\setminus \{F'\}}\cup \FD{S}.
\end{equation}
Hence, $x_i\notin \keyclosure{(F')}{q'}$ implies $D$ contains an edge of the form $\fd{y}{z}\in \FD{S}$. 
Moreover, $Z\subseteq \nkeyqueryvars{q}$ implies $z$ is the last node in $D$.
Thus, $D$ can be written as
\[
D=t, \dots ,y,z, 
\]
where  $z=x_i$ and possibly $t=y$. 
 By simplicity of $D$,  no atom generates more than one edge in $D$. In particular, no edge in the prefix $D[t,y]$ 
  belongs to $\FD{S}$. Hence, \eqref{eq:huh} entails $D[t,y]$ 
  is a path in 
$\FD{q'\setminus \{F'\}}$.

We now prove the following claim.
\begin{claim}
  $t\neq y$.
  \end{claim}
  \begin{claimproof}
  Assume toward contradiction that $t=y$. That is, $D$ consists of the single edge $\fd{y}{z}$.
Since $y\neq \emp$, we obtain $\keyvars{F}=y$, and consequently, $\keyvars{F'}=y$. On the other hand, $\rootvar{y}=\emp$, and thus by $y\in \reach{x}{\igraph{q}}$ we obtain $\rootvar{x}=\emp$. 
We now consider the two cases in parallel.

It follows that $\igraph{q}$ has a path of the form
\begin{itemize}
\item \Cref{it:atomel}: $D'= \emp, \dots ,x,y,z$,
\item \Cref{it:atomel2}: $D'= \emp, \dots ,x,y',y,z$,
\end{itemize}
where possibly $x=\emp$.

We may assume that $D'$ is simple by \Cref{lem:forest}.
Hence $D'[\emp,x]$ is a path in $\FD{q\setminus\{F\}}$.
By \eqref{eq:huh}, and since $D'[\emp,x]$ does not contain any edge from $\FD{S}$ (by simplicity of $D'$), we conclude that $D'[\emp,x]$ is a path in $\FD{q'\setminus \{F'\}}$. Moreover, since $R'\neq F'$ by $\keyvars{R'}=x\neq y=\keyvars{F'}$ (which follows by \Cref{lem:forest}), we have that $\fd{x}{z}$ belongs to $\FD{q'\setminus \{F'\}}$. Therefore, $x_i\in \keyclosure{(F')}{q'}$, a contradiction.
\end{claimproof}
Given the claim above, we can write 
\[
D=t, \dots ,v,y,z,
\]
where $z=x_i$, and possibly $v=t$. 
Consider the following claims.
\begin{claim}
$\fd{v}{y}\notin \igraph{q}$.
 \end{claim}
 \begin{claimproof}
 (\Cref{it:atomel})
We first show that $x\neq v$. 
Toward contradiction, suppose  $x=v$. Since $\fd{x}{y}$ belongs to $\WFD{q}$ and is uniquely generated by $R$, this entails $R\neq F$. In particular, we then have $R'\neq F'$, and $\fd{v}{z}$ belongs to $\FD{q'\setminus \{F'\}}$.
Again, using simplicity of $D$ and \eqref{eq:huh} it is possible to argue that $D[t,v]$ is a path in $\FD{q'\setminus \{F'\}}$.
Therefore, $x_i\in \keyclosure{(F')}{q'}$, a contradiction.  Hence $x\neq v$ obtains.
\Cref{lem:backward} then implies the claim.  

(\Cref{it:atomel2}) If $v \not\sib  y'$, then \Cref{lem:backward} entails $\fd{v}{y}\notin \igraph{q}$. Thus, suppose
$v \sib  y'$.
We first observe that $v\in\notkeyvars{R}$ and $\fd{x}{v}\in\WFD{q}$, with the latter edge uniquely generated by $R$.
This is immediate if $v=y'$. Otherwise, let $H\in q$ witness $v\sib y'$. By normality, $H$ generates inconsistent edges to both $v$ and $y'$. Since the inconsistent edge into $y'$ is uniquely generated by $R$, we obtain $H=R$.

Assume toward contradiction that $\fd{v}{y}\in\igraph{q}$. Then $\fd{v}{y}\in\SFD{q}$: this holds by assumption if $v=y'$, and by \Cref{lem:backward} otherwise.
We first exclude $t=v$. In this case, since $v$ is a variable, we have $\keyvars{F}=v$.
As in the preceding claim, take a path $D'=\emp,\dots,x,y',y,z$ in $\igraph{q}$.
Since $\fd{x}{v}\in\igraph{q}$ and $\igraph{q}$ is acyclic, $v$ does not occur in $D'[\emp,x]$. Nor does $y$ occur in this prefix.
Consequently, this prefix avoids edges generated by $F$ or $S$, and by \eqref{eq:huh} it is a path in $\FD{q'\setminus\{F'\}}$.
Moreover, $x\neq v=\keyvars{F}$ entails $R'\neq F'$, so appending $\fd{x}{z}\in\FD{R'}$ yields $z\in\keyclosure{(F')}{q'}$, a contradiction.

Thus $t\neq v$, and $v$ has an immediate predecessor $u$ in $D$.
If $u=x$, then $D$ uses the edge $\fd{x}{v}$ uniquely generated by $R$, so $R\neq F$ and $R'\neq F'$.
By simplicity of $D$, its prefix $D[t,x]$ avoids edges generated by $S$ and hence, by \eqref{eq:huh}, survives in $\FD{q'\setminus\{F'\}}$.
Appending $\fd{x}{z}\in\FD{R'}$ again gives $z\in\keyclosure{(F')}{q'}$, a contradiction.
If $u\neq x$, then \Cref{lem:samesourcesat} and $\fd{x}{v}\in\WFD{q}$ imply $\fd{u}{v}\in\SFD{q}$.
Together with $\fd{v}{y}\in\SFD{q}$, \Cref{lem:ctrans} yields $\fd{u}{y}\in\SFD{q}\subseteq\FD{q\setminus\{F\}}$.
Replacing $u,v,y$ by $u,y$ therefore shortens $D$, a contradiction.
This proves the claim in the second case.
\end{claimproof}

In particular, the claim entails  $\fd{v}{y}\in\SFD{q}$.
Suppose, toward contradiction, that there some  $u$ precedes $v$ in $D$. 
\Cref{lem:verysimple} states that if $\fd{u}{v}\in \FD{q}$, $\fd{v}{y}\in \SFD{q} \setminus \igraph{q}$, and $v\neq y$, then $\fd{u}{y}\in \SFD{q}$. However, due to $\SFD{q} \subseteq \FD{q\setminus \{F\}}$ this contradicts the assumption on $D$ being shortest.
Hence no such $u$ exists. This means that $t=v$.
%

We have now established that
\[
D=t,y,z,
\]
 witnessing $x_i\in \keyclosure{F}{q}$. 
By the above, $\fd{t}{y}\in \SFD{q}\setminus \igraph{q}$, $\fd{y}{z}\in \FD{S}$, and $z=x_i$.

On the other hand, by \Cref{it:aeone}, we have $y \in \outvarstar{\igraph{q}}{\emp}$.
 Consequently,  $\igraph{q}$ contains a path of the form
  \begin{itemize}
\item \Cref{it:atomel}: $D'= \emp, \dots ,x,y,z$,
\item \Cref{it:atomel2}: $D'= \emp, \dots ,x,y',y,z$.
\end{itemize}
Again, $D'$ is simple  as $\igraph{q}$ is acyclic. 
Note that $t\neq x$, as $\fd{t}{y}\notin  \igraph{q}$. Hence, by \Cref{lem:noupstrong} and $\fd{t}{y}\in \SFD{q}$, we obtain that $t$ does not appear in $D'[\emp,x]$. In particular, we have $t\neq \emp$.

Hence $t=\keyvars{F}$, so we obtain that $D'[\emp,x]$ is a path in $\FD{q\setminus \{F\}}$.
Since $D'[\emp,x]$ does not contain any edge from $\FD{S}$ (by simplicity of $D'$), we conclude that $D'[\emp,x]$ is a path in $\FD{q'\setminus \{F'\}}$. Moreover, as $R\neq F$ (by $\keyvars{R}=x\neq t=\keyvars{F}$), we obtain $R'\neq F'$ and thus $\fd{x}{z}$ belongs to $ \FD{q'\setminus \{F'\}}$. We conclude that $x_i\in \keyclosure{(F')}{q'}$, a contradiction.

We conclude by contradiction that $F'\attacks{q'} G'$ implies $F\attacks{q} G$.

\textbf{(\cref{it:qaeone})}  Suppose  that for some $H'\in q'$, we have    $\fd{u}{v}\in \FD{H'}$,
      $\fd{v}{w}\in \SFD{q'}$, and
    $\fd{u}{w}\in \FDclosure{q'\setminus \{H'\}}$. We need to show that $\fd{u}{w}\in \SFD{q'}$.
    
    We consider first the case where $H'=R'$. Then $u=x$ and $v\in\notkeyvars{R'}$.
    Note that $\fd{v}{w}\in \SFD{q}$ by \Cref{claim:simple}. Moreover, we have  $\fd{u}{w}\in \FDclosure{q\setminus \{R\}}$ due to $q'\setminus \{R'\} = q\setminus \{R,S\}$.
    
    If $v\in\notkeyvars{R}$, then
    $\fd{u}{v}\in \FD{R}$. 
    Thus, by saturation of $q$ it follows that  $\fd{u}{w}\in \SFD{q}$. Hence also $\fd{u}{w}\in \SFD{q'}$ by \Cref{claim:simple}.

     Suppose then that $v\notin\notkeyvars{R}$, i.e., $v$ appears in $\vec{z}$. We show that this case is not possible, considering the two constructions of $q'$ side by side. First, we prove that in both cases $\fd{y}{w}\notin \SFD{q}$.
     \begin{itemize}
     \item \Cref{it:atomel}:
     If $\fd{y}{w}\in \SFD{q}$, we obtain by $\fd{u}{y}\in \FD{R}$, $\fd{u}{w}\in \FDclosure{q\setminus \{R\}}$, and by saturation of $q$ that 
    $\fd{u}{w}\in \SFD{q}$, hence $\fd{u}{w}\in \SFD{q'}$ by \Cref{claim:simple}.
        Thus, we may assume  $\fd{y}{w}\notin \SFD{q}$ in which case $y\neq w$.     

     \item \Cref{it:atomel2}:
         If $\fd{y}{w}\in \SFD{q}$, we obtain by $\fd{u}{y'}\in \FD{R}$, $\fd{y'}{y}\in \SFD{q}$, $\fd{u}{w}\in \FDclosure{q\setminus \{R\}}$, and saturation of $q$ that 
    $\fd{u}{w}\in \SFD{q}$, hence $\fd{u}{w}\in \SFD{q'}$ by \Cref{claim:simple}.    Thus,         we may assume $\fd{y}{w}\notin \SFD{q}$ in which case $y\neq w$.
\end{itemize}

From $\fd{y}{w}\notin \SFD{q}$, $\fd{y}{v}\in \FD{S}$, and $\fd{v}{w}\in \SFD{q}$,
      we obtain by saturation of $q$ that  $\fd{y}{w}\notin \FDclosure{q\setminus \{S\}}$.
     
     Now, we recall that $\fd{u}{w}\in \FDclosure{q'\setminus \{R'\}}$. From $q'\setminus \{R'\} = q\setminus \{R,S\}$, we obtain $\fd{u}{w}\in \FDclosure{q\setminus \{S\}}$. Let \[D= t, \dots ,w\] be some path from $t$ to $w$ in $\FD{q\setminus \{S\}}$, where $t\in \{u,\emp\}$. 
    Then,  for each element $x'$ in $D$, 
    $\fd{y}{x'}\notin \FDclosure{q\setminus \{S\}}$, because otherwise we obtain a contradiction with $\fd{y}{w}\notin \FDclosure{q\setminus \{S\}}$.
    Furthermore,  $\fd{y}{w}\notin \FDclosure{q\setminus \{S\}}$ entails $t\neq \emp$, i.e., 
 $t=u$. Recall also that $u=x$.  
    Hence, \[U=x, \dots ,w,v,\] 
    obtained by appending $D$ with the undirected edge $\{w,v\}$, 
     is a path in $\gaifman{q}$ that is not separated by $\keyclosure{S}{q}$ (as $\fd{v}{w}\in \SFD{q}\subseteq \FDclosure{q\setminus \{S\}}$).
    Since this path connects $\notkeyvars{S}$ and $\atomvars{R}$, it follows  $S\attacks{q} R$. Since $R\attacks{q} S$ by \Cref{lem:nextattack} or \Cref{lem:upattack} (depending on the case), this contradicts the assumption that the attack graph of $q$ is acyclic. Hence, the case where $v$ is mentioned in $\vec{z}$ is not possible. This concludes the case where $H'=R$.

    Suppose then that $H'\neq R'$. Then, we have $H'\in q$, and, by \Cref{claim:simple}, $\fd{v}{w}\in \SFD{q}$. Consequently, to establish $\fd{u}{w}\in \SFD{q'}$,  by saturation of $q$ we only need to show that $\fd{u}{w}\in \FDclosure{q\setminus \{H'\}}$. We observe that the path from $u$ to $w$ in $\FD{q'\setminus \{H'\}}$, provided by assumption, is readily in $\FD{q\setminus \{H'\}}$, except that we may need to replace an edge induced by $R'$ with two edges induced by $R,S$ in \cref{it:atomel}, or three edges induced by $R,T,S$ in \cref{it:atomel2}. Note that this is possible as $S\neq H'$; additionally, in \cref{it:atomel2}, the edge generated by $T$ belongs to $\SFD{q}\subseteq \FD{q\setminus \{H'\}}$. Hence $\fd{u}{w}\in \FDclosure{q\setminus \{H'\}}$, which concludes the case $H'\neq R'$.
    
 Thus the case that $v$  appears in $\vec{z}$ is not possible. This concludes the proof of this item.
    
  \textbf{(\cref{it:qaeoneone})} 
In what follows we consider an atom $F'\in q'$. 

(\cref{it:normal1}) Suppose toward contradiction that $|\notkeyvars{F'}|>1$, and $\fd{p}{u} \in \SFD{q'}$, for $p= \keyvars{F'}$
and some $u\in \notkeyvars{F'}$. 
Then $\fd{p}{u} \in \FD{F'}$, and there exists $G'\in q'\setminus \{F'\}$ such that $\fd{p}{u} \in \FDclosure{G'}$.

{\bf Case 1: $F'\neq R'$}. Then $F=F'$ and hence $|\notkeyvars{F}|>1$ and $\fd{p}{u} \in \FD{F}$. 
By normality of $q$, it holds $\fd{p}{u} \in \WFD{q}$.

If $\fd{p}{u} \in \FDclosure{G}$, 
then $\fd{p}{u} \in \SFD{q}$, a contradiction. Thus $\fd{p}{u} \notin \FDclosure{G}$. This implies $G=R$ and $u\in Z$. Then by $\fd{y}{u},\fd{p}{u}\in \WFD{q}$ and 
 \Cref{lem:samesource}, we obtain $p=y$, hence $\fd{y}{u} \in \FD{F}$. But then $F=S$, and $F'$ is undefined, a contradiction.

{\bf Case 2: $F'= R'$}. Since $G'\neq F'$, we obtain $G=G'$ and hence $\fd{p}{u} \in \FDclosure{G}$. Thus, if
$\fd{p}{u} \in \FD{F}$, we obtain a contradiction with the normality of $q$. Hence $\fd{p}{u} \notin \FD{F}$.
Since $\fd{p}{u} \in \FD{F'}$, this entails $p=x$ and $u\in Z$. Now, $\fd{x}{u} \in \FDclosure{G}$
entails   $\fd{x}{u} \in \FD{G}$ or $\fd{\emp}{u} \in \FD{G}$. Since $\fd{x}{u} \notin \FD{q}$ by \Cref{lem:noupstrong}, 
we obtain $\fd{\emp}{u} \in \FD{q}$. By normality of $q$, in particular, this means $\fd{\emp}{u} \in \WFD{q}$.
By \Cref{lem:samesource}, this leads to $\emp = y$, a contradiction with \cref{it:aeone}, i.e., $y\in \nrootqueryvars{q}$.

(\cref{it:normal2}) Suppose toward contradiction that $\fd{u}{v} \in \SFD{q'}$, for some distinct $u,v\in \notkeyvars{F'}$. 
By \Cref{claim:simple}, we have $\fd{u}{v} \in \SFD{q}$.
This means that there exist $G'\in q'$ and $H'\in q'\setminus \{G'\}$ such that $\fd{u}{v} \in \FD{G'}$ and $\fd{u}{v} \in\FDclosure{H'}$.
Additionally, by $Z\subseteq \nkeyqueryvars{q}$ we obtain $u\notin Z$. 

{\bf Case 1: $u,v\in \notkeyvars{F}$.} Since $\fd{u}{v} \in \SFD{q}$, this violates the normality of $q$, a contradiction.



{\bf Case 2: $u,v\notin \notkeyvars{F}$.} Then $F=R$ and $F'=R'$. Note that we have $u\notin Z$. 
We proceed by case distinction.
     \begin{itemize}
     \item \Cref{it:atomel}:
the case where both $u,v$ belong to $(y,\vec{u})$ leads to a 
$F$ witnessing a violation of normality. Hence, $u$ belongs to $(y,\vec{u})$ and $v$ belongs to $Z$.

Suppose first that $u=y$. Then $\fd{u}{v} \in \SFD{q}$, $u=\keyvars{S}$, and $v\in \notkeyvars{S}$ contradict the assumption that $S\in \modei{q}$.
 
Suppose then that $u$ belongs to $\vec{u}$. Now, we have $y\neq u$, since $q$ is variable-repetition-free, and $\fd{y}{v}\in \WFD{q}$. Hence, by \Cref{lem:backward}, we obtain $\fd{u}{v} \in \SFD{q}\setminus \igraph{q}$. 
Since $\fd{x}{u}\in \WFD{q}$ and $u\neq v$, this leads to $\fd{x}{v} \in \SFD{q}$ by \Cref{lem:verysimple}. However, this raises a contradiction with \Cref{lem:noupstrong}.

     \item \Cref{it:atomel2}: the case where both $u,v$ belong to $(y',\vec{u})$ leads to a 
$F$ witnessing a violation of normality. Hence, $u$ belongs to $(y',\vec{u})$ and $v$ belongs to $Z$.

Suppose first that $u=y'$. Then $\fd{u}{v} \in \SFD{q}$ and $u\neq v$ imply $\fd{u}{v}\in \FD{q}$, which contradicts \Cref{lem:noupstrong},
as $y',y,v$ is a path in $\igraph{q}$.

Suppose then that $u$ belongs to $\vec{u}$. It follows that $y'\neq u$ since $q$ is variable-repetition-free. Then, arguing as in the previous case, we obtain a contradiction.
\end{itemize}

We conclude by contradiction that this item holds.

(\cref{it:normal4}) Suppose, toward contradiction that $\fd{u}{v},\fd{v}{u}\in \SFD{q'}$, for some distinct $u,v\in \queryvars{q'}$.
Thus, by \Cref{claim:simple}, we obtain
$\fd{u}{v},\fd{v}{u}\in \SFD{q}$, 
 a contradiction with the normality of $q$.

(\cref{it:normal3}) This item follows readily from \Cref{claim:simple} and the normality of $q$.
\end{proof}



 
 \medskip
 
\begin{lemma}[Inconsistent Edge Elimination]\label{lem:aedb}
Let $\mathcal{C}$ consist of all $(\db,q)\in \DB\times \sjfuabcq$, where $q$ is a saturated and normal query and $\db$ is a $q$-saturated database.
For each $(\db,q)\in \mathcal{C}$ and eliminable $S\in \modei{q}$, 
$(\db,q)\mapsto \atomeliminate_S(\db,q)$ is a strict, guarded, equivalence-preserving, and  polynomial-time reduction from $\mathcal{C}$ to $\mathcal{C}\cup\{\true,\false\}$.
\end{lemma}
\begin{proof}
Let $R(\underline{x},y,\vec{u}), S(\underline{y},\vec{z}) \in \modei{q}$ be the atoms witnessing that $S$ is eliminable.
Let us also write $Z=\notkeyvars{S}$. Recall that both $Z$ is non-empty by \Cref{it:aenotempty}.

The reduction $\atomeliminate_S$ is clearly strict and computable in polynomial time. It is into $\mathcal{C}\cup\{\true,\false\}$
by \Cref{lem:qnew,lem:dbsat}. It thus suffices to prove that it is guarded and equivalence-preserving.

{\bf (Guardedness)}
\Cref{it:inv:guarded,it:inv:guarding} of guardedness follow 
by \Cref{lem:root3,lem:nr2}. \Cref{it:sizeinv,it:inv:2,it:sizeinv,it:dummy,it:dummy2} are clear by construction.
We prove  
  \Cref{it:inv:const,it:inv:1}.

To see why  \Cref{it:inv:const} holds,
assume that $H\in \modec{q}$, and either $\keyvars{H}\in \initqueryvars{q}$ or $\atomvars{H}= \emptyset$.
 Toward contradiction, suppose first that $H\notin {q'}$.
 Then, in both cases \ref{it:atomel} and \ref{it:atomel2} of the construction of $q'$, it follows that $\rootof{\keyvars{H}}=\bot$ and $\notkeyvars{H}\neq \emptyset$, contradicting the assumption. This proves that  $H\in {q'}$. 
 We claim that furthermore $H\in \modec{(q')}$. Suppose, toward contradiction, that $H\in \modei{(q')}$. Then there exists an edge $\fd{u}{v}\in \WFD{q'}\cap \FD{H}$. By \Cref{claim:simple}, since $\fd{u}{v}\in \FD{q}$, this implies $\fd{u}{v}\in \WFD{q}$, contradicting $H\in \modec{q}$.
Hence, the claim is true, and therefore \Cref{it:inv:const} holds.  

To prove \Cref{it:inv:1}, 
Assume that for each $T\in q$,
\(
|\restrict{T^{\db}}{V}| \leq N.
\)
  Now, \Cref{lem:nr2} and the fact that $\queryvars{q}=\queryvars{q'}$ imply $\rqueryvars{q'}= \rqueryvars{q}$.
      For any $U\in q'\setminus \{R'\}$, we have $U^{\db}=U^{\db_0}$, and thus, using \Cref{lem:dbsat},
\[|\restrict{U^{\db'}}{V'}|\leq |\restrict{U^{\db_0}}{V'}|\leq | \restrict{U^{\db}}{V}| \leq  N,\]
 where  $V'\defeq \atomvars{U}\cap \rqueryvars{q'}$ and $V\defeq \atomvars{U}\cap \rqueryvars{q}$.

Consider then the atom $R'$. Since $ \notkeyvars{S} \subseteq \nrqueryvars{q}$,
letting $V' \defeq \atomvars{R'}\cap \rqueryvars{q'}$ and $V\defeq \atomvars{R}\cap \rqueryvars{q}$, we have that
 $V' = V$.
 Hence, by the construction and \Cref{lem:dbsat},
\[
|\restrict{(R')^{\db'}}{V'}| \leq |\restrict{(R')^{\db_0}}{V'}| = |\restrict{R^{\db}}{V'}| \leq |\restrict{R^{\db}}{V}| \leq  N.
\]
Thus, \Cref{it:inv:1} holds as well. We conclude that the stated reduction is guarded.

   {\bf (Equivalence preservation)}
   By \Cref{lem:dbsat},  it suffices to prove that $(\db,q)\in \cert$ if and only if $(\db_0,q')\in \cert$, where $q'$ and $\db_0$ are as in \eqref{eq:atomel2} and \eqref{eq:dbel2}.

    We consider case \ref{it:atomel} of the construction of $q'$ and assume that $x$ is a variable; the remaining cases are then analogous.

    Now, since $\rootof{y}=\emp$ (by \Cref{it:aeone}) and $\fd{x}{y}\in \igraph{q}$, the graph $\igraph{q}$ contains a path of the form
    \[
    D=\emp,\dots ,x,y,z.
    \]
    (Note that if $x$ is a constant, then in case \ref{it:atomel} this path takes the form   \(
    D=\emp,y,z.
    \))
    
    This path corresponds to a sequence of
     of atoms 
     \[
     R_1(\underline{f},x_1,\uk), R_2(\underline{x_1},x_2,\uk),\dots ,R_{n+1}(\underline{x_n},x,\uk), R(\underline{x},y,\vec{u}),S(\underline{y},\vec{z})\in q
     \]
    where $f$ is a constant, and $x_1, \dots ,x_n$ is a (possibly empty) sequence of variables.
     For every repair $\rep$ of $\db$, there exists a unique sequence of constants $f_1, \dots ,f_n,a,b,\vec{c},\vec{d}$ such that
         \[
     R_1(\underline{f},f_1,\uk), R_2(\underline{f_1},f_2,\uk),\dots ,R_{n+1}(\underline{f_n},a,\uk), R(\underline{a},b,\vec{d}),S(\underline{b},\vec{c})\in \rep.
     \]
     We say that the mapping $(x,y,\vec{z},\vec{u})\mapsto (a,b,\vec{c},\vec{d})$ is \emph{determined by $\rep$}. For a repair $\rep_0$ of $\db_0$, by replacing the atoms $R(\underline{x},y,\vec{u}),S(\underline{y},\vec{z})$ with $R'(\underline{x},y,\vec{z},\vec{u})$ and the facts $R(\underline{a},b,\vec{d}),S(\underline{b},\vec{c})$ with $R'(\underline{a},b,\vec{c},\vec{d})$ in the above sequences, we can similarly describe how the mapping $(x,y,\vec{z},\vec{u})\mapsto (a,b,\vec{c},\vec{d})$ is \emph{determined by $\rep_0$}.

    Clearly, if $\theta(q)\subseteq \rep$, for a valuation $\theta$, then $(x,y,\vec{u},\vec{z})\mapsto (\theta(x),\theta(y),\theta(\vec{u}),\theta(\vec{z}))$
    is the mapping that is determined by $\rep$, and similarly for $\rep_0$.
     In particular, whenever the repairs $\rep$ and $\rep_0$ agree on $R_1, \dots ,R_n$, then they determine the same mapping $(x,y,\vec{u},\vec{z})\mapsto (\theta(x),\theta(y),\theta(\vec{u}),\theta(\vec{z}))$.
    
     ($\Rightarrow$)
    Suppose  $(\db,q)\in \cert$, and let  $\rep_0$ be an arbitrary repair of $\db_0$. 
    Let $(x,y,\vec{z},\vec{u})\mapsto (a,b,\vec{c},\vec{d})$ be determined by $\rep_0$. In particular, $R'(\underline{a},b,\vec{c},\vec{d})\in \rep_0$.
We let $\rep$ extend $\restrict{\rep_0}{q\setminus \{R,S\}}$ as follows:
\begin{itemize}
\item Add $R(\underline{a},b,\vec{d})$ to $\rep$, and from the remaining $R$-blocks of $\db$ add an arbitrary fact to $\rep$.
\item Add $S(\underline{b},\vec{c})$ to $\rep$, and from the remaining  $S$-blocks of $\db$ add an arbitrary fact to $\rep$.
\end{itemize}
Clearly, $\rep$ is a repair of $\db$. Furthermore, $\rep$ and $\rep_0$ determine the same mapping $(x,y,\vec{z},\vec{u})\mapsto (a,b,\vec{c},\vec{d})$. By assumption the exists a valuation $\theta$ on $\queryvars{q}$ such that $\theta(q)\subseteq \rep$. 
Thus $\theta $ maps $(x,y,\vec{z},\vec{u})$ to $ (a,b,\vec{c},\vec{d})$. By construction we obtain $\theta(q')\subseteq \rep_0$. Thus $(\db_0,q')\in \cert$.

     ($\Leftarrow$)
 Suppose  $(\db_0,q')\in \cert$, and let $\rep$ be an arbitrary repair of $\db$. 
    Let $(x,y,\vec{z},\vec{u})\mapsto (a,b,\vec{c},\vec{d})$ be the mapping that is determined by $\rep$. In particular, $R(a,b,\vec{d}),S(\underline{b},\vec{c})\in \rep$.
    We let $\rep_0$ extend $\restrict{\rep}{q\setminus \{R,S\}}$ as follows:
\begin{itemize}
\item Add $R'(\underline{a},b,\vec{d},\vec{c})$ to $\rep_0$, and from the remaining $R'$-blocks of $\db_0$ add an arbitrary fact to $\rep_0$.
\end{itemize}
Analogously to the first direction, $\rep_0$ is a repair of $\db_0$. Furthermore, $\rep$ and $\rep_0$ determine the same mapping $(x,y,\vec{z},\vec{u})\mapsto (a,b,\vec{c},\vec{d})$. By assumption the exists a valuation $\theta$ on $\queryvars{q'}$ such that $\theta(q')\subseteq \rep_0$. Thus $\theta$ maps $ (x,y,\vec{z},\vec{u})$ to $ (a,b,\vec{c},\vec{d})$. By construction we obtain $\theta(q)\subseteq \rep$. Thus $(\db,q)\in \cert$.
\end{proof}

\subsection{Pruning Step}\label{sect:prunestep}
We first illustrate query pruning with the following example.
\begin{example}\label{ex:prev}
Consider the query $q'$ given in \eqref{eq:q'} of \Cref{ex:q'}.
The key-nonkey graph of $q'$ appeared already in \Cref{fig:guarding}.
We recall that the attack-propagation graph of $q'$ consists of a single $\emp$-rooted tree together with an isolated guard variable $u$. 
Suppose $\db$ is a $q'$-saturated database.
Unless $\textsc{Prune}(\db,q')$ produces $\false$, it outputs a pair $(\db',q'')$, where the query takes the form \eqref{eq:q''}.
The stages of the pruning reduction are depicted in \Cref{fig:prune}.
\end{example}

\begin{figure}[t]
\[
\hspace{-1cm}
\begin{array}{ccc}
\scalebox{0.55}{%
\begin{tikzpicture}[
  every node/.style={font=\small}
]

\node[empnode] (emp) at (0,0) {$\bot$};

\node[unguarded] (x1) at (-1,2) {$x_1$};
\node[unguarded] (x2) at (1,2) {$x_2$};

\node[guarded] (y1) at (-2,4) {$y_1$};
\node[guarded] (y2) at (0,4) {$y_2$};

\node[unguarded] (z1) at (-1,6) {$z_1$};
\node[guarded] (z2) at ( 1,6) {$z_2$};

\node[levelbox, fit=(z1)(z2)] {};

\node[levelbox, fit=(x1)(x2)] {};
\node[levelbox, fit=(y1)] {};
\node[levelbox, fit=(y2)] {};

\sarr{emp}{x1}
\sarr{emp}{x2}

\sarr{x1}{y1}
\bdarr{x1}{y2}
\bdarr{x2}{y2}
\rdarr{y1}{y2}

\sarr{y2}{z1}
\sarr{y2}{z2}

\end{tikzpicture}
}
&
\hspace{-.3cm}
\scalebox{0.55}{%
\begin{tikzpicture}[
  every node/.style={font=\small}
]
\node[
    anchor=west,
    font=\small\bfseries
] at ([xshift=-2cm]emp.west)
{(LP)};
\node[empnode] (emp) at (0,0) {$\bot$};

\node[unguarded] (x1) at (-1,2) {$x_1$};
\node[unguarded] (x2) at (1,2) {$x_2$};

\node[guarded] (y1) at (-2,4) {$y_1$};
\node[guarded] (y2) at (0,4) {$y_2$};

\node[unguarded] (z1) at (-1,6) {$z_1$};
\node[guarded] (z2) at ( 1,6) {$z_2$};

\node[levelbox, fit=(z1)(z2)] {};

\node[levelbox, fit=(x1)(x2)] {};
\node[levelbox, fit=(y1)] {};
\node[levelbox, fit=(y2)] {};

\sarr{emp}{x1}
\sarr{emp}{x2}

\sarr{x1}{y1}
\bdarr{x1}{y2}
\bdarr{x2}{y2}

\sarr{y2}{z1}
\sarr{y2}{z2}

\end{tikzpicture}
}
&
\hspace{-.5cm}
\scalebox{0.55}{%
\begin{tikzpicture}[
  every node/.style={font=\small}
]
\node[
    anchor=west,
    font=\small\bfseries
] at ([xshift=-2cm]emp.west)
{(LE)};
\node[empnode] (emp) at (0,0) {$\bot$};

\node[unguarded] (x1) at (-1,2) {$x_1$};
\node[unguarded] (x2) at (1,2) {$x_2$};

\node[guarded] (y1) at (-2,4) {$y_1$};
\node[guarded] (y2) at (0,4) {$y_2$};

\node[guarded] (z2) at ( 1,6) {$z_2$};

\node[levelbox, fit=(z2)] {};

\node[levelbox, fit=(x1)(x2)] {};
\node[levelbox, fit=(y1)] {};
\node[levelbox, fit=(y2)] {};

\sarr{emp}{x1}
\sarr{emp}{x2}

\sarr{x1}{y1}
\bdarr{x1}{y2}
\bdarr{x2}{y2}

\sarr{y2}{z2}

\end{tikzpicture}
}
\\[2ex]
\hspace{-.5cm}
\scalebox{0.55}{%
\begin{tikzpicture}[
  every node/.style={font=\small}
]
\node[
    anchor=west,
    font=\small\bfseries
] at ([xshift=-2cm]emp.west)
{(IEE)};
\node[empnode] (emp) at (0,0) {$\bot$};

\node[guarded] (y1) at (-3,2) {$y_1$};
\node[unguarded] (x1) at (-1,2) {$x_1$};
\node[unguarded] (x2) at (1,2) {$x_2$};

\node[guarded] (y2) at (0,4) {$y_2$};

\node[guarded] (z2) at ( 1,6) {$z_2$};

\node[levelbox, fit=(z2)] {};

\node[levelbox, fit=(y1)(x1)(x2)] {};
\node[levelbox, fit=(y2)] {};

\sarr{emp}{x1}
\sarr{emp}{x2}
\sarr{emp}{y1}

\bdarr{x1}{y2}
\bdarr{x2}{y2}

\sarr{y2}{z2}

\end{tikzpicture}
}
&
\scalebox{0.55}{%
\phantom{kkkk}
\begin{tikzpicture}[
  every node/.style={font=\small}
]
\node[
    anchor=west,
    font=\small\bfseries
] at ([xshift=-2cm]emp.west)
{(IEE)};
\node[empnode] (emp) at (0,0) {$\bot$};

\node[guarded] (y1) at (-3,2) {$y_1$};
\node[unguarded] (x1) at (-1,2) {$x_1$};
\node[unguarded] (x2) at (1,2) {$x_2$};

\node[guarded] (y2) at (0,4) {$y_2$};

\node[guarded] (z2) at ( 3,2) {$z_2$};


\node[levelbox, fit=(y1)(x1)(x2)(z2)] {};
\node[levelbox, fit=(y2)] {};

\sarr{emp}{x1}
\sarr{emp}{x2}
\sarr{emp}{y1}
\sarr{emp}{z2}

\bdarr{x1}{y2}
\bdarr{x2}{y2}


\end{tikzpicture}
}
&
\scalebox{0.55}{%
\hspace{0cm}
\begin{tikzpicture}[
  every node/.style={font=\small}
]
\node[
    anchor=west,
    font=\small\bfseries
] at ([xshift=-2cm]emp.west)
{(CEE)};
\node[empnode] (emp) at (0,0) {$\bot$};

\node[guarded] (y1) at (-4,2) {$y_1$};
\node[unguarded] (x1) at (-2,2) {$x_1$};
\node[unguarded] (x2) at (0,2) {$x_2$};

\node[guarded] (y2) at (2,2) {$y_2$};

\node[guarded] (z2) at ( 4,2) {$z_2$};


\node[levelbox, fit=(y1)(x1)(x2)(z2)(y2)] {};

\sarr{emp}{x1}
\sarr{emp}{x2}
\sarr{emp}{y1}
\sarr{emp}{z2}
\sarr{emp}{y2}


\end{tikzpicture}
}
\\[2ex]
\hspace{.3cm}
\scalebox{0.55}{%
\phantom{kkkk}
\begin{tikzpicture}[
  every node/.style={font=\small}
]
\node[
    anchor=west,
    font=\small\bfseries
] at ([xshift=-2cm]emp.west)
{(LE)};
\node[empnode] (emp) at (0,0) {$\bot$};

\node[guarded] (y1) at (-3,2) {$y_1$};
\node[unguarded] (x2) at (-1,2) {$x_2$};

\node[guarded] (y2) at (1,2) {$y_2$};

\node[guarded] (z2) at ( 3,2) {$z_2$};


\node[levelbox, fit=(y1)(x2)(z2)(y2)] {};

\sarr{emp}{x2}
\sarr{emp}{y1}
\sarr{emp}{z2}
\sarr{emp}{y2}


\end{tikzpicture}
}
&
\hspace{-.6cm}
\scalebox{0.55}{%
\hspace{2cm}
\begin{tikzpicture}[
  every node/.style={font=\small}
]
\node[
    anchor=west,
    font=\small\bfseries
] at ([xshift=-2cm]emp.west)
{(LE)};
\node[empnode] (emp) at (0,0) {$\bot$};

\node[guarded] (y1) at (-2,2) {$y_1$};

\node[guarded] (y2) at (0,2) {$y_2$};

\node[guarded] (z2) at ( 2,2) {$z_2$};


\node[levelbox, fit=(y1)(z2)(y2)] {};

\sarr{emp}{y1}
\sarr{emp}{z2}
\sarr{emp}{y2}


\end{tikzpicture}
}
\end{array}
\]
\caption{The key-nonkey-graphs $\FD{q}$ in query pruning (the edges corresponding to guarding relations are omitted). 
Every edge belongs to the attack-propagation graph $\igraph{q}$, except for the edge represented by the red arrow.
The red vertices represent unguarded variables and blue vertices guarded variables. 
The dashed rectangles enclose the non-primary-key variables of a single atom.
The pruning reduction is composed of leaf pruning (LP), leaf elimination (LE), consistent edge elimination (CEE), and inconsistent edge elimination (IEE). }\label{fig:prune}
\end{figure}
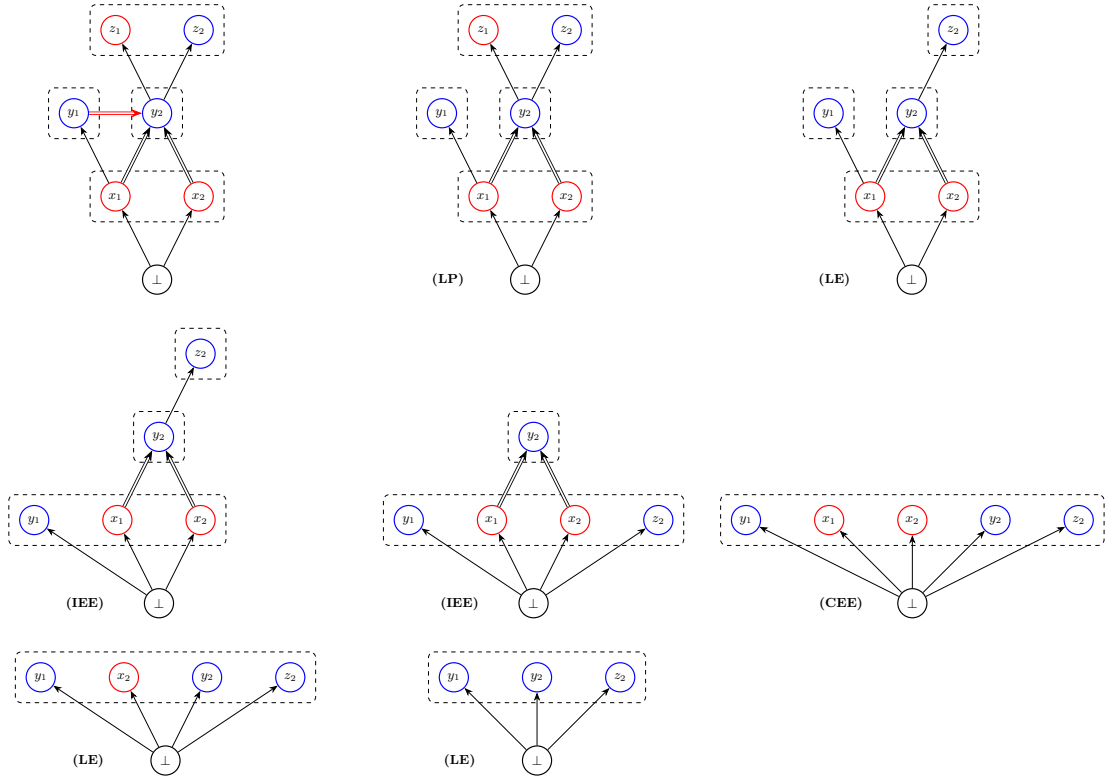

\begin{algorithm}[t]
\caption{$\textsc{Prune}(\db,q)$ \label{alg:prune}}
\Input{A saturated and normal query $q \in \sjfuabcq$ and a $q$-saturated database $\db$}
\Output{$\false$ or $(\db',q')$, where  $q'$ is a saturated  query from $\sjfuabcq$, $\db'$ is a $q'$-saturated database, and $\outvar{\igraph{q'}}{\emp} \subseteq \nrqueryvars{q'} \cap \leafqueryvars{q'} \cap \nkeyqueryvars{q'}$
}

\SetKwRepeat{Repeat}{repeat}{until}
\Repeat{none of the while-loops is entered}{

    \While{ exists prunable $y\in \queryvars{q}$}{
        choose the $\prec$-least prunable $y$\;
        $(\db,q) \gets \Try(\leafprune_y(\db,q))$\;
        \tcp{\Cref{lem:vedb111}; $y$ is prunable if $y\in  \leafqueryvars{q} \cap \keyqueryvars{q}$}
    }
    
        \While{ exists eliminable $y\in \queryvars{q}$}{
        choose the $\prec$-least eliminable
        $y $\;
                $(\db,q) \gets \Try(\leafeliminate_y(\db,q))$\;
        \tcp{\Cref{lem:vedb0}; $y$ is eliminable if $y\in    \leafqueryvars{q} \cap \nkeyqueryvars{q}\cap \rqueryvars{q}$}
    }

    \While{exists extendable $R \in \modei{q}$}{
        choose the $\prec$-least extendable $R$\;
                        $(\db,q) \gets \Try(\atomextend_R(\db,q))$\;
        \tcp{\Cref{lem:vedb01}; $R$ is extendable if,
given $Y \defeq \notkeyvars{R}$,
 $\outvar{\WFD{q}}{Y}=\emptyset$ and $\emptyset \neq \outvar{\igraph{q}}{Y}\subseteq  \leafqueryvars{q} \cap \nkeyqueryvars{q} \cap\nrqueryvars{q} $;
    }
    }
    \While{exists eliminable $S \in \modei{q}$}{
        choose the $\prec$-least eliminable $S$\;
                 $(\db,q) \gets \Try(\atomeliminate_S(\db,q))$\;
        \tcp{\Cref{lem:aedb}; $S$ is eliminable if, given $Y \defeq \notkeyvars{S}$,  $\emptyset \neq Y \subseteq    \leafqueryvars{q} \cap \nkeyqueryvars{q}\cap \nrqueryvars{q}$ and
        $\keyvars{S} \in \reach{\bot}{\igraph{q}}\cap \nrootqueryvars{q}$.}
    }
}
\Return{$(\db,q)$}\;
\end{algorithm}

Then, let $\textsc{Prune}$ be the reduction described in \Cref{alg:prune}. We now prove the pruning lemma.

\prune*
\begin{proof}
 $\textsc{Prune}$ is guarded and equivalence-preserving by \cref{prop:compose,lem:vedb111,lem:vedb0,lem:vedb01,lem:aedb}. By the same token,
 $\textsc{Prune}$ is a mapping into $\mathcal{C}$.
 
 It remains to prove \eqref{eq:setprune} and the polynomial-time computability.

{\bf Time complexity.}
Each iteration strictly reduces $\size{q}$, as stated in \cref{lem:vedb111,lem:vedb0,lem:vedb01,lem:aedb}.
Hence the number of iterations is polynomial. Also, each step can be carried out in polynomial time in the temporary input $(\db',q')$  by 
the same lemmas. In particular, checking the conditions of the while-clauses, including the construction of the graphs $\FD{q'},\SFD{q'},\WFD{q'},\igraph{q'}$, can be done in polynomial time in $\size{q'}$.

It holds that $(\db,q)$ is $N$-unguarded-bounded, for $N\defeq \size{\db}$. By guardedness (\Cref{def:faith}), then, $\db'$ is $N$-unguarded-bounded,  $|\adom{\db'}|\leq N$, and $|\initqueryvars{q}|=|\initqueryvars{q'}|$.  Hence, by \Cref{lem:grounding,lem:isone}, for each $R\in q'$,
 \[
 |R^{\db'}| \leq  |\restrict{R^{\db'}}{V}| \cdot |\adom{\db'}|^{|\initqueryvars{q'}|} \leq   N^{c+1},
 \]
  where $V\defeq \atomvars{R}\cap \rqueryvars{q'}$.
Furthermore, by guardedness, it holds that $\size{q'}\leq \size{q}$. Therefore, 
\[\size{\db'} \leq \size{q'}\cdot N^{c+1} \leq \size{q} \cdot \size{\db}^{c+1},\]
meaning that each step is computable in polynomial time in the initial input $(\db,q)$.
 We conclude that each step can be carried out in polynomial time in the initial input. Since each reduction inside the while-loop is strict, it is correct to conclude that the algorithm runs in polynomial time.

\smallskip
\noindent\textbf{\cref{eq:setprune}.} Let us consider the increasing effect each while-loop has.
\begin{itemize}
\item Once the while-loop 1 does not apply, 
\[ \leafqueryvars{q} \subseteq \nkeyqueryvars{q}.\]
\item Once the while-loops 1-2 do not apply, 
\begin{equation}\label{it:second}
 \leafqueryvars{q} \subseteq \nkeyqueryvars{q} \cap \nrqueryvars{q}.
\end{equation}
\item Once the while-loops 1-3 do not apply, given $F\in \modei{q}$, $Y \defeq \notkeyvars{F}$, such that $\outvar{\WFD{q}}{Y}=\emptyset$, 
\begin{equation}\label{it:third}
\text{ if $\outvar{\igraph{q}}{Y} \subseteq  \leafqueryvars{q}$, then $Y\subseteq \leafqueryvars{q}$ (i.e., $\outvar{\igraph{q}}{Y}=\emptyset$).}
\end{equation}
\item Once the while-loops 1-4 do not apply, 
\begin{equation}\label{eq:false}
\outvar{\igraph{q}}{\emp} \subseteq  \leafqueryvars{q} \cap \nkeyqueryvars{q} \cap \nrqueryvars{q}.
\end{equation}
\end{itemize}
The first three items are straightforward to verify.
For the last item, assume toward contradiction that none of the while-loops 1-4 applies, while \eqref{eq:false} is false. 

Since \eqref{it:second} holds but \eqref{eq:false} does not, we have  $ \outvar{\igraph{q}}{\emp} \not\subseteq \leafqueryvars{q}$. 
By normality of $q$, we find an atom $F\in \modei{q}$ such that $\keyvars{F}=\emp$ and $Y\not\subseteq \leafqueryvars{q}$, where $Y\defeq \notkeyvars{F}$. Note that vacuously 
 $Y\subseteq \outvarstar{\igraph{q}}{\bot} \cap \nrootqueryvars{q}$. We consider two possible cases.

Suppose first that $\outvar{\WFD{q}}{Y}\neq \emptyset$. 
 Then, we can pick an atom $G \in \modei{q}$ such that $\keyvars{G}\in Y$ and $\notkeyvars{G}\neq \emptyset$.
 Consequently, $\keyvars{G}\in  \outvarstar{\igraph{q}}{\bot} \cap \nrootqueryvars{q}$.

Suppose then that $\outvar{\WFD{q}}{Y}= \emptyset$.
 By \eqref{it:third}, since $Y\not\subseteq \leafqueryvars{q}$, we obtain $\outvar{\igraph{q}}{Y}\not \subseteq  \leafqueryvars{q}$.
We can now pick $y$ such that $\outvar{\igraph{q}}{y}\not\subseteq  \leafqueryvars{q}$. 
Hence,  and due to $\outvar{\WFD{q}}{y}= \emptyset$, there exist two edges of the form $\fd{y}{y'}\in \SFD{q} \cap \igraph{q}$ and $\fd{y'}{v}\in \WFD{q}$; for this, recall that by \Cref{lem:consecutive} $\igraph{q}$ does not have two consecutive edges from $\SFD{q}$. 
Hence, there exists an atom $G \in \modei{q}$,
where $\keyvars{G}=y'\in \outvarstar{\igraph{q}}{\bot} \cap \nrootqueryvars{q}$ and $\notkeyvars{G}\neq \emptyset$.

We conclude that
there exists an atom $G \in \modei{q}$
such that $\keyvars{G}\in \outvarstar{\igraph{q}}{\bot} \cap \nrootqueryvars{q}$ and $Z\neq \emptyset$, where $Z \defeq \notkeyvars{G}$. We may assume that $\keyvars{G}$ has a maximal distance from $\emp$ in $\igraph{q}$ among all such $G$. The reasoning above establishes that $Z\subseteq  \leafqueryvars{q}$; otherwise, due to acyclicity of $\igraph{q}$ (\Cref{lem:forest}),
one would find an atom $H \in \modei{q}$ such that $\keyvars{H}$ has a greater distance to $\emp$ than $\keyvars{G}$.
However, then $G$ is eliminable, contradicting the assumption that none of the while-loops 1-4 applies. Hence, by contradiction, the last item follows.
\end{proof}

\section{Definitions and Lemmas for \Cref{sect:sub}}\label{sect:appsub}
In this section we provide auxiliary definitions, lemmas, and missing proofs for \Cref{sect:sub}.

First, we present what we call substitution-guarded-reductions, which are similar to the guarded reductions given earlier in \Cref{def:faith}. 
 
 \smallskip
\begin{definition}[Substitution-Guarded reduction]\label{def:subguard}
A many-one reduction $f\colon \DB\times \sjfuabcq \to \DB\times \sjfuabcq \cup \{\true,\false\}$ is \emph{substitution-guarded} if for every $(\db,q)\in \DB\times \sjfuabcq$, 
assuming $(\db',q')= f(\db,q) \notin\{\true, \false\}$,
 the conditions of guardedness hold,
  except that 
\cref{it:inv:const,it:inv:guarded,it:inv:guarding} are replaced by the following conditions:
\begin{itemize}
\item[4$^*$]\label{it:4bstar} 
If $F\in \modec{q}$, and either $\keyvars{F}\in \initqueryvars{q}$ or $\atomvars{F}= \emptyset$,
the following statements hold.
\begin{itemize}
\item
 If  $\atomvars{F}= \emptyset$, then $F\in \modec{(q')}$.
 \item
  If $\keyvars{F}\in  \initqueryvars{q}$ and $F= R(\underline{x},y,\vec{c})$, where $y$ is a variable and $\vec{c}$ is a sequence of constants, then, for some constants $a$ and $b$, one of $F,R(\underline{x},b,\vec{c}),R(\underline{a},b,\vec{c}) $ is in $ \modec{(q')}$.
  \item 
  If $\keyvars{F}\in  \initqueryvars{q}$ and $F= R(\underline{x},\vec{c})$, where $\vec{c}$ is a sequence of constants, then, for some constant $a$, $F$ or $R(\underline{a},\vec{c})$ is in $ \modec{(q')}$.
\end{itemize}
\item[5$^*$]\label{it:4cstar} $\nrqueryvars{q'}= \nrqueryvars{q}\cap \queryvars{q'}$. 
\item[6$^*$]\label{it:4dstar} $\initqueryvars{q'} = \initqueryvars{q}\cap \queryvars{q'}$. 
\end{itemize}
\end{definition}
A substitution-guarded reduction is always {strict}, i.e.,  $\size{q'}<\size{q}$.

Note that if $F\in \modec{q}$ and  $\keyvars{F}\in \initqueryvars{q}$, then $F$ by normality is either of the form $R(\underline{x},y,\vec{c})$ or $R(\underline{x},\vec{c})$, where $y$ is a variable and $\vec{c}$ a sequence of constants. 
One may now observe that the following propositions hold. The first one is clear by definition. 
\bigskip
\begin{proposition}\label{prop:subisfaith}
 Guarded reductions are substitution-guarded.
\end{proposition}
\medskip
\begin{proposition}\label{prop:subcompose}
 Substitution-guarded reductions are closed under composition, provided that newly
introduced relation names are fresh.
\end{proposition}
\begin{proof}
 The proof is analogous to that of \Cref{prop:compose}.
 Here, additionally, Items 5$^*$ and 6$^*$ rely on \Cref{it:dummy}.
 \end{proof}

\subsection{Substitution Step}
Next, we consider the substitution step, showing that it is substitution-guarded. Again, we start with preparatory lemmas.
Recall the notion of a substitutable set from \Cref{def:sub}.

\smallskip
\begin{lemma}\label{lem:rootnorm}
   Let $q\in \sjfuabcq$ be saturated and normal.
Let $Y\subseteq \queryvars{q}$ be substitutable.  
   Let $f\colon Y\to \Const$ be a function.
Define $q'\defeq q_{f}$. 
Then $\initqueryvars{q} \setminus  Y=\initqueryvars{q'}$.
\end{lemma}
\begin{proof}
 We need to prove $(\rootqueryvars{q}\cap \srcqueryvars{q}) \setminus  Y =\rootqueryvars{q'} \cap \srcqueryvars{q'}$.
 We divide this to several claims. For a variable $x\in \queryvars{q}$, let us define $x'\defeq \emp$ if $x\in Y$, and otherwise $x'\defeq x$.

\begin{claim}
 $\rootqueryvars{q} \setminus  Y\subseteq\rootqueryvars{q'}$. 
 \end{claim}
 \begin{claimproof}
 Toward contradiction, suppose the claim is false.
Let $u\in \rootqueryvars{q} \setminus  (Y\cup \rootqueryvars{q'})$.
 Then there exists an edge of the form $\fd{p}{u}\in\igraph{q'}$.  
 Clearly, there exists $\fd{w}{u}\in\FD{q}$ where $w'\defeq p$.
 Since $q$ is variable-repetition-free, it holds that $w\neq u$.
 Since $u\in \rootqueryvars{q}$, we also have $\fd{w}{u}\notin\WFD{q}$.
We consider two possible cases.

Suppose first that $\fd{p}{u}\in\WFD{q'}$.   
Now, $\fd{w}{u}\in\SFD{q}$ implies $\fd{p}{u}\in \SFD{q'}$ by \Cref{claim:notsimple0}, which contradicts the assumption.
Hence, we obtain $\fd{w}{u}\notin\FD{q}$, a contradiction. 

Suppose then that $\fd{p}{u}\in\SFD{q'}$.
Then we find another edge
 $\fd{r}{p}\in \WFD{q'}$ such that $\fd{r}{u}\notin\SFD{q'}$. 
This entails $p=w \notin Y$ (as $p$ is a variable), i.e., we have $\fd{r}{w}\in \WFD{q'}$, $\fd{w}{u}\in\SFD{q'}$, and $\fd{r}{u}\notin\SFD{q'}$.
Again, there exists $\fd{s}{w}\in\FD{q}$ such that $s'=r$.
By \Cref{claim:notsimple0}, we have $\fd{w}{u}\in\SFD{q}$ and $\fd{s}{u}\notin\SFD{q}$. 
Note that $\fd{w}{u}\in\FD{q}$ due to $w\neq u$.
Thus, if $\fd{s}{w}\in \WFD{q}$, then $\fd{w}{u}\in\igraph{q}$, contradicting the fact that $u\in \rootqueryvars{q}$. 
Moreover, if $\fd{s}{w}\in \SFD{q}$, then by \Cref{lem:ctrans} we obtain $\fd{s}{u}\in\SFD{q}$, again a contradiction. 

We conclude by contradiction that  $\rootqueryvars{q}\subseteq\rootqueryvars{q'}$.
\end{claimproof}

\begin{claim}
 $\rootqueryvars{q'} \subseteq \rootqueryvars{q} \setminus Y$. 
 \end{claim}
 \begin{claimproof}
Since $\rootqueryvars{q'}$ does not intersect $Y$, 
it suffices to prove $\rootqueryvars{q'} \subseteq \rootqueryvars{q}$.
Assume toward contradiction that this is false. Let $u \in \rootqueryvars{q'} \setminus  \rootqueryvars{q} $.
In particular, $u\notin Y$.
 By assumption there exists an edge of the form $\fd{w}{u}\in\igraph{q}$.  
If $\fd{w}{u}\in\WFD{q}$, then $\fd{w'}{u}\in\WFD{q'}$  by \Cref{claim:notsimple0}, i.e., $u\notin \rootqueryvars{q'}$, a contradiction.
 Thus, suppose $\fd{w}{u}\in\SFD{q}$. Then we find another edge $\fd{p}{w}\in \WFD{q}$ such that $\fd{p}{u}\notin\SFD{q}$.
 Since $\outvar{\SFD{q}}{Y} \subseteq Y$ and $u\notin Y$, we obtain $w\notin Y$, hence $w=w'$.
 Consequently, \Cref{claim:notsimple0} implies $\fd{p'}{w}\in \WFD{q'}$, $\fd{w}{u}\in\SFD{q'}$, and $\fd{p'}{u}\notin\SFD{q'}$,
 which is tantamount 
to $\fd{p'}{u}\in\igraph{q'}$, i.e., $u \notin \rootqueryvars{q'}$. This contradicts our assumption, so we obtain $\rootqueryvars{q'} \subseteq \rootqueryvars{q} \setminus Y$.
\end{claimproof}

\begin{claim}
 $\srcqueryvars{q} \cap \rootqueryvars{q} \subseteq \srcqueryvars{q'}$.
 \end{claim}
 \begin{claimproof}
Assume  toward contradiction that $u\in (\srcqueryvars{q} \cap \rootqueryvars{q}) \setminus \srcqueryvars{q'}$. Then, there exists an edge of the form $\fd{u_0}{u}\in {\FD{q'}}$ such that $u_0\notin \reach{u}{\FD{q'}}$. In particular, $u\in \queryvars{q'}$.
We claim that $\fd{u_0}{u}\in {\FD{q}}$. Assume toward contradiction that $\fd{u_0}{u}\notin {\FD{q}}$. This is only possible when 
$\fd{\emp}{u}\in {\FD{q'}}$ and $\fd{y}{u}\in {\FD{q}}$, for some $y\in Y$. However, assuming  $y\in Y$, $\fd{y}{u}\in {\WFD{q}}$ contradicts $u\in \rootqueryvars{q}$ and $\fd{y}{u}\in {\SFD{q}}$  contradicts $u\notin Y$. Therefore, the claim that $\fd{u_0}{u}\in {\FD{q}}$ holds.

Now, from $u\in \srcqueryvars{q}$
  we obtain $u_0\in \reach{u}{\FD{q}}$. In particular, there exists a directed path $D$ from $u$ to $u_0$ in $\FD{q}$. 
  We may assume that $D$ is shortest among all such paths.
  Since no such path exists in $\FD{q'}$, we obtain that $D$ contains a vertex  $y\in Y$.
  Moreover, $y\neq u$ as $y$ does not belong to $\queryvars{q'}$ while $u$ does.
  We may select $y$ so that, beside it, no other vertex in the path $D[y,u]$ belongs to $Y$. Since $\outvar{\SFD{q}}{Y} \subseteq Y$ (due to the substitability of $Y$), it follows that the first edge in $D[y,u]$ belongs to $\WFD{q}$. 
  Since $D[y,u]$ is also a shortest path from $y$ to $u$, \Cref{lem:shortestpath} implies  $u\in \reach{y}{\igraph{q}}$,
  contradicting $u\in \rootqueryvars{q}$. 
   We conclude by contradiction that $\srcqueryvars{q} \cap \rootqueryvars{q} \subseteq \srcqueryvars{q'}$.
  \end{claimproof}

\begin{claim}
 $\srcqueryvars{q'} \subseteq \srcqueryvars{q}$.
 \end{claim}
 \begin{claimproof}
Assume toward contradiction that $u\in \srcqueryvars{q'} \setminus \srcqueryvars{q}$. 
In particular, we have $u\notin Y$.
Then we find an edge $\fd{u_0}{u} \in \FD{q}$ such that $u_0 \notin \reach{u}{ \FD{q}}$.

Suppose first that $\fd{u_0}{u} \in \FD{q'}$, in which case by assumption we have $u_0 \in \reach{u}{ \FD{q'}}$. Since $u\in \srcqueryvars{q'}$ (as opposed to being $\emp$), it can be seen that any path from $u$ to $u_0$ in $\FD{q'}$ belongs also to 
$\FD{q}$. Hence, $u_0 \in \reach{u}{ \FD{q}}$, contradicting the assumption. 

Suppose then that $\fd{u_0}{u} \notin \FD{q'}$. This entails $u_0\in Y$. Now, $\fd{u_0}{u}\in {\WFD{q}}$ implies by \Cref{claim:notsimple0} 
that $\fd{\emp}{u}\in {\WFD{q'}}$, contradicting the assumption that $u\in \srcqueryvars{q'}$.
Furthermore, $\fd{u_0}{u}\in {\SFD{q}}$ implies $u\in Y$ by the fact that $\outvar{\SFD{q}}{Y} \subseteq Y$, again a contradiction. Therefore, by contradiction we obtain $\srcqueryvars{q'} \subseteq \srcqueryvars{q}$.
\end{claimproof}
This concludes the proof.
\end{proof}

\begin{lemma}\label{lem:nrsub}
   Let $q\in \sjfuabcq$ be saturated and normal.
Let $Y\subseteq \queryvars{q}$ be substitutable.  
   Let $f\colon Y\to \Const$ be a function.
Define $q'\defeq q_{f}$. 
Then $\nrqueryvars{q} \setminus  Y=\nrqueryvars{q'}$.
\end{lemma}
\begin{proof}
Assume first that $v\in \nrqueryvars{q}\setminus Y$. Then there exists $u\in \initqueryvars{q}$ such that $\fd{u}{v}\in \SFD{q}$.
Note that we have  $u\notin Y$ by $v\notin Y$ and $\outvar{\SFD{q}}{Y} \subseteq Y$.
Hence,  by \Cref{lem:rootnorm} we obtain $u\in \initqueryvars{q'}$, and by \Cref{claim:notsimple0} we have $\fd{u}{v}\in \SFD{q'}$.
Therefore, $v\in \nrqueryvars{q'}$.

Assume then that $v\in \nrqueryvars{q'}$.  Then there exists $u\in \initqueryvars{q'}$ such that $\fd{u}{v}\in \SFD{q'}$.
It follows that  $u,v\notin Y$.
Hence,  by \Cref{lem:rootnorm} we obtain $u\in \initqueryvars{q}$, and by \Cref{claim:notsimple0} we have $\fd{u}{v}\in \SFD{q}$. 
Therefore, $v\in \nrqueryvars{q}$.
\end{proof}

We can then prove the following lemma. Recall that the reduction $\textsc{Substitute}_f$ is given in \eqref{eq:subf}.

\smallskip

\sub*
\begin{proof}
The reduction is strict, polynomial-time, and into $\mathcal{C}^*\cup\{\true,\false\}$ by \Cref{lem:removecfd3}.
It remains to show that the  reduction is substitution-guarded.

\Cref{it:inv:2,it:sizeinv,it:dummy} are immediate, and 
Items 5$^*$ and 6$^*$ follow by \Cref{lem:nrsub,lem:rootnorm}.
It remains to prove Items 4$^*$ and \Cref{it:inv:1}. We start with the former.

Let $q$ be a saturated and normal query from $\sjfuabcq$. 
Let $Y\subseteq \queryvars{q}$ be substitutable.  
Let $f\colon Y \to \adom{\db}$.

Let $F\in \modec{q}$, and suppose that either $\keyvars{F}\in \initqueryvars{q}$ or $\atomvars{F}= \emptyset$.
 We consider the three statements of Item 4$^*$ in the following.
\begin{itemize}
\item
 If  $\atomvars{F}= \emptyset$, then clearly $F\in \modec{(q')}$.
 \item
  Suppose $\keyvars{F}\in  \initqueryvars{q}$ and $F= R(\underline{x},y,\vec{c})$, where $y$ is a variable and $\vec{c}$ is a sequence of constants. In particular, we have that $\notkeyvars{F}=\{y\}$
  and $\fd{x}{y}\in \SFD{q}$.
  If $F\in q'$, then $x,y\notin Y$, and hence by \Cref{claim:notsimple0},  $\fd{x}{y}\in \SFD{q'}$, yielding $F\in \modec{(q')}$.
  Thus, suppose $F\notin q'$. Then, one of $R(\underline{f(x)},y,\vec{c})$, $R(\underline{x},f(y),\vec{c})$, $R(\underline{f(x)},f(y),\vec{c})$ belongs to $q'$. We show that only the middle atom can belong to $q'$.
  
  Toward contradiction, suppose first that $R(\underline{f(x)},f(y),\vec{c})\in q'$. Then $x,y\in Y$, and as $q$ is variable-repetition-free, we moreover have $x\neq y$. In particular, $Y$ is not a singleton, so by \Cref{it:sub3} of substitutability we obtain $Y= \notkeyvars{H}$, for some $H\in q$ such that $\keyvars{H}=\emp$.
  As $\fd{x}{y}\in \SFD{q}$, we obtain a contradiction with the normality of $q$. Hence, $R(\underline{f(x)},f(y),\vec{c})\notin q'$.
  
  Furthermore, toward contradiction, suppose that $R(\underline{f(x)},y,\vec{c})\in q'$. Then $x\in Y$, and by \Cref{it:sub2} of substitutability, also $y\in Y$, contradicting the fact that $y$ has not been substituted. Hence, $R(\underline{f(x)},y,\vec{c})\notin q'$.
  
  It follows that $R(\underline{x},f(y),\vec{c})\in q'$, and by definition this means that $R(\underline{x},f(y),\vec{c})\in \modec{(q')}$
  Thus, the second statement follows.
  \item 
  If $\keyvars{F}\in  \initqueryvars{q}$ and $F= R(\underline{x},\vec{c})$, where $\vec{c}$ is a sequence of constants, then clearly either $F$ or $R(\underline{f(x)},\vec{c})$ is in $ \modec{(q')}$.
\end{itemize}
This concludes the proof of Item 4$^*$.

To prove \Cref{it:inv:1}, assume that $(\db,q)$ is $N$-unguarded-bounded (see \Cref{def:Nguarded}).
In particular, for each $T\in q$, it holds that
\(
|\restrict{T^{\db}}{V}| \leq N,
\)
where $V\defeq \atomvars{T}\cap \rqueryvars{q}$. 
By \Cref{lem:nrsub}, $\nrqueryvars{q'}= \nrqueryvars{q} \setminus Y$. Since $\queryvars{q'}$ and $Y$ are disjoint, 
this leads to $\rqueryvars{q'}\subseteq \rqueryvars{q}$.
 Thus, using \Cref{lem:dbsat} and \Cref{it:sizeinv}, for each $T' \defeq T_{Y\mapsto f(Y)} \in q'$, where $T\in q$, we have
\[|\restrict{T^{\db'}}{V'}|\leq |\restrict{T^{\db}}{V'}|\leq | \restrict{T^{\db}}{V}| \leq  N,\]
 where  $V'\defeq \atomvars{T'}\cap \rqueryvars{q'}$ and $V\defeq \atomvars{T}\cap \rqueryvars{q}$.
Hence, \Cref{it:inv:1} holds as well. 

We conclude that the stated reduction is substitution-guarded.
This concludes the proof.
\end{proof}

\section{Proof of Main Lemma}\label{sect:appmain}

In the proof of the main lemma, we use the following result.

\smallskip
\begin{lemma}[{\cite[Lemma~4.4]{KoutrisW17}}]\label{lem:non-attacked}
Let $q$ be a query from $\sjfbcq$. Let $F\in q$ be an atom that
has no incoming attack in the attack graph of $q$.
Then, the following are equivalent for
every database $\db$:
\begin{enumerate}
\item every repair of $\db$ satisfies $q$;
\item for some
$f\colon \keyvarsvars{F} \to \adom{\db}$,
  every repair of $\db$ satisfies $q_f$.
\end{enumerate}
\end{lemma}

\usetikzlibrary{shapes.geometric,positioning,arrows.meta}
\usetikzlibrary{decorations.pathmorphing}

\usetikzlibrary{decorations.pathmorphing}

\usetikzlibrary{calc,decorations.pathmorphing}

\newcommand{\subtreeicon}[1]{%
\begin{scope}[shift={(#1)},rotate=180]

    \coordinate (A) at (90:7mm);
    \coordinate (B) at (210:7mm);
    \coordinate (C) at (330:7mm);
    \coordinate (M) at ($(B)!0.5!(C)$);

    \begin{scope}
        \clip (A)--(B)--(C)--cycle;

        \fill[red!35]
            (-2,-2) --
            plot[
                decorate,
                decoration={snake,amplitude=.2mm,segment length=1.2mm}
            ] coordinates {(0,-2) (0,2)}
            -- (-2,2) -- cycle;

        \fill[blue!30]
            (2,-2) --
            plot[
                decorate,
                decoration={snake,amplitude=.2mm,segment length=1.2mm}
            ] coordinates {(0,-2) (0,2)}
            -- (2,2) -- cycle;
    \end{scope}

    \draw[
        decorate,
        decoration={snake,amplitude=.2mm,segment length=1.2mm}
    ]
    (M) -- (A);

    \draw (A)--(B)--(C)--cycle;

\end{scope}
}

\newcommand{\subtreeiconblue}[1]{%
\begin{scope}[shift={(#1)},rotate=180]

    \coordinate (A) at (0,7mm);

    \coordinate (B) at (-11mm,4.1mm);
    \coordinate (C) at ( 11mm,4.1mm);

    \fill[blue!30] (A)--(B)--(C)--cycle;
    \draw (A)--(B)--(C)--cycle;

\end{scope}
}

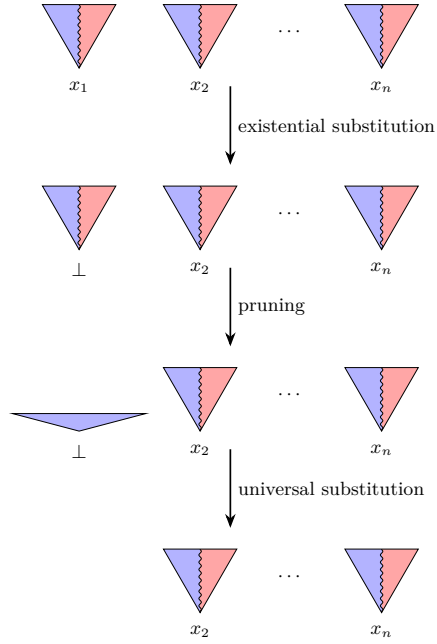
\begin{figure}
\hspace{3cm}
\scalebox{.8}{
\begin{tikzpicture}[
    >=Stealth,
    every node/.style={font=\small},
    subtree/.style={
        draw,
        regular polygon,
        regular polygon sides=3,
        rotate=180,
        minimum size=20mm,
        inner sep=0pt
    }
]


\coordinate (t1) at (0,1.5);
\subtreeicon{t1}
\node[below=8mm of t1] {$x_1$};

\coordinate (t2) at (2,1.5);
\subtreeicon{t2}
\node[below=8mm of t2] {$x_2$};

\node at (3.5,1.4) {$\cdots$};

\coordinate (tn) at (5,1.5);
\subtreeicon{tn}
\node[below=8mm of tn] {$x_n$};


\draw[->,thick]
    (2.5,0.5) -- node[right]{existential substitution} (2.5,-0.8);


\coordinate (s1) at (0,-1.5);
\subtreeicon{s1}
\node[below=8mm of s1] {$\emp$};

\coordinate (s2) at (2,-1.5);
\subtreeicon{s2}
\node[below=8mm of s2] {$x_2$};

\node at (3.5,-1.6) {$\cdots$};

\coordinate (sn) at (5,-1.5);
\subtreeicon{sn}
\node[below=8mm of sn] {$x_n$};


\draw[->,thick]
    (2.5,-2.5) -- node[right]{pruning} (2.5,-3.8);


\coordinate (p1) at (0,-4.5);
\subtreeiconblue{p1}
\node[below=8mm of p1] {$\emp$};

\coordinate (p2) at (2,-4.5);
\subtreeicon{p2}
\node[below=8mm of p2] {$x_2$};

\node at (3.5,-4.6) {$\cdots$};

\coordinate (pn) at (5,-4.5);
\subtreeicon{pn}
\node[below=8mm of pn] {$x_n$};

\draw[->,thick]
    (2.5,-5.5) --
    node[right]{universal substitution}
    (2.5,-6.8);



\coordinate (u2) at (2,-7.5);
\subtreeicon{u2}
\node[below=8mm of u2] {$x_2$};

\node at (3.5,-7.6) {$\cdots$};

\coordinate (un) at (5,-7.5);
\subtreeicon{un}
\node[below=8mm of un] {$x_n$};
\end{tikzpicture}
}
\caption{Schematic outline for the main lemma. The triangles represent the trees of the attack-propagation graph $\igraph{q}$, with the red area representing unguarded and the blue area guarded variables. \label{fig:schematic}}
\end{figure}

A schematic outline of the reduction performed by \textsc{Solve} (see \Cref{alg:solve}) is shown in \Cref{fig:schematic}. Let $q$ be a saturated and normal query. Unless $\emp$ already has outgoing edges in the attack-propagation graph $\igraph{q}$, the algorithm first selects a suitable non-attacked root variable $x_1$ and existentially substitutes it with a constant. It then flattens the tree rooted at $\emp$ by pruning it, thereby removing all non-guarded variables from the tree. Finally, the guarded variables that remain in the tree are universally substituted by a constant. This removes one tree from $\igraph{q}$, after which the process repeats on the remaining graph. At any stage, a computation branch may terminate with $\false$; if it reaches the end, it returns $\true$. Guarding guarantees that the number of branches is polynomially bounded, with each branch running in polynomial time.

The function $\textsc{Solve}$ is given by \Cref{alg:solve}. Recall the meaning of $\Try$ from \Cref{rem:try}
and $\textsc{Substitute}_f$ from \eqref{eq:subf}.


\smallskip
\main*
\begin{proof}
We claim that, given $(\db,q) \in \mathcal{C}$, $\textsc{Solve}(\db,q)$ runs in polynomial time, and 
\begin{equation}\label{eq:solve}
\textsc{Solve}(\db,q)=\true \iff (\db,q)\in \cert.
\end{equation}
By \Cref{lem:isone} and the definition of $\mathcal{C}$, we have
\begin{equation}\label{eq:c-up}
 |\initqueryvars{q}|\leq \cgs{q} \leq c.
\end{equation}
We first consider (A) correctness of the algorithm \textsc{Solve}, and then (B) its running time.

{\bf{(A) Correctness.}} 
 By \Cref{lem:prune,lem:dbsat}, every reduction involved in the algorithm is equivalence-preserving, except for   
 the substitution steps in lines 5 and 14. Thus, it suffices to focus on these variable substitutions.
 We start with the case where the condition of the if-statement holds, and then proceed to the case where it does not hold. 

{\bf Case: $\queryvars{q}\neq \emptyset$ and $\outvar{\igraph{q}}{\emp} = \emptyset$.}
By \Cref{lem:non-attacked}, 
\[(\db,q)\in \cert \iff  (\db,q_{x\mapsto a})\in \cert  \text{ for some }a \in \adom{\db}\]
Moreover, by \Cref{lem:dbsat}, 
either
\(
\dbsaturate(\db,q_{x\mapsto a})=\false,
\)
in which case
\(
(\db,q_{x\mapsto a})\notin\cert,
\)
or, defining 
\(
\db' \defeq \dbsaturate(\db,q_{{x}\mapsto {a}}),
\) 
we have
\[
(\db,q_{{x}\mapsto {a}})\in\cert
\iff
(\db',q_{{x}\mapsto {a}})\in\cert.
\]

Suppose that the algorithm continues from the reduced instance $(\db',q_{{x}\mapsto {a}})$.
By \Cref{lem:rootnotattacked} and the normality of $q$, it can be seen that the singleton set $\{x\}$ is substitutable.
Hence, by \Cref{lem:pruneconst}, $q_{{x}\mapsto {a}}$ remains in $\sjfuabcq$ and is saturated and normal, while $\db'$ is $q_{{x}\mapsto {a}}$-saturated, as required for the recursive call of \textsc{Solve}. 

Thus, it can be seen that the first part of the algorithm is correct.

{\bf Case: $\queryvars{q}\neq \emptyset$ and $\outvar{\igraph{q}}{\emp} \neq \emptyset$.}
By \Cref{lem:prune}, pruning is equivalence-preserving and, unless it returns $\true$ or $\false$, produces a saturated and normal query in $\sjfuabcq$ together with a database saturated for that query. 

It is possible that after pruning we have $\outvar{\igraph{q}}{\emp}=\emptyset$. This can happen if before pruning $\emp$ is only connected to non-guarded variables in $\igraph{q}$; in this case, these variables are simply pruned away and the algorithm continues the while-loop.

Suppose that after pruning $\outvar{\igraph{q}}{\emp}\neq \emptyset$. Then there exists a variable $y\in \queryvars{q}$ such that $\fd{\emp}{y}\in \igraph{q}$. In particular, there exists an atom  $R\in \modei{q}$ such that $\keyvars{R}=\emp$ and $y\in Y\defeq \notkeyvars{R}$.  
 It is clear that
\begin{align*}&(\db,q)\in \cert \\
\iff &  (\db^*,q_f)\in \cert \text{ for all } f\colon Y \to \adom{\db}\text{ such that } R_f\in \db,
\end{align*}
where, for each such $f$, we set $\db^*\defeq (\db\setminus R^\db)\cup\{R_f\}$.
Therefore, by \Cref{lem:dbsat},
either
\(
\dbsaturate(\db^*,q_{f})=\false,
\)
in which case
\(
(\db^*,q_{f})\notin\cert,
\)
or, defining  
\(
\db' \defeq \dbsaturate(\db^*,q_{f}),
\)
we have
\[
(\db^*,q_{f})\in\cert
\iff
(\db',q_{f})\in\cert.
\]

Suppose that the algorithm continues from the reduced instance $(\db',q_{f})$.
It holds that $Y\subseteq\outvar{\igraph{q}}{\emp} $, and
by \Cref{lem:prune}  we obtain that $\outvar{\igraph{q}}{\emp} \subseteq   \nkeyqueryvars{q}$. 
Thus, by $Y \subseteq   \nkeyqueryvars{q}$ we obtain
$\outvar{\SFD{q}}{Y} \subseteq Y$. Furthermore, $\outvar{\SFD{q}}{\emp} \subseteq Y$ since $q$ is normal, meaning that
 \Cref{it:sub2} of substitutability obtains for $Y$. As \Cref{it:sub1} and \Cref{it:iii} are clear, we conclude that $Y$ is substitutable.
Applied to $(\db^*,q)$ and $f$, \Cref{lem:pruneconst} entails that $q_f$ belongs to $\sjfuabcq$ and is saturated and normal, while $\db'$ is $q_f$-saturated.
Hence the recursive call of \textsc{Solve} satisfies the required invariants.

It can now be seen that the second part of the algorithm is correct.

{\bf Case: $\queryvars{q}= \emptyset$.} Then, $q$ is simply a set of facts. Since $\db$ is $q$-saturated, $q\subseteq \db$ and $\db$ is consistent over the schema of $q$. Hence $(\db,q)\in \cert$.

 It can be verified that between every recursive call there is at least one strict reduction.
 If $\queryvars{q}\neq \emptyset$ and the first branch is taken, then the first reduction involved in the sequel, i.e.,  $(\db,q)\to \textsc{Substitute}_f(\db,q)$ is strict. Suppose that $\queryvars{q}\neq \emptyset$ and the second branch is taken. If the pruning step is vacuous, i.e., $(\db,q)$ is already at a pruned state, then the condition of the if-statement is triggered, leading to a strict reduction; otherwise, a non-vacuous pruning step involves at least one strict reduction. Therefore, we may conclude that the computation terminates. 
 
 We conclude by these considerations that \eqref{eq:solve} holds. This concludes step (A).
  
 {\bf (B) Running time.} 
We begin by noting  the following property. A \emph{state} of the computation is a pair $(\db',q')\notin \{\true,\false\}$
arising during an execution of the algorithm.

\begin{claim}\label{claim:inv-pres}
For every state $(\db',q')$ of the computation, the query $q'$ is a saturated and normal query from $\sjfuabcq$, the database $\db'$ is $q'$-saturated, and the reduction from $(\db,q)$ to $(\db',q')$ is substitution-guarded. Consequently,
$
|\initqueryvars{q'}|\leq |\initqueryvars{q}|\leq c.
$
  \end{claim}
  \begin{claimproof}
The reduction from $(\db,q)$ to $(\db',q')$ is composed of those in lines 5, 10, and 14. In line 14, the database is first restricted to $\db^*=(\db\setminus R^\db)\cup\{R_f\}$. Since $R_f\in\db$, we have $\db^*\subseteq\db$, so $(\db,q)\mapsto(\db^*,q)$ is guarded by \Cref{prop:dbrestriction}. The subsequent application of $\textsc{Substitute}_f$ to $(\db^*,q)$ is substitution-guarded by \Cref{lem:pruneconst} and produces a saturated database together with a saturated and normal query from $\sjfuabcq$. The same query and database properties hold after lines 5 and 10 by \Cref{lem:pruneconst,lem:prune}, respectively. Thus \Cref{prop:subisfaith,prop:subcompose} imply that the composition is substitution-guarded. The final claim follows from Item~6$^*$ of substitution-guardedness and \eqref{eq:c-up}.
    \end{claimproof}

 Let $N\defeq \size{\db}$ be the size of the initial input database $\db$.  
 Specifically, $\db$  is $N$-unguarded-bounded.
 In the following, we first consider (B1) the running time of each computation branch, and then (B2) the number of such computation branches. By establishing a polynomial upper bound for both, 
 it follows that the total running time of the algorithm is polynomial.
 
{\bf (B1) Branch running time.}  
Let us first consider the size of an arbitrary intermediate input. For this, let $(\db',q')$ be an arbitrary state of the algorithm. By \Cref{claim:inv-pres}, $q'$ is saturated and normal and $\db'$ is $q'$-saturated. Moreover, substitution-guardedness gives that $\db'$ is $N$-unguarded-bounded, $\nrqueryvars{q'} \subseteq \nrqueryvars{q}$, $\adom{\db'}\subseteq \adom{\db}$, and $\size{q'}\leq \size{q}$.
  Thus, by \Cref{lem:grounding} and \eqref{eq:c-up},
 we obtain, for each $R\in q'$, that
 \[|R^{\db'}| \leq  |\restrict{R^{\db'}}{V}| \cdot |\adom{\db'}|^{|\initqueryvars{q'}|} 
   \leq N \cdot |\adom{\db}|^{|\initqueryvars{q}|}
 \leq N^{c+1},\]
  where $V\defeq \atomvars{R}\cap \rqueryvars{q'}$. Hence, 
  \begin{equation}\label{eq:dbupper}
  \size{\db'} \leq \size{q'} \cdot N^{c+1} \leq \size{q} \cdot N^{c+1} \leq \size{q} \cdot \size{\db}^{c+1}. 
  \end{equation}
  Consequently, there exists a polynomial $p$ such that $\size{\db'} + \size{q'} \leq p(\size{\db},\size{q})$, for  any state $(\db',q')$. 
  By \Cref{lem:prune,lem:removecfd3}, every reduction in the algorithm runs in polynomial time in its local input $(\db',q')$ (with the same fixed constant $c$ applying to every invocation of $\textsc{Prune}$ by \Cref{claim:inv-pres}).
  Hence, every reduction runs in polynomial time in the initial input $(\db,q)$.
  Since the recursion depth is bounded by $\size{q}$, following the argument in step (A), we may conclude that each computation branch runs in polynomial time in the input $(\db,q)$. This concludes step (B1).
  
  {\bf (B2) Number of branches.}
  Fix $\db$ as the initial  database, and $q$ as the initial  query. The computation can be viewed to maintain a partial function $f \colon \queryvars{q}\to \adom{\db}$, corresponding to the substitution of constants for variables. 
  Let \(\subs\) denote the set of all such substitutions \(f\), including the empty function, that arise during the computation. Any two surviving states with the same substitution domain have the same query structure, up to the substituted constants: database saturation may terminate a branch with \(\false\), but otherwise leaves the query unchanged, while the subsequent transformations depend only on its structure. Hence the next substitution domain is the same in all such states, and since substitutions only extend their domains, the sets \(\Dom(f)\), for \(f\in\subs\), form a chain under inclusion.
   In particular,  writing $\dom{f}$ for the domain of a function $f$, for any $f,g \in \subs$ it holds that $\dom{f}\subseteq\dom{g}$ or $\dom{g}\subseteq\dom{f}$.
   For every non-empty $f\in \subs$, we write $\Pred(f)$ for the \emph{immediate predecessor} of $f$, defined as the unique inclusion-maximal $f'\in \subs$ such that $f'\subsetneq f$. Moreover, $f'$ is a \emph{predecessor} of $f$ if $\Pred^i(f)=f'$, for some $i\geq 0$, where $\Pred^{i+1}(f)\defeq \Pred(\Pred^{i}(f))$ and $\Pred^0$ is the identity.


Let $g\in \subs$. A state $(\db',q')$ is called \emph{$g$-consistent}
if it arises from the initial input $(\db,q)$ via the substitutions
specified by $g$. Note that $q'$ may be distinct from the simple substitution $q_{g}$ due to the changes
made to the atoms during the pruning step.
We say that $g$ is \emph{saturated} if the computation reaches a state $(\db',q')$ where $q'$ is $g$-consistent and $\db'$ is $q'$-saturated.



  \begin{claim}\label{claim:isguarded}
  For every $f\in \subs$, $\dom{f} \subseteq \nrqueryvars{q}$.
  \end{claim}
  \begin{claimproof}
  Let $x\in \dom{f}$. Let $g$ be the $\subseteq$-maximal predecessor of $f$ such that $x\notin \dom{g}$.
  Then, the substitution $x\mapsto f(x)$ arises when line 5 or 14 is applied to some $g$-consistent state $(\db',q')$. 
   By  \Cref{claim:inv-pres}, $q'$ is saturated and normal. 

  Consider first the case where $x$ is replaced with a constant according to line 5 of the algorithm.
  Then $x\in \rootqueryvars{q'}$, and
 by \Cref{lem:rootguarded} we have $\rootqueryvars{q'}\subseteq \nrqueryvars{q'}$. By substitution-guardedness, 
 $\nrqueryvars{q'} \subseteq \nrqueryvars{q}$. In particular, $\rootqueryvars{q'}\subseteq\nrqueryvars{q} $, hence  $x \in \nrqueryvars{q}$. 
 
 Consider then the case where $x$ is replaced with a constant according to line 14 of the algorithm.
  This substitution is relative to some $F\in q'$ such that $\keyvars{F}=\emp$ and $x\in   \notkeyvars{F}$.
 Since $q'$ is normal, $\fd{\emp}{x}\in \WFD{q'}$. Consequently, as $q'$ is pruned, we obtain $x\in  \outvar{\igraph{q'}}{\emp} \subseteq  \nrqueryvars{q'}$. By substitution-guardedness, this leads to $x\in   \nrqueryvars{q}$.
  \end{claimproof}

  A function $f\in \subs$ is called \emph{maximal} if it is $\subseteq$-maximal among all functions of $\subs$; it is  \emph{saturated-maximal} if it is $\subseteq$-maximal among all saturated functions of $\subs$. 
  We denote by $\maxsubs$ the set of all maximal functions, and by $\satmaxsubs$ the set of all saturated-maximal functions.
  
    Each function in $\maxsubs$ corresponds to one computation branch. Our aim is to construct a polynomial upper bound for $|\maxsubs|$. 
  If $f\in \maxsubs$ is not saturated, then $\Pred(f)$ is saturated, extending to a function $f'\in \satmaxsubs$ (where possibly $\Pred(f)=f'$).
  Therefore, we can define a function $G \colon \maxsubs\to \satmaxsubs$ such that 
  $\Pred(f) \subseteq G(f)$. We claim that $|\maxsubs|$ is polynomial in $(\db,q)$ if $|\satmaxsubs|$ is polynomial in $(\db,q)$.
  \begin{claim}\label{claim:Mpoly}
  Let $(\db,q)$ be the initial input. 
  There exists a polynomial $p$ such that
\begin{equation}
|\maxsubs|
   \le |\satmaxsubs|\cdot p(\size{\db},\size{q}).
\end{equation}   \end{claim}
   \begin{claimproof}
     It suffices to construct $p$ such that  $|G^{-1}[\{f\}]| \leq p(\size{\db},\size{q})$, for each $f \in \satmaxsubs$.
   Fix $f\in\satmaxsubs$.
The function $f$ has at most $\size{q}$ predecessors in $\subs$.
Moreover, each predecessor $f'$ has at most $\size{\db'}$
extensions $f^*\in\maxsubs$ such that $G(f^*)=f$, for some intermediate database $\db'$ satisfying \eqref{eq:dbupper}.
 The claim follows from this.
   \end{claimproof}

Next, we prove that $|\satmaxsubs|$ is polynomial in $(\db,q)$. 
We begin with the following claim. 

      \begin{claim}\label{claim:contains}
          Let $(\db,q)$ be the initial input. 
      Let $f\in \subs$, and let $(\db',q')$ be an $f$-consistent state.  Let $F\defeq S(\underline{x},y,\vec{c})\in \modec{q}$
      be such that $x\in \initqueryvars{q}$ and $\fd{x}{y}\in \SFD{q}$.
      Then $\modec{(q')}$ contains one of $F, S(\underline{x},f(y),\vec{c}), S(\underline{f(x)},f(y),\vec{c})$.
  \end{claim}
\begin{claimproof}
The query $\modec{(q')}$ contains an atom of the form  $S(\underline{a},b,\vec{c})$, $S(\underline{x},b,\vec{c})$, or $ S(\underline{x},y,\vec{c})$.
This is an immediate consequence of \Cref{claim:inv-pres} and Item 4$^*$ 
 of substitution-guardedness (\Cref{def:subguard}). Since $(\db',q')$ is an $f$-consistent state, it must be the case that $b=f(y)$ and $a=f(x)$.
\end{claimproof}

Let $f\in \satmaxsubs$. By definition, there exists a state $(\db',q')$ where $q'$ is $f$-consistent and $\db'$ is $q'$-saturated.
For a variable $x\in \queryvars{q'}$, choose $G\in q'$ such that $x\in \atomvars{G} $.
Then, since $\db'$ is $q'$-saturated (in particular, it satisfies non-emptiness and pairwise consistency over $q'$), we can choose  a valuation $\theta$ over $\atomvars{G}$ such that $\theta(G)\in \db'$. 
Let us extend $f$ by setting $x \mapsto \theta(x)$. Let $f^*$ denote the extension of $f$ obtained by repeating this procedure for any  $x\in \queryvars{q'}$.
As $\adom{\db'} \subseteq \adom{\db}$ by substitution-guardedness, we note that $f^*$ remains a partial function $\queryvars{q} $ to $\adom{\db}$.
Moreover, it can be seen that $\initqueryvars{q}\subseteq \dom{f^*}$; informally, at substitution any guard over $q$ is either itself substituted or remains a guard over $q'$ (by Item 6$^*$ of substitution-guardedness), and at pruning any guard remains a guard (by \Cref{it:inv:guarding} of guardedness).
 
\begin{claim}
          Let $(\db,q)$ be the initial input. 
Let $f,g\in \satmaxsubs$. Then $\restrict{f^*}{\initqueryvars{q}} =\restrict{g^*}{\initqueryvars{q}}$ implies $f=g$.
\end{claim}  
\begin{claimproof}
Suppose that $f\neq g$. 
By symmetry, we may assume that $\dom{f}\subseteq \dom{g}$. Since both functions are saturated-maximal, it holds that
$f\not\subseteq g$. In particular, there exists $y\in \dom{f}$ such that $f(y)\neq g(y)$. If $y\in \initqueryvars{q}$, we are done. 
Thus, suppose $y\notin \initqueryvars{q}$.

By \Cref{claim:isguarded}, it holds that $y\in \nrqueryvars{q}$. Then,  
there exists an edge $\fd{x}{y}\in \SFD{q}$ with $x\in \initqueryvars{q}$.
Since $x\neq y$ and $q$ is normal, this is witnessed by an 
 atom $F=S(\underline{x},y,\vec{c})\in \modec{q}$, where $\vec{c}$ is a sequence of constants.

Consider the state $(\db^f,q^f)$ such that $q^f$ is $f$-consistent and $\db^f$ is $q^f$-saturated.
By \Cref{claim:contains}, $q^f$ contains one of $F$, $S(\underline{x},f(y),\vec{c})$, $S(\underline{f(x)},f(y),\vec{c})$.
Since $y\in \dom{f}$, it follows that $y\notin \queryvars{q^f}$, and consequently it is the case that 
$q^f$ contains either $S(\underline{x},f(y),\vec{c})$ or $S(\underline{f(x)},f(y),\vec{c})$.
Consider also the state $(\db^g,q^g)$ such that $q^g$ is $g$-consistent and $\db^g$ is $q^g$-saturated. 
Similarly, we obtain that $q^g$ contains either $S(\underline{x},g(y),\vec{c})$ or $S(\underline{g(x)},g(y),\vec{c})$.

Assume first that $q^f$ contains $S(\underline{f(x)},f(y),\vec{c})$. Since $\dom{f}\subseteq \dom{g}$, we obtain that $q^g$
contains $S(\underline{g(x)},g(y),\vec{c})$.
Since $\db^f$ is $q^f$-saturated and $\db^g$ is $q^g$-saturated, we have that $S(\underline{f(x)},f(y),\vec{c}) \in \db^f $
and $ S(\underline{g(x)},g(y),\vec{c})\in \db^g$.
By \Cref{it:sizeinv} of substitution-guardedness, we have $\restrict{\db^f}{S}\subseteq \restrict{\db}{S}$  and $\restrict{\db^g}{S}\subseteq \restrict{\db}{S}$.
As $\db$ is $q$-saturated and $\fd{x}{y}\in \SFD{q}$, it follows from $f(y)\neq g(y)$ that $f(x)\neq g(x)$.
Specifically, this means that $f^*(x)\neq g^*(x)$.

Assume then that $q^f$ contains $S(\underline{x},f(y),\vec{c})$. 
By construction, there exists an atom $G\in q^f$
and a valuation $\theta$ over $\atomvars{G}$ such that $\theta(G)\in \db^f$ and $\theta(x)=f^*(x)$.
Since $\db^f$ is $q^f$-saturated, we can extend $\theta$ to a valuation $\theta^*$
such that $S(\underline{\theta^*(x)},f(y),\vec{c})\in \db^f$. In particular, this means that $S(\underline{f^*(x)},f(y),\vec{c}) \in \db^f$.
Furthermore, $q_g$ contains $S(\underline{x},g(y),\vec{c})$ or $S(\underline{g(x)},g(y),\vec{c})$, 
and analogously we obtain, in both cases, that $S(\underline{g^*(x)},g(y),\vec{c}) \in \db^g$. Then,
as in the previous paragraph, we obtain from $f(y)\neq g(y)$ that $f^*(x)\neq g^*(x)$.

Hence we have established that $\restrict{f^*}{\initqueryvars{q}} \neq \restrict{g^*}{\initqueryvars{q}}$.
\end{claimproof}

By the previous claim and \Cref{lem:isone}, we now obtain the following sequence of inequalities:
  \[
  |\satmaxsubs|\leq |\adom{\db}|^{|\initqueryvars{q}|} \leq   |\adom{\db}|^{\cgs{q}}\leq   |\adom{\db}|^{c}.
  \]
  This, together with \Cref{claim:Mpoly}, establishes the existence of a polynomial $p'$ such that 
  \[|\maxsubs|
   \le  p'(\size{\db},\size{q}),\]
concluding thereby step (B2).

The combination of (B1) and (B2) proves that $\textsc{Solve}(\db,q)$ runs in polynomial time in the input, concluding step (B).
\end{proof}

\begin{algorithm}[t]
\caption{$\textsc{Solve}(\db,q)$ \label{alg:solve}}
\Input{A saturated database $\db$ and a saturated and normal query $q\in\sjfuabcq$.}
\Output{\textnormal{true} iff $(\db,q)\in \cert$.}


\While{$\queryvars{q}\neq \emptyset$}{
\If{$\outvar{\igraph{q}}{\emp} = \emptyset$}{
    choose the $\prec$-least non-attacked variable $x\in\rootqueryvars{q}$ with $\outvar{\SFD{q}}{x}= \{x\}$\;
    \tcp{\Cref{lem:rootnotattacked}; $\{x\}$ is substitutable, case \ref{it:ii}}
    \ForEach{$f\colon \{x\} \to \adom{\db}$}{
        \(s\gets\textsc{Substitute}_f(\db,q)\)\;
    \tcp{\Cref{lem:pruneconst}} 
    \If{
    $s\neq \false$ and $\textsc{Solve}(s)=\true$}{
        \Return{$\true$}\;
}
    }
    \Return{$\false$}\;
        \tcp{\Cref{lem:non-attacked}}
}

\Else{
$(\db,q) \gets \Try(\textsc{Prune}(\db,q))$\;
\tcp{\Cref{lem:prune}
}
	\If{$\outvar{\igraph{q}}{\emp}\neq \emptyset$}{
		  Select $\prec$-least $R\in q$ with $\keyvars{R}=\emp$ and $\notkeyvars{R}\neq \emptyset$\;
	  \tcp{Pruning implies 
	  that $\notkeyvars{R}$ is substitutable, case \ref{it:iii}
}
    \ForEach{$f\colon \notkeyvars{R} \to  \adom{\db}$ such that
$R_{f}\in\db$}{
    $s \gets \textsc{Substitute}_{f}((\db\setminus R^\db)\cup\{R_f\},q)$\;
\tcp{\Cref{lem:pruneconst}}
    \If{
    $s=\false$ or $\textsc{Solve}(s)=\false$}{
        \Return{$\false$}\;
    }
}
\Return{$\true$}\;
    }
}
}
\Return $\true$\; 
\tcp{Saturation of $\db$ guarantees that $q \subseteq \db$ and $\db$ is consistent.}
\end{algorithm}

\section{Hardness Proofs with Unary Primary Keys}\label{sect:apphard}
We here provide the proofs for the lower-bound results stated in \Cref{sect:mainstatement}.

\smallskip
 \hard*
\begin{proof}
The upper bound is simple and well-known. We prove the lower bound. 
We use the following well-known $\Pi^{\sf P}_2$-complete problem.
An instance of $\aesat$ consists of a quantified Boolean formula
\[
\psi = \forall x_1\dots \forall x_m\exists y_{1}\dots\exists y_{n}\varphi,
\]
where $\varphi = C_1 \land \dots \land C_\ell$ is a 3-CNF formula over the variables
$x_1,\ldots,x_{m},y_1,\dots ,y_n$.
The question is whether the formula is true.
It is well known that $\aesat$ is $\Pi^{\sf P}_2$-complete.
For a clause $C_i$, let $C_{i,j}$ denote its $j$th literal, and for a literal $l \in \{x, \lnot x\}$, let $v(l)=x$ denote the underlying variable.
Given $(b_1,b_2,b_3)\in\{0,1\}^3$, we write
$(b_1,b_2,b_3)\models C_i$
if assigning $b_j$ to the variable $v(C_{i,j})$ satisfies $C_i$.

We define
\begin{equation}\label{eq:reduction}
q\defeq \{R_i(\underline{c},x_i) \mid i\in [m]\} \cup \{S(\underline{z_i},v(C_{i,1}),v(C_{i,2}),v(C_{i,3})) \mid S\in \{T_i,U_{i}\}, i\in [\ell]\}
\end{equation}
and
\begin{align*}
\db \defeq &\{R_i(\underline{c},b) \mid i\in [m], b \in \{0,1\}\} \\
&\cup \{S(\underline{(b_1,b_2,b_3)},b_1,b_2,b_3) \mid S\in \{T_i,U_{i}\}, i\in [\ell], (b_1,b_2,b_3)\models C_i\},
\end{align*}
where $c$ is a constant.

Clearly, $(\db,q)$ can be constructed from $\psi$ in polynomial time. It is straightforward to verify that $q$ belongs to $\sjfuabcq$. In particular, 
the atoms $T_i,U_i$ do not attack any atom, and the atoms $R_i$ do not attack any atoms $R_j$ where $j\neq i$.
Furthermore, we observe that the relations $T^\db_{i},U_i^\db$ are consistent, while the relations $R^\db_i$ are inconsistent.

We claim that $\psi$ is true if and only if $(\db,q)\in \cqa{\sjfuabcq}$.

($\Rightarrow$) Let $\rep$ be a repair of $\db$. Let $\theta$ be the valuation over $x_1, \dots ,x_m$ such that
$R_i(\underline{c},\theta(x_i))\in \rep$, for $i\in [m]$. Since $\psi$ is true, we can extend $\theta$ to a valuation $\theta'$ over
$x_1,\ldots,x_{m}$, $y_1, \dots ,y_n$ such that $\theta'\models \varphi$. In particular, 
$\theta'(v(C_{i,1}),v(C_{i,2}),v(C_{i,3}))\models C_i$ for each clause $C_i$. Hence, 
\[S(\underline{\theta'(v(C_{i,1}),v(C_{i,2}),v(C_{i,3}))},\theta'(v(C_{i,1})),\theta'(v(C_{i,2})),\theta'(v(C_{i,3})))\in \db,
\]
 for $S\in \{T_{i},U_{i}\}$. Extend $\theta'$ by setting $\theta'(z_i)=(\theta'(v(C_{i,1})),\theta'(v(C_{i,2})),\theta'(v(C_{i,3})))$ for each $i\in[\ell]$. By consistency of $T^\db_{i},U_i^\db$, it follows that $\theta'(q)\subseteq \rep$.

($\Leftarrow$) Let $\theta$ be an arbitrary valuation over $x_1, \dots ,x_m$. Choose a repair $\rep$ such that $R_i(\underline{c},\theta(x_i))\in \rep$, for $i\in [m]$. Assuming $(\db,q)\in \cqa{\sjfuabcq}$, there exists a valuation $\xi$ over $\queryvars{q}$ such that $\xi(q)\subseteq \rep$. Clearly, $\xi$ must agree with $\theta$ on the variables $x_1, \dots ,x_m$, so it is an extension of $\theta$ to the variables $y_{1},\dots ,y_{n}$.  Now, for $S\in \{T_{i},U_{i}\}$, it holds that 
\[
S(\underline{\xi(z_i)},\xi(v(C_{i,1})),\xi(v(C_{i,2})),\xi(v(C_{i,3})))\in \rep,
\]
 which entails $\xi(v(C_{i,1}),v(C_{i,2}),v(C_{i,3}))\models C_i$. Thus $\xi\models \varphi$, hence $\psi$ is true.
\end{proof}
Note that the above reduction does not leverage alternation implicit in repair semantics, beyond the second level. In the following parameterized perspective, we make better use of this phenomenon.

\paragraph*{Parameterized perspective}
We briefly recall the parameterized complexity notions used below. 
A parameterized problem belongs to (uniform) $\mathsf{XP}$ if there exists an algorithm that decides instances of size $n$ and parameter $k$ in time $O(n^{f(k)})$ for some computable function $f$.
The classes $\mathrm{W}[t]$ form the hierarchy
\[
\mathsf{FPT}\subseteq\mathrm{W}[1]\subseteq\mathrm{W}[2]\subseteq\cdots\subseteq\mathsf{XP},
\]
whose inclusions are conjectured to be strict; $\mathrm{co}\text{-}\mathrm{W}[t]$ consists of the complements of problems in $\mathrm{W}[t]$. For the hardness proof, we refer to the model-checking characterization of Flum and Grohe \cite{FlumGrohe2005}. 
Let $\Sigma_{t,1}$ denote sentences of the form
\[
\phi=
\exists x_1\cdots \exists x_m
\forall y_1\exists y_2\cdots Qy_{t-1}
\theta,
\]
where $\theta$ is quantifier-free and the $t-1$ quantifier blocks following the leading existential block are alternating and contain one variable each. Together, \cite[Theorem~4.13 and Corollary~4.15]{FlumGrohe2005} imply that, fixing $t\geq 2$,  model checking $\Sigma_{t,1}$-sentences
$\phi$ over graphs $G$, parameterized by the size $k$ of the whole sentence, is $\mathrm{W}[t]$-complete.
Without loss of generality, we may assume that every quantified variable occurs in $\theta$.

If $\phi$ is a quantifier-free first-order formula, $\vec{x}$ is a tuple of variables occurring in $\phi$, and $\vec{a}$ is an equally long tuple of constants, then $\phi[\vec{x}\mapsto\vec{a}]$ denotes the result of simultaneously substituting $\vec{a}$ for $\vec{x}$ in $\phi$.
\paramresult*
\begin{proof}
The algorithm underlying \Cref{thm:main} is uniform and yields membership in $\mathsf{XP}$.
Since $\mathrm{W}[t]\subseteq \mathrm{W}[t+1]$, it suffices to establish the hardness statement for odd $t \geq 3$.

Write $t=2p+1$ and rename variables so that
\[
\phi=
\exists x_1\cdots \exists x_m
\forall u_1\exists v_1\cdots \forall u_p\exists v_p\,
\theta.
\]

First, we apply a Tseitin transformation to the quantifier-free formula $\neg\theta$.
This yields an expansion $\mathcal A$ of $G$ such that
\[
G\not\models \phi
\quad\Longleftrightarrow\quad
\mathcal A\models \phi^-,
\]
where
\[
\phi^-= \forall x_1\cdots \forall x_m
\exists u_1\forall v_1\cdots \exists u_p\forall v_p \exists w_1 \dots \exists w_r\bigwedge_{h=1}^{M}P_h(\vec v_h),
\]
$M,r=O(k)$, and every $P_h$ has bounded arity. In particular, the transformation is FPT
with parameter $k$. Moreover, every quantified variable appears in some atom $P_h(\vec v_h)$.

Then, we construct from $(\mathcal A,\phi^-)$ an instance $(\db,q)$ of 
$\cqa{\sjfuabcq}$. Let
\begin{align*}
q \coloneqq {}&
\{R_i(\underline c,x_i)\mid i\in[m]\}
\cup
\{V_i(\underline{u'_i},v'_i)\mid i\in[p]\}
\\
&{}\cup
\{S(\underline{v'_i},u'_i)
\mid S\in\{A_i,B_i\},\ i\in[p]\}
\\
&{}\cup
\{S(\underline{u'_{i+1}},v'_i)
\mid S\in\{C_i,D_i\},\ i\in[p-1]\}
\\
&{}\cup
\{S(\underline{u'_i},u_i)
\mid S\in\{E_i,F_i\},\ i\in[p]\}
\\
&{}\cup
\{S(\underline{v'_i},v_i)
\mid S\in\{G_i,H_i\},\ i\in[p]\}
\\
&{}\cup
\{S(\underline{z_h},\vec v_h)
\mid S\in\{T_h,U_h\},\ h\in[M]\},
\end{align*}
where $c\in A$ is arbitrary, $u'_1,v'_1, \dots ,u'_p,v'_p,z_1,\dots,z_M$ are fresh variables, and
all  relation symbols are fresh. The quantifier gadgets in the resulting key-nonkey graph are illustrated in \Cref{fig:parameterized-gadget}. Assuming that $A$ is the universe of $\mathcal A$, the
database $\db$ contains
\begin{itemize}
\item
$R_i(\underline c,a)$ for each $i\in[m]$ and $a\in A$;
\item
$V_i(\underline{(a_1,b_1,\dots,a_{i-1},b_{i-1},a_i)},
(a_1,b_1,\dots,a_{i-1},b_{i-1},a_i,b_i))$
for each $i\in[p]$ and $a_1,b_1,\dots,a_i,b_i\in A$;
\item
$S(\underline{(a_1,b_1,\dots,a_i,b_i)},
(a_1,b_1,\dots,a_i))$
for each $i\in[p]$, $a_1,b_1,\dots,a_i,b_i\in A$, and $S\in\{A_i,B_i\}$;
\item
$S(\underline{(a_1,b_1,\dots,a_i,b_i,a_{i+1})},
(a_1,b_1,\dots,a_i,b_i))$
for each $i\in[p-1]$, $a_1,b_1,\dots,a_i,b_i,a_{i+1}\in A$, and $S\in\{C_i,D_i\}$;
\item
$S(\underline{(a_1,b_1,\dots,a_i)},a_i)$
for each $i\in[p]$, $a_1,b_1,\dots,a_i\in A$, and $S\in\{E_i,F_i\}$;
\item
$S(\underline{(a_1,b_1,\dots,a_i,b_i)},b_i)$
for each $i\in[p]$, $a_1,b_1,\dots,a_i,b_i\in A$, and $S\in\{G_i,H_i\}$.
\end{itemize}
Moreover, for every $\vec a\in P_h^{\mathcal A}$, choose a fresh identifier
$e_{h,\vec a}$, unique to $(h,\vec a)$, and add the twin facts
$T_h(\underline{e_{h,\vec a}},\vec a)$ and
$U_h(\underline{e_{h,\vec a}},\vec a)$ to $\db$.

Note that $\adom{\db}$ extends $A$ with the fresh identifiers.
Furthermore, it can be verified that every attack respects the order
\begin{equation}\label{eq:attack-order}
\{R_1,\dots,R_m\}
<V_1<\cdots<V_p
<\{A_i,B_i,C_i,D_i,E_i,F_i,G_i,H_i,T_h,U_h\},
\end{equation}
where $i,h$ range over appropriate integers.
 Thus, the attack graph is acyclic and
$q\in\sjfuabcq$.

We claim that 
\[
\mathcal A\models\phi^-
\iff 
\text{every repair of $\db$ satisfies $q$}.
\]
For the claim, observe first that
\begin{align*}
 & \text{ every repair of $\db$ satisfies $q$}\\
\iff &\forall c_1, \dots, c_m\in A: \text{ every repair of $\db$ satisfies $(q \setminus \{R_1, \dots ,R_m\})_{x_1, \dots ,x_m\mapsto c_1, \dots ,c_m}$}.
\end{align*}
Since the atoms $R_1, \dots ,R_m$ do not attack any atom in $q_{x_1, \dots ,x_m\mapsto c_1, \dots ,c_m}$, by \Cref{lem:replaceconstant2} and \eqref{eq:attack-order} the atom $V_1$ is unattacked in this query; hence, it is unattacked also in $q^0\defeq (q \setminus \{R_1, \dots ,R_m\})_{x_1, \dots ,x_m\mapsto c_1, \dots ,c_m}$.
Consequently, by construction and  \Cref{lem:non-attacked} we obtain that
\begin{align*}
& \text{every repair of $\db$ satisfies $ q^0 $}\\
\iff &\text{ $\exists a_1 \in A: $ every repair of $\db$ satisfies $q^0_{u'_1 \mapsto (a_1)}$}\\
\iff &\text{ $\exists a_1 \in A\forall b_1 \in A:$ every repair of $\db$ satisfies $(q^0 \setminus \{V_1\})_{u'_1,v'_1 \mapsto  (a_1),(a_1,b_1)}$.}
\end{align*}

Now, the atom $V_2$ is unattacked in $(q^0 \setminus \{V_1\})_{u'_1,v'_1 \mapsto  (a_1),(a_1,b_1)}$, so the same alternation repeats for $u'_2,v'_2$, and, by induction, for the remaining variable pairs $u'_i,v'_i$.
Suppose 
\[
q^1=(q^0 \setminus \{V_1, \dots ,V_p\})_{u'_1,v'_1, \dots ,u'_p,v'_p \mapsto  (a_1),(a_1,b_1), \dots ,(a_1,b_1, \dots ,a_p),(a_1,b_1, \dots ,a_p,b_p)}
\]
is some query obtained as a result of this alternation. 

By construction, every repair of $\db$ satisfies $q^1$ if and only if every repair of $\db$ satisfies 
\[
q^2 \defeq (q^1 \setminus \{A_i,B_i,C_i,D_i,E_i,F_i,G_i,H_i\})_{u_1,v_1, \dots ,u_p,v_p \mapsto a_1,b_1, \dots ,a_p,b_p},
\]
where $i$ ranges over appropriate integers. 
Note that $\queryvars{q^2}$ consists of only the variables listed in $w_1, \dots ,w_r,z_1, \dots ,z_M$ and $q^2$ only of the atoms $T_h, U_h$. 
Moreover, $\db$ restricted to $T_h, U_h$ is consistent.
Thus, every repair of $\db$ satisfies $q^2$ if and only if there exists $d_1, \dots ,d_r \in A$ and identifiers $e_1, \dots ,e_M$ such that $\db$ satisfies 
\[
q^* \defeq q^2_{w_1, \dots ,w_r, z_1, \dots ,z_M \mapsto d_1, \dots ,d_r, e_1, \dots ,e_M}
\] 
Finally, $\db$ satisfies $q^*$ precisely when $\mathcal A$ satisfies $\bigwedge_{h=1}^M P_h(\vec v_h)$ over the generated variable assignment. From this, the claim follows.

Finally, it suffices to show that the reduction from $(\mathcal A,\phi^-)$ to $(\db,q)$ is FPT.
In the key-nonkey graph, the source SCCs are 
\[
\{\emp\},\{u'_p,v'_p\},\{z_1\},\dots,\{z_M\}.
\]
Consequently, $\cgs{q}=M+1=O(k)$.
If $N=|A|$, the construction has
size at most $f(k)N^{c_t}$ for a constant $c_t$ depending only on the fixed
$t$ (one may take $c_t$ to be the maximum of $2p$ and the arity bound in the
normalization). Thus, the final reduction is FPT for every fixed $t$. This
proves the hardness claim.
\end{proof}

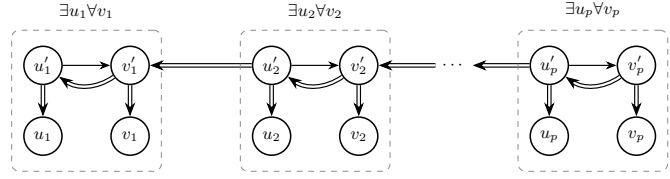
\begin{figure}[t]
\centering
\scalebox{0.72}{%
\begin{tikzpicture}[
  every node/.style={font=\small},
  quantifier node/.style={
    circle,
    draw=black,
    thick,
    inner sep=1.5pt,
    minimum size=7mm
  },
  quantifier box/.style={
    draw=black!45,
    dashed,
    rounded corners,
    inner xsep=6pt,
    inner ysep=8pt
  }
]

\node[quantifier node] (up1) at (0,1.3) {$u'_1$};
\node[quantifier node] (vp1) at (1.6,1.3) {$v'_1$};
\node[quantifier node] (u1)  at (0,0) {$u_1$};
\node[quantifier node] (v1)  at (1.6,0) {$v_1$};
\sarr{up1}{vp1}
\darr[bend left=35]{vp1}{up1}
\darr{up1}{u1}
\darr{vp1}{v1}
\node[quantifier box,fit=(up1)(vp1)(u1)(v1)] (block1) {};
\node[anchor=south] at ([yshift=1mm]block1.north)
  {$\exists u_1\forall v_1$};

\node[quantifier node] (up2) at (4.2,1.3) {$u'_2$};
\node[quantifier node] (vp2) at (5.8,1.3) {$v'_2$};
\node[quantifier node] (u2)  at (4.2,0) {$u_2$};
\node[quantifier node] (v2)  at (5.8,0) {$v_2$};
\sarr{up2}{vp2}
\darr[bend left=35]{vp2}{up2}
\darr{up2}{u2}
\darr{vp2}{v2}
\darr{up2}{vp1}
\node[quantifier box,fit=(up2)(vp2)(u2)(v2)] (block2) {};
\node[anchor=south] at ([yshift=1mm]block2.north)
  {$\exists u_2\forall v_2$};

\node (continuation) at (7.55,1.3) {$\cdots$};
\darr{continuation}{vp2}

\node[quantifier node] (upp) at (9.3,1.3) {$u'_p$};
\node[quantifier node] (vpp) at (10.9,1.3) {$v'_p$};
\node[quantifier node] (up)  at (9.3,0) {$u_p$};
\node[quantifier node] (vp)  at (10.9,0) {$v_p$};
\sarr{upp}{vpp}
\darr[bend left=35]{vpp}{upp}
\darr{upp}{up}
\darr{vpp}{vp}
\darr{upp}{continuation}
\node[quantifier box,fit=(upp)(vpp)(up)(vp)] (blockp) {};
\node[anchor=south] at ([yshift=1mm]blockp.north)
  {$\exists u_p\forall v_p$};

\end{tikzpicture}%
}
\caption{The quantifier-gadget chain; all other vertices and edges of the reduction are omitted. Each primed pair $(u'_i,v'_i)$ realizes the block $\exists u_i\forall v_i$; double arrows are consistent edges, single arrows are inconsistent edges, and the consistent edges $u'_{i+1}\to v'_i$ connect consecutive blocks.}
\label{fig:parameterized-gadget}
\end{figure}

The parameterized complexity of $\cqa{\sjfuabcq}$ thus lies between membership in $\mathsf{XP}$ and $\mathrm{co}\text{-}\mathrm{W}[t]$-hardness for every fixed $t$. 
Whether the gap can be tightened is left as an open question;  such a gap is not unprecedented: an analogous $\mathsf{XP}$ upper bound together with $\mathrm{W}[t]$-hardness for every $t$ has been obtained for chain scheduling with exact delays~\cite{BodlaenderW20}.

\section{Hardness Proof with Binary Primary Keys}\label{sect:apphard2}

We conclude by showing that the unary primary-key assumption in
\Cref{thm:main} cannot be omitted. In fact, already binary primary keys make
the problem \coNP-hard for queries of closure generator size zero. The
attack graphs in the reduction are not only acyclic, but contain no directed
path of length two.

\binaryhardness*
\begin{proof}
Note that the one-bag tree decomposition yields $\ctw{q}\leq\cgs{q}$
and $\hctw{q}\leq\hcgs{q}$. Moreover, it holds that $\hctw{q}\leq\ctw{q}$. Indeed, let $B$ be a bag and
let $V$ witness that its closure width is at most $c$. For every $v\in V$,
choose an atom $F_v\in q$ containing $v$. The subquery consisting of the atoms
$F_v$, for $v\in V$, contains $V$ and therefore witnesses that the hyperclosure width of $B$ is at
most $c$. Hence it suffices to prove the upper bound
for the hyperclosure treewidth alone.

For the upper bound, then, consider the complement of the problem. We first guess a repair $\rep$ of the input database $\db$ and verify that $\rep$ contains exactly one fact from every block of $\db$. It remains to verify that $\rep\not\models q$.  We now guess a tree decomposition $\mathcal{T}=(T,(B_t)_{t\in T})$ of $q$ and, for every $t\in T$, a subquery $q_t\subseteq q$ of size at most
$c$ such that $B_t\subseteq\cl{\queryvars{q_t}}{q}$. We may assume that $T$ has at most $\lvert\queryvars{q}\rvert$ nodes
\cite[Remark~3.2]{Adler08}. Hence, the decomposition and its hyperclosure-width witnesses have polynomial size, and their validity can
be checked in polynomial time. Regarding the primary-key dependencies as partial functions, the repair $\rep$ is a structure in which these functions are well defined. Therefore, given the guessed decomposition, Adler's dynamic-programming algorithm \cite[Theorem~4.2]{Adler08} decides in polynomial time whether $\rep\models q$. The algorithm extends immediately to queries with constants by requiring that homomorphisms fix constants. Consequently, we can verify $\rep\not\models q$ in polynomial time. Thus non-certainty is in $\NP$, and $\cqa{\mathcal{Q}}$ is in $\coNP$.

For the lower bound, we reduce from the complement of \textsc{3Sat}. Let $\varphi$ be a
3-CNF formula over variables $x_1,\ldots,x_m$. First, one constructs, in an obvious way, a Boolean circuit $C$ computing $\varphi$. 
We may assume that
in this construction the two inputs of every binary gate are distinct. 
 Let $o$ be the
output gate of $C$. We interpret every gate $g$ as a variable.

 Fix
three database constants $0,1,c$. The query $q_C$ contains the atoms
\begin{itemize}
\item $O(\underline o)$,
\item $R_i(\underline c,x_i)$ for each $i\in[m]$,
\item  $S_g(\underline{a,b},g),S'_g(\underline{a,b},g)$ for each binary gate $g$ with input gates $a,b$, and
\item $S_g(\underline{a},g),S'_g(\underline{a},g)$ for each NOT gate $g$ with input gate $a$.
\end{itemize}
All relation names in the construction are distinct, so $q_C$ is
self-join-free. Since only the atoms $R_i$ have outgoing attacks, but they do not attack each other, the attack graph has no cycles.
Moreover, from $\cl{\emptyset}{q_C}=\queryvars{q_C}$
 we obtain that
$ \cgs{q_C}=0$.

The database $\db_C$ contains
\begin{itemize}
\item $O(\underline 0)$,
\item $R_i(\underline c,0),R_i(\underline c,1)$  for each $i\in[m]$,
\item  $S_g(\underline{a,b},f_g(a,b)),
 S'_g(\underline{a,b},f_g(a,b))$ for $a,b\in\{0,1\}$, each binary gate $g$, and the Boolean function $f_g$ computing $g$,
\item $S_g(\underline{0},1),S'_g(\underline{0},1),S_g(\underline{1},0),S'_g(\underline{1},0)$ for each NOT gate $g$.
\end{itemize}
All gate relations are consistent, and the relations $S_g,S'_g,O$ belong to every repair. The relations $R_i$ are inconsistent.

It is not difficult to check that
 every repair of $\db_C$ satisfies $q_C$
 if and only if
 C is not satisfiable. The reduction can be computed in logarithmic space. This concludes the proof.
\end{proof}

\end{document}